\documentclass[10pt,english]{article}
\usepackage[T1]{fontenc}
\usepackage[latin9]{inputenc}
\usepackage[a4paper,margin=1in]{geometry}
\usepackage{color}
\usepackage{babel}
\usepackage{verbatim}
\usepackage{float}
\usepackage{amsmath}
\usepackage{amsthm}
\usepackage{amssymb}
\usepackage{graphicx}
\usepackage{esint}
\usepackage[numbers]{natbib}
\usepackage[unicode=true,pdfusetitle,
 bookmarks=true,bookmarksnumbered=false,bookmarksopen=false,
 breaklinks=false,pdfborder={0 0 0},pdfborderstyle={},backref=false,colorlinks=true]
 {hyperref}
\hypersetup{
 pdfborderstyle={},citecolor=blue,urlcolor=blue}

\makeatletter

\providecommand{\tabularnewline}{\\}
\floatstyle{ruled}
\newfloat{algorithm}{tbp}{loa}
\providecommand{\algorithmname}{Algorithm}
\floatname{algorithm}{\protect\algorithmname}

\newcommand{\lyxaddress}[1]{
\par {\raggedright #1
\vspace{1.4em}
\noindent\par}
}
\theoremstyle{plain}
\newtheorem{thm}{\protect\theoremname}
  \theoremstyle{remark}
  \newtheorem*{rem*}{\protect\remarkname}
  \theoremstyle{plain}
  \newtheorem{assumption}{\protect\assumptionname}
  \theoremstyle{plain}
  \newtheorem{prop}[thm]{\protect\propositionname}
  \theoremstyle{plain}
  \newtheorem{lem}[thm]{\protect\lemmaname}
  \theoremstyle{plain}
  \newtheorem{cor}[thm]{\protect\corollaryname}
  \theoremstyle{remark}
  \newtheorem{rem}[thm]{\protect\remarkname}

\usepackage{babel}

  \providecommand{\assumptionname}{Assumption}
  \providecommand{\lemmaname}{Lemma}
  \providecommand{\propositionname}{Proposition}
  \providecommand{\remarkname}{Remark}
\providecommand{\corollaryname}{Corollary}
\providecommand{\theoremname}{Theorem}

\makeatother

  \providecommand{\assumptionname}{Assumption}
  \providecommand{\corollaryname}{Corollary}
  \providecommand{\lemmaname}{Lemma}
  \providecommand{\propositionname}{Proposition}
  \providecommand{\remarkname}{Remark}
\providecommand{\theoremname}{Theorem}

\begin{document}

\title{The Correlated Pseudo-Marginal Method}

\author{George Deligiannidis$^{\dagger}$, Arnaud Doucet$^{\ddagger}$ and
Michael K. Pitt$^{\dagger}$\thanks{corresponding author}}
\maketitle

\lyxaddress{$^{\dagger}$Department of Mathematics, King's College London, UK.}

\lyxaddress{$^{\ddagger}$Department of Statistics, University of Oxford, UK.}
\begin{abstract}
The pseudo-marginal algorithm is a Metropolis\textendash Hastings-type
scheme which samples asymptotically from a target probability density
when we are only able to estimate unbiasedly an unnormalised version
of it. In a Bayesian context, it is a state-of-the-art posterior simulation
technique when the likelihood function is intractable but can be estimated
unbiasedly using Monte Carlo samples. However, for the performance
of this scheme not to degrade as the number $T$ of data points increases,
it is typically necessary for the number $N$ of Monte Carlo samples
to be proportional to$\ T$ to control the relative variance of the
likelihood ratio estimator appearing in the acceptance probability
of this algorithm.\ The correlated pseudo-marginal algorithm is a
modification of the pseudo-marginal method using a likelihood ratio
estimator computed using two correlated likelihood estimators. For
random effects models, we show under regularity conditions that the
parameters of this scheme can be selected such that the relative variance
of this likelihood ratio estimator is controlled when $N$ increases
sublinearly with $T$ and we provide guidelines on how to optimise
the parameters of the algorithm based on a non-standard weak convergence
analysis. The efficiency of computations for Bayesian inference relative
to the pseudo-marginal method empirically increases with $T$ and
is higher than two orders of magnitude in some of our examples. 
\end{abstract}
{\small{}{}Keywords: }Asymptotic posterior normality; Correlated
random numbers; Intractable likelihood; \linebreak{}
 Metropolis\textendash Hastings algorithm; Particle filter; Random
effects model; State-space model; Weak convergence.

\section{Introduction\label{section:introduction}}

Consider a Bayesian model where the likelihood of the observations
$y$ is denoted by $p(y\mid\theta)$ and the prior for the parameter
$\theta\in\Theta\subseteq\mathbb{R}^{d}$ admits a density $p(\theta)$
with respect to Lebesgue measure $\mathrm{d}\theta$. Then the posterior
density of interest is $\pi(\theta)\propto p(y\mid\theta)p(\theta)$.
We slightly abuse notation by using the same symbols for distributions
and densities.

A standard approach to compute expectations with respect to $\pi\left(\theta\right)$
is to use the Metropolis\textendash Hastings (MH)\ algorithm to generate
an ergodic Markov chain of invariant density $\pi\left(\theta\right)$.
Given the current state $\theta$ of the Markov chain, one samples
at each iteration a candidate $\theta^{\prime}$ which is accepted
with a probability which depends on the likelihood ratio $p(y\mid\theta^{\prime})/p(y\mid\theta)$.
For many latent variable models, the likelihood is intractable and
it is thus impossible to implement the MH\ algorithm. In this context,
Markov chain Monte Carlo (MCMC) schemes targeting the joint posterior
density of the parameter and latent variables are often inefficient
as the parameter and latent variables can be strongly correlated under
the posterior, or cannot even be used if only forward simulation of
the latent variables is feasible; see \citep{ionides2006}, \citep{Johndrow2016},
and \citep[Section 2.3]{andrieu:doucet:holenstein2010} for a\ detailed
discussion.

Contrary to these approaches, the pseudo-marginal (PM) algorithm directly
mimicks the MH scheme targeting the marginal $\pi\left(\theta\right)$
by substituting an estimator of the likelihood ratio $p(y\mid\theta^{\prime})/p(y\mid\theta)$
for the true likelihood ratio in the MH\ acceptance probability \citep{linliuSloan2000},
\citep{beaumont2003estimation}, \citep{andrieu2009pseudo}. This
estimator is obtained by computing a non-negative unbiased estimator
of $p(y\mid\theta^{\prime})$ and dividing it by the estimator of
$p(y\mid\theta)$ computed when $\theta$ was accepted. This simple
yet powerful idea has become very popular as it is often possible
to obtain a non-negative unbiased estimator of intractable likelihoods
and it provides state-of-the-art performance in many scenarios; see,
e.g., \citep{andrieu:doucet:holenstein2010}, \citep{fluryshephard2011}.
Qualitative convergence results for this procedure have been obtained
in \citep{andrieu2009pseudo} and \citep{andrieuvihola2015}.

Assuming that the likelihood estimator is evaluated using importance
sampling or particle filters for state-space models with $N$ particles,
it has also been shown under various assumptions in \citep{PittSilvaGiordaniKohn(12)},
\citep{doucet2015efficient} and \citep{Sherlock2015efficiency} that
$N$ should be selected such that the variance of the loglikelihood
ratio estimator should take a value between 1.0 and 2.0 in regions
of high probability mass to minimise the computational resources necessary
to achieve a specific asymptotic variance for a particular PM average.
As the number $T$ of data $y=\left(y_{1},...,y_{T}\right)$ increases,
this implies that $N$ should increase linearly with $T$ \citep[Theorem 1]{berarddelmoraldoucet2014}
and the computational cost of PM is thus of order $T^{2}$ at each
iteration. This can be prohibitive for large datasets.

The reason for this is that the PM algorithm is based on an estimator
of $p(y\mid\theta^{\prime})/p(y\mid\theta)$ obtained by dividing
estimators of $p(y\mid\theta^{\prime})$ and $p(y\mid\theta)$ which
are independent given $\theta$ and $\theta^{\prime}$. In contrast,
the correlated pseudo-marginal (CPM) method correlates the estimators
of $p(y\mid\theta^{\prime})$ and $p(y\mid\theta)$ so as to reduce
the variance of the resulting ratio. Correlation between these estimators
is introduced by correlating the auxiliary random variates used to
obtain these estimators. Two implementations of this generic idea
are detailed. We show how to correlate importance sampling estimators
for random effects models and how to correlate particle filter estimators
for state-space models using the Hilbert sort procedure proposed in
\citep{gerber2015}.

We study in detail large sample properties of the CPM\ scheme for
random effects models. In this scenario, the loglikelihood ratio estimator
based on our correlation scheme is shown to satisfy a conditional
Central Limit Theorem (CLT) whenever $N$ grows to infinity sublinearly
with $T$ and the Euclidean distance between $\theta$ and $\theta^{\prime}$
is of order $1/\sqrt{T}$. When the posterior concentrates towards
a normal of standard deviation $1/\sqrt{T}$, this CLT\ can be used
to show that a space-rescaled version of the CPM chain converges weakly
to a discrete-time Markov chain on the parameter space. The Integrated
Autocorrelation Time (IACT) of the weak limit is not impacted by how
fast $N$ goes to infinity with $T$. However the lower this growth
rate is, the more correlated the auxiliary variables need to be to
control the variance of this estimator. We provide results suggesting
we need $N$ to grow at least at rate $\sqrt{T}$ for the IACT of
the original CPM chain to remain finite as $T\rightarrow\infty$.
We use these results to provide practical guidelines on how to optimise
performance of the algorithm for large data sets which are validated
experimentally. In our numerical examples on random effects models
and state-space models, CPM always outperforms PM and the improvement
increases with $T$ from 20 to 50 times when $T$ is a few hundreds
to more than 100 times when $T$ is a few thousands.

The rest of the paper is organised as follows. In Section \ref{SS: sim likelihood},
we introduce the CPM algorithm and detail its implementation for random
effects and state-space models. In Section \ref{sec:scaling}, we
present various CLTs for the loglikelihood estimator and loglikelihood
ratio estimators used by PM\ and CPM. In Section \ref{sec: optimisation},
we exploit these results to analyze and optimize the CPM kernel in
the large sample regime. We demonstrate experimentally the efficiency
of this methodology in Section\ \ref{sec:Applications} and discuss
various potential extensions in Section \ref{sec:discussion}. All
the proofs are given in the Supplementary Material.

\section{Metropolis\textendash Hastings and correlated pseudo-marginal schemes\label{SS: sim likelihood}}

\subsection{Metropolis\textendash Hastings algorithm}

The transition kernel $Q_{\textsc{ex}}$ of the MH algorithm targeting
$\pi\left(\theta\right)$ using a proposal distribution $q\left(\theta,\mathrm{d}\theta^{\prime}\right)=q\left(\theta,\theta^{\prime}\right)\mathrm{d}\theta^{\prime}$
is given by 
\begin{equation}
Q_{\textsc{ex}}\left(\theta,\mathrm{d}\theta^{\prime}\right)=q\left(\theta,\mathrm{d}\theta^{\prime}\right)\alpha_{\textsc{ex}}(\theta,\theta^{\prime})+\left\{ 1-\varrho_{\textsc{ex}}\left(\theta\right)\right\} \delta_{\theta}\left(\mathrm{d}\theta^{\prime}\right),\label{eq:Q_EX}
\end{equation}
where 
\begin{equation}
r_{\textsc{ex}}(\theta,\theta^{\prime})=\frac{\pi(\theta^{\prime})q\left(\theta^{\prime},\theta\right)}{\pi(\theta)q\left(\theta,\theta^{\prime}\right)}=\frac{p(y\mid\theta^{\prime})p\left(\theta^{\prime}\right)q\left(\theta^{\prime},\theta\right)}{p(y\mid\theta)p\left(\theta\right)q\left(\theta,\theta^{\prime}\right)},\label{eq:MHexactacceptanceratio}
\end{equation}
and 
\begin{equation}
\alpha_{\textsc{ex}}(\theta,\theta^{\prime})=\min\{1,r_{\textsc{ex}}(\theta,\theta^{\prime})\}\text{, \  \ }\varrho_{\textsc{ex}}\left(\theta\right)=\int q\left(\theta,\mathrm{d}\theta^{\prime}\right)\alpha_{\textsc{ex}}(\theta,\theta^{\prime}).\label{eq: accept EX}
\end{equation}
Implementing this MH\ scheme requires being able to evaluate the
likelihood ratio $p(y\mid\theta^{\prime})/p(y\mid\theta)$ appearing
in the expression of $r_{\textsc{ex}}(\theta,\theta^{\prime})$.

\subsection{The correlated pseudo-marginal algorithm}

Assume $\widehat{p}(y\mid\theta,U)$ is a non-negative unbiased estimator
of the intractable likelihood $p(y\mid\theta)$ when $U\sim m$. Here
$U$ corresponds to the $\mathcal{U}$-valued auxiliary random variables
used to obtain the estimator. We assume that $m\left(\mathrm{d}u\right)=m\left(u\right)\mathrm{d}u$
and introduce the joint density $\overline{\pi}(\theta,u)$ on $\Theta\times\mathcal{U}$,
where 
\begin{equation}
\overline{\pi}(\theta,u)=\pi(\theta)m(u)\thinspace\widehat{p}(y\mid\theta,u)/p(y\mid\theta).\label{eq:pseudomarginaltarget}
\end{equation}
As $\widehat{p}(y\mid\theta,U)$ is unbiased, $\overline{\pi}(\theta,u)$
admits $\pi\left(\theta\right)$ as marginal density. The CPM\ algorithm
is a MH scheme targeting (\ref{eq:pseudomarginaltarget}) with proposal
density $q\left(\theta,\mathrm{d}\theta^{\prime}\right)K\left(u,\mathrm{d}u^{\prime}\right)$
where $K$ admits an $m$-reversible Markov transition density, i.e.
\begin{equation}
m\left(u\right)K\left(u,u^{\prime}\right)=m\left(u^{\prime}\right)K\left(u^{\prime},u\right).\label{eq:Kismreversible}
\end{equation}
This yields the acceptance probability 
\begin{equation}
\alpha_{Q}\left\{ \left(\theta,u\right),\left(\theta^{\prime},u^{\prime}\right)\right\} =\min\left\{ 1,r_{\textsc{ex}}(\theta,\theta^{\prime})\frac{\widehat{p}(y\mid\theta^{\prime},u^{\prime})/p(y\mid\theta^{\prime})}{\widehat{p}(y\mid\theta,u)/p(y\mid\theta)}\right\} .\label{eq:acceptanceprobabilityCPM}
\end{equation}
Hence, the CPM\ algorithm admits $\overline{\pi}\left(\theta,u\right)$
as an invariant density by construction and its transition kernel
$Q$ is given by 
\begin{equation}
Q\left\{ \left(\theta,u\right),\left(\mathrm{d}\theta^{\prime},\mathrm{d}u^{\prime}\right)\right\} =q\left(\theta,\mathrm{d}\theta^{\prime}\right)K\left(u,\mathrm{d}u^{\prime}\right)\alpha_{Q}\left\{ \left(\theta,u\right),\left(\theta^{\prime},u^{\prime}\right)\right\} +\left\{ 1-\varrho_{Q}\left(\theta,u\right)\right\} \delta_{\left(\theta,u\right)}\left(\mathrm{d}\theta^{\prime},\mathrm{d}u^{\prime}\right),\label{eq:transitionCPM}
\end{equation}
where $1-\varrho_{Q}\left(\theta,u\right)$ is the corresponding rejection
probability. For $K\left(u,u^{\prime}\right)=m\left(u^{\prime}\right)$,
we recover the PM algorithm.

Let $\varphi\left(z;\mu,\Sigma\right)$ be the multivariate normal
density of argument $z$, mean $\mu$ and covariance\ matrix $\Sigma$
and let $X\sim\mathcal{N}\left(\mu,\Sigma\right)$ denote a sample
from this distribution. Henceforth, we assume the likelihood estimator
is computed using $M\geq1$ standard normal random variables so 
\begin{equation}
m\left(u\right)=\varphi\left(u;0_{M},I_{M}\right)\text{ and }K_{\rho}\left(u,u^{\prime}\right)=\varphi\left(u^{\prime};\rho u,\left(1-\rho^{2}\right)I_{M}\right),\label{eq:invariantuandARkernel}
\end{equation}
where $\rho\in\left(-1,1\right)$, $0_{M}$ is the $M\times1$ vector
with zero entries and $I_{M}$ the $M\times M$ identity matrix. It
is straightforward to check that $K_{\rho}$ is $m-$reversible. There
is no loss of generality to select $m$ as a\ normal density since
inversion techniques can be used to form any random variable of interest\footnote{For example, in Section \ref{sec:PF}, it is necessary to generate
uniform random variates and these may be constructed as $\Phi(u_{i})$
where $u_{i}$ is a scalar element of $u$ and $\Phi$ the cumulative
distribution function of the standard normal.}. In addition the kernel $K_{\rho}$ has the advantage that it can
be regarded as a discretized Ornstein\textendash Uhlenbeck process.
This property is exploited to establish the main result of Section
\ref{sec:scaling}.

Algorithm \ref{alg:CPM} summarizes how one samples from $Q\left\{ \left(\theta,U\right),\cdot\right\} $.

\begin{algorithm}[H]
\caption{\textbf{Correlated Pseudo-Marginal Algorithm}\label{alg:CPM}}

\begin{enumerate}
\item Sample $\theta^{\prime}\sim q\left(\theta,\cdot\right)$. 
\item Sample $\varepsilon\sim\mathcal{N}\left(0_{M},I_{M}\right)$ and set
$U^{\prime}=\rho U+\sqrt{1-\rho^{2}}\varepsilon.$ 
\item Compute the estimator $\widehat{p}(y\mid\theta^{\prime},U^{\prime})$
of $p(y\mid\theta^{\prime}).$ 
\item With probability 
\begin{equation}
\alpha_{Q}\left\{ \left(\theta,U\right),\left(\theta^{\prime},U^{\prime}\right)\right\} =\min\left\{ 1,\frac{\widehat{p}(y\mid\theta^{\prime},U^{\prime})}{\widehat{p}(y\mid\theta,U)}\frac{p(\theta^{\prime})}{p(\theta)}\frac{q\left(\theta^{\prime},\theta\right)}{q\left(\theta,\theta^{\prime}\right)}\right\} ,\label{eq:alg_accept}
\end{equation}
output $\left(\theta^{\prime},U^{\prime}\right)$. Otherwise, output
$\left(\theta,U\right)$. 
\end{enumerate}
\end{algorithm}

Contrary to the PM method corresponding to $\rho=0$, we need to store
the vector $u$ instead of $\widehat{p}(y\mid\theta,u)$ to implement
the algorithm when $\rho\neq0$. In the applications considered, this
overhead is mild.

The rationale behind the CPM\ scheme is that if $\left(\theta,u\right)\longmapsto\widehat{p}(y\mid\theta,u)$
is a regular enough function and $\left(\theta,U\right)$ and $\left(\theta^{\prime},U^{\prime}\right)$
are ``close\textquotedblright \ enough then we expect the ratio
estimator $\widehat{p}(y\mid\theta^{\prime},U^{\prime})/\widehat{p}(y\mid\theta,U)$
to have small relative variance and therefore to better mimick the
noiseless MH scheme $Q_{\textsc{ex}}$. In many situations, the posterior
$\pi\left(\theta\right)$ will be approximately normal for large data
sets with a covariance scaling in $1/\sqrt{T}$ so an appropriately
scaled MH random walk or autoregressive proposal $q\left(\theta,\mathrm{d}\theta^{\prime}\right)$
will ensure that $\theta$ and $\theta^{\prime}$ are ``close\textquotedblright .
We explain in Section \ref{sec:scaling}\ how $\rho$ can be selected
as a function of $T$ to ensure that $U$ and $U^{\prime}$ are ``close\textquotedblright \ enough
so that the loglikelihood ratio estimator $\log\{\widehat{p}(y\mid\theta^{\prime},U^{\prime})/\widehat{p}(y\mid\theta,U)\}$
satisfies a conditional CLT at stationarity. As alluded to in the
introduction, properties of this estimator and in particular its asymptotic
distribution and variance at stationarity are critical to our analysis
of the CPM\ scheme in the large sample regime detailed in Section
\ref{sec: optimisation}.

\subsection{Application to latent variable models\label{sec:latentvariablemodels}}

\subsubsection{Random effects models\label{subsec:panel}}

Consider the model 
\begin{equation}
X_{t}\overset{\text{i.i.d.}}{\sim}f_{\theta}(\cdot)\text{,}\text{\hspace{1cm}}\left.Y_{t}\right\vert X_{t}\sim g_{\theta}(\cdot\mid X_{t}),\label{eq:independentlatentvariablemodels}
\end{equation}
where $\left\{ X_{t};t\geq1\right\} $ are $\mathbb{R}^{k}$-valued
latent variables and $\left\{ Y_{t};t\geq1\right\} $ are $\mathsf{Y}$-valued
observations. For any $i<j$, let $i:j=\left\{ i,i+1,...,j\right\} $.
For a realization $Y_{1:T}=y_{1:T}$, the likelihood satisfies 
\begin{equation}
p(y_{1:T}\mid\theta)={\displaystyle \prod\limits _{t=1}^{T}}\;p(y_{t}\mid\theta),\text{\hspace{1cm}}p(y_{t}\mid\theta)=\int g_{\theta}(y_{t}\mid x_{t})\,f_{\theta}(x_{t})\mathrm{d}x_{t}.\label{eq:likelihoodpaneldata}
\end{equation}
If the $T$ integrals appearing in (\ref{eq:likelihoodpaneldata})
are intractable, we can estimate them using importance sampling 
\begin{equation}
\widehat{p}(y_{1:T}\mid\theta,U)={\displaystyle \prod\limits _{t=1}^{T}}\;\widehat{p}(y_{t}\mid\theta,U_{t}),\text{\hspace{1cm}}\widehat{p}(y_{t}\mid\theta,U_{t})=\frac{1}{N}\sum_{i=1}^{N}\omega\left(y_{t},U_{t,i};\theta\right),\label{eq:likelihoodestimator}
\end{equation}
where the importance weights\ $\omega\left(y,U_{t,i};\theta\right)$
are given by 
\begin{equation}
\omega(y_{t},U_{t,i};\theta)=\frac{g_{\theta}(y_{t}\mid X_{t,i})\thinspace f_{\theta}(X_{t,i})}{q_{\theta}(X_{t,i}\mid y_{t})},\label{eq:ISweights}
\end{equation}
assuming that there exists a deterministic map $\Xi_{t}:\mathbb{R}^{p}\times\Theta\rightarrow\mathbb{R}^{k}$
such that $X_{t,i}=\Xi_{t}(U_{t,i};\theta)\sim q_{\theta}(\cdot\mid y_{t})$
for $U_{t,i}\sim\mathcal{N}\left(0_{p},I_{p}\right)$. In this case,
we have $U=\left(U_{1},\ldots,U_{T}\right),$ $U_{t}=\left(U_{t,1},...,U_{t,N}\right)$
so $U\sim\mathcal{N}\left(0_{M},I_{M}\right)$ where $M=TNp$.

\subsubsection{State-space models\label{sec:PF}}

Consider a generalization of the model (\ref{eq:independentlatentvariablemodels})
where the latent variables $\left\{ X_{t};t\geq1\right\} $ now arise
from a homogeneous $\mathbb{R}^{k}$-valued Markov process of initial
density $\nu_{\theta}$ and Markov transition density $f_{\theta}$,
i.e. for $t\geq1$ 
\begin{equation}
X_{1}\sim\nu_{\theta},\text{ }\left.X_{t+1}\right\vert X_{t}\sim f_{\theta}(\cdot\mid X_{t}),\text{\hspace{1cm}}\left.Y_{t}\right\vert X_{t}\sim g_{\theta}(\cdot\mid X_{t}).\label{eq:statespacemodels}
\end{equation}
For a realization $Y_{1:T}=y_{1:T}$, the likelihood satisfies the
predictive decomposition 
\begin{equation}
p(y_{1:T}\mid\theta)=p(y_{1}\mid\theta)\,{\displaystyle \prod\limits _{t=1}^{T}}\text{ }p(y_{t}\mid y_{1:t-1},\theta),\label{eq:predictivedecompositionlikelihood}
\end{equation}
with 
\begin{equation}
p(y_{t}\mid y_{1:t-1},\theta)=\int g_{\theta}(y_{t}\mid x_{t}).p_{\theta}(x_{t}\mid y_{1:t-1})\mathrm{d}x_{t},\label{eq:predictiveSSM}
\end{equation}
where $p_{\theta}(x_{1}\mid y_{1:0})=\nu_{\theta}(x_{1})$ and $p_{\theta}(x_{t}\mid y_{1:t-1})$
denotes the posterior density of $X_{t}$ given $Y_{1:t-1}=y_{1:t-1}$
for $t\geq2$. Importance sampling estimators of the likelihood have
relative variance typically increasing exponentially with $T$ so
the likelihood is usually estimated using particle filters.

Particle filters propagate $N$ random samples, termed particles,
over time using a sequence of resampling steps and importance sampling
steps using the importance densities $q_{\theta}\left(\left.x_{1}\right\vert y_{1}\right)$
at time $1$ and $q_{\theta}\left(\left.x_{t}\right\vert y_{t},x_{t-1}\right)$
at times $t\geq2$. Let $\Xi_{1}:\mathbb{R}^{p}\times\Theta\rightarrow\mathbb{R}^{k}$
and $\Xi_{t}:\mathbb{R}^{k}\times\mathbb{R}^{p}\times\Theta\rightarrow\mathbb{R}^{k}$
for $t\geq2$ be deterministic maps such that $X_{1}=\Xi_{1}(V;\theta)\sim q_{\theta}(\cdot\mid y_{1})$
and $X_{t}=\Xi_{t}(x_{t-1},V;\theta)\sim q_{\theta}(\cdot\mid y_{t},x_{t-1})$
for $t\geq2$ if $V\sim\mathcal{N}\left(0_{p},I_{p}\right)$. If we
use these representations to sample the particles and normal random
variables to obtain uniform\ random variables to sample the categorical
distributions appearing in the resampling steps then we can obtain
an unbiased estimator $\widehat{p}(y_{t}\mid\theta,U)$ of the likelihood
where $U$ follows a multivariate normal \citep{DelMoral2004}. When
this estimator is used within a PM scheme, the resulting algorithm
is known as the particle marginal MH \citep{andrieu:doucet:holenstein2010}.
However if this likelihood estimator is used in the CPM\ context,
the likelihood ratio estimator $\widehat{p}(y_{1:T}\mid\theta^{\prime},u^{\prime})/\widehat{p}(y_{1:T}\mid\theta,u)$
can significantly deviate from $1$ even when $\left(\theta,u\right)$
is close to $\left(\theta^{\prime},u^{\prime}\right)$ and the true
likelihood is continuous at $\theta.$ This is because the resampling
steps introduce discontinuities in the particles that are selected
when $\theta$ and $u$ are modified, even slightly.

To reduce the variability of this likelihood ratio estimator, we use
a resampling scheme based on the Hilbert sort procedure proposed in
\citep{gerber2015}. This procedure is based on the Hilbert space-filling
curve which is a continuous fractal map $H:\left[0,1\right]\rightarrow\left[0,1\right]^{k}$
whose image is $\left[0,1\right]^{k}$. It admits a pseudo-inverse
$h:\left[0,1\right]^{k}\rightarrow\left[0,1\right]$, that is $H\circ h\left(x\right)=x$
for all $x\in\left[0,1\right]^{k}$. For most points $x,x^{\prime}$\ that
are close in $\left[0,1\right]^{k}$, their images $h\left(x\right)$
and $h\left(x^{\prime}\right)$ tend to be close. This property can
be used to build a ``sorted\textquotedblright \ resampling procedure
which will ensure that when the parameter or auxiliary variables change
only slightly the particles that are selected remain close. Practically,
this resampling procedure proceeds as follows: 1) the $\mathbb{R}^{k}-$valued
particles are projected in the hypercube $\left[0,1\right]^{k}$ using
a bijection $\varkappa:\mathbb{R}^{k}\rightarrow\left[0,1\right]^{k}$,
2) The resulting $\left[0,1\right]^{k}-$valued particles are projected
on $\left[0,1\right]$ using the pseudo-inverse $h$, 3) These projected
$\left[0,1\right]-$valued particles are sorted, 4)\ The systematic
resampling scheme proposed in \citep{carpenter1999} is used on the
sorted points.

Let us introduce the importance weights $\omega_{1}\left(u_{1};\theta\right)=\nu_{\theta}(x_{1})\thinspace g_{\theta}(y_{1}\mid x_{1})/\thinspace q_{\theta}(x_{1}\mid y_{1})$
and $\omega_{t}\left(x_{t-1},u_{t};\theta\right)=f_{\theta}(x_{t}\mid x_{t-1})\thinspace g_{\theta}(y_{t}\mid x_{t})/q_{\theta}(x_{t}\mid y_{t},x_{t-1})$
for $t\geq2$. The only difference between the resulting particle
filter presented below and the algorithm proposed in \citep{gerber2015}
is that we use normal random variates instead of randomized quasi-Monte
Carlo points in $\left[0,1\right]^{p}$. For the mapping $\varkappa$,
we adopt the logistic transform used in \citep{gerber2015}.

\begin{algorithm}[H]
\caption{\textbf{Particle filter using Hilbert sort}\label{alg:PFHilbert}}

\begin{enumerate}
\item Sample $U_{1,i}\sim\mathcal{N}(0_{p},I_{p})$ \textsf{and set} $X_{1,i}=\Xi_{1}(U_{1,i};\theta)\ $\textsf{for
}$i\in1:N$\textsf{.} 
\item For $t=1,\ldots,T-1$
\begin{enumerate}
\item \textsf{Find the permutation }$\sigma_{t}$\textsf{ such that }$h\circ\varkappa\left(X_{t,\sigma_{t}\left(1\right)}\right)\leq\ldots\leq h\circ\varkappa\left(X_{t,\sigma_{t}\left(N\right)}\right)$
\textsf{if} $t\geq2$\textsf{, or }$X_{t,\sigma_{t}\left(1\right)}\leq\ldots\leq X_{t,\sigma_{t}\left(N\right)}$
\textsf{if} $t=1$. 
\item \textsf{Sample }$U_{t}^{R}\sim\mathcal{N}\left(0,1\right)$\textsf{,
set }$\overline{U}_{t,i}=(i-1)/N$ $+$ $\Phi(U_{t}^{R})/N$ \textsf{for}
$i\in1:N$\textsf{.} 
\item \textsf{Sample }$A_{t,i}\sim F_{t}^{-1}\left(\overline{U}_{t,i}\right)$
\textsf{for} $i\in1:N$\textsf{ where }$F_{t}^{-1}$\textsf{ is the
generalized inverse distribution function of the categorical distribution
with weights }$\{\omega_{1}(U_{1,\sigma_{1}\left(i\right)};\theta);i\in1:N\}$
\textsf{if }$t=1$\textsf{ and $\{\omega_{t}(X_{t-1,\sigma_{t-1}\left(A_{t-1,\sigma_{t}\left(i\right)}\right)},U_{t,\sigma_{t}\left(i\right)};\theta);i\in1:N\}$
for }$t\geq2.$ 
\item \textsf{Sample }$U_{t+1,i}\sim\mathcal{N}\left(0_{p},I_{p}\right)$\textsf{
and set }$X_{t+1,i}=\Xi_{t+1}(X_{t,\sigma_{t}\left(A_{t,i}\right)},U_{t+1,i};\theta)$\textsf{
for }$i\in1:N$. 
\end{enumerate}
\end{enumerate}
\end{algorithm}

If we denote by $U=\left(U_{1,1},...,U_{T,N},U_{1}^{R},...,U_{T-1}^{R}\right)$
the vector of all the standard normal variables used within this particle
filter, the corresponding unbiased likelihood estimator is given by
\begin{equation}
\widehat{p}(y_{1:T}\mid\theta,U)=\left\{ \frac{1}{N}\sum_{i=1}^{N}\omega_{1}(U_{1,i};\theta)\right\} {\displaystyle \prod\limits _{t=2}^{T}}\left\{ \frac{1}{N}\sum_{i=1}^{N}\omega_{t}(X_{t-1,\sigma_{t-1}\left(A_{t-1,i}\right)},U_{t,i};\theta)\right\} .\label{eq:likelihoodestimatorPF}
\end{equation}
In this case, we have $M=TNp+T-1$. We can now directly use this estimator
within the CPM\ method.

\subsection{Discussion}

Ideas related to the CPM\ scheme have previously been proposed: \citep{Lee2010}
suggest combining PM steps with updates where $U$ is held fixed and
only $\theta$ is updated but this scheme will scale poorly with $T$
as it still uses PM\ steps. In \citep{andrieudoucetlee2012}, the
authors propose combining PM steps with steps where $\theta$ is held
fixed and correlation between $\widehat{p}(y\mid\theta,U)$ and $\widehat{p}(y\mid\theta,U^{\prime})$
is introduced by sampling $U^{\prime}$ using a $m-$reversible Markov
kernel $K$. However, the crucial selection of $K$ was not discussed
in \citep{andrieudoucetlee2012}. After the first version of this
work was made available, \citep{Dahlin2015} proposed independently
to use the correlation scheme (\ref{eq:invariantuandARkernel}) but
their guidelines for the correlation parameter $\rho$ differ from
the ones we give in subsequent sections and do not ensure that the
variance of the loglikelihood ratio estimator is controlled as $T$
increases. They also use a standard particle filter which will .

As the density $m$ of $U$ is independent of $\theta$, it might
be argued that a Gibbs algorithm sampling alternately from the full
conditional densities $\overline{\pi}(\left.\theta\right\vert u)$
and $\overline{\pi}(\left.u\right\vert \theta)$ of $\overline{\pi}(\theta,u)$
could mix well and that it is unnecessary to update $\theta$ and
$U$ jointly as in the CPM\ scheme. Related ideas have been explored
in \citep{papaspiliopoulos2007}. However, such a Gibbs strategy is
usually not implementable in the applications considered here. Sophisticated
particle Gibbs samplers have been proposed to mimick it but their
computational complexity is of order $T^{2}N$ per iteration for state-space
models when using such a parameterisation \citep[Section 6.2]{Lindsten2014}.
Thus they are not even competitive to the standard PM\ algorithm
whose cost per iteration is of order $T^{2}$.

\section{Asymptotics of the loglikelihood ratio 
estimators\label{sec:scaling}}

To understand the quantitative properties of the CPM scheme, it is
key to establish the statistical properties of the likelihood ratio
estimator appearing in its acceptance probability (\ref{eq:acceptanceprobabilityCPM}).
For the random effects models introduced in Section \ref{subsec:panel},
we establish conditional CLTs for the\ loglikelihood estimator (\ref{eq:likelihoodestimator})
and the corresponding loglikelihood ratio estimators used by the PM\ and
the CPM\ algorithms when $N,T\rightarrow\infty$. Here $N$ will
be a deterministic function of $T$ denoted by $N_{T}$. We show that
these estimators exhibit very different behaviours, underlining the
benefits of the CPM\ over the PM.

Consider a sequence of random variables $\{M^{T};T\geq1\}$ defined
on a common probability space $\left(\Omega,\mathcal{G},P\right)$
and sub-$\sigma$-algebras $\{\mathcal{G}^{T};T\geq1\}$ of $\mathcal{G}$
and write $\overset{P}{\rightarrow}$ to denote convergence in probability.
We write subsequently $\left.M^{T}\right\vert \mathcal{G}^{T}\Rightarrow\lambda$
if $M\sim\lambda$ and $\mathbb{E}[\left.f\left(M^{T}\right)\right\vert \mathcal{G}^{T}]\overset{P}{\rightarrow}\mathbb{E}[f\left(M\right)]$
as $T\rightarrow\infty$ for any bounded continuous function $f$.

Henceforth, we will make the assumption that $Y_{t}\overset{\text{i.i.d.}}{\sim}\mu$
and denote by $\mathcal{Y}^{T}$ the $\sigma$-field spanned by $Y_{1:T}$.
When additionally $U\sim m$, we denote the associated probability
measure, expectation and variance by $\mathbb{P}$,\ $\mathbb{E}$
and $\mathbb{V}$. As our limit theorems consider the asymptotic regime
where $N_{T},T\rightarrow\infty$, we should write $m_{T},\pi_{T}$
instead $m,\pi$ and similarly $U^{T},$ $U_{t}^{T}$ and $U_{t,i}^{T}$
instead of $U,$ $U_{t}$ and $U_{t,i}.$ The probability space is
defined precisely in Supplementary Material~\ref{Appendix:notation}. We do not
emphasise here this dependency on $T$ for notational simplicity but
it should be kept in mind that we deal with triangular arrays of random
variables. We can write unambiguously $\mathbb{E}(\psi(Y_{1},U_{1,1}^{T};\theta))$
as $\mathbb{E}(\psi(Y_{1},U_{1,1};\theta))$ because $U_{1,1}^{T}\sim\mathcal{N}\left(0_{p},I_{p}\right)$
under $\mathbb{P}$ for any $T\geq1.$

\subsection{Asymptotic distribution of the loglikelihood error\label{subsec:MarginalCLT}}

Let $\gamma(y_{1};\theta)^{2}=\mathbb{V}(\varpi(y_{1},U_{1,1};\theta))$
be the variance conditional upon $Y_{1}=y_{1}$ and $\gamma\left(\theta\right)^{2}=\mathbb{V}(\varpi(Y{}_{1},U_{1,1};\theta))=\mathbb{E}(\gamma(Y_{1};\theta)^{2})$
the unconditional variance of the normalized importance weight 
\begin{equation}
\varpi(Y_{t},U_{1,1};\theta)=\frac{\omega(Y_{t},U_{1,1};\theta)}{p(Y_{t}\mid\theta)},\label{eq:normalisedISweight}
\end{equation}
where $\omega\left(Y_{t},U_{1,1};\theta\right)$ is defined in (\ref{eq:ISweights}).

We present conditional CLTs for the loglikelihood error 
\begin{equation}
Z_{T}\left(\theta\right)=\log\widehat{p}(Y_{1:T}\mid\theta,U)-\log p(Y_{1:T}\mid\theta),\label{eq:loglikelihoodestimatorCLT}
\end{equation}
when $U$ arises from the proposal $m$ or from
\begin{equation}
\overline{\pi}(\left.u\right\vert \theta)=\frac{\overline{\pi}(\theta,u)}{\pi(\theta)}={\displaystyle \prod\limits _{t=1}^{T}}\frac{\widehat{p}(Y_{t}\mid\theta,u_{t})}{p(Y_{t}\mid\theta)}\varphi(u_{t};0_{pN_{T}},I_{pN_{T}}),\label{eq:posteriorpanelfactorizes}
\end{equation}
$\overline{\pi}(\theta,u)$ being given in (\ref{eq:pseudomarginaltarget}).
The density $\overline{\pi}(\left.u\right\vert \theta)$ depends upon
$N_{T}$ as the estimator of $p(Y_{t}\mid\theta)$ is obtained using
$N_{T}$ samples. 
\begin{thm}
\label{Theorem:CLTmarginal}Let $N_{T}=\left\lceil \beta T^{\alpha}\right\rceil $
with $1/3<\alpha\leq1$, $\beta>0$ and $Y_{t}\overset{\text{i.i.d.}}{\sim}\mu$.
\begin{enumerate}
\item If $\mathbb{E}\left(\varpi(Y,U_{1,1};\theta)^{8}\right)<\infty$ and
$U\sim m$ then 
\begin{equation}
\left.T^{\left(\alpha-1\right)/2}Z_{T}\left(\theta\right)+\frac{1}{2}T^{\left(1-\alpha\right)/2}\beta^{-1}\gamma\left(\theta\right)^{2}\right\vert \mathcal{Y}^{T}\Rightarrow\mathcal{N}\left(0,\beta^{-1}\gamma\left(\theta\right)^{2}\right).\label{eq:CLTexpressionmarginalproposal}
\end{equation}
\item If $\mathbb{E}\left(\varpi(Y_{1},U_{1,1};\theta)^{9}\right)+\mathbb{E}\left(\gamma(Y_{1};\theta)^{4}\right)<\infty$
and $U\sim\overline{\pi}(\left.\cdot\right\vert \theta)$ then 
\begin{equation}
\left.T^{\left(\alpha-1\right)/2}Z_{T}\left(\theta\right)-\frac{1}{2}T^{\left(1-\alpha\right)/2}\beta^{-1}\gamma\left(\theta\right)^{2}\right\vert \mathcal{Y}^{T}\Rightarrow\mathcal{N}\left(0,\beta^{-1}\gamma\left(\theta\right)^{2}\right).\label{eq:CLTexpressionmarginalequilibrium}
\end{equation}
\end{enumerate}
\end{thm}
\begin{rem*}
To establish (\ref{eq:CLTexpressionmarginalproposal}), respectively
(\ref{eq:CLTexpressionmarginalequilibrium}), for $1/2<\alpha\leq1$,
the condition $\mathbb{E}\left(\varpi(Y_{1},U_{1,1};\theta)^{4}\right)<\infty$,
respectively $\mathbb{E}\left(\varpi(Y_{1},U_{1,1};\theta)^{5}\right)<\infty$,
is sufficient. 
\end{rem*}
For particle filters, a CLT for $Z_{T}\left(\theta\right)$ of the
form (\ref{eq:CLTexpressionmarginalproposal}) has already been established
for the case $\alpha=1$ in \citep{berarddelmoraldoucet2014} when
using multinomial resampling under strong mixing assumptions. We conjecture
that both (\ref{eq:CLTexpressionmarginalproposal}) and (\ref{eq:CLTexpressionmarginalequilibrium})
hold under weaker assumptions for $1/3<\alpha<1$ and the Hilbert
sort resampling scheme. However, it is technically very challenging
to establish this result\footnote{In the simpler scenario where one uses systematic resampling, such
a CLT has not yet been established. Some of the technical problems
arising when attempting to carry out such an analysis are detailed
in \citep{gentil2008}. }. 

The results (\ref{eq:CLTexpressionmarginalproposal}) and (\ref{eq:CLTexpressionmarginalequilibrium})
imply that for large $T$ we expect that, under the proposal, $Z_{T}\left(\theta\right)$
is approximately normal with mean $-\beta^{-1}T^{1-\alpha}\gamma\left(\theta\right)^{2}/2$
and variance $\beta^{-1}T^{1-\alpha}\gamma\left(\theta\right)^{2}$,
whereas at equilibrium it is also approximately normal with the same
variance but opposite mean. The empirical distribution of $Z_{T}\left(\theta\right)$
is examined for random effects models and state-space models in Section
\ref{sec:Applications} and shown to closely match these limiting
distributions.

\subsection{Asymptotic distribution of the loglikelihood ratio error\label{subsec:conditionalCLT}}

Assume that we are at state $\left(\theta,U\right)$ and propose $\left(\theta^{\prime},U^{\prime}\right)$
using $\theta^{\prime}\sim q\left(\theta,\cdot\right)$, $U^{\prime}\sim m$
as in the PM algorithm or $\theta^{\prime}\sim q\left(\theta,\cdot\right)$,
$U^{\prime}\sim K_{\rho}\left(U,\cdot\right)$ as in the CPM algorithm.
In both cases, the acceptance ratio (\ref{eq:acceptanceprobabilityCPM})
depends on the loglikelihood ratio error 
\begin{equation}
R_{T}\left(\theta,\theta^{\prime}\right)=\log\text{\thinspace}\frac{\widehat{p}(Y_{1:T}\mid\theta^{\prime},U^{\prime})}{\widehat{p}(Y_{1:T}\mid\theta,U)}-\log\frac{p(Y_{1:T}\mid\theta^{\prime})}{p(Y_{1:T}\mid\theta)}.\label{eq:loglikelihoodratioerror}
\end{equation}
We examine here the limiting distribution of $R_{T}(\theta,\theta+\xi/\sqrt{T})$
for fixed $\theta$ and $\xi$. The rationale for examining this ratio
is that the posterior typically concentrates at rate $1/\sqrt{T}$
when $T$ increases so a correctly scaled random walk proposal for
a MH\ algorithm will be of the form $\theta^{\prime}=\theta+\xi/\sqrt{T}$
for $\xi$ a random variable of distribution independent of $T$.

For the PM\ algorithm, we have the following conditional CLT. 
\begin{thm}
\label{Theorem:CLTlikelihoodratiostandardpseudomarginal}Let $\theta$,$\xi$
be fixed. Assume that $\vartheta\mapsto\varpi\left(y_{1},u_{1,1};\vartheta\right)\ $and
$\vartheta\mapsto\mathbb{E}\left(\varpi(Y_{1},U_{1,1};\vartheta)^{9}\right)$
are continuous at $\vartheta=\theta$ for any $(y_{1},u_{1,1})\in\mathsf{Y\times}\mathbb{R}^{p}$,
$\vartheta\mapsto$ $\gamma\left(\vartheta\right)$ is continuously
differentiable at $\vartheta=\theta$ and $\mathbb{E}\left(\varpi(Y_{1},U_{1,1};\vartheta)^{9}\right)+\mathbb{E}\left(\gamma(Y_{1};\theta)^{4}\right)<\infty$.
For $N_{T}=\left\lceil \beta T^{\alpha}\right\rceil $ with $1/3<\alpha\leq1,$
$\beta>0,$ $Y_{t}\overset{\text{i.i.d.}}{\sim}\mu$, $U\sim\overline{\pi}\left(\left.\cdot\right\vert \theta\right)$
and $U^{\prime}\sim m$ where $U$ and $U^{\prime}$ are independent,
we have 
\begin{equation}
\left.T^{\left(\alpha-1\right)/2}R_{T}(\theta,\theta+\xi/\sqrt{T})+T^{\left(1-\alpha\right)/2}\beta^{-1}\gamma\left(\theta\right)^{2}\right\vert \mathcal{Y}^{T}\Rightarrow\mathcal{N}\left(0,2\beta^{-1}\gamma\left(\theta\right)^{2}\right).\label{eq:CLTloglikelihoodratioestimatorPMH}
\end{equation}
\end{thm}
This result shows that the loglikelihood ratio error in the PM case
can only have a limiting variance of order 1 if $N_{T}$ is proportional
to $T$. The loglikelihood ratio estimator used by the CPM\ exhibits
a markedly different behaviour if we consider $U^{\prime}\sim K_{\rho_{T}}\left(U,\cdot\right)$
with 
\begin{equation}
\rho_{T}=\exp\left(-\psi\frac{N_{T}}{T}\right),\label{eq:correlationscaling}
\end{equation}
for some $\psi>0$. Let us denote by $\mathcal{F}^{T}$ the $\sigma$-field
spanned by $\left\{ Y_{t};t\in1:T\right\} $ and $\left\{ U_{t,i};t\in1:T,i\in1:N\right\} $.
We also denote the Euclidean norm by $\left\Vert \cdot\right\Vert $
and we write $\nabla_{u}f=\left(\partial_{u^{1}}f,...,\partial_{u^{p}}f\right)^{\prime}$
for a real-valued function $f:\mathbb{R}^{p}\rightarrow\mathbb{R}$
where $u=\left(u^{1},...,u^{p}\right)$. 
\begin{thm}
\label{Theorem:conditionalCLTthetathetacand}Let $\theta$,$\xi$
be fixed. Let $Y_{t}\overset{\text{i.i.d.}}{\sim}\mu$, $U\sim\overline{\pi}\left(\left.\cdot\right\vert \theta\right)$
and $U^{\prime}\sim K_{\rho_{T}}\left(U,\cdot\right)$ where $\rho_{T}$
is given by (\ref{eq:correlationscaling}) then if Assumptions \ref{ass:momentsofW}-\ref{ass:duSteinfourthmoment}
in Supplementary Material~\ref{Section:ProofCLTnightmare} \ hold and if $N_{T}\rightarrow\infty$
as $T\rightarrow\infty$ with $N_{T}/T\rightarrow0$, we have 
\begin{equation}
\left.R_{T}(\theta,\theta+\xi/\sqrt{T})\right\vert \mathcal{F}^{T}\Rightarrow\mathcal{N}\left(-\kappa\left(\theta\right)^{2}/2,\kappa\left(\theta\right)^{2}\right),\label{eq:conditionalCLTloglikelihoodratioestimatorCPMH}
\end{equation}
where 
\begin{equation}
\kappa\left(\theta\right)^{2}=2\psi\mathbb{E}\left(\left\Vert \nabla_{u}\varpi(Y_{1},U_{1,1};\theta)\right\Vert ^{2}\right).\label{eq:varianceloglikelihoodratiocorrelated}
\end{equation}
\end{thm}
We do not make any structural assumption on $\varpi(y,u;\theta)$
to establish Theorem \ref{Theorem:conditionalCLTthetathetacand}.
Assumptions \ref{ass:momentsofW}-\ref{ass:duSteinfourthmoment} are
differentiability and integrability assumptions of this quantity with
respect to $y,$ $u$ and $\theta$. For CPM, this result states that
the limiting variance of the loglikelihood ratio is of order 1 when
$N_{T}$ grows sublinearly with $T$. Moreover, it shows that the
distribution of the loglikelihood ratio error becomes asymptotically
independent of $U$, suggesting that the CPM\ chain is less prone
to sticking than the PM at stationarity.%

This conditional CLT\ has not been established for particle filters.
For univariate state-space models, i.e. $k=1$, we have observed experimentally
on various stationary state-space models that a similar conditional
CLT appears to hold. For multivariate state-space models, the CLT
only appears to hold conditional upon $\mathcal{Y}^{T}$ when $N_{T}$
grows at least at rate $T^{k/\left(k+1\right)}$; see Section \ref{sec:Applications}.

\section{Analysis and optimisation\label{sec: optimisation}}

\subsection{Weak convergence in the large sample regime}

The use of weak convergence techniques to analyse and optimise MCMC
schemes was pioneered in \citep{RobertsGelmanGilks1997} and has found
numerous applications ever since; see, e.g., \citep{Sherlock2015efficiency}
for a recent application to PM. To the best of our knowledge, all
these contributions consider the asymptotic regime where the parameter
dimension $d\rightarrow\infty$ while $T$ is fixed. In these scenarios,
a time rescaling is introduced and the limiting Markov process is
usually a diffusion process. We analyze here the CPM\ scheme for
random effects models after space rescaling under the large sample
regime standard in asymptotic statistics; i.e., $d$ is fixed while
$T\rightarrow\infty$. Our analysis assumes the statistical model
is regular enough to ensure that the posteriors $\{\pi_{T}(\theta);T\geq1\}$
can be approximated by normal densities which concentrate. Here $\pi_{T}(\theta)$
is a random probability density dependent on $Y_{1:T}$ assumed to
be measurable w.r.t. $\mathcal{Y}^{T}$; see, e.g., \citep{bertipratellirigo2006,crauel2003}
for a formal definition. We write $\overset{\mathbb{P}^{Y}}{\rightarrow}$
to denote convergence in probability with respect to the law of $\left\{ Y_{t};t\geq1\right\} $. 
\begin{assumption}
\label{ass:BVM}The sequence of random probability densities $\{\pi_{T}(\theta);T\geq1\}$
satisfies at $T\rightarrow\infty$
\[
\int\left\vert \pi_{T}(\theta)-\varphi(\theta;\widehat{\theta}_{T},\overline{\Sigma}/T)\right\vert \mathrm{d}\theta\overset{\mathbb{P}^{Y}}{\rightarrow}0
\]
where $\{\widehat{\theta}_{T};T\geq1\}$ is a random sequence such
that $\widehat{\theta}_{T}$ is $\mathcal{Y}^{T}-$measurable, $\widehat{\theta}_{T}\overset{}{\rightarrow_{\mathbb{P}^{Y}}}\overline{\theta}$
and $\overline{\Sigma}$ is a positive definite matrix. 
\end{assumption}
This assumption will be satisfied if a Berstein-von\ Mises theorem
holds; see \citep[Section 10.2]{vandervaart2000} for sufficient conditions.

Consider the stationary CPM chain $\{(\vartheta_{n}^{T},\mathsf{U}_{n}^{T});n\geq0\}$
of proposal $q_{T}\left(\theta,\theta^{\prime}\right)$ targeting
the random measure (\ref{eq:pseudomarginaltarget}) $\overline{\pi}_{T}\left(\mathrm{d}\theta,\mathrm{d}u\right)=\pi_{T}(\mathrm{d}\theta)\overline{\pi}_{T}(\left.\mathrm{d}u\right\vert \theta)$
associated with $Y_{1:T}$. By rescaling the parameter component of
the CPM chain using $\widetilde{\vartheta}_{n}^{T}:=\sqrt{T}(\vartheta_{n}^{T}-\widehat{\theta}_{T})$,
we obtain the stationary Markov chain $\{(\widetilde{\vartheta}_{n}^{T},\mathsf{U}_{n}^{T});n\geq0\}$
of initial distribution $(\widetilde{\vartheta}_{0}^{T},\mathsf{U}_{0}^{T})\sim\widetilde{\pi}_{T}$
where 
\begin{equation}
\widetilde{\pi}_{T}(\widetilde{\theta},u)=\widetilde{\pi}_{T}(\widetilde{\theta})\widetilde{\pi}_{T}(\left.u\right\vert \widetilde{\theta}),\quad\widetilde{\pi}_{T}(\widetilde{\theta})=\pi_{T}(\widehat{\theta}_{T}+\widetilde{\theta}/\sqrt{T})/\sqrt{T},\quad\widetilde{\pi}_{T}(\left.u\right\vert \widetilde{\theta})=\overline{\pi}_{T}(u|\widehat{\theta}_{T}+\widetilde{\theta}/\sqrt{T}),\label{eq:rescaledtargetCPM}
\end{equation}
and the associated proposal density for the parameter becomes 
\begin{equation}
\widetilde{q}_{T}(\widetilde{\theta},\widetilde{\theta}^{\prime})=q_{T}(\widehat{\theta}_{T}+\widetilde{\theta}/\sqrt{T},\widehat{\theta}_{T}+\widetilde{\theta}^{\prime}/\sqrt{T})/\sqrt{T}.\label{eq:rescaledproposal}
\end{equation}
We will assume here that we use a random walk proposal scaled appropriately. 
\begin{assumption}
\label{Assumption:proposal}The proposal density is of the form 
\begin{equation}
q_{T}\left(\theta,\theta^{\prime}\right)=\sqrt{T}\upsilon(\sqrt{T}(\theta^{\prime}-\theta)),\label{eq:randomwalkproposal}
\end{equation}
where $\upsilon$ is a probability density on $\mathbb{R}^{d}$; that
is $\theta^{\prime}\sim q_{T}\left(\theta,\cdot\right)$ when $\theta^{\prime}=\theta+\xi/\sqrt{T}$
with $\xi\sim\upsilon$. 
\end{assumption}
Finally, we assume that we can control uniformly the rate at which
convergence of the CLT in Theorem \ref{Theorem:conditionalCLTthetathetacand}
holds in a neighbourhood of $\overline{\theta}$ specified in Assumption
\ref{ass:BVM}. We denote by $d_{\mathrm{BL}}\left(\mu,\nu\right)$
the bounded Lipschitz metric between probability measures $\mu,\nu$;
see, e.g., \citep[p. 332]{vandervaart2000} or Supplementary Material~\ref{Section:Conditionalweakconvergence}.
\begin{assumption}
\label{Assumption:uniformCLT}There exists a neighbourhood $N(\bar{\theta})$
of \textup{$\bar{\theta}$} such that the loglikelihood ratio error
considered in Theorem \ref{Theorem:conditionalCLTthetathetacand}
with $\xi\sim\upsilon\left(\cdot\right)$ satisfies as $T\rightarrow\infty$
\[
\sup_{\theta\in N(\bar{\theta})}\,\widetilde{\mathbb{E}}\left[\left.d_{\mathrm{BL}}\left\{ \mathcal{L}\mathrm{aw}\left(\left.R_{T}(\theta,\theta+\xi/\sqrt{T})\right\vert \mathcal{F}^{T}\right),\mathcal{N}\left(-\kappa\left(\theta\right)^{2}/2,\kappa\left(\theta\right)^{2}\right)\right\} \right|\mathcal{Y}^{T}\right]\overset{\mathbb{P}^{Y}}{\rightarrow}0.
\]
\end{assumption}
We prove that Assumption \ref{Assumption:uniformCLT} holds under
regularity conditions in Supplementary Material~\ref{Appendix:uniformCLT}. 

Under Assumption \ref{Assumption:proposal}, the proposal $\widetilde{q}_{T}(\widetilde{\theta},\widetilde{\theta}^{\prime})$
defined in (\ref{eq:rescaledproposal}) satisfies $\widetilde{q}_{T}(\widetilde{\theta},\widetilde{\theta}^{\prime})=\upsilon(\widetilde{\theta}^{\prime}-\widetilde{\theta})$
and will be denoted $\widetilde{q}(\widetilde{\theta},\widetilde{\theta}^{\prime})$.
In this case, the corresponding transition kernel of the rescaled
CPM chain is given by 
\begin{equation}
Q_{T}\{(\widetilde{\theta},u),(\mathrm{d}\widetilde{\theta}^{\prime},\mathrm{d}u^{\prime})\}=\widetilde{q}(\widetilde{\theta},\mathrm{d}\widetilde{\theta}^{\prime})K_{\rho_{T}}\left(u,\mathrm{d}u^{\prime}\right)\alpha_{Q_{T}}\{(\widetilde{\theta},u),(\widetilde{\theta}^{\prime},u^{\prime})\}+\{1-\varrho_{Q_{T}}(\widetilde{\theta},u)\}\delta_{(\widetilde{\theta},u)}(\mathrm{d}\widetilde{\theta}^{\prime},\mathrm{d}u^{\prime})\label{eq:rescaledtransitionkernelr}
\end{equation}
with acceptance probability 
\[
\alpha_{Q_{T}}\{(\widetilde{\theta},u),(\widetilde{\theta}^{\prime},u^{\prime})\}=\min\left\{ 1,\frac{\widetilde{\pi}_{T}(\widetilde{\theta}^{\prime},u^{\prime})\widetilde{q}(\widetilde{\theta}^{\prime},\widetilde{\theta})K_{\rho_{T}}\left(u^{\prime},u\right)}{\widetilde{\pi}_{T}(\widetilde{\theta},u)\widetilde{q}(\widetilde{\theta},\widetilde{\theta}^{\prime})K_{\rho_{T}}\left(u,u^{\prime}\right)}\right\} ,
\]
and corresponding rejection probability $1-\varrho_{Q_{T}}(\widetilde{\theta},u)$.
The kernel $Q_{T}$ depends on $Y_{1:T}$ and is assumed to be measurable
w.r.t. $\mathcal{Y}^{T}$. Let $\Theta_{T}=\{\widetilde{\vartheta}_{n}^{T};n\geq0\}$.
The following result shows that the non-Markov stationary sequences
$\left\{ \Theta_{T};T\geq1\right\} $ converge weakly as $T\rightarrow\infty$
to a stationary Markov chain corresponding to the Penalty method\textendash an
``ideal\textquotedblright{} Monte Carlo technique which cannot be
practically implemented \citep{ceperley1999}, \citep[p. 7]{nicholls2012}. 
\begin{thm}
\label{Theorem:weakconvergence}If Assumptions \ref{ass:BVM}, \ref{Assumption:proposal}
and \textcolor{black}{\ref{Assumption:uniformCLT} hold }and $\vartheta\mapsto\kappa\left(\vartheta\right)$
is continuous at $\vartheta=\overline{\theta}$ then the random probability
measures on $\left(\mathbb{R}^{d}\right)^{\infty}$ given by the laws
of $\left\{ \Theta_{T};T\geq1\right\} $ converge weakly in probability
$\mathbb{P}^{Y}$ as $T\rightarrow\infty$ to the law of a stationary
Markov chain $\{\widetilde{\vartheta}_{n};n\geq0\}$ defined by $\widetilde{\vartheta}_{0}\sim\mathcal{N}(0,\overline{\Sigma})$
and $\widetilde{\vartheta}_{n}\sim P(\widetilde{\vartheta}_{n},\cdot)$
for $n\geq1$ with 
\begin{equation}
P(\widetilde{\theta},\mathrm{d}\widetilde{\theta}^{\prime})=\widetilde{q}(\widetilde{\theta},\mathrm{d}\widetilde{\theta}^{\prime})\alpha_{P}(\widetilde{\theta},\widetilde{\theta}^{\prime})+\{1-\varrho_{P}(\widetilde{\theta})\}\delta_{\widetilde{\theta}}(\mathrm{d}\widetilde{\theta}^{\prime}),\label{eq:penaltykernel}
\end{equation}
and
\[
\alpha_{P}(\widetilde{\theta},\widetilde{\theta}^{\prime})=\int\varphi\left(\mathrm{d}w;-\kappa^{2}/2,\kappa^{2}\right)\min\left\{ 1,\frac{\varphi(\widetilde{\theta}^{\prime};0,\overline{\Sigma})\widetilde{q}(\widetilde{\theta}^{\prime},\widetilde{\theta})}{\varphi(\widetilde{\theta};0,\overline{\Sigma})\widetilde{q}(\widetilde{\theta},\widetilde{\theta}^{\prime})}\exp\left(w\right)\right\} ,
\]
$1-\varrho_{P}(\widetilde{\theta})$ being the corresponding rejection
probability and \textup{$\kappa:=\kappa(\overline{\theta})$.}
\end{thm}
The consequence of this result is that, as $T\rightarrow\infty$,
only the asymptotic distribution of the loglikelihood ratio error
at the central parameter value $\overline{\theta}$ impacts the acceptance
probability of the limiting chain. For large $T$ and a proposal of
the form specified in Assumption \ref{Assumption:proposal}, we thus
expect some of the quantitative properties of the CPM kernel $Q$,
where we now omit $T$ from notation, to be captured by the Markov
kernel 
\begin{equation}
\widehat{Q}\left(\theta,\mathrm{d}\theta^{\prime}\right)=q\left(\theta,\mathrm{d}\theta^{\prime}\right)\alpha_{\widehat{Q}}\left(\theta,\theta^{\prime}\right)+\{1-\varrho_{\widehat{Q}}\left(\theta\right)\}\delta_{\theta}\left(\mathrm{d}\theta^{\prime}\right),\label{eq:CPMkernelapproximation}
\end{equation}
where 
\[
\alpha_{\widehat{Q}}\left(\theta,\theta^{\prime}\right)=\int\varphi(\mathrm{d}w;-\kappa^{2}/2,\kappa^{2})\min\left\{ 1,r_{\textsc{ex}}(\theta,\theta^{\prime})\exp\left(w\right)\right\} ,
\]
$1-\varrho_{\widehat{Q}}\left(\theta\right)$ being the corresponding
rejection probability. We have obtained (\ref{eq:CPMkernelapproximation})
by using the change of variables $\theta=\widehat{\theta}_{T}+\widetilde{\theta}/\sqrt{T}$
and substituting the true target for its normal approximation in (\ref{eq:penaltykernel}),
hence removing a level of approximation.

\subsection{A\ bounding Markov chain\label{subsec:Abounding-Markov-chain}}

We analyse here the stationary Markov chain of transition kernel $\widehat{Q}$
arising from our weak convergence analysis. To state our results,
we need the following notation. For any real-valued measurable function
$h$, probability measure and Markov kernel $K$ on a measurable space
$(E,\mathcal{E})$, we write $\mu\left(h\right)=\int_{E}h\left(x\right)\mu\left(\mathrm{d}x\right)$,
$Kh\left(x\right)=\int_{E}K\left(x,\mathrm{d}x^{\prime}\right)h\left(x^{\prime}\right)$
and $K^{n}h\left(x\right)=\int_{E}K^{n-1}\left(x,\mathrm{d}z\right)K\left(z,\mathrm{d}x^{\prime}\right)h\left(x^{\prime}\right)$
for $n\geq2$ with $K^{1}=K$. We also introduce the Hilbert space
$L^{2}\left(\mu\right)=\left\{ h:E\rightarrow\mathbb{R}:\mu\left(h^{2}\right)<\infty\right\} $
equipped with the inner product $\left\langle g,h\right\rangle _{\mu}=\int_{E}g\left(x\right)h\left(x\right)\mu\left(\mathrm{d}x\right)$.
For any $h\in L^{2}\left(\mu\right)$, the autocorrelation at lag
$n\geq0$ is\ $\phi_{n}\left(h,K\right)=\left\langle \overline{h},K^{n}h\right\rangle _{\mu}/\mu(\overline{h}^{2})$
where $\overline{h}=h-\mu\left(h\right)$. The IACT associated with
a function $h$ under a Markov kernel $K$ is given\ by $\mathrm{IF}\left(h,K\right)=1+2\sum_{n=1}^{\infty}\phi_{n}\left(h,K\right)$
and will be referred to subsequently as the \emph{inefficiency}. For
$\mu\left(\mathrm{d}x\right)=\mu\left(\mathrm{d}x_{1},\mathrm{d}x_{2}\right)$,
we will slightly abuse notation and write $\mathrm{IF}(h,K)$ instead
of $\mathrm{IF}(g,K)$ when $g\left(x_{1},x_{2}\right)=h\left(x_{1}\right)$
or $g\left(x_{1},x{}_{2}\right)=h\left(x_{2}\right)$. When estimating
$\mu\left(h\right)$, $n\mathrm{IF}\left(h,K\right)$ samples from
a stationary Markov chain of $\mu$-invariant transition kernel $K$
are necessary to obtain approximately an estimator of the same precision
as an average of $n$ independent draws from $\mu$; see, e.g., \citep{geyer1992}. 

We provide an upper bound on $\mathrm{IF}(h,\widehat{Q})$ which we
exploit to provide guidelines on how to optimise the performance of
the CPM\ scheme in Subsection \ref{subsec:Optimization}. The inefficiency
$\mathrm{IF}(h,\widehat{Q})$ is difficult to work with but we give
an upper bound that only depends on $\mathrm{IF}(h,Q_{\textsc{ex}})$
and $\kappa$. To proceed, we introduce an auxiliary Markov kernel
$Q^{\ast}$ given by 
\begin{equation}
Q^{\ast}\left(\theta,\mathrm{d}\theta^{\prime}\right)=\varrho_{\text{\textsc{U}}}\left(\kappa\right)Q_{\textsc{ex}}\left(\theta,\mathrm{d}\theta^{\prime}\right)+\left\{ 1-\varrho_{\text{\textsc{U}}}\left(\kappa\right)\right\} \delta_{\theta}\left(\mathrm{d}\theta^{\prime}\right),\label{eq:Boundingchainkernel}
\end{equation}
where 
\begin{equation}
\varrho_{\text{\textsc{U}}}\left(\kappa\right)=\int\varphi(\mathrm{d}w;-\kappa^{2}/2,\kappa^{2})\min\left\{ 1,\exp\left(w\right)\right\} =2\Phi\left(-\kappa/2\right).\label{eq:lazyproba}
\end{equation}
We denote by $\mathrm{\bar{\varrho}_{Q^{*}}}\left(\kappa\right)$,
respectively $\mathrm{\bar{\varrho}_{\widehat{Q}}}\left(\kappa\right)$,
the average acceptance probability of $Q^{*}$, respectively $\widehat{Q}$,
at stationarity. The kernel $Q^{\ast}$ is a ``lazy\textquotedblright \ version
of $Q_{\textsc{ex}}$ which satisfies the following properties. 
\begin{prop}
\label{Proposition:Qstarproperties}The kernel $Q^{\ast}$ is reversible
w.r.t. $\pi$ and $\mathrm{IF(}h,\widehat{Q})\leq\mathrm{IF}\left(h,Q^{\ast}\right)$
for any $h\in L^{2}\left(\pi\right)$ where 
\begin{equation}
\mathrm{IF}\left(h,Q^{\ast}\right)=\{1+\mathrm{IF}(h,Q_{\textsc{ex}})\}/\varrho_{\text{\textsc{U}}}(\kappa)-1,\label{eq:inefficiencyQstar}
\end{equation}
with equality when $\varrho_{\textsc{ex}}\left(\theta\right)=1$ for
all $\theta\in\Theta$ and 
\begin{equation}
\mathrm{\bar{\varrho}_{Q^{*}}\left(\kappa\right)}=\varrho_{\text{\textsc{U}}}\left(\kappa\right)\pi(\varrho_{\textsc{ex}})\leq\mathrm{\bar{\varrho}_{\widehat{Q}}}\left(\kappa\right).\label{eq:inequalityaverageacceptanceproba}
\end{equation}
Moreover, $Q^{\ast}$ is geometrically ergodic if $Q_{\textsc{ex}}$
is geometrically ergodic. 
\end{prop}

For any transition kernel $K$ admitting an invariant distribution
given by $\pi$, or admitting $\pi$ as a marginal, we define the
relative inefficiency $\mathrm{RIF}\left(h,K\right)$ of the kernel
$K$ with respect to the known likelihood kernel $Q_{\textsc{ex}}$
and the auxiliary relative computing time $\mathrm{ARCT}\left(h,K\right)$
by 
\begin{equation}
\mathrm{RIF}\left(h,K\right):=\frac{\mathrm{IF}\left(h,K\right)}{\mathrm{IF}(h,Q_{\textsc{ex}})},\text{\hspace{1cm}}\mathrm{ARCT}\left(h,K\right):=\sqrt{\frac{\mathrm{RIF}\left(h,K\right)}{\kappa^{2}\varrho_{\text{\textsc{U}}}(\kappa)}}.\label{eq:RCTdef}
\end{equation}
We next minimise $\mathrm{ARCT}\left(h,Q^{\ast}\right)$, an upper
bound on $\mathrm{ARCT}(h,\widehat{Q})$, w.r.t. $\kappa$\textendash{}
this quantity is a component of the function we need to minimize in
order to optimize the performance of the CPM algorithm; see Section
\ref{subsec:Optimization}.
\begin{prop}
\label{cor:RIF}The following results hold:
\begin{enumerate}
\item If $\mathrm{IF}(h,Q_{\textsc{ex}})=1,$ then 
\[
\mathrm{RIF}\left(h,Q^{\ast}\right)=\{2-\varrho_{\text{\textsc{U}}}(\kappa)\}/\varrho_{\text{\textsc{U}}}(\kappa),
\]
and $\mathrm{ARCT}\left(h,Q^{\ast}\right)$ is minimised at $\kappa=1.35$,
at which point $\varrho_{\text{\textsc{U}}}(\kappa)=0.50$, $\mathrm{RIF}(h,Q^{\ast})=2.99$
and $\mathrm{ARCT}\left(h,Q^{\ast}\right)=1.81$. 
\item As $\mathrm{IF}(h,Q_{\textsc{ex}})\longrightarrow\infty,$ 
\[
\mathrm{RIF}\left(h,Q^{\ast}\right)=1/\varrho_{\text{\textsc{U}}}(\kappa),
\]
and $\mathrm{ARCT}\left(h,Q^{\ast}\right)$ is minimised at $\kappa=1.50$,
at which point $\varrho_{\text{\textsc{U}}}(\kappa)=0.43$, $\mathrm{RIF}(h,Q^{\ast})=2.20$
and $\mathrm{ARCT}\left(h,Q^{\ast}\right)=1.47$. 
\item $\mathrm{RIF}\left(h,Q^{\ast}\right)$ and $\mathrm{ARCT}\left(h,Q^{\ast}\right)$
are decreasing functions of $\mathrm{IF}(h,Q_{\textsc{ex}})$. The
minimising argument rises monotonically from $1.35$ to $1.50$ as
$\mathrm{IF}(h,Q_{\textsc{ex}})$ increases from $1$ to $\infty$. 
\end{enumerate}
\end{prop}
Figure \ref{fig:RCT_limit} displays $\varrho_{\text{\textsc{U}}}(\kappa)$,
$\mathrm{RIF}\left(h,Q^{\ast}\right)$ and $\mathrm{ARCT}\left(h,Q^{\ast}\right)$
against $\kappa$. The two scenarios displayed are for $\mathrm{IF}(h,Q_{\textsc{ex}})=1$,
corresponding to the ``perfect\textquotedblright \ proposal case
where $q\left(\theta,\theta^{\prime}\right)=\pi(\theta^{\prime})$,
and for the limiting case where $\mathrm{IF}(h,Q_{\textsc{ex}})\longrightarrow\infty$.
These are parts (i) and (ii) of Proposition \ref{cor:RIF}. From Figure
\ref{fig:RCT_limit}, it is also clear that $\mathrm{ARCT}\left(h,Q^{\ast}\right)$,
for both scenarios, is fairly flat as a function of $\kappa$. The
function only approximately doubles relative to the minimum at $\kappa=1$
or $4$.

\begin{figure}[ptb]
\begin{centering}
\includegraphics[height=2.5in]{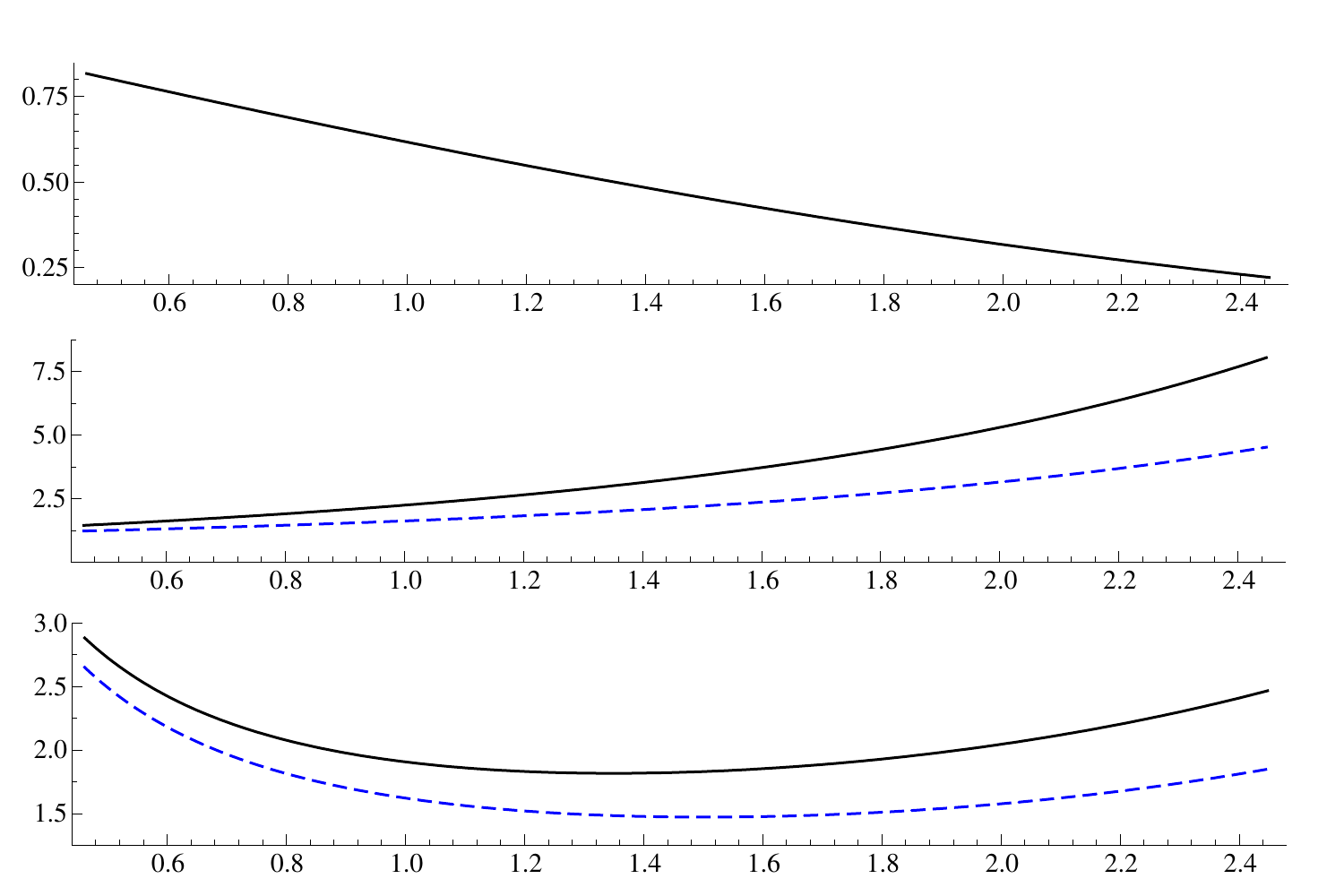}
\par\end{centering}
\caption{Illustrations of Proposition \ref{cor:RIF}. Top: Acceptance probability
$\varrho_{\text{\textsc{U}}}(\kappa)$ against $\kappa$. Middle:
Relative inefficiency $\mathrm{RIF}\left(h,Q^{\ast}\right)$ against
$\kappa$ (solid line $\mathrm{IF}(h,Q_{\textsc{ex}})=1$, dashed
line $\mathrm{IF}(h,Q_{\textsc{ex}})\rightarrow\infty$). Bottom:
Auxiliary relative computing time $\mathrm{ARCT}\left(h,Q^{\ast}\right)$
against $\kappa$ (solid line $\mathrm{IF}(h,Q_{\textsc{ex}})=1$,
dashed line $\mathrm{IF}(h,Q_{\textsc{ex}})\rightarrow\infty$). }
\label{fig:RCT_limit} 
\end{figure}

\subsection{A lower bound on the integrated autocorrelation time\label{subsec:A-lower-bound}}

The weak convergence result presented in Theorem \ref{Theorem:weakconvergence}%
{} does not imply that $\mathrm{IF}(h,Q)\overset{\mathbb{P}^{Y}}{\rightarrow}\mathrm{IF}\left(h,P\right)$
as $T\rightarrow\infty$ for any test function $h:\Theta\rightarrow\mathbb{R}$
such that $\mathrm{IF}(h,Q_{T})+\mathrm{IF}(h,P)<\infty$. Whereas
our weak convergence result holds whenever $N_{T}\rightarrow\infty$
as $T\rightarrow\infty$ and $N_{T}/T\rightarrow0$, we establish
here a result suggesting that we need $N_{T}$ to increase with $T$
at least at rate $\sqrt{T}$ for $\mathrm{IF}(h,Q_{T})$ to remain
controlled. To simplify the presentation in this section, we assume
further on that $d=1$.

In the CPM\ context, the auxiliary variables sequence $\{\mathsf{U}_{n};n\geq0\}$
evolve at a much slower scale than $\{\vartheta_{n};n\geq0\}$ according
to the kernel $K_{\rho_{T}}$ where $\rho_{T}$ is given by (\ref{eq:correlationscaling}).
We expect and observe empirically that the inefficiency $\mathrm{IF}(h,Q_{T})$
is of the same order as the inefficiency of $\left\{ \mathbb{E}\left[h(\vartheta_{n})|\mathsf{U}_{n}\right];n\geq0\right\} $
when $N_{T}$ increases too slowly to infinity with $T$. Moreover,
under large $T$, we have under regularity conditions, see, e.g.,
\citep[Lemma 2]{doucet2013}
\begin{equation}
\mathbb{E}\left[h(\vartheta_{n})|\mathsf{U}_{n}\right]=h(\widehat{\theta}_{T})+\frac{\overline{\Sigma}}{2T}\nabla_{\vartheta,\vartheta}h(\widehat{\theta}_{T})+\frac{\overline{\Sigma}}{T}\nabla_{\vartheta}h(\widehat{\theta}_{T})\text{ }\Psi(\widehat{\theta}_{T},\mathsf{U}_{n})+O_{\widetilde{\mathbb{P}}}\left(T^{-2}\right),\label{eq:expansion}
\end{equation}
where 
\begin{equation}
\Psi(\widehat{\theta}_{T},U)=\nabla_{\vartheta}\log\{\widehat{p}(Y_{1:T}\mid\widehat{\theta}_{T},U)/p(Y_{1:T}\mid\widehat{\theta}_{T})\}\label{eq:scoreerror}
\end{equation}
is the error in the simulated score at $\widehat{\theta}_{T}$, which
we will refer to as the score error. As\ a first step, we compute
the inefficiency $\mathrm{IF}(\Psi,Q_{T})$ of the score error.
\begin{prop}
\label{Proposition:slowdown}Under regularity conditions given in
Section \ref{Section:Breakdownproof}, there exists $C>0$ such that
$\mathrm{IF}(\Psi,Q_{T})\geq C\mathbb{V}_{\overline{\pi}}\left(\Psi\right)$
\textup{~$\mathbb{P}^{Y}-\mathrm{a.s.}$}
\end{prop}
It follows from calculations similar to Supplementary Material~\ref{Appendix:breakdown},
see also \citep[Proposition 3]{Lindsten2016}, that under regularity
conditions there exists $A>0$ such that $\mathbb{V}_{\overline{\pi}}\left(\Psi\right)\sim AT/N$
$\mathbb{P}^{Y}-\mathrm{a.s.}$ By combining (\ref{eq:expansion})
and Proposition \ref{Proposition:slowdown}, we thus expect the inefficiency
of $\left\{ \mathbb{E}\left[h(\vartheta_{n})|\mathsf{U}_{n}\right];n\geq0\right\} $
to be lower bounded by a term of order 
\[
\frac{\mathrm{IF}(\Psi,Q_{T})\text{ }\mathbb{V}_{\overline{\pi}}\left(\Psi/T\right)}{\mathbb{V}_{\pi}\left(h\right)}\gtrsim B\frac{T}{N_{T}}\frac{T^{1-\alpha}}{T^{2}}T=BT^{1-2\alpha}
\]
for $N_{T}=\left\lceil \beta T^{\alpha}\right\rceil $ and a constant
$B>0$. This result suggests that a necessary condition for $\mathrm{IF}(h,Q_{T})$
to remain finite as $T\rightarrow\infty$ is to have $N_{T}$ growing
at least at rate $\sqrt{T}$. This is validated by our experimental
results of Section \ref{sec:Applications} which also suggest that
this rate is sufficient.

\subsection{Optimization\label{subsec:Optimization}}

We provide a heuristic to select the parameters of the CPM so as to
optimise its performance which is validated by experimental results
in Section \ref{sec:Applications}. Again, we assume here that $d=1$.
For a test function $h:\Theta\rightarrow\mathbb{R}$, we want to minimize
\begin{equation}
\smash{\mathrm{CT}(h,Q_{T})=N_{T}\times\mathrm{IF}(h,Q_{T}),}\label{eq:CTcpm}
\end{equation}
where the factor $N_{T}$ arises from the fact that the computational
cost of obtaining the likelihood estimator is proportional to $N_{T}$
for random effects models. The results of Section \ref{subsec:A-lower-bound}
suggest that we should choose the number of Monte Carlo samples to
scale as $N_{T}=\beta T^{1/2}$ so that $\rho_{T}=\mathrm{exp}\left(-\psi\beta T^{-1/2}\right)$.
It remains to determine both $\psi$ and $\beta$.

To evaluate (\ref{eq:CTcpm}), we first decompose the functional of
interest evaluated at the parameter at the $n$-th iteration as 
\begin{equation*}
\smash{h(\vartheta_{n})=f(\mathsf{U}_{n})+g(\vartheta_{n},\mathsf{U}_{n}).}
\end{equation*}
where 
\begin{equation}
\smash{f(\mathsf{U}):=\mathbb{E}{}_{\bar{\pi}_{T}}\left[h(\vartheta)|\mathsf{U}\right],\hspace{1cm}g(\vartheta,\mathsf{U}):=h(\vartheta)-\mathbb{E}{}_{\bar{\pi}_{T}}\left[h(\vartheta)|\mathsf{U}\right].}\label{eq:fastslow}
\end{equation}
From \citep{KipnisVaradhan86}, we have that
\[
\mathbb{V}{}_{\pi_{T}}\left(h\right)\mathrm{IF}(h,Q_{T})\leq2\mathbb{V}{}_{\bar{\pi}_{T}}(f)\mathrm{IF}(f,Q_{T})+2\mathbb{V}{}_{\bar{\pi}_{T}}(g)\mathrm{IF}(g,Q_{T}).
\]

Assumption \ref{ass:BVM} combined to mild regularity assumptions
on $h$ and integrability conditions shows that
$
\mathbb{V}{}_{\bar{\pi}_{T}}\left(h\left(\vartheta_{n}\right)\right)\approx {\overline{\Sigma}_{h}} / {T}
$,
where $\overline{\Sigma}_{h}=|h'(\bar{\theta})|^{2}\overline{\Sigma}.$
Since $f(\vartheta_{n},\mathsf{U}_{n})$ and $g(\vartheta_{n},\mathsf{U}_{n})$
are clearly uncorrelated, this implies that
$
\mathbb{V}{}_{\bar{\pi}_{T}}\left(h\right)=\mathbb{V}{}_{\bar{\pi}_{T}}\left(f\right)+\mathbb{V}{}_{\bar{\pi}_{T}}\left(g\right)$.
From (\ref{eq:expansion}) we have 
$
\mathbb{V}{}_{\bar{\pi}_{T}}(f)\approx\overline{\Sigma}^{2}\mathbb{V}{}_{\bar{\pi}_{T}}\left(\Psi/T\right)\approx {\overline{\Sigma}_{f}} / {(TN_{T})}
$,
so it follows that
\begin{equation*}
\smash{\mathbb{V}{}_{\bar{\pi}_{T}}(g) \approx\frac{\overline{\Sigma}_{h}}{T}-\frac{\overline{\Sigma}_{f}}{TN_{T}}\approx\frac{\overline{\Sigma}_{h}}{T}.}
\end{equation*}
Using the reasoning of Section \ref{subsec:A-lower-bound} and the
calculations above we obtain
\begin{align}
\mathrm{IF}\left(h,Q_{T}\right) & \leq\frac{2}{\overline{\Sigma}_{h}}\left(\mathbb{V}{}_{\bar{\pi}_{T}}\left(\sqrt{T}f\right)\mathrm{IF}\left(f,Q_{T}\right)+\mathbb{V}{}_{\bar{\pi}_{T}}\left(\sqrt{T}g\right)\mathrm{IF}\left(g,Q_{T}\right)\right)\nonumber \\
 & \approx\frac{2}{\overline{\Sigma}_{h}}\left(\frac{\overline{\Sigma}_{f}}{N_{T}}\mathrm{IF}\left(\Psi,Q_{T}\right)+\overline{\Sigma}_{h}\mathrm{IF}\left(g,Q_{T}\right)\right).\label{eq:IFinequality}
\end{align}

Proposition \ref{Proposition:slowdown} states that $\mathrm{IF}(\Psi,Q_{T})$
is of order at least $T/N_{T}$ in probability as $T\to\infty$. Numerical
results suggest that in fact we have $\mathrm{IF}(\Psi,Q_{T})\approx A/\left(\delta_{T}\varrho_{\text{\textsc{U}}}(\kappa)\right)$
where $\delta_{T}=\psi N_{T}/T=-\mathrm{log}\thinspace\rho_{T}$ as
detailed in Section \ref{Sec: app_RE}, Figure \ref{fig:IFscoreasafunctionofdelta}.%
{} Hence, by substituting this expression of $\mathrm{IF}(\Psi,Q_{T})$
in (\ref{eq:IFinequality}), it follows that
\begin{align*}
\mathrm{IF}\left(h,Q_{T}\right) & \lesssim\frac{2}{\overline{\Sigma}_{h}}\left(\frac{\overline{\Sigma}_{f}}{N_{T}}\frac{A}{\delta_{T}\varrho_{\text{\textsc{U}}}(\kappa)}+\overline{\Sigma}_{h}\thinspace\mathrm{IF}\left(g,Q_{T}\right)\right).
\end{align*}
It can also be observed empirically from Figure \ref{Fig:longrange1},
described in Section \ref{Sec: app_RE}, that the autocorrelations
of $g(\vartheta_{n},\mathsf{U}_{n})$ decay exponentially, at a rate
independent of $T$. Thus we expect that, at least approximately,
we have $\mathrm{IF}\left(g,Q_{T}\right)\approx\mathrm{IF}(h,\widehat{Q}_{T})$
in probability. Therefore overall we have that for some constant $B$
\[
\mathrm{IF}\left(h,Q_{T}\right)\lesssim2\left(\frac{B}{\varrho_{\text{\textsc{U}}}(\kappa)\delta_{T}N_{T}}+\mathrm{IF}\left(h,\widehat{Q}_{T}\right)\right).
\]

We are interested in optimizing $\mathrm{CT}(h,Q_{T})\simeq N_{T}\times\mathrm{IF}(h,Q_{T})$
w.r.t. $\psi$ and $\beta$ where we recall from (\ref{eq:varianceloglikelihoodratiocorrelated})
that $\delta_{T}=\psi N_{T}/T=\psi\beta/\sqrt{T}=(\kappa^{2}\beta)/(\gamma^{2}\sqrt{T})$
as $\kappa^{2}=\psi\gamma^{2}$. Therefore 
\begin{equation}
\mathrm{CT}(h,Q_{T})\lesssim2T^{1/2}\left(\frac{C}{\beta\varrho_{\text{\textsc{U}}}(\kappa)\kappa^{2}}+\beta\mathrm{IF}\left(h,\widehat{Q}_{T}\right)\right),\label{eq:upperboundCT}
\end{equation}
where $C=B\gamma^{2}$, and thus the upper bound on $\mathrm{CT}(h,Q_{T})$
is minimized for
\[
\beta=\sqrt{\frac{C}{\varrho_{\text{\textsc{U}}}(\kappa)\kappa^{2}\mathrm{IF}\left(h,\widehat{Q}_{T}\right)}}.
\]
Substituting this expression in the upper bound on $\mathrm{IF}(h,\widehat{Q}_{T})$,
we obtain 
$
\mathrm{IF}\left(h,Q_{T}\right)\lesssim4\mathrm{IF}\big(h,\widehat{Q}_{T}\big)
$.
Then by further substituting the last expression in the resulting
upper bound on $\mathrm{CT}(h,Q_{T})$, we obtain 
\begin{equation}
\mathrm{CT}(h,Q_{T})\lesssim4\sqrt{C}T^{1/2}\mathrm{ARCT}\left(h,\widehat{Q}_{T}\right)\lesssim4\sqrt{C}T^{1/2}\mathrm{ARCT}(h,Q_{T}^{*})\label{eq:RCT}
\end{equation}
where $\mathrm{ARCT}$ is introduced in (\ref{eq:RCTdef}). Therefore
in practice we %
{} will minimize $\mathrm{ARCT}(h,Q_{T}^{*})$ w.r.t. $\kappa$, which
has already been addressed in Proposition \ref{cor:RIF}. The minimiser
$\hat{\kappa}$ is a function of $\mathrm{IF}(h,Q_{\textsc{ex}})$
which varies only slightly as $\mathrm{IF}(h,Q_{\textsc{ex}})$ varies
from $1$ to $\infty$ as observed in Figure \ref{fig:RCT_limit}.
Consequently, we propose the following procedure to optimize the performance
of CPM. Let $T$ be fixed and large enough for the asymptotic assumptions
to hold approximately. First, we choose a candidate value for $N$
and determine $\hat{\psi}$ such that the standard deviation of the
log-likelihood ratio estimator around the mode of the posterior satisfies
$\hat{\kappa}\approx1.4$. Second, fixing $\psi$ to $\hat{\psi}$,
we evaluate for several values of $\beta$ the computational time
$\mathrm{CT}(h,Q_{T})$ which we assume is of the form of the upper
bound (\ref{eq:upperboundCT}), i.e.,
\begin{equation}
\mathrm{CT}(h,Q_{T})={C_{0}}/{\beta}+C_{1}\beta,\label{eq:CTbeta}
\end{equation}
with $\kappa$ and $T$ kept constant; see Figure \ref{fig:limit_delta}
in Section \ref{Sec: app_RE} for empirical results. This function
is minimized for $\beta=\sqrt{C_{0}/C_{1}}$. Practically we evaluate
$\mathrm{CT}(h,Q_{T})$ only on a subset of the data. We then estimate
through regression the constants $C_{0}$ and $C_{1}$ by $\hat{C}_{0}$
and $\hat{C}_{1}$ which in turn provide the following estimate of
$\beta$
\begin{equation}
\hat{\beta}=\sqrt{\hat{C}_{0}/\hat{C}_{1}}.\label{eq:estimatebeta}
\end{equation}
We examine in Section \ref{Sec: app_RE} the assumptions made here,
illustrate this procedure and demonstrate its robustness.

\section{Applications\label{sec:Applications}}

\subsection{Random effects model\label{Sec: app_RE}}

We illustrate the performance of\ the PM and CPM\ schemes on a simple
Gaussian random effects model where 
\begin{equation}
X_{t}\overset{\text{i.i.d.}}{\sim}\mathcal{N}(\theta,1)\text{,\text{\hspace{1cm}}}\left.Y_{t}\right\vert X_{t}\sim\mathcal{N}(X_{t},1).\label{eq:Gaussianrandomeffect}
\end{equation}
We are interested in estimating $\theta$ (which has a true value
of $0.5$) to which we assign a zero-mean Gaussian prior with large
variance. In this scenario, the likelihood is known as $Y_{t}\sim\mathcal{N}(\theta,2)$.
This allows for detailed experimental analysis of the loglikelihood
error and the loglikelihood ratio error. This also allows us to implement
the MH algorithm with the true likelihood. The same normal random
walk proposal is used for all three schemes (MH, PM and CPM) and the
following unbiased estimator of the likelihood is used for the PM
and CPM schemes: 
\begin{equation}
\widehat{p}(y_{1:T}\mid\theta,U)={\displaystyle \prod\limits _{t=1}^{T}}\text{ }\text{ }\widehat{p}(y_{t}\mid\theta,U_{t}),\text{\hspace{1cm}}\widehat{p}(y_{t}\mid\theta,U_{t})=\frac{1}{N}\sum\limits _{i=1}^{N}\varphi\left(y_{t};\theta+U_{t,i},1\right),\text{\hspace{1cm}}U_{t,i}\overset{\text{i.i.d.}}{\sim}\mathcal{N}(0,1).\label{eq:likelihoodestimate}
\end{equation}
The inefficiency is estimated for all three schemes for $h\left(\theta\right)=\theta$
using $1+2\sum_{n=1}^{L}\widehat{\phi}_{n}$ where $\widehat{\phi}_{n}$
is the estimated correlation for $\theta$ at lag $n$ and $L$ is
a suitable cutoff value. We use the notation $Z=\log\left\{ \widehat{p}(y_{1:T}\mid\theta,U)/p(y_{1:T}\mid\theta)\right\} $
and $W=\log\left\{ \widehat{p}(y_{1:T}\mid\theta^{\prime},U')/p(y_{1:T}\mid\theta^{\prime})\right\} $
where $\theta^{\prime}\sim q\left(\theta,\cdot\right)$, $U^{\prime}\sim K_{\rho}\left(U,\cdot\right)$
and write $R=W-Z$ for $R_{T}\left(\theta,\theta'\right)$ defined
in (\ref{eq:loglikelihoodratioerror}).

As discussed in Section \ref{sec: optimisation}, for large datasets,
the relative inefficiency $\mathrm{RIF=\mathrm{IF}/\mathrm{IF}_{\mathrm{MH}}}$
and associated relative computing time $\mathrm{RCT}=N\times\mathrm{RIF}$
of the CPM scheme depend on the standard deviation $\kappa$ of $R$
at stationarity and the correlation parameter $\rho$. To validate
experimentally the results of Section \ref{sec:scaling}, we first
analyze the case where $T=8192$ in more detail. We run CPM using
a random walk proposal for $N=80$ and $\rho=0.9963$, so that $\kappa=1.145$.
The draws of $W$ and $Z$ at equilibrium, together with $R$, are
displayed in Figure \ref{fig_GE_8k_N80}. The draws of $Z$ are approximately
distributed according to $\mathcal{N}(\sigma^{2}/2,\sigma^{2})$ (middle
left), where the variance $\sigma^{2}$ is high. The draws of $R$
appear uncorrelated (in unreported tests) and their histogram is indistinguishable
from the expected theoretical distribution $\mathcal{N}(-\kappa^{2}/2,\kappa^{2})$
established in Theorem \ref{Theorem:conditionalCLTthetathetacand}
(middle right). This is in agreement with Theorem \ref{Theorem:CLTmarginal},
equation (\ref{eq:CLTexpressionmarginalequilibrium}), the posterior
of $\theta$ being concentrated. The resulting draws and correlogram
(bottom panel) of $\theta$ demonstrate low persistence. 

\begin{figure}[ptb]
\begin{centering}
\includegraphics[height=2.5in]{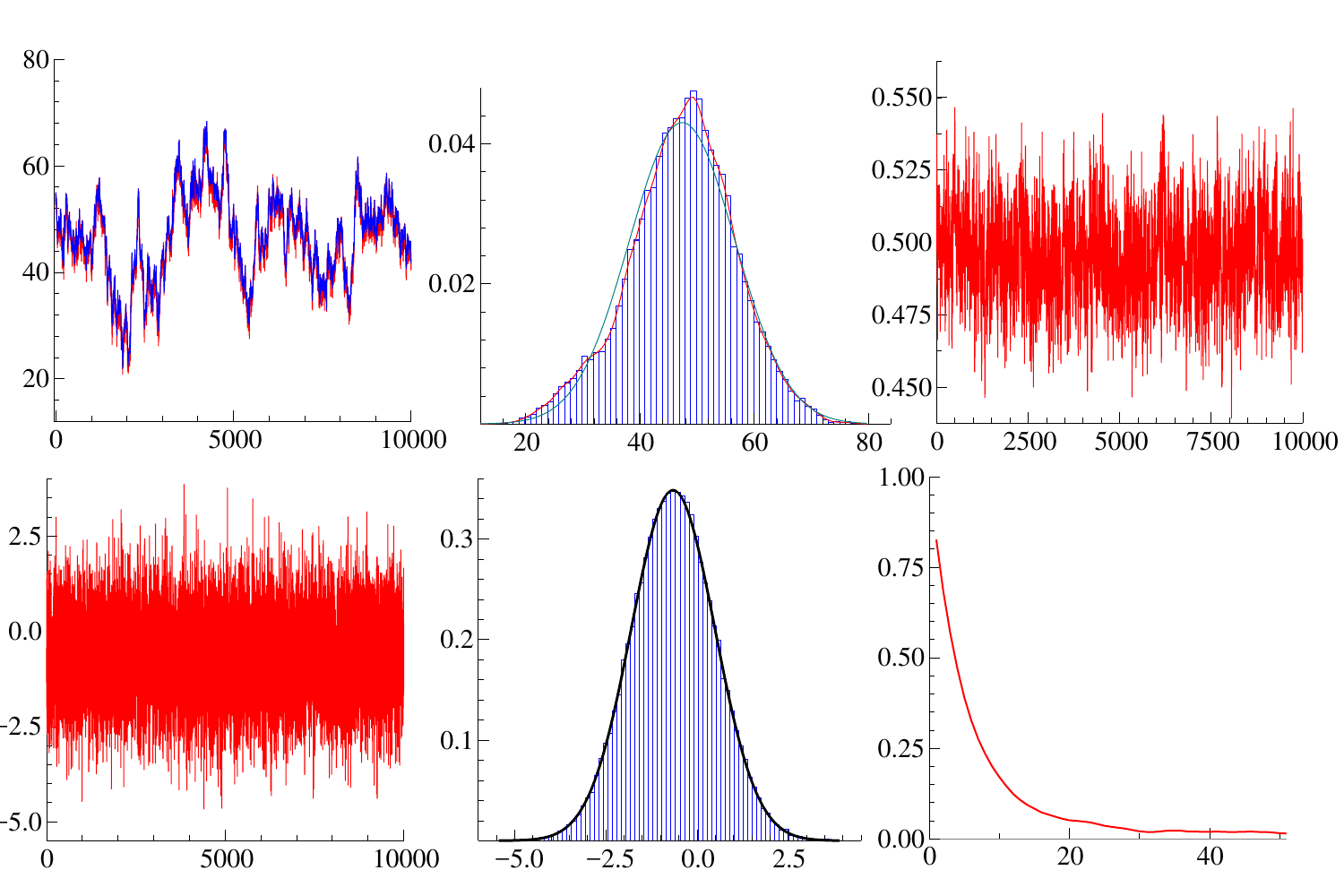} 
\par\end{centering}
\caption{Random effects model using CPM: $T=8192,$ $N=80$, $\rho=0.9963$.
Left: the first $10,000$ iterations of $W$ (blue) and $Z$ (red)
(top), the difference $R$ (bottom). Middle: Histograms of $Z$ (top)
and $R$ (bottom) and the theoretical Gaussian densities. Right: draws
of $\theta$ (top) and the corresponding correlogram (bottom). }
\label{fig_GE_8k_N80} 
\end{figure}

For the PM scheme, it is necessary to take $N=5000$ samples to ensure
that the variance of $Z$ evaluated at a central value $\widehat{\theta}$
is approximately one \citep{doucet2015efficient}. We next validate
experimentally the theoretical results of Section \ref{sec: optimisation}
by investigating the performance of CPM for this dataset, varying
$N$, and thus also hence $\kappa^{2}=\mathbb{V}(R)$, while keeping
$\rho=0.9963$.%
{} \ Figure \ref{analyseRE} displays the values of $\mathrm{RIF}$
and $\mathrm{RCT}$ against $\kappa$ as well as the marginal acceptance
probabilities, showing that $\mathrm{RCT}$ is approximately minimized
around $\kappa=1.6$ close to the minimizing argument of $\mathrm{ARCT}(h,Q_{T}^{*})$
established in Proposition \ref{cor:RIF} which satisfies (\ref{eq:RCT}).
The bottom two plots show that $\log\thinspace\kappa^{2}$ decreases
linearly with $\log\thinspace N$ as expected (bottom right) and that
the marginal probability of acceptance in the CPM\ scheme is close
to the asymptotic lower bound (bottom left) given by (\ref{eq:inequalityaverageacceptanceproba}).
From these experimental results, it is clear that for all values of
$N$ considered, the gains of the CPM scheme over the PM method in
terms of $\mathrm{RCT}$ are very significant. The optimal value of
$N$ for the CPM scheme is $35$ ($\kappa=1.6$) which gives $\mathrm{RCT}$
$=$ $61$ against a value of $\mathrm{RCT}=$ $14100$ for the PM\ scheme.
As a consequence, PM would take more than $200$ times as long in
computational time to produce an estimate of the posterior mean of
$\theta$ of the same accuracy.%
{} 

\begin{figure}[ptb]
\begin{centering}
\includegraphics[height=2.5in]{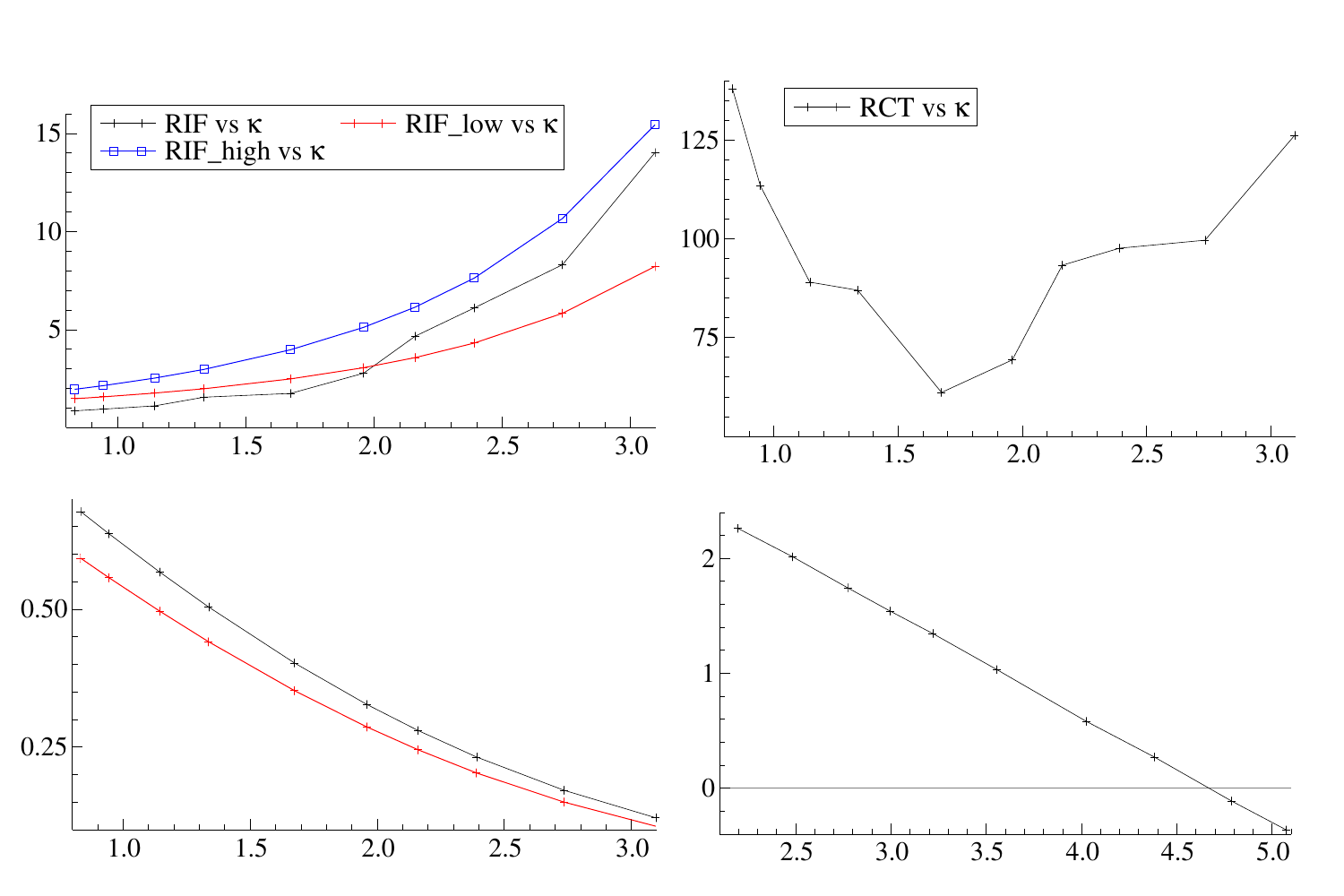} 
\par\end{centering}
\caption{Random effects model using CPM: $T=8192$, $\rho$ fixed and various
$N$. $\mathrm{RIF}_{\mathrm{CPM}}$ (top left) and $\mathrm{RIF}_{Q^{*}}$
for $\mathrm{IF}(h,Q_{\textsc{ex}})=1$ and $\mathrm{IF}(h,Q_{\textsc{ex}})=\infty$
against $\kappa$, see Corollary \ref{cor:RIF}. $\mathrm{RCT}_{\mathrm{CPM}}$
against $\kappa$ (top right). The acceptance probability of the CPM\textbf{
}and the theoretical lower bound, of (\ref{eq:inequalityaverageacceptanceproba}),
against $\kappa$ (bottom left). $\log(\kappa^{2})$ against $\log(N)$
(bottom right).}
\label{analyseRE} 
\end{figure}

We next investigate the performance of CPM when $T$ and $N=\beta\sqrt{T}$
vary while $\psi$, equivalently $\rho$, is scaled such that $\kappa$
is approximately constant. The results are recorded in Table \ref{table_GE_diifT}.
They suggest that the scaling $N=\beta\sqrt{T}$ is successful as
$\mathrm{IF}_{\mathrm{CPM}}$ appears to stabilize whereas the scaling
$N=\beta T$ would be necessary for $\mathrm{IF}_{\mathrm{PM}}$ to
stabilize. Experimental results not reported here confirm that if
$N$ grows at a slower rate than $\sqrt{T}$ , then IF$_{\mathrm{CPM}}$
increases without bound with $T$. 

\begin{table}[h!]
\centering%
\begin{tabular}{lllllllll}
\hline 
$T$  & $N$  & $\rho$  & $\kappa^{2}$  & $\bar{\varrho}_{\mathrm{MH}}$  & IF$_{\mathrm{MH}}$  & $\bar{\varrho}_{\mathrm{CPM}}$  & IF$_{\mathrm{CPM}}$  & RIF$_{\mathrm{CPM}}$\tabularnewline
\hline 
$1024$  & $19$  & $0.9894$  & $2.0$  & $0.71$  & \multicolumn{1}{r}{$10.71$} & $0.48$  & $43.26$  & \multicolumn{1}{r}{$4.0$4}\tabularnewline
$2048$  & $28$  & $0.9925$  & $1.9$  & $0.69$  & \multicolumn{1}{r}{$8.21$} & $0.49$  & $38.50$  & \multicolumn{1}{r}{$4.61$}\tabularnewline
$4096$  & $39$  & $0.9947$  & $1.7$  & $0.72$  & \multicolumn{1}{r}{$11.75$} & $0.51$  & $21.0$1  & \multicolumn{1}{r}{$1.79$}\tabularnewline
$8192$  & $56$  & $0.9962$  & $1.8$  & $0.81$  & \multicolumn{1}{r}{$15.61$} & $0.50$  & $24.25$  & \multicolumn{1}{r}{$1.55$}\tabularnewline
$16384$  & $79$  & $0.9974$  & $1.8$  & $0.70$  & \multicolumn{1}{r}{$9.37$} & $0.50$  & $20.05$  & \multicolumn{1}{r}{$2.14$}\tabularnewline
\hline 
\end{tabular}\caption{Random effects model. Inefficiency and acceptance probabilities for
MH and CPM, $N=\beta\sqrt{T}$ and $\rho$ selected such that $\kappa^{2}$
is approximately constant.}
\label{table_GE_diifT} 
\end{table}

We now justify empirically some of the assumptions made in Section
\ref{sec: optimisation} on how to select the parameters $\psi$ and
$\beta$. First, we show that the CPM process can be thought of as
a combination of two different processes: a `slow' moving component
$f(\mathsf{U}_{n})\approx\hat{f}(\mathsf{U}_{n})=\widehat{\theta}_{T}+\overline{\Sigma}T^{-1}\Psi(\widehat{\theta}_{T},\mathsf{U}_{n})$,
the modified score error associated to the score error $\Psi(\widehat{\theta}_{T},\mathsf{U}_{n})$
defined in (\ref{eq:scoreerror}), and a `fast' component $g(\vartheta_{n},\mathsf{U}_{n})=\vartheta_{n}-f(\mathsf{U}_{n})\approx\hat{g}(\vartheta_{n},\mathsf{U}_{n})=\vartheta_{n}-\hat{f}(\mathsf{U}_{n})$.
We display these components for a CPM run and the associated correlograms
in Figure \ref{Fig:longrange1}%
{} for fixed $\kappa$. We also illustrate in Figure \ref{fig:IFscoreasafunctionofdelta}
that $\mathrm{IF}(\Psi,Q_{T})\approx A/\left(\delta_{T}\varrho_{\text{\textsc{U}}}(\kappa)\right)$
where $\delta_{T}=\psi N_{T}/T=-\mathrm{log}\thinspace\rho_{T}$.
The mixing of the score error deteriorates as $\rho$ approaches one.
For fixed $\kappa$, as $N$ increases then $\rho$ decreases so the
autocorrelation function of the score decays faster and, additionally,
the variability of the modified score error decreases. Hence its contribution
to the autocorrelation function of the CPM scheme decreases. The optimization
scheme developed in Section \ref{subsec:Optimization} essentially
selects $\beta$ such that the asymptotic variances of both the slow
and fast components $\hat{f}(\mathsf{U}_{n})$ and $\hat{g}(\vartheta_{n},\mathsf{U}_{n})$
are of the same order.

\begin{figure}[ptb]
\begin{centering}
\includegraphics[height=2.5in]{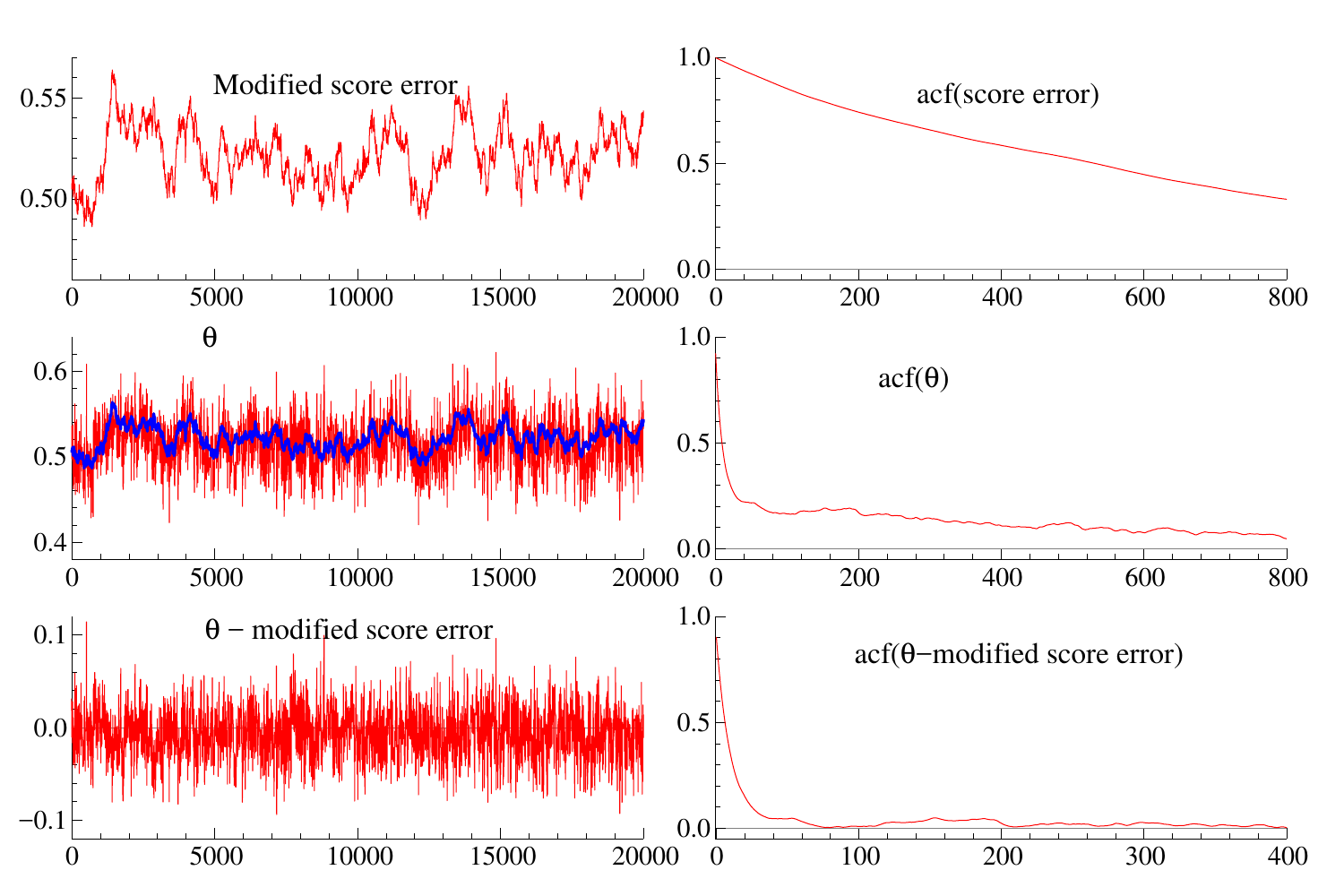} 
\par\end{centering}
\caption{Random effects model using CPM: $T=2,560$, $\beta=0.12,$ $N=6$,
$\rho=0.9977$. Top: modified score error $\hat{f}(\mathsf{U}_{n})$
(left) and its correlogram (right). Middle: parameter $\vartheta_{n}$
(red) and modified score error (blue) (left) and correlogram $\vartheta_{n}$
(right). Bottom: residual $\hat{g}(\vartheta_{n},\mathsf{U}_{n})$
(left) and correlogram (right).}

\label{Fig:longrange1}
\end{figure}

\begin{figure}[ptb]
\begin{centering}
\includegraphics[height=2.5in]{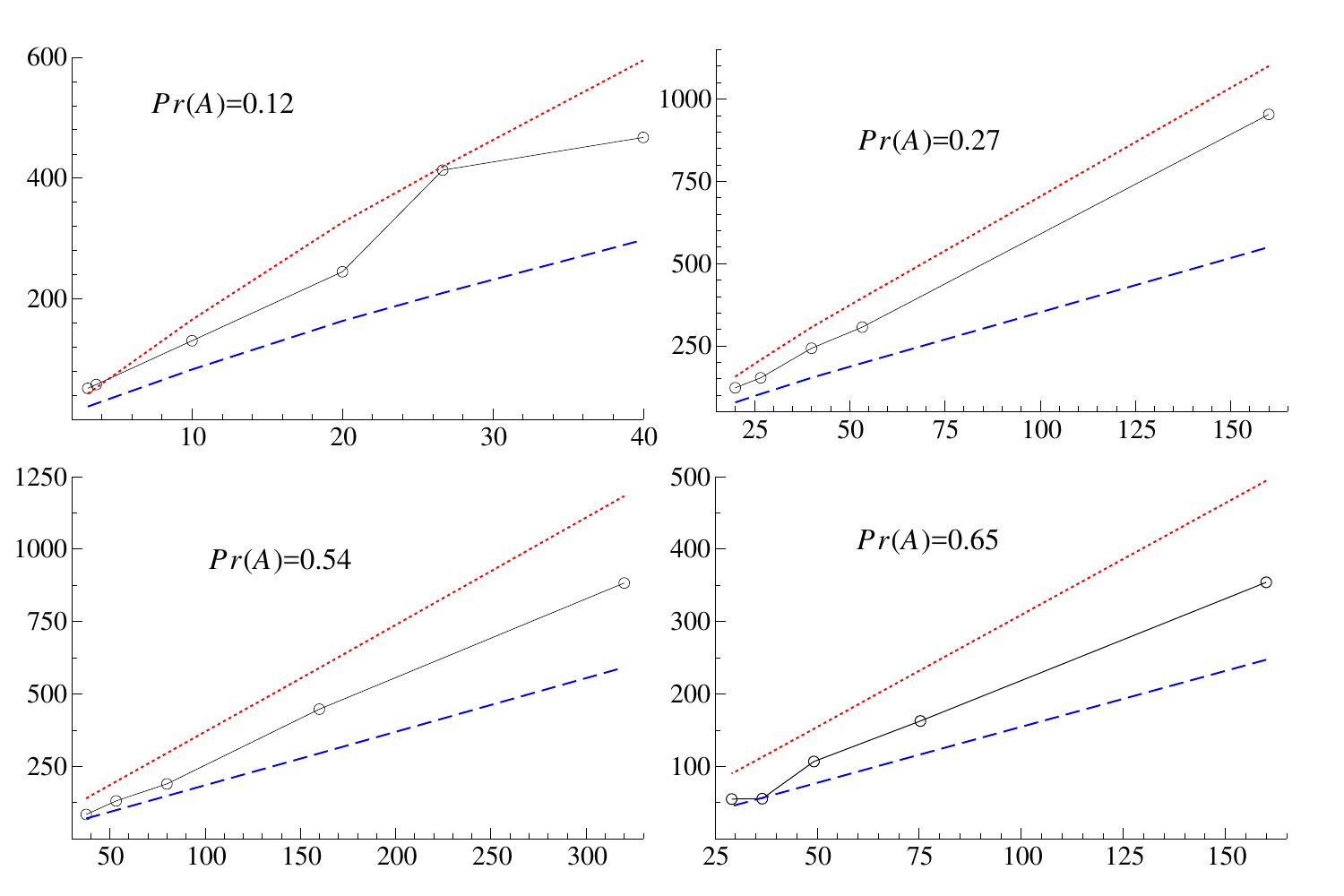} 
\par\end{centering}
\caption{Random effects model using CPM: $T=320$. Inefficiency of the score
error (black line) plotted against $1/\delta$ for four different
values of $\kappa^{2}=9.5,4.9,1.42,0.75$ from top left to bottom
right clockwise and corresponding acceptance probability $\bar{\varrho}_{\mathrm{CPM}}$.
Upper bound $2/(\delta\bar{\varrho}_{\mathrm{CPM}})$ (dotted red)
and lower bound $1/(\delta\bar{\varrho}_{\mathrm{CPM}})$ (dotted
blue).}
\label{fig:IFscoreasafunctionofdelta} 
\end{figure}

To apply the optimization procedure, we first run the algorithm for
$N=20$ and tune $\psi$ to get $\hat{\kappa}\approx1.4$. For the
resulting value $\hat{\psi}$, we then evaluate $\mathrm{CT}_{\mathrm{CPM}}=N\times\mathrm{IF}{}_{\mathrm{CPM}}$
for various values of $\beta$ and perform a regression based on (\ref{eq:CTbeta})-(\ref{eq:estimatebeta}).
Practically, we can only use a subset of the data to perform this
optimization to speed up computation. The results are fairly insensitive
to the size of this subset as illustrated in Figure \ref{fig:limit_delta}
and suggest selecting $\beta$ around 0.25.

\begin{figure}[ptb]
\begin{centering}
\includegraphics[height=2.2in]{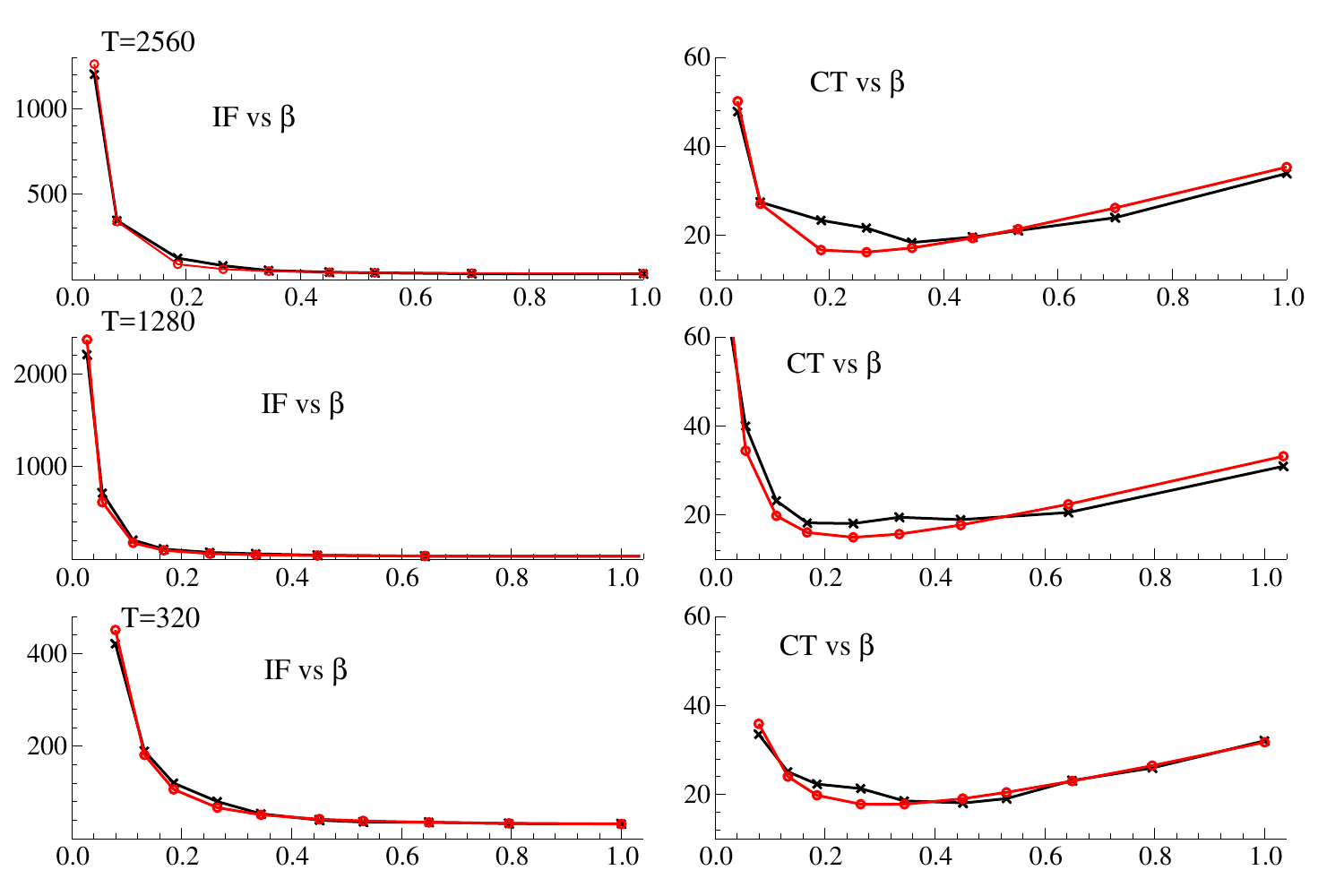} 
\par\end{centering}
\caption{Random effects model. IF and CT\ as a function of $\beta$. Top to
bottom: $T=2560,1280,320$. Left: $\mathrm{IF}=\mathrm{IACT}$ vs
$\beta$. Right: $\mathrm{CT}=\mathrm{IF}\times\beta$ vs $\beta$.
The regression fit based upon estimated CT is included in red.}
\label{fig:limit_delta} 
\end{figure}

\subsection{Heston stochastic volatility model\label{Sec:SVscalar}}

We investigate here the empirical performance of CPM on the Heston
model \citep{Heston1993,chopin2017}, a popular stochastic volatility
model with leverage which is a partially observed diffusion model.
The logarithm of observed price $P\left(t\right)$ evolves according
to 
\begin{align*}
\mathrm{d\log}P(t) & =\sigma(t)\mathrm{d}B(t),\\
\mathrm{d}\sigma^{2}(t) & =\upsilon\left\{ \mu-\sigma^{2}\left(t\right)\right\} \mathrm{d}t+\omega\sigma\left(t\right)\mathrm{d}W(t),
\end{align*}
where $\sigma(t)$ is a stationary latent spot stochastic volatility
process such that $\sigma^{2}\left(t\right)\sim\mathcal{G}\left(\alpha,\beta\right)$
where $\mathcal{G}\left(\alpha,\beta\right)$ is the gamma distribution
of shape $\alpha=2\mu\upsilon/\omega^{2}$ and rate $\beta=2\upsilon/\omega^{2}$.
The Brownian motions $B(t)$ and $W\left(t\right)$ are correlated
with $\chi=$corr$\left\{ B(t),W(t)\right\} $. We shall suppose that
the log prices are observed at equally spaced times $\tau_{0}<\cdots<\tau_{T}$,
where $\triangle=\tau_{s}-\tau_{s-1}$ for all $s$ and we denote
$Y_{s}=\log P(\tau_{s})-\log P(\tau_{s-1})$ for $s=1,...,T$. Conditional
on the volatility and driving processes $\sigma^{2}(t)$ and $W(t)$,
the distribution of these returns is given by
\begin{align}
Y_{s} & \sim\mathcal{N}\left\{ \chi\gamma_{s};(1-\chi^{2})\sigma_{s}^{2\ast}\right\} ,\label{eq:vol_meas}\\
\sigma_{s}^{2\ast} & =\int_{\tau_{s-1}}^{\tau_{s}}\sigma^{2}(t)\mathrm{d}t,\quad\gamma_{s}=\int_{\tau_{s-1}}^{\tau_{s}}\sigma(t)\mathrm{d}W(t).\label{eq:vol_state}
\end{align}
To perform inference, we first reparameterise the model in terms of
$x(t)=\log\sigma^{2}(t)$. We apply Itô's lemma to $x(t)$ and discretize
the resulting diffusion using an Euler scheme. %
We denote by $x_{k}^{s}=x\left(\tau_{s}+\epsilon i\right)$ where
$\epsilon=\triangle/K$ for $i=0,...,I$ so that $x_{I}^{s}=x_{0}^{s+1}$.
The evolution of these latent variables is given by
\[
x_{i+1}^{s}=x_{i}^{s}+\epsilon\left[\upsilon\left\{ \mu e^{-x_{i}^{s}}-1\right\} -\frac{\omega^{2}}{2}e^{-x_{i}^{s}}\right]+\sqrt{\epsilon}\omega e^{-x_{i}^{s}/2}\eta_{i},
\]
where $\eta_{i}\overset{\text{i.i.d.}}{\sim}\mathcal{N}\left(0,1\right)$
for $i=0,...,I-1$. Under the Euler scheme, the distribution of the
returns is given by
\begin{align}
Y_{s} & \sim\mathcal{N}\left\{ \chi\widehat{\gamma}_{s};(1-\chi^{2})\widehat{\sigma}_{s}^{2\ast}\right\} ,\label{eq:obsEuler}\\
\widehat{\sigma}_{s}^{2\ast} & =\epsilon{\textstyle \sum_{i=1}^{I}}\exp(x_{i}^{s}),\quad\widehat{\gamma}_{t}=\sqrt{\epsilon}{\textstyle \sum_{i=1}^{I}}\exp(x_{i}^{s}/2)\eta_{i}.\label{eq:volatilityEuler}
\end{align}
where $\widehat{\sigma}_{t}^{2\ast}$ and $\widehat{\gamma}_{t}$
are the Euler approximations to the corresponding expressions in (\ref{eq:vol_state}).
We are interested in inferring $\theta=(\mu,\upsilon,\omega,\chi)$
given $T=4,000$ daily returns $y_{1:T}$ from the S\&P 500 index
from 15/08/1990 to 03/07/2006%
. We use here $I=10$. Although the state is scalar, it is very difficult
to perform inference using standard MCMC techniques as this involves
$T\times I=40000$ highly correlated latent variates.

We first run the CPM by keeping the parameter fixed at the posterior
mean $\hat{\theta}$, estimated from a full CPM run, and only updating
the auxiliary variables. We display the histograms of $Z=\log\widehat{p}(y_{1:T}\mid\hat{\theta},U),$
$W=\log\widehat{p}(y_{1:T}\mid\hat{\theta},U')$ and $R=\log\{\widehat{p}(y_{1:T}\mid\hat{\theta},U')/\widehat{p}(y_{1:T}\mid\hat{\theta},U)\}$
in Figure \ref{fig:EmpiricalCLTSV} for $N=80$ and $N=300$ using
the parameters given in Table \ref{table_Heston-2}%
. We observe that $R$ is approximately distributed according to $\mathcal{N}(-\kappa^{2}/2,\kappa^{2})$
for $\kappa=1.35$ in both cases. Additionally the sequence of estimates
is almost uncorrelated across CPM iterations.

\begin{figure}[h!]
\centering\includegraphics[height=2.5in]{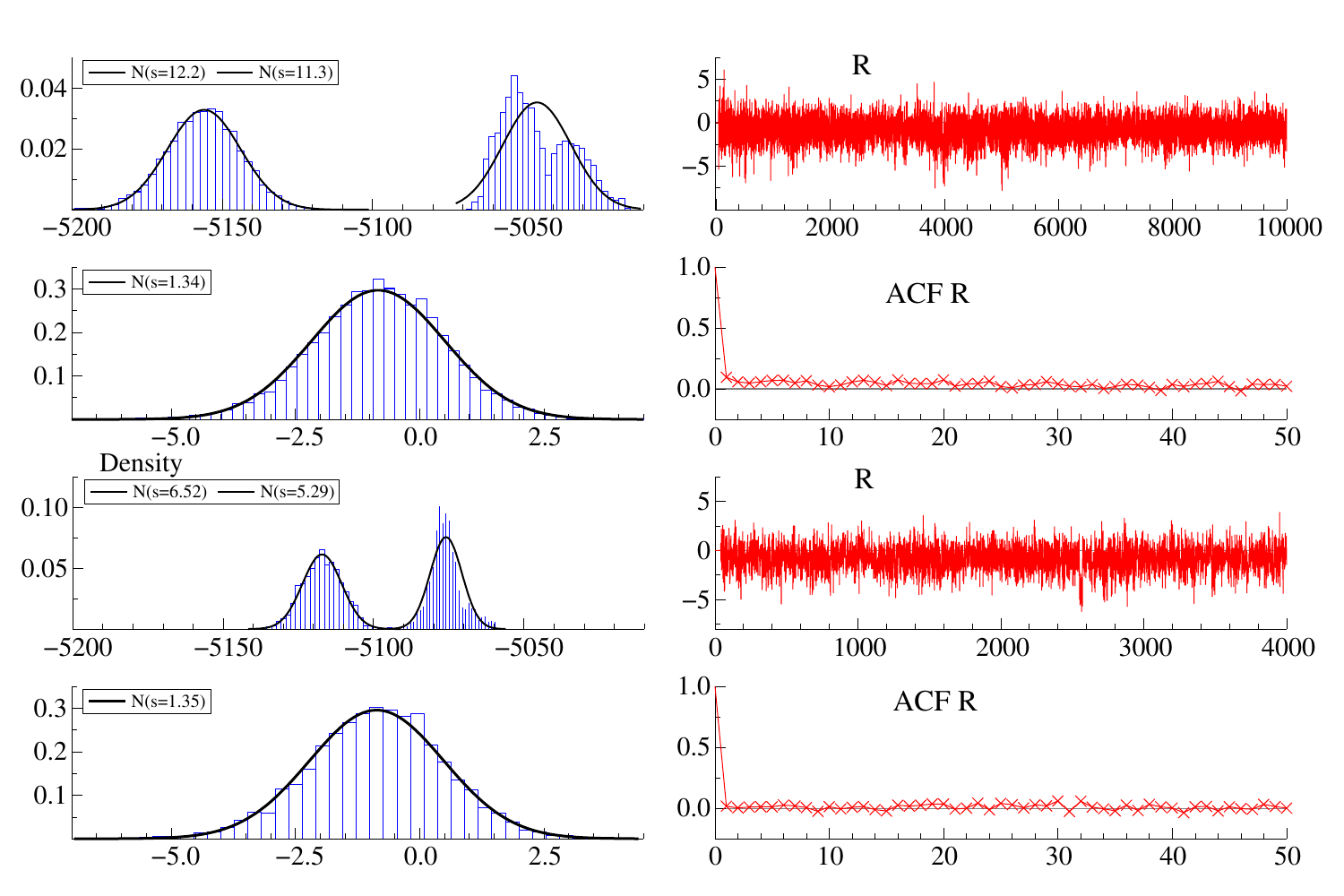}
\caption{Histograms of $W,Z$ for $N=80$ (1st left), $N=300$ (3rd left),
histograms of $R$ for $N=80$ (2nd left), $N=300$ (4th left). $R$
across CPM iterations and associated correlograms for $N=80$ (1st
right, 2nd right), $N=300$ (3rd right, 4th right).}
\label{fig:EmpiricalCLTSV} 
\end{figure}

We then run the CPM using a random walk proposal. Using $N=300$,
we first select $\psi=0.125$ to get the standard deviation of the
loglikelihood ratio estimator at $\hat{\theta}$ around $\kappa=1.4$.
We then run the CPM schemes for other values of $N$, $N=\beta\sqrt{T}$,
and compute $\mathrm{CT}=N\times\mathrm{IF}$. %
These results are summarized in Table \ref{table_Heston-2}. The posterior
estimates are in very close agreement across the different values
of $N$. In unreported results, we observe empirically that the dependence
of $\mathrm{CT}$ on of $\beta$ for parameters $\left(\mu,\phi:=e^{-\upsilon},\omega,\chi\right)$,
matches (\ref{eq:CTbeta}) which can be optimized, suggesting that
an optimal value of $N$ around 70-80. As in the random effect scenario,
we also observe on datasets of increasing length that the scaling
$N=\beta\sqrt{T}$ is successful as $\mathrm{IF}_{\mathrm{CPM}}$
appears to stabilize. In this context, the PM procedure is extremely
expensive computationally as we need approximately $N=20000$ to obtain
a standard deviation of $Z$ around one \citep{doucet2015efficient},
our implementation taking 7 minutes per iteration to run on a standard
desktop. In terms of $\mathrm{CT}$, CPM is approximately 100 times
more efficent than PM.

\begin{table}[h!]
\centering%
\begin{tabular}{llllll}
\hline 
$\mathbb{E}(\theta)$ (SD$(\theta)$ )  & $\mu$  & $\phi$  & $\omega$  & $\chi$  & CPM $\rho$\tabularnewline
\hline 
$N=80$  & 1.258 (0.098)  & 0.981 (0.0027)  & 0.142 (0.0099)  & -0.676 (0.027)  & 0.9975\tabularnewline
$N=150$  & 1.253 (0.098)  & 0.981 (0.0028)  & 0.142 (0.0105)  & -0.672 (0.034)  & 0.9953\tabularnewline
$N=300$  & 1.255 (0.099)  & 0.981 (0.0028)  & 0.142 (0.0110)  & -0.671 (0.032)  & 0.9907\tabularnewline
 &  &  &  &  & \\
\hline 
CT($\theta)$  & $\mu$  & $\phi$  & $\omega$  & $\chi$  & $\bar{\varrho}_{\mathrm{CPM}}$ \tabularnewline
\hline 
$N=80$  & 9,995 & 12,555 & 13,571 & 33,794  & 0.276\tabularnewline
$N=150$  & 19,691  & 20,256  & 17,931  & 32,588  & 0.272\tabularnewline
$N=300$  & 32,970 & 30,432  & 35,103  & 35,505  & 0.281\tabularnewline
\hline 
\end{tabular}\caption{Heston model. Posterior means and standard deviations over 10,000
iterations (top). $\mathrm{CT}=\mathrm{IF}\times N$ for the CPM scheme
for $N=\beta\sqrt{T}$ and $\rho$ selected such that $\kappa\approx1.4$
at $\hat{\theta}$.}
\label{table_Heston-2} 
\end{table}

\subsection{Linear Gaussian state-space model\label{subsec:LinearGaussianSSM}}

We examine empirically the performance of the CPM for multivariate
state-space models using the particle filter with Hilbert sort described
in Algorithm \ref{alg:PFHilbert} and compare it to the PM procedure.
Attention is restricted to a linear Gaussian state-space model which
allows exact calculation of the likelihood and of the loglikelihood
error $Z_{T}\left(\theta,U\right)=\log\widehat{p}(Y_{1:T}\mid\theta,U)-\log p(Y_{1:T}\mid\theta)$.
Very similar empirical results for non-linear non-Gaussian state-space
models were observed.

We consider the model discussed in \citep{AdamAnthony2016,Jacob2016}
where $\left\{ X_{t};t\geq1\right\} $ and $\left\{ Y_{t};t\geq1\right\} $
are $\mathbb{R}^{k}$-valued with
\begin{equation}
X_{1}\sim\mathcal{N}\left(0,I_{n}\right),\text{\hspace{1cm}}\text{ }X_{t+1}=A_{\theta}X_{t}+V_{t+1},\text{\hspace{1cm}}Y_{t}=X_{t}+W_{t},\label{eq:lineargaussianSSM}
\end{equation}
where $V_{t}\overset{\text{i.i.d.}}{\sim}\mathcal{N}\left(0_{k},I_{k}\right),W_{t}\overset{\text{i.i.d.}}{\sim}\mathcal{N}\left(0_{k},I_{k}\right)$
and $A_{\theta}^{i,j}=\theta^{\left|i-j\right|+1}$. %

We use for the proposal density within the particle filter the transition
density of $\left\{ X_{t};t\geq1\right\} $. We first examine the
achieved correlation between successive draws of $Z=\log\left\{ \widehat{p}(y_{1:T}\mid\theta,U)/p(y_{1:T}\mid\theta)\right\} $
by running the CPM procedure holding the parameter fixed and equal
to its true value $\theta=0.4$. This allows for the examination of
the variance of $R=\log\left\{ \widehat{p}(y_{1:T}\mid\theta',U')/p(y_{1:T}\mid\theta)\right\} -Z$
where $U'\sim K_{\rho}\left(U,\cdot\right)$ is the proposal when
$\theta'=\theta$. We examine the results for various values of $T$,
with $N=\left\lceil \beta T^{\alpha}\right\rceil $ and $\rho=\exp$$\left(-\psi N/T\right)$
for $k\in\{2,3,4\}$.

We will now discuss the choice of $\alpha$ for state-space models.
In sharp contrast to random effects models, we found empirically that
there are dimension dependent limitations to the realized correlation
that can be achieved through the particle filter with Hilbert sort.
In particular we found that, due to resampling, the realized correlation
is limited by $\min\left\{ 1-c_{1}N^{-1/k},1-c_{2}\delta\right\} $
for some constants $c_{1},c_{2}$, unless we set $\delta$ extremely
small. A back of the envelope calculation suggests that to overcome
this limitation we would need to choose $\delta$ around $N^{-2/k}$.
Since $\delta\propto N/T$ this would result in choosing $N\propto T^{k/(k+2)}$
in which case empirical results suggest that the number of particles
is too small to control $\kappa^{2}$, the variance of the loglikelihood
ratio error. Since the inefficiency tends to increase if we set $\delta$
too small, the above considerations suggest that we set $\delta=N^{-1/k}$,
or equivalently scaling $N\propto T^{k/(k+1)}$. Thus for the following
examples we set $\alpha=k/k+1$.

We run the simulated chain for 1000 iterations recording $\kappa^{2}=\mathbb{V}\left(R\right)$
and $\sigma^{2}=\mathbb{V}\left(Z\right)$. The values of $\beta$
and $\psi$ have been chosen so that they result in a particular target
value of $\kappa^{2}$ as will be evident from the following tables.
The asymptotic acceptance probability of the CPM scheme is thus in
this case given by$\varrho_{\text{\textsc{CPM}}}\left(\kappa\right):=\varrho_{\text{\textsc{U}}}\left(\kappa\right)=2\Phi\left(-\kappa/2\right)$
while it is $\varrho_{\text{\textsc{PM}}}\left(\sigma\right)=2\Phi\left(-\sigma/\sqrt{2}\right)$
for the PM \citep{doucet2015efficient}.

The results for $k=2$ are reported in Table \ref{table:dim2}, where
the two eigenvalues of $A_{\theta}$ are 0.56 and 0.24. It is clear
that the proposed scaling rules result in values of $\kappa^{2}$
which are approximately constant, remaining at values of around $2$
for $T\geq1600$. The implied acceptance probability of the CPM scheme
$\varrho_{\text{\textsc{CPM}}}\left(\kappa\right)$ therefore settles
at a value just below 0.5. By contrast the marginal variance $\sigma^{2}$
increases at the expected rate $T^{1-\alpha}$ and accordingly the
implied acceptance probability $\varrho_{\text{\textsc{PM}}}\left(\sigma\right)$
deteriorates rapidly, even for $T=100$. Similar results are found
for the case $k=3$, reported in Table \ref{table:dim3}, where the
eigenvalues of $\Lambda$ are ($0.6605$,$0.3360$,$0.2035$) resulting
in a model with moderately high persistence. In this case we set $\alpha=3/4$.
Although less dramatic, the implied gain of the CPM method over the
PM is substantial even for $T=100$ and increases as $T$ goes up.
The variance $\kappa^{2}$ appears again to stabilize at a value less
than $3$. 

\begin{table}[h!]
\centering %
\begin{tabular}{lllllll}
\hline 
\multicolumn{7}{l}{State dimension $k=2$ with $\beta=0.854,\thinspace\psi=0.12,\thinspace\alpha=2/3$}\tabularnewline
\hline 
$T$  & $N$  & $\delta=-\log\thinspace\rho$  & $\kappa^{2}$  & $\sigma^{2}$  & $\varrho_{\text{\textsc{CPM}}}\left(\kappa\right)$  & $\varrho_{\text{\textsc{PM}}}\left(\sigma\right)$\tabularnewline
\hline 
$100$  & $18$  & $0.0216$  & $2.59$  & $16.3$  & $0.42$  & $0.004$\tabularnewline
$400$  & $46$  & $0.0138$  & $2.71$  & $20.5$  & $0.41$  & $0.0013$\tabularnewline
$1600$  & $116$  & $0.0087$  & $2.01$  & $34.1$  & $0.48$  & $3.6\times10^{-5}$\tabularnewline
$6400$  & $294$  & $0.0055$  & $2.07$  & $49.7$  & $0.47$  & $6.0\times10^{-7}$\tabularnewline
$25600$  & $742$  & $0.0034$  & $1.97$  & $105.9$  & $0.48$  & $3.4\times10^{-13}$\tabularnewline
\hline 
\end{tabular}\caption{Linear state-space model. Results for $k=2$ for varying $T$.}
\label{table:dim2}
\end{table}
\begin{table}[h!]
\centering %
\begin{tabular}{lllllll}
\hline 
\multicolumn{7}{l}{State dimension $k=3$ with $\beta=1.57,\thinspace\psi=0.042,\thinspace$
$\alpha=3/4$}\tabularnewline
\hline 
$T$  & $N$  & $\delta=-\log\thinspace\rho$  & $\kappa^{2}$  & $\sigma^{2}$  & $\varrho_{\text{\textsc{CPM}}}\left(\kappa\right)$  & $\varrho_{\text{\textsc{PM}}}\left(\sigma\right)$\tabularnewline
\hline 
$100$  & $49$  & $0.0205$  & $3.15$  & $13.7$  & $0.37$  & $0.0089$\tabularnewline
$400$  & $140$  & $0.0147$  & $2.97$  & $16.6$  & $0.39$  & $0.0039$\tabularnewline
$1600$  & $397$  & $0.0104$  & $3.44$  & $26.7$  & $0.35$  & $0.00025$\tabularnewline
$6400$  & $1124$  & $0.0074$  & $3.03$  & $34.1$  & $0.38$  & $3.66\times10^{-5}$\tabularnewline
$25600$  & $3181$  & $0.0052$  & $2.69$  & $49.4$  & $0.41$  & $6.74\times10^{-7}$\tabularnewline
\hline 
\end{tabular}\caption{Linear state-space model. Results for $k=3$ for varying $T$. }
\label{table:dim3}
\end{table}
%

The full CPM procedure is now implemented for $T=400$ and $T=6400$
when $k=2$ and $k=3$ using the parameters of Tables \ref{table:dim2}
and \ref{table:dim3}. An autoregressive proposal in the Metropolis
algorithm is employed for $\theta$ which is based on the posterior
mode and the second derivative at this point \citep{Tran2016}.%
{} %

The results for $k=3$ and $T=6400$ are shown in Figure \ref{fig:CPMd3T6400}.
The mixing for $\theta$ is fairly rapid from the achieved value of
$\kappa=2.26$. The empirical distributions of $Z$ under $m$ and
$\bar{\pi}$ are plotted (middle left) and are close to the theoretical
distributions $\mathcal{N}\left(-\sigma^{2}/2,\sigma^{2}\right)$
and $\mathcal{N}\left(\sigma^{2}/2,\sigma^{2}\right)$ respectively,
where $\sigma=7.5$. The middle right plot and the third row show
the draws of $R$, the empirical distribution and the associated correlogram
arising from the CPM scheme. It is clear that $R$ is approximately
distributed according to $\mathcal{N}\left(-\kappa^{2}/2,\kappa^{2}\right)$
for some $\kappa$, which is overlaid, but the correlations across
iterations vanish slowler than for random effect models and one-dimensional
state-space models. The gain over the PM method is around $\sigma^{2}$
meaning we need around 50 times as many particles in the PM method
to achieve similar results to the CPM scheme. When $T=400,$ we obtained
$\kappa=1.92$ and $\sigma=4.30$ resulting in gains over the PM of
approximately 18 fold. When $k=2$, the gains are more impressive
and are around $25$ fold for $T=400$ and $80$ fold when $T=6400$.
\begin{figure}[ptb]
\centering \includegraphics[width=3.5558in,height=2.3728in]{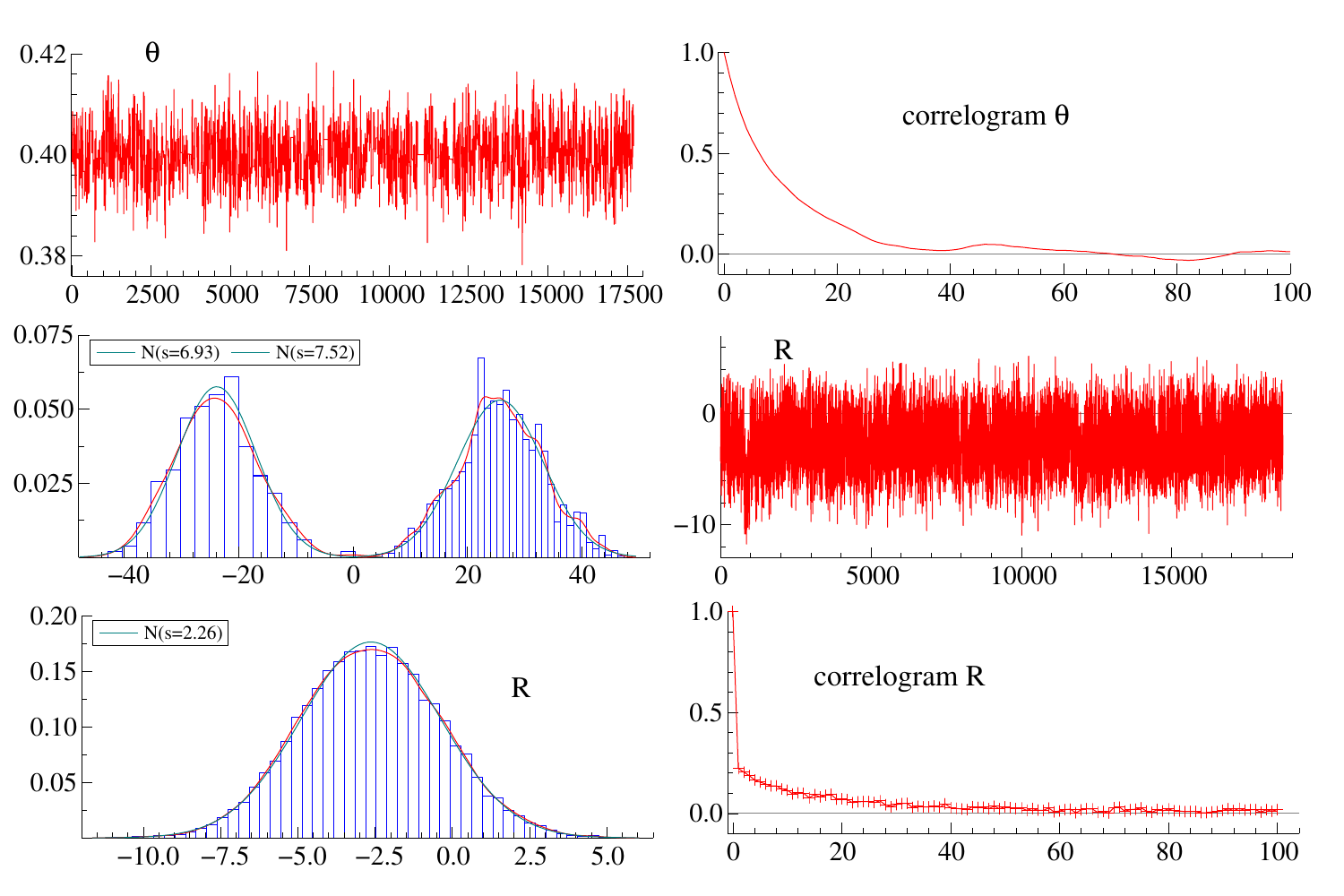}\caption{The CPM\ results for the 3-dimensional state space model with $T=6400$.
Top: parameter samples (left) and corresponding correlogram (right).
Middle: Histograms of $Z$ arising from $m$ and $\bar{\pi}$ (left),
draws of $R$ (right). Bottom: Histogram of $R$ (left) and correlogram
(right).}
\label{fig:CPMd3T6400}
\end{figure}

\section{Discussion\label{sec:discussion}}

The CPM method is a generic extension of the PM\ method based on
an estimator of the likelihood ratio appearing in its acceptance probability
which is obtained by correlating the estimators of its numerator and
denominator. We have detailed two implementations of this idea for
random effects and state-space models. For random effects models,
we have provided theory to perform an efficient implementation of
the methodology and have verified empirically that this methodology
is also useful for state-space models. In our examples, the efficiency
of computations using the CPM relative to the PM method increases
with $T$ and is improved by more than two order orders of magnitude
for large data sets. This methodology is particularly beneficial in
scenarios such as partially observed diffusions where sophisticated
MCMC alternatives, such as particle Gibbs techniques, are inefficient.

From a theoretical point of view, we have obtained for random effects
models a result suggesting that a necessary condition to ensure finiteness
of the IACT of the CPM parameter sequence as $T$ increases is to
have $N_{T}$ growing at least at rate $\sqrt{T}$. Our experimental
results suggest that this condition is also sufficient for a large
class of functions and thus that the computational complexity of CPM
for random effects models is $O(T^{\frac{3}{2}})$ versus $O(T^{2})$
for PM. For state-space models, our empirical results indicate that
this scaling degrades with the state dimension $k$ and that we need
$N_{T}$ to grow at rate $T^{\frac{k}{k+1}}$, suggesting that the
computational complexity of CPM in this context is thus $O(T^{\frac{2k+1}{k+1}})$
up to a logarithmic factor\footnote{The particle filter with Hilbert sort has computational complexity
$N_{T}$log$N_{T}$ per observation. } versus $O(T^{2})$ for PM. It would be of interest but technically
very involved to establish these results rigorously.

From a methodological point of view, it is possible in the state-space
context to use alternatives to the Hilbert resampling sort to implement
the CPM algorithm \citep[Section 6]{MalikPitt2011}, \citep{Lecuyer2016}
and several such methods have been proposed (for example \citep{Jacob2016},
\citep{SenThieryJasra2016}) following the first version of this work
(arXiv:1511.04992). Our empirical results suggest that all these procedures
provide roughly similar improvements over the PM. It could also be
beneficial to use the sequential randomised Quasi Monte Carlo (QMC)
algorithm proposed in \citep{gerber2015}, \citep{chopin2017} within
the CPM scheme by correlating the single uniform used to randomize
the QMC grid. In a random effects context, it has already be demonstrated
that this can provide significant improvements \citep{Tran2016b}.
Finally, a sequential extension of the particle marginal Metropolis\textendash Hastings
algorithm \citep{andrieu:doucet:holenstein2010}, a PM method, has
been proposed in \citep{chopin2013} and it would be interesting to
develop an efficient sequential version of CPM.

\section{Acknowledgments}

Arnaud Doucet's research is partially supported by the Engineering
and Physical Sciences Research Council, grant EP/K000276/1. We thank
Sebastian Schmon for his comments.

\newpage 

\appendix

\section{Supplementary Material}

\subsection{Notation\label{Appendix:notation}}

We define a reference probability space $(\Omega_{T},\mathcal{G}_{T},\mathbb{P}_{T})$
which supports the following random variables: 
\begin{enumerate}
\item $\theta^{T}\sim\pi_{T}$ where $\pi_{T}$ denote the posterior distribution
associated to observations $y_{1:T},$ 
\item $\{U_{t,i}^{T}:t\in1:T,i\in1:N\}$ independent and identically distributed
$\mathcal{N}(0_{p},I_{p})$ random variables, 
\item $\{B_{t,i}^{T}(\cdot):t\in1:T,i\in1:N\}$ where the $B_{t,i}^{T}(\cdot)$
are mutually independent, $p-$dimensional standard Brownian motions. 
\end{enumerate}
We set 
\[
\Omega_{T}:=\Theta\times\mathbb{R}^{pNT}\times C^{p}[0,\infty)^{NT},
\]
and 
\[
\mathbb{P}_{T}\left(\mathrm{d}\theta^{T},\left\{ \mathrm{d}u_{t,i}^{T}\right\} _{t,i},\left\{ \mathrm{d}\left(\beta_{t,i}^{T}(\cdot)\right)\right\} _{t,i}\right)=\pi_{T}(\mathrm{d}\theta^{T})\prod_{t=1}^{T}\prod_{i=1}^{N}\varphi(\mathrm{d}u_{t,i}^{T};0_{p},I_{p})\prod_{t=1}^{T}\prod_{i=1}^{N}\mathbb{W}^{p}\left(\mathrm{d}\beta_{t,i}^{T}(\cdot)\right),
\]
where $\mathbb{W}^{p}(\mathrm{d}\cdot)$ denotes the Wiener measure
on $C^{p}[0,\infty)$, $C^{p}[0,\infty)$ being the space of $\mathbb{R}^{p}$-
valued continuous paths on $[0,\infty)$.

Let $\mathfrak{Y=\otimes}_{t=1}^{\infty}\mathsf{Y}$ where $\mathsf{Y\ }$is
a topological space, $y_{t}$ is $\mathsf{Y}$-valued and $\mathcal{B\mathfrak{(Y)}}$
the associated Borel $\sigma$-algebra. Then we consider the product
space $(\Omega,\mathcal{G},\mathbb{P})$ where 
\[
\Omega=\mathfrak{Y}\times\prod_{T}\Omega_{T},\text{ }\mathcal{G}=B\mathfrak{(Y)}\otimes\left(\otimes_{T}\mathcal{G}_{T}\right),
\]
and 
\[
\mathbb{P}=\left(\prod_{t=1}^{\infty}\mu(\mathrm{d}y_{t})\right)\otimes\left(\otimes_{T}\mathbb{P}_{T}\right).
\]

In most cases we will be working with the probability measure $\widetilde{\mathbb{P}}$
capturing the scenario when the CPM\ algorithm is in the stationary
regime, which is defined as follows. For every $T\geq1$ and sequence
of observations $y_{1:T}$, we define the probability measure $\widetilde{\mathbb{P}}_{T}^{y_{1:T}}$
by 
\[
\frac{\mathrm{d}\widetilde{\mathbb{P}}_{T}^{y_{1:T}}}{\mathrm{d}\mathbb{P}_{T}}\left(\theta^{T},\left\{ u_{t,i}^{T}\right\} _{i,t},\left\{ \beta_{t,i}^{T}(\cdot)\right\} _{t,i}\right)=\prod_{t=1}^{T}\frac{1}{N}\sum_{i=1}^{N}\varpi(y_{t};\theta^{T},u_{t,i}^{T}),
\]
and let 
\[
\widetilde{\mathbb{P}}=\left(\prod_{t=1}^{\infty}\mu(\mathrm{d}y_{t})\right)\otimes\left(\otimes_{T}\widetilde{\mathbb{P}}_{T}^{y_{1:T}}\right).
\]
We will denote by $\mathbb{E}$, $\mathbb{V}$ and $\widetilde{\mathbb{E}}$,
$\widetilde{\mathbb{V}}$ the expectation and variance under $\mathbb{P}$
and $\widetilde{\mathbb{P}}$ respectively.

When $T$ and $\theta^{T}$ are understood fixed, allowing some abuse
of notation, we will write $\mathbb{P}$ to denote the measure 
\[
\mathbb{P}\left(\mathrm{d}y_{1},\dots,\mathrm{d}y_{T},\left\{ \mathrm{d}u_{t,i}^{T}\right\} _{t,i},\left\{ \mathrm{d}\left(\beta_{t,i}^{T}(\cdot)\right)\right\} _{t,i}\right)=\prod_{t=1}^{T}\mu(\mathrm{d}y_{t})\prod_{t=1}^{T}\prod_{i=1}^{N}\varphi(\mathrm{d}u_{t,i}^{T};0_{p},I_{p})\prod_{t=1}^{T}\prod_{i=1}^{N}\mathbb{W}^{p}\left(\mathrm{d}\beta_{t,i}^{T}(\cdot)\right),
\]
and similarly 
\begin{align*}
 & \widetilde{\mathbb{P}}\left(\mathrm{d}y_{1},\dots,\mathrm{d}y_{T},\left\{ \mathrm{d}u_{t,i}^{T}\right\} _{t,i},\left\{ \mathrm{d}\left(\beta_{t,i}^{T}(\cdot)\right)\right\} _{t,i}\right)\\
 & =\bar{\pi}_{T}\left(\{\mathrm{d}u_{t,i}^{T}\}\mid\theta^{T}\right)\prod_{t=1}^{T}\mu(\mathrm{d}y_{t})\prod_{t=1}^{T}\prod_{i=1}^{N}\mathbb{W}^{p}\left(\mathrm{d}\beta_{t,i}^{T}(\cdot)\right)\\
 & =\prod_{t=1}^{T}\frac{1}{N}\sum_{i=1}^{N}\varpi(y_{t};\theta^{T},u_{t,i}^{T})\prod_{t=1}^{T}\prod_{i=1}^{N}\varphi(\mathrm{d}u_{t,i}^{T};0_{p},I_{p})\prod_{t=1}^{T}\mu(\mathrm{d}y_{t})\prod_{t=1}^{T}\prod_{i=1}^{N}\mathbb{W}^{p}\left(\mathrm{d}\beta_{t,i}^{T}(\cdot)\right).
\end{align*}
To ease notation, we will often drop the superscript $T$, since we
will always be considering variables belonging to the same row. In
addition we will write $N$ for $N_{T}$ in the proofs, omitting the
explicit dependence of $N_{T}$ on $T$. In the proofs of Theorem
\ref{Theorem:CLTmarginal}, Theorem \ref{Theorem:CLTlikelihoodratiostandardpseudomarginal}
and Theorem \ref{Theorem:conditionalCLTthetathetacand}, we also write
$m,$ $\bar{\pi}\left(\mathrm{d}u\mid\theta\right)$, $B_{t,i}$,
$U_{t,i}$ instead of $m_{T},$ $\bar{\pi}_{T}\left(\mathrm{d}u^{T}\mid\theta^{T}\right)$,
$B_{t,i}^{T}$ and $U_{t,i}^{T}$. Notice that $\mathbb{E}\left(\varpi\left(Y_{1},U_{1,1}^{T};\theta\right)^{j}\right)$
is independent of $T$ for any $j$ as $U_{1,1}^{T}\sim\mathcal{N}(0_{p},I_{p})$
under $\mathbb{P}$.

\subsection{Proof of Part 1 of Theorem \ref{Theorem:CLTmarginal}}

The starting point of our analysis is the following decomposition
\begin{align}
\log\widehat{p}\left(\left.Y_{1:T}\right\vert \theta\right)-\log p\left(\left.Y_{1:T}\right\vert \theta\right) & =\sum_{t=1}^{T}\log\left\{ 1+\frac{\varepsilon_{N}(Y_{t};\theta)}{\sqrt{N}}\right\} \label{eq:loglikelihooderror}
\end{align}
with 
\[
\varepsilon_{N}(Y_{t};\theta):=\sqrt{N}\frac{\widehat{p}\left(\left.Y_{t}\right\vert \theta\right)-p\left(\left.Y_{t}\right\vert \theta\right)}{p\left(\left.Y_{t}\right\vert \theta\right)}=\frac{1}{\sqrt{N}}\sum_{i=1}^{N}\left\{ \varpi(Y_{t},U_{t,i};\theta)-1\right\} .
\]
We will denote by $\rho_{i}\left(\theta\right)$ the $i^{\text{th}}$
order cumulant of the normalized importance weight $\varpi(Y_{1},U_{1,1};\theta)$
given in (\ref{eq:normalisedISweight}) under $\mathbb{P}$ and by
$\gamma\left(\theta\right)^{2}$ its variance, so that $\rho_{2}\left(\theta\right)=\gamma\left(\theta\right)^{2}$.

We first establish three preliminary lemmas. 
\begin{lem}
\label{Lemma:momentsofepsilonunderproposal}The terms $\left\{ \varepsilon_{N}\left(y_{t};\theta\right)\right\} _{t=1}^{T}$
are independent with for any $y\in\mathsf{Y}$, $\ \mathbb{E}\left(\varepsilon_{N}\left(y;\theta\right)\right)=0$
and 
\begin{align}
\mathbb{E}\left(\varepsilon_{N}\left(y;\theta\right)^{2}\right) & =\mathbb{V}\left(\varpi(y,U_{1,1};\theta)\right)=\rho_{2}\left(y;\theta\right):=\gamma\left(y;\theta\right)^{2},\label{eq:secondordermoment}\\
\mathbb{E}\left(\varepsilon_{N}\left(y;\theta\right)^{3}\right) & =\frac{\rho_{3}\left(y;\theta\right)}{\sqrt{N}},\label{eq:thirdordermoment}\\
\mathbb{E}\left(\varepsilon_{N}\left(y;\theta\right)^{4}\right) & =3\gamma\left(y;\theta\right)^{4}+\frac{\rho_{4}\left(y;\theta\right)}{N},\label{eq:fourthordermoment}\\
\mathbb{E}\left(\varepsilon_{N}\left(y;\theta\right)^{5}\right) & =\frac{10\rho_{2}\left(y;\theta\right)\rho_{3}\left(y;\theta\right)}{\sqrt{N}}+\frac{\rho_{5}\left(y;\theta\right)}{N\sqrt{N}},\nonumber 
\end{align}
where $\rho_{i}\left(y;\theta\right)$ denotes the $i$th-order cumulant
of $\varpi(y,U_{1,1};\theta)$ and $\rho_{i}\left(\theta\right)=\mathbb{E}\left(\rho_{i}(Y;\theta)\right)=\mathbb{\int}\rho_{i}\left(y;\theta\right)\mu\left(\mathrm{d}y\right)$. 
\end{lem}
The proof of Lemma \ref{Lemma:momentsofepsilonunderproposal} follows
from direct calculations so it is omitted. 
\begin{lem}
\label{Lemma:Marcinkiewicz}For any $k\geq2,$ if $\mathbb{E}\left(\varpi(Y_{1},U_{1,1};\theta)^{k}\right)<\infty$
then $\lim\sup_{T\rightarrow\infty}\mathbb{E}\left(\left\vert \varepsilon_{N}(Y_{1};\theta)\right\vert ^{k}\right)<\infty$. 
\end{lem}
\begin{proof}[Proof of Lemma~\ref{Lemma:Marcinkiewicz}]
\textit{ }It follows from a successive application of Marcinkiewicz-Zygmund,
Jensen and C$_{p}$ inequalities that for any $k\geq2$, there exist
$b\left(k\right),c\left(k\right)<\infty$ such that 
\begin{align*}
\mathbb{E}\left(\left\vert \varepsilon_{N}(Y_{1};\theta)\right\vert ^{k}\right) & =\mathbb{E}\left(\left\vert \frac{1}{\sqrt{N}}\sum_{i=1}^{N}\left\{ \varpi(Y_{1},U_{1,i};\theta)-1\right\} \right\vert ^{k}\right)\\
 & \leq b\left(k\right)\mathbb{E}\left(\left\vert \frac{1}{N}\sum_{i=1}^{N}\left\{ \varpi(Y_{1},U_{1,i};\theta)-1\right\} ^{2}\right\vert ^{k/2}\right)\\
 & \leq b\left(k\right)\frac{1}{N}\sum_{i=1}^{N}\mathbb{E}\left(\left\vert \varpi(Y_{1},U_{1,i};\theta)-1\right\vert ^{k}\right)\\
 & =b\left(k\right)\left(\mathbb{E}\left\vert \varpi(Y_{1},U_{1,1};\theta)-1\right\vert ^{k}\right)\\
 & \leq b\left(k\right)c\left(k\right)\left(\mathbb{E}\left(\varpi(Y_{1},U_{1,1};\theta)^{k}\right)+1\right).
\end{align*}

This concludes the proof. 
\end{proof}
\begin{lem}
\label{Lemma:WLLN}Consider the triangular array $\left\{ \varepsilon_{N}\left(Y_{t};\theta\right)\right\} $
and let $k\geq2$. If there exists $\delta>0$ such that\linebreak{}
 $\mathbb{E}\left(\varpi(Y_{1},U_{1,1};\theta)^{k+\delta}\right)<\infty$
then 
\begin{equation}
T^{-1}\sum_{t=1}^{T}\varepsilon_{N}(Y_{t};\theta)^{k}-\mathbb{E}\left(\varepsilon_{N}(Y_{1};\theta)^{k}\right)\overset{\mathbb{P}}{\rightarrow}0.\label{eq:WLLN1}
\end{equation}
If $\mathbb{E}\left(\varpi(Y_{1},U_{1,1};\theta)^{2k}\right)$ then
we have for any $\lambda>0$ 
\begin{equation}
T^{-\frac{(1+\lambda)}{2}}\sum_{t=1}^{T}\varepsilon_{N}(Y_{t};\theta)^{k}-T^{\frac{(1-\lambda)}{2}}\mathbb{E}\left(\varepsilon_{N}(Y_{1};\theta)^{k}\right)\overset{\mathbb{P}}{\rightarrow}0.\label{eq:WLNNscale1}
\end{equation}
\end{lem}
\begin{proof}[Proof of Lemma~\ref{Lemma:WLLN}]
\textit{ }The results (\ref{eq:WLLN1}) follows directly from a weak
law of large numbers (WLLN)\ applied to the triangular array $\varepsilon_{N}(Y_{t};\theta)^{k}-\mathbb{E}\left(\varepsilon_{N}(Y_{1};\theta)^{k}\right)$;
see, e.g., \citep[Theorem B.18]{doucmoulinesstoffer2014}. This results
holds as $\mathbb{E}\left(\varpi(Y_{1},U_{1,1};\theta)^{k+\delta}\right)<\infty$
so $\lim\sup_{T\rightarrow\infty}\mathbb{E}\left(\left\vert \varepsilon_{N}(Y_{1};\theta)\right\vert ^{k+\delta}\right)<\infty$
from Lemma \ref{Lemma:Marcinkiewicz}. For the second result (\ref{eq:WLNNscale1}),
we have for any $\epsilon>0$ 
\begin{align*}
\mathbb{P}\left\{ \left\vert T^{-\frac{(1+\lambda)}{2}}\sum_{t=1}^{T}\left\{ \varepsilon_{N}(Y_{t};\theta)^{k}-\mathbb{E}\left(\varepsilon_{N}(Y_{1};\theta)^{k}\right)\right\} \right\vert \geq\epsilon\right\}  & \leq\frac{\mathbb{E}\left[\left(\sum_{t=1}^{T}\left[\varepsilon_{N}(Y_{t};\theta)^{k}-\mathbb{E}\left(\varepsilon_{N}(Y_{1};\theta)^{k}\right)\right]\right)^{2}\right]}{T^{(1+\lambda)}\epsilon^{2}}\\
 & =\frac{\mathbb{E}\left(\left[\varepsilon_{N}(Y_{1};\theta)^{k}-\mathbb{E}\left(\varepsilon_{N}(Y_{1};\theta)^{k}\right)\right]^{2}\right)}{T^{\lambda}\epsilon^{2}}\\
 & \rightarrow0.
\end{align*}
The result follows. 
\end{proof}
We can now give the proof of Theorem \ref{Theorem:CLTmarginal} Part
1. 
\begin{proof}[Proof of Part 1 of Theorem ~\ref{Theorem:CLTmarginal}]
\textit{ }We first perform a fourth order Taylor expansion of each
term appearing in (\ref{eq:loglikelihooderror}), i.e. 
\begin{equation}
\log\left\{ 1+\frac{\varepsilon_{N}(Y_{t};\theta)}{\sqrt{N}}\right\} =\frac{\varepsilon_{N}(Y_{t};\theta)}{\sqrt{N}}-\frac{\varepsilon_{N}(Y_{t};\theta)^{2}}{2N}+\frac{\varepsilon_{N}(Y_{t};\theta)^{3}}{3N\sqrt{N}}-\frac{\varepsilon_{N}(Y_{t};\theta)^{4}}{4N^{2}}+R_{t,N}(Y_{t};\theta)\label{eq:Taylor}
\end{equation}
where 
\begin{equation}
R_{t,N}(Y_{t};\theta)=\frac{1}{5}\frac{1}{\left(1+\xi_{N}(Y_{t};\theta)\right)^{5}}\left\{ \frac{\varepsilon_{N}(Y_{t};\theta)}{\sqrt{N}}\right\} ^{5}\label{eq:remainderexpression}
\end{equation}
with $\left\vert \xi_{N}(Y_{t};\theta)\right\vert \leq\left\vert \frac{\varepsilon_{N}\left(Y_{t};\theta\right)}{\sqrt{N}}\right\vert .$
We need to ensure that these Taylor expansions are valid for $t\in1:T$
so we control the probability of the event $B\left(Y^{T},\epsilon\right)=\left\{ \underset{t\leq T}{\max}\left\vert \frac{\varepsilon_{N}\left(Y_{t};\theta\right)}{\sqrt{N}}\right\vert >\epsilon\right\} $.
We have for any $\epsilon>0$ 
\begin{align*}
\mathbb{P}\left\{ B\left(Y^{T},\epsilon\right)\right\}  & \leq\sum_{t=1}^{T}\mathbb{P}\left(\left\vert \frac{\varepsilon_{N}\left(Y_{t};\theta\right)}{\sqrt{N}}\right\vert >\epsilon\right)\\
 & =T\mathbb{P}\left(\left\vert \frac{\varepsilon_{N}\left(Y_{1};\theta\right)}{\sqrt{N}}\right\vert >\epsilon\right)\\
 & \leq T\frac{\mathbb{E}\left(\varepsilon_{N}\left(Y_{1};\theta\right)^{8}\right)}{\epsilon^{8}N^{4}}\\
 & \leq\frac{\mathbb{E}\left(\varepsilon_{N}\left(Y_{1};\theta\right)^{8}\right)}{\epsilon^{8}\beta^{4}T^{4\alpha-1}}\text{.}
\end{align*}
As $\mathbb{E}\left(\varpi\left(Y_{1},U_{1,1}^{T};\theta\right)^{8}\right)<\infty$
under assumption, the complementary event satisfies for $\alpha>1/4$
\begin{equation}
\underset{T\rightarrow\infty}{\lim}\mathbb{P}\left(\left(B\left(Y^{T},\epsilon\right)\right)^{\mathtt{C}}\right)=1.\label{eq:probaTaylorvalid}
\end{equation}
On the event $\left(B\left(Y^{T},\epsilon\right)\right)^{\mathtt{C}}$,
the Taylor expansion (\ref{eq:Taylor}) holds for all $t\in1:T$ so
we can write 
\begin{align}
\frac{\log\widehat{p}\left(\left.Y_{1:T}\right\vert \theta\right)-\log p\left(\left.Y_{1:T}\right\vert \theta\right)}{T^{\left(1-\alpha\right)/2}}= & \frac{1}{\beta^{1/2}T^{1/2}}\sum_{t=1}^{T}\varepsilon_{N}\left(Y_{t};\theta\right)\label{eq:firstterm}\\
 & -\frac{1}{2\beta T^{\left(1+\alpha\right)/2}}\sum_{t=1}^{T}\varepsilon_{N}\left(Y_{t};\theta\right)^{2}\label{eq:secondterm}\\
 & +\frac{1}{3\beta^{3/2}T^{\left(1+2\alpha\right)/2}}\sum_{t=1}^{T}\varepsilon_{N}\left(Y_{t};\theta\right)^{3}\label{eq:thirdterm}\\
 & -\frac{1}{4\beta^{2}T^{\left(1+3\alpha\right)/2}}\sum_{t=1}^{T}\varepsilon_{N}\left(Y_{t};\theta\right)^{4}\label{eq:fourthterm}\\
 & +\frac{1}{T^{\left(1-\alpha\right)/2}}\sum_{t=1}^{T}R_{t,N}\left(Y_{t};\theta\right)\label{eq:remainder}\\
 & +o_{\mathbb{P}}\left(1\right)\nonumber 
\end{align}
where the $o_{\mathbb{P}}\left(1\right)$ arises from substituting
$\beta T^{\alpha}$ to $N=\left\lceil \beta T^{\alpha}\right\rceil $.

We first control the remainder (\ref{eq:remainder}), using the fact
that (\ref{eq:remainderexpression}) can be controlled on the event
$B^{\mathtt{C}}\left(Y^{T},\epsilon\right)$, as follows 
\begin{align*}
\frac{1}{T^{\left(1-\alpha\right)/2}}\left\vert \sum_{t=1}^{T}R_{t,N}\left(Y_{t};\theta\right)\right\vert  & \leq\frac{1}{5\beta^{5/2}}\frac{1}{\left(1-\epsilon\right)^{5}}\frac{1}{T^{\left(1-\alpha\right)/2}N^{5/2}}\sum_{t=1}^{T}\left\vert \varepsilon_{N}\left(Y_{t};\theta\right)\right\vert ^{5}\\
 & \leq\frac{1}{5\beta^{5/2}}\frac{1}{\left(1-\epsilon\right)^{5}}\frac{1}{T^{\left(4\alpha-1\right)/2}}\frac{1}{T}\sum_{t=1}^{T}\left\vert \varepsilon_{N}\left(Y_{t};\theta\right)\right\vert ^{5}.
\end{align*}

The WLLN for triangular arrays holds by a similar argument to Lemma
\ref{Lemma:WLLN} so we have 
\[
\frac{1}{T}\sum_{t=1}^{T}\left\vert \varepsilon_{N}\left(Y_{t};\theta\right)\right\vert ^{5}-\mathbb{E}\left(\left\vert \varepsilon_{N}\left(Y_{1};\theta\right)\right\vert ^{5}\right)\overset{\mathbb{P}}{\rightarrow}0.
\]
Hence as $\alpha>1/4$, we have 
\begin{equation}
\frac{1}{T^{\left(1-\alpha\right)/2}}\left\vert \sum_{t=1}^{T}R_{t,N}\left(Y_{t};\theta\right)\right\vert \overset{\mathbb{P}}{\rightarrow}0.\label{eq:remaindergoestozero}
\end{equation}

The term on the r.h.s. of (\ref{eq:firstterm}) satisfies a conditional
CLT\ for triangular arrays; see Lemma \ref{lem:conditionalLindebergCLT}
in Section \ref{Section:Conditionalweakconvergence}. Indeed, we have
for any $\epsilon>0$ 
\begin{align}
\mathbb{E}\left[T^{-1}\sum_{t=1}^{T}\mathbb{E}\left(\left.\varepsilon_{N}\left(Y_{t};\theta\right)^{2}\mathbb{I}_{\left\{ \left\vert \varepsilon_{N}\left(Y_{t};\theta\right)\right\vert \geq\sqrt{T}\epsilon\right\} }\right\vert \mathcal{Y}^{T}\right)\right] & =\mathbb{E}\left[\epsilon^{2}\sum_{t=1}^{T}\mathbb{E}\left(\left.\frac{\varepsilon_{N}\left(Y_{t};\theta\right)^{2}}{\epsilon^{2}T}\mathbb{I}_{\left\{ \left\vert \varepsilon_{N}\left(Y_{t};\theta\right)\right\vert \geq\sqrt{T}\epsilon\right\} }\right\vert \mathcal{Y}^{T}\right)\right]\label{eq:checkingLindeberg}\\
 & \leq\epsilon^{2}\sum_{t=1}^{T}\mathbb{E}\left(\frac{\varepsilon_{N}\left(Y_{t};\theta\right)^{4}}{\epsilon^{4}T^{2}}\right)\nonumber \\
 & =\frac{1}{T\epsilon^{2}}\mathbb{E}\left(\varepsilon_{N}\left(Y_{t};\theta\right)^{4}\right)\nonumber \\
 & =\frac{1}{T\epsilon^{2}}\left\{ 3\gamma\left(\theta\right)^{4}+\frac{\rho_{4}\left(\theta\right)}{N}\right\} \nonumber \\
 & \rightarrow0,\nonumber 
\end{align}
so the following conditional Lindeberg condition holds 
\[
T^{-1}\sum_{t=1}^{T}\mathbb{E}\left(\left.\varepsilon_{N}\left(Y_{t};\theta\right)^{2}\mathbb{I}_{\left\{ \left\vert \varepsilon_{N}\left(Y_{t};\theta\right)\right\vert \geq\sqrt{T}\epsilon\right\} }\right\vert \mathcal{Y}^{T}\right)\overset{\mathbb{P}}{\rightarrow}0\ .
\]

As (\ref{eq:secondordermoment}) holds, by the strong law of large
numbers (SLLN), the limiting variance is given by 
\[
\lim_{T\rightarrow\infty}\frac{1}{\beta T}\sum_{t=1}^{T}\mathbb{E}\left(\left.\varepsilon_{N}\left(Y_{t};\theta\right)^{2}\right\vert \mathcal{Y}^{T}\right)=\lim_{T\rightarrow\infty}\frac{1}{\beta T}\sum_{t=1}^{T}\gamma\left(Y_{t};\theta\right)^{2}=\beta^{-1}\gamma\left(\theta\right)^{2}.
\]

Lemma \ref{Lemma:WLLN} shows that the second term (\ref{eq:secondterm})
satisfies 
\begin{equation}
\frac{1}{T^{\left(1+\alpha\right)/2}}\sum_{t=1}^{T}\varepsilon_{N}\left(Y_{t};\theta\right)^{2}-T^{\left(1-\alpha\right)/2}\gamma\left(\theta\right)^{2}\overset{\mathbb{P}}{\rightarrow}0,\label{eq:CVsecondterminproba}
\end{equation}
while the third term (\ref{eq:thirdterm}) satisfies 
\begin{equation}
\frac{1}{T^{\left(1+2\alpha\right)/2}}\sum_{t=1}^{T}\varepsilon_{N}\left(Y_{t};\theta\right)^{3}-\frac{\rho_{3}\left(\theta\right)}{\beta^{1/2}T^{\left(3\alpha-1\right)/2}}\overset{\mathbb{P}}{\rightarrow}0,\label{eq:CVthirdterminproba}
\end{equation}
hence it vanishes for $\alpha>1/3$. Similarly, Lemma \ref{Lemma:WLLN}
and (\ref{eq:fourthordermoment}) show that 
\begin{equation}
\frac{1}{T^{\left(1+3\alpha\right)/2}}\sum_{t=1}^{T}\varepsilon_{N}\left(Y_{t};\theta\right)^{4}-\frac{3\gamma\left(\theta\right)^{4}}{T^{\left(3\alpha-1\right)/2}}-\frac{\rho_{4}\left(\theta\right)}{\beta T^{\left(5\alpha-1\right)/2}}\overset{\mathbb{P}}{\rightarrow}0\label{eq:CVfourthterminproba}
\end{equation}
where $\frac{\rho_{4}\left(\theta\right)}{\beta T^{\left(5\alpha-1\right)/2}}\rightarrow0$
for any $\alpha>1/5.$

The term $T^{-\left(1-\alpha\right)/2}\left\{ \log\widehat{p}\left(\left.Y_{1:T}\right\vert \theta\right)-\log p\left(\left.Y_{1:T}\right\vert \theta\right)\right\} $
is asymptotically equivalent in distribution to the sum of the terms
(\ref{eq:firstterm}), (\ref{eq:secondterm}), (\ref{eq:thirdterm}),
(\ref{eq:fourthterm}) and (\ref{eq:remainder}). By combining (\ref{eq:probaTaylorvalid})
to the fact that (\ref{eq:firstterm}) satisfies a conditional CLT,
(\ref{eq:CVsecondterminproba}), (\ref{eq:CVthirdterminproba}), (\ref{eq:CVfourthterminproba}),
(\ref{eq:remaindergoestozero}) hold and Lemma \ref{propn1}, the
result follows. 
\end{proof}

\subsection{Proof of Part 2 of Theorem \ref{Theorem:CLTmarginal}}
\begin{lem}
\label{Lemma:relationshipmomentproposalequilibrium}For any $y\in\mathsf{Y}$
and integer $k\geq1$, if $\mathbb{E}\left[\left\vert \varepsilon_{N}\left(y;\theta\right)\right\vert ^{k+1}\right]<\infty$
then $\widetilde{\mathbb{E}}\left[\left\vert \varepsilon_{N}\left(y;\theta\right)\right\vert ^{k}\right]<\infty$
and 
\[
\widetilde{\mathbb{E}}\left[\varepsilon_{N}\left(y;\theta\right)^{k}\right]=\mathbb{E}\left[\varepsilon_{N}\left(y;\theta\right)^{k}\right]+\frac{1}{\sqrt{N}}\mathbb{E}\left[\varepsilon_{N}\left(y;\theta\right)^{k+1}\right].
\]
\end{lem}
\begin{proof}[Proof of Lemma ~\ref{Lemma:relationshipmomentproposalequilibrium}]
\textit{ }We have 
\begin{align*}
\widetilde{\mathbb{E}}\left[\varepsilon_{N}\left(y;\theta\right)^{k}\right] & =\frac{1}{N^{k/2}}\dotsint\left[\sum_{i=1}^{N}\left\{ \varpi\left(y,u_{1,i};\theta\right)-1\right\} \right]^{k}\overline{\pi}(\mathrm{d}u_{1,1:N}|\theta)\\
 & =\frac{1}{N^{1+k/2}}\dotsint\left[\sum_{i=1}^{N}\left\{ \varpi\left(y,u_{1,i};\theta\right)-1\right\} \right]^{k}\left[N+\sum_{i=1}^{N}\left\{ \varpi\left(y,u_{1,i};\theta\right)-1\right\} \right]{\displaystyle \prod\nolimits _{j=1}^{N}}\varphi\left(\mathrm{d}u_{1,j};0_{p},I_{p}\right)\\
 & =\frac{1}{N^{k/2}}\dotsint\left[\sum_{i=1}^{N}\left\{ \varpi\left(y,u_{1,i};\theta\right)-1\right\} \right]^{k}{\displaystyle \prod\nolimits _{j=1}^{N}}\varphi\left(\mathrm{d}u_{1,j};0_{p},I_{p}\right)\\
 & +\frac{1}{N^{1+k/2}}\dotsint\left[\sum_{i=1}^{N}\left\{ \varpi\left(y,u_{1,i};\theta\right)-1\right\} \right]^{k+1}{\displaystyle \prod\nolimits _{j=1}^{N}}\varphi\left(\mathrm{d}u_{1,j};0_{p},I_{p}\right).
\end{align*}

The result follows directly. 
\end{proof}
\begin{cor}
By combining Lemma \ref{Lemma:momentsofepsilonunderproposal} and
Lemma \ref{Lemma:relationshipmomentproposalequilibrium}, we obtain
\begin{align}
\widetilde{\mathbb{E}}\left[\varepsilon_{N}\left(y;\theta\right)\right] & =\frac{\gamma\left(y;\theta\right)^{2}}{\sqrt{N}},\text{ }\label{eq:equilibriumfirstordermoment}\\
\widetilde{\mathbb{E}}\left[\varepsilon_{N}\left(y;\theta\right)^{2}\right] & =\gamma\left(y;\theta\right)^{2}+\frac{\rho_{3}\left(y;\theta\right)}{N},\label{eq:equilibriumsecondordermoment}\\
\widetilde{\mathbb{E}}\left[\varepsilon_{N}\left(y;\theta\right)^{3}\right] & =\frac{3\gamma\left(y;\theta\right)^{4}+\rho_{3}\left(y;\theta\right)}{\sqrt{N}}+\frac{\rho_{4}\left(y;\theta\right)}{N\sqrt{N}},\label{eq:equilibriumthirdordermoment}\\
\widetilde{\mathbb{E}}\left[\varepsilon_{N}\left(y;\theta\right)^{4}\right] & =3\gamma\left(y;\theta\right)^{4}+\frac{\rho_{4}\left(y;\theta\right)+10\rho_{2}\left(y;\theta\right)\rho_{3}\left(y;\theta\right)}{N}+\frac{\rho_{5}\left(y;\theta\right)}{N^{2}}.\label{eq:equilibriumfourthordermoment}
\end{align}
Similarly, we have $\widetilde{\mathbb{E}}\left[\varepsilon_{N}\left(Y_{1};\theta\right)\right]=\gamma\left(\theta\right)^{2}/\sqrt{N},$
$\widetilde{\mathbb{E}}\left[\varepsilon_{N}\left(Y_{1};\theta\right)^{2}\right]=\gamma\left(\theta\right)^{2}+\rho_{3}\left(\theta\right)/N$,
etc. 
\end{cor}
We can now give the proof of Part 2 of Theorem \ref{Theorem:CLTmarginal}.
For $\alpha=1$, it is possible to combine Part 1 of Theorem \ref{Theorem:CLTmarginal}\textit{
}to a uniform integrability argument to establish this result but
this argument does not extend to $1/3<\alpha<1$. 
\begin{proof}[Proof of Part 2 of Theorem ~\ref{Theorem:CLTmarginal}]
\textit{ }The proof of this CLT is very similar to the proof of Part
1 of Theorem \ref{Theorem:CLTmarginal}\ so we skip some details.
We again first perform a fourth order Taylor expansion of each term
appearing in (\ref{eq:loglikelihooderror}), i.e. see (\ref{eq:Taylor})
and (\ref{eq:remainderexpression}). We also need to ensure that these
Taylor expansions are valid for $t\in1:T$ so we need to control the
probability of the event $B\left(Y^{T},\epsilon\right)=\left\{ \underset{t\leq T}{\max}\left\vert N^{-1/2}\varepsilon_{N}\left(Y_{t};\theta\right)\right\vert >\epsilon\right\} $.
We have for any $\epsilon>0$ 
\[
\widetilde{\mathbb{P}}\left\{ B\left(Y^{T},\epsilon\right)\right\} \leq\frac{\widetilde{\mathbb{E}}\left(\varepsilon_{N}\left(Y_{1};\theta\right)^{8}\right)}{\epsilon^{8}\beta^{4}T^{4\alpha-1}}\text{.}
\]
As $\mathbb{E}\left(\varpi\left(Y_{1},U_{1,1}^{T};\theta\right)^{9}\right)<\infty$
holds, Lemma \ref{Lemma:relationshipmomentproposalequilibrium} ensures
that $\widetilde{\mathbb{E}}\left(\varepsilon_{N}\left(Y_{1};\theta\right)^{8}\right)<\infty$
so 
\begin{equation}
\underset{T\rightarrow\infty}{\lim}\widetilde{\mathbb{P}}\left(\left(B\left(Y^{T},\epsilon\right)\right)^{^{\mathtt{C}}}\right)=1\label{eq:probaTaylorvalid2}
\end{equation}
for $\alpha>1/4$. On the event $\left(B\left(Y^{T},\epsilon\right)\right)^{^{\mathtt{C}}}$,
the Taylor expansion (\ref{eq:Taylor}) holds for all $t\in1:T$ so
we can similarly decompose $T^{-\left(1-\alpha\right)/2}\left\{ \log\widehat{p}\left(\left.Y_{1:T}\right\vert \theta\right)-\log p\left(\left.Y_{1:T}\right\vert \theta\right)\right\} $
as the sum of the terms (\ref{eq:firstterm}), (\ref{eq:secondterm}),
(\ref{eq:thirdterm}), (\ref{eq:fourthterm}), (\ref{eq:remainder})
and an additional $o_{\widetilde{\mathbb{P}}}\left(1\right)$ term.

We can show that as $\alpha>1/4$ the remainder vanishes 
\begin{equation}
\frac{1}{T^{\left(1-\alpha\right)/2}}\left\vert \sum_{t=1}^{T}R_{t,N}\left(Y_{t};\theta\right)\right\vert \overset{\widetilde{\mathbb{P}}}{\rightarrow}0\label{eq:remaindergoestozero2}
\end{equation}
because the WLLN for triangular arrays holds so we have 
\[
\frac{1}{T}\sum_{t=1}^{T}\left\vert \varepsilon_{N}\left(Y_{t};\theta\right)\right\vert ^{5}-\widetilde{\mathbb{E}}\left(\left\vert \varepsilon_{N}\left(Y_{1};\theta\right)\right\vert ^{5}\right)\overset{\widetilde{\mathbb{P}}}{\rightarrow}0.
\]
Using (\ref{eq:equilibriumfirstordermoment}), we can rewrite the
first term (\ref{eq:firstterm}) as follows 
\begin{align}
\frac{1}{\beta^{1/2}T^{1/2}}\sum_{t=1}^{T}\varepsilon_{N}\left(Y_{t};\theta\right) & =\frac{1}{\beta^{1/2}T^{1/2}}\sum_{t=1}^{T}\left\{ \varepsilon_{N}\left(Y_{t};\theta\right)-\widetilde{\mathbb{E}}\left(\left.\varepsilon_{N}\left(Y_{t};\theta\right)\right\vert \mathcal{Y}^{T}\right)\right\} \label{eq:firsttermdecomposition1}\\
 & +\frac{1}{\beta^{1/2}T^{1/2}}\sum_{t=1}^{T}\left\{ \widetilde{\mathbb{E}}\left(\left.\varepsilon_{N}\left(Y_{t};\theta\right)\right\vert \mathcal{Y}^{T}\right)-\frac{\gamma\left(\theta\right)^{2}}{\sqrt{N}}\right\} \label{eq:firsttermdecomposition2}\\
 & +\frac{T^{1/2}}{\beta^{1/2}}\frac{\gamma\left(\theta\right)^{2}}{\sqrt{N}}.\label{eq:firsttermdecomposition3}
\end{align}

The r.h.s. of (\ref{eq:firsttermdecomposition1}) satisfies a conditional
CLT, see Lemma \ref{lem:conditionalLindebergCLT}. Indeed the conditional
Lindeberg condition holds using arguments similar to (\ref{eq:checkingLindeberg})
as $T^{-1}\widetilde{\mathbb{E}}\left(\varepsilon_{N}^{4}\left(Y_{t};\theta\right)\right)\rightarrow0$.
By Lemma \ref{Lemma:relationshipmomentproposalequilibrium} and the
SLLN, the limiting variance is given by 
\begin{align*}
\lim_{T\rightarrow\infty}\frac{1}{\beta T}\sum_{t=1}^{T}\widetilde{\mathbb{E}}\left(\left.\varepsilon_{N}\left(Y_{t};\theta\right)^{2}\right\vert \mathcal{Y}^{T}\right)-\widetilde{\mathbb{E}}\left(\left.\varepsilon_{N}\left(Y_{t};\theta\right)\right\vert \mathcal{Y}^{T}\right)^{2} & =\lim_{T\rightarrow\infty}\frac{1}{\beta T}\sum_{t=1}^{T}\left\{ \gamma\left(Y_{t};\theta\right)^{2}+\frac{\rho_{3}\left(Y_{t};\theta\right)}{N}-\frac{\gamma\left(Y_{t};\theta\right)^{4}}{N}\right\} \\
 & =\beta^{-1}\gamma\left(\theta\right)^{2}
\end{align*}
almost surely, (\ref{eq:equilibriumfirstordermoment})-(\ref{eq:equilibriumsecondordermoment})
and using the assumption $\widetilde{\mathbb{E}}\left[\gamma\left(Y_{1};\theta\right)^{4}\right]<\infty$.
The term (\ref{eq:firsttermdecomposition2}) satisfies 
\begin{equation}
\frac{1}{T^{1/2}}\sum_{t=1}^{T}\left\{ \widetilde{\mathbb{E}}\left(\left.\varepsilon_{N}\left(Y_{t};\theta\right)\right\vert \mathcal{Y}^{T}\right)-\frac{\gamma\left(\theta\right)^{2}}{\sqrt{N}}\right\} =\frac{1}{T^{1/2}\sqrt{N}}\sum_{t=1}^{T}\left\{ \gamma\left(Y_{t};\theta\right)^{2}-\gamma\left(\theta\right)^{2}\right\} \overset{\widetilde{\mathbb{P}}}{\rightarrow}0\label{eq:firsttermdecomposition2converges}
\end{equation}
by the SLLN, the assumption $\widetilde{\mathbb{E}}\left[\gamma\left(Y_{1};\theta\right)^{4}\right]<\infty$
and Chebyshev's inequality. Finally we have for (\ref{eq:firsttermdecomposition3})
\begin{equation}
\frac{T^{1/2}}{\beta^{1/2}}\frac{\gamma\left(\theta\right)^{2}}{\sqrt{N}}-\frac{T^{\left(1-\alpha\right)/2}}{\beta}\gamma\left(\theta\right)^{2}\rightarrow0.\label{eq:remainderfirstterm}
\end{equation}
For the second term (\ref{eq:secondterm}), using (\ref{eq:equilibriumsecondordermoment}),
we obtain using Lemma \ref{Lemma:WLLN} 
\begin{equation}
\frac{1}{T^{\left(1+\alpha\right)/2}}\sum_{t=1}^{T}\varepsilon_{N}\left(Y_{t};\theta\right)^{2}-T^{\left(1-\alpha\right)/2}\gamma\left(\theta\right)^{2}-\beta^{-1}T^{\left(1-3\alpha\right)/2}\rho_{3}\left(\theta\right)\overset{\widetilde{\mathbb{P}}}{\rightarrow}0,\label{eq:CVsecondterminproba2}
\end{equation}
where the third term on the l.h.s. vanishes for $\alpha>1/3.$ For
the third term (\ref{eq:thirdterm}), we obtain using (\ref{eq:equilibriumthirdordermoment})
and Lemma \ref{Lemma:WLLN} 
\begin{equation}
\frac{1}{T^{\left(1+2\alpha\right)/2}}\sum_{t=1}^{T}\varepsilon_{N}\left(Y_{t};\theta\right)^{3}-\frac{3\gamma\left(\theta\right)^{4}+\rho_{3}\left(\theta\right)}{\beta^{1/2}T^{\left(3\alpha-1\right)/2}}-\frac{\rho_{4}\left(\theta\right)}{\beta^{3/2}T^{\left(5\alpha-1\right)/2}}\overset{\widetilde{\mathbb{P}}}{\rightarrow}0.\label{eq:CVthirdterminproba2}
\end{equation}
Hence, (\ref{eq:thirdterm}) vanishes for $\alpha>1/3$. Finally for
the fourth term (\ref{eq:fourthterm}), we obtain using (\ref{eq:equilibriumfourthordermoment})
and Lemma \ref{Lemma:WLLN} 
\begin{equation}
\frac{1}{T^{\left(1+3\alpha\right)/2}}\sum_{t=1}^{T}\varepsilon_{N}\left(Y_{t};\theta\right)^{4}-\frac{3\gamma\left(\theta\right)^{4}}{T^{\left(3\alpha-1\right)/2}}-\frac{\rho_{4}\left(\theta\right)+10\rho_{2}\left(\theta\right)\rho_{3}\left(\theta\right)}{\beta T^{\left(5\alpha-1\right)/2}}-\frac{\rho_{5}\left(\theta\right)}{\beta^{2}T^{\left(7\alpha-1\right)/2}}\overset{\widetilde{\mathbb{P}}}{\rightarrow}0\label{eq:CVfourthterminproba2}
\end{equation}
where $T^{-\left(5\alpha-1\right)/2}\left\{ \rho_{4}\left(\theta\right)+10\rho_{2}\left(\theta\right)\rho_{3}\left(\theta\right)\right\} \rightarrow0$
and $T^{-\left(7\alpha-1\right)/2}\rho_{5}\left(\theta\right)\rightarrow0$
for any $\alpha>1/5.$

The term $T^{-\left(1-\alpha\right)/2}\left\{ \log\widehat{p}\left(\left.Y_{1:T}\right\vert \theta\right)-\log p\left(\left.Y_{1:T}\right\vert \theta\right)\right\} $
is asymptotically equivalent in distribution to the sum of the terms
(\ref{eq:firstterm}), (\ref{eq:secondterm}), (\ref{eq:thirdterm}),
(\ref{eq:fourthterm}) and (\ref{eq:remainder}). By combining (\ref{eq:probaTaylorvalid2})
to the fact that (\ref{eq:firsttermdecomposition1}) satisfies a conditional
CLT, (\ref{eq:firsttermdecomposition2converges}), (\ref{eq:remainderfirstterm}),
(\ref{eq:CVsecondterminproba2}), (\ref{eq:CVthirdterminproba2}),
(\ref{eq:CVfourthterminproba2}), (\ref{eq:remaindergoestozero2})
and Lemma \ref{propn1}, the result follows. 
\end{proof}
\begin{rem}
It follows directly from our proof that for $\frac{1}{4}<\alpha\leq\frac{1}{3}$
\begin{equation}
\left.\frac{Z_{T}\left(\theta\right)}{T^{\left(1-\alpha\right)/2}}-\frac{T^{\left(1-\alpha\right)/2}}{2}\beta^{-1}\gamma\left(\theta\right)^{2}+T^{\left(1-3\alpha\right)/2}\beta^{-2}\left\{ \frac{\rho_{3}\left(\theta\right)}{6}-\frac{1}{4}\gamma\left(\theta\right)^{4}\right\} \right\vert \mathcal{Y}^{T}\Rightarrow\mathcal{N}\left\{ 0,\beta^{-1}\gamma\left(\theta\right)^{2}\right\} .\label{eq:CLTexpressionmarginalequilibriumalphasmall}
\end{equation}
We also note that if we assume that higher order moments of $\varpi\left(Y_{1},U_{1,1};\theta\right)$
under $\widetilde{\mathbb{P}}$ are finite then we obtain different
expressions in the CLT for $\frac{1}{2k+3}<\alpha\leq\frac{1}{2k+1}$
where $k\in\mathbb{N}$. 
\end{rem}

\subsection{Proof of Theorem \ref{Theorem:CLTlikelihoodratiostandardpseudomarginal}}

To simplify presentation, we only give the proof when $\theta$ is
a scalar parameter, the multivariate extension is direct. We have
$Z_{T}\left(\theta\right)=\log\widehat{p}(Y_{1:T}\mid\theta,U)-\log p(Y_{1:T}\mid\theta,U)$
with $U\sim\overline{\pi}\left(\left.\cdot\right\vert \theta\right)$.
We define $W_{T}\left(\theta+\xi/\sqrt{T}\right)=\log\widehat{p}(Y_{1:T}\mid\theta+\xi/\sqrt{T},U^{\prime})-\log p(Y_{1:T}\mid\theta,U^{\prime})$
with $U^{\prime}\sim m$.

The result will follow by the arguments used in the proof of Theorem~\ref{Theorem:CLTmarginal},
replacing 
\[
\epsilon_{N}(Y_{t},U_{t};\theta)=\frac{1}{\sqrt{N}}\sum_{i=1}^{N}\left[\varpi(Y_{t},U_{t,i};\theta)-1\right],
\]
with 
\[
\zeta_{N}(Y_{t};\theta)=\epsilon_{N}(Y_{t},U_{t}^{\prime};\theta+\xi/\sqrt{T})-\epsilon_{N}(Y_{t},U_{t};\theta).
\]
We make here the dependence of $\epsilon_{N}$ on $U_{t}$ or $U_{t}^{\prime}$
explicit. We need to check that the moment conditions used for $\epsilon_{N}$
carry over to $\zeta_{N}$. We have by the C$_{p}$ inequality and
Lemma~\ref{Lemma:relationshipmomentproposalequilibrium} that there
exists $c<\infty$ 
\begin{align*}
\widetilde{\mathbb{E}}\left(\zeta_{N}(Y_{1};\theta)^{8}\right) & \leq c\left\{ \mathbb{E}\left(\epsilon_{N}\left(Y_{1},U_{1}^{\prime};\theta+\xi/\sqrt{T}\right)^{8}\right)+\widetilde{\mathbb{E}}\left(\epsilon_{N}\left(Y_{1},U_{1};\theta\right)^{8}\right)\right\} \\
 & \leq c\left\{ \mathbb{E}\left(\epsilon_{N}\left(Y_{1},U_{1};\theta+\xi/\sqrt{T}\right)^{8}\right)\right.\\
 & \qquad\left.+\mathbb{E}\left(\epsilon_{N}\left(Y_{1},U_{1};\theta\right)^{8}\right)+\frac{1}{\sqrt{N}}\left\vert \mathbb{E}\left(\epsilon_{N}\left(Y_{1},U_{1};\theta\right)^{9}\right)\right\vert \right\} .
\end{align*}
As $\vartheta\mapsto\varpi(Y_{1},U_{1,1};\vartheta)$ and $\vartheta\mapsto\widetilde{\mathbb{E}}(\varpi(Y_{1},U_{1,1};\vartheta)^{9})$
are continuous by assumption, it is straightforward to check that
lower order moments are also continuous. Therefore for $T$ large
enough 
\[
\widetilde{\mathbb{E}}\left(\zeta_{N}(Y_{1};\theta)^{8}\right)\leq c\left\{ 2\mathbb{E}\left(\epsilon_{N}\left(Y_{1},U_{1};\theta\right)^{8}\right)+\frac{1}{\sqrt{N}}\mathbb{E}\left(\left\vert \epsilon_{N}\left(Y_{1},U_{1};\theta\right)\right\vert ^{9}\right)\right\} ,
\]
and similar results hold for lower order moments.

We use a Taylor expansion similarly to Theorem~\ref{Theorem:CLTmarginal}
Part 1 and Part 2, 
\begin{align*}
\frac{W_{T}\left(\theta+\xi/\sqrt{T}\right)}{T^{\left(1-\alpha\right)/2}}-\frac{Z_{T}\left(\theta\right)}{T^{\left(1-\alpha\right)/2}} & =\frac{1}{\beta^{1/2}T^{1/2}}\sum_{t=1}^{T}\left[\varepsilon_{N}\left(Y_{t},U_{t}^{\prime};\theta+\xi/\sqrt{T}\right)-\varepsilon_{N}\left(Y_{t},U_{t};\theta\right)\right]\\
 & -\frac{1}{2\beta T^{\left(1+\alpha\right)/2}}\sum_{t=1}^{T}\left[\varepsilon_{N}\left(Y_{t},U_{t}^{\prime};\theta+\xi/\sqrt{T}\right)^{2}-\varepsilon_{N}\left(Y_{t},U_{t};\theta\right)^{2}\right]\\
 & +\frac{1}{3\beta^{3/2}T^{\left(1+2\alpha\right)/2}}\sum_{t=1}^{T}\left[\varepsilon_{N}\left(Y_{t},U_{t}^{\prime};\theta+\xi/\sqrt{T}\right)^{3}-\varepsilon_{N}\left(Y_{t},U_{t};\theta\right)^{3}\right]\\
 & -\frac{1}{4\beta^{2}T^{\left(1+3\alpha\right)/2}}\sum_{t=1}^{T}\left[\varepsilon_{N}\left(Y_{t},U_{t}^{\prime};\theta+\xi/\sqrt{T}\right)^{4}-\varepsilon_{N}\left(Y_{t},U_{t};\theta\right)^{4}\right]\\
 & +\frac{1}{T^{\left(1-\alpha\right)/2}}\sum_{t=1}^{T}R_{t,N}^{\prime}\left(Y_{t};\theta,\xi\right)+o_{\overline{\mathbb{P}}}\left(1\right),
\end{align*}
where $\overline{\mathbb{P}}$ denotes the probability over $U\sim\overline{\pi}\left(\left.\cdot\right\vert \theta\right),$
$U^{\prime}\sim m$ and $Y_{t}\overset{\text{i.i.d.}}{\sim}\mu$ and
$\overline{\mathbb{E}}$ the associated expectation. By inspecting
the proofs of Parts 1 and 2 of Theorem~\ref{Theorem:CLTmarginal},
we can rewrite this as 
\begin{align}
\lefteqn{\frac{W_{T}\left(\theta+\xi/\sqrt{T}\right)}{T^{\left(1-\alpha\right)/2}}-\frac{Z_{T}\left(\theta\right)}{T^{\left(1-\alpha\right)/2}}}\label{eq:taylorratio}\\
 & =\frac{1}{\beta^{1/2}T^{1/2}}\sum_{t=1}^{T}\left[\varepsilon_{N}\left(Y_{t},U_{t}^{\prime};\theta+\xi/\sqrt{T}\right)-\varepsilon_{N}\left(Y_{t},U_{t};\theta\right)\right]\label{eq:CLTtermTheorem3}\\
 & \qquad-\frac{1}{2\beta T^{\left(1+\alpha\right)/2}}\sum_{t=1}^{T}\left[\varepsilon_{N}\left(Y_{t},U_{t}^{\prime};\theta+\xi/\sqrt{T}\right)^{2}-\varepsilon_{N}\left(Y_{t},U_{t};\theta\right)^{2}\right]+o_{\overline{\mathbb{P}}}\left(1\right).\label{eq:constanttermTheorem3}
\end{align}
The term (\ref{eq:CLTtermTheorem3}) satisfies a conditional CLT\ for
triangular arrays (Lemma \ref{lem:conditionalLindebergCLT}) as the
conditional Lindeberg condition is verified 
\begin{align*}
 & \overline{\mathbb{E}}\left[T^{-1}\sum_{t=1}^{T}\overline{\mathbb{E}}\left(\left.\left\{ \varepsilon_{N}\left(Y_{t},U_{t}^{\prime};\theta+\xi/\sqrt{T}\right)-\varepsilon_{N}\left(Y_{t},U_{t};\theta\right)\right\} ^{2}\mathbb{I}_{\left\{ \left\vert \varepsilon_{N}\left(Y_{t};\theta\right)-\varepsilon_{N}\left(Y_{t},U_{t};\theta\right)\right\vert \geq\sqrt{T}\epsilon\right\} }\right\vert \mathcal{Y}^{T}\right)\right]\\
 & =\overline{\mathbb{E}}\left[\epsilon^{2}\sum_{t=1}^{T}\overline{\mathbb{E}}\left(\left.\frac{\left\{ \varepsilon_{N}\left(Y_{t},U_{t}^{\prime};\theta+\xi/\sqrt{T}\right)-\varepsilon_{N}\left(Y_{t},U_{t};\theta\right)\right\} ^{2}}{\epsilon^{2}T}\mathbb{I}_{\left\{ \left\vert \varepsilon_{N}\left(Y_{t};\theta\right)-\varepsilon_{N}\left(Y_{t},U_{t};\theta\right)\right\vert \geq\sqrt{T}\epsilon\right\} }\right\vert \mathcal{Y}^{T}\right)\right]\\
 & \leq\frac{1}{T\epsilon^{2}}\overline{\mathbb{E}}\left(\left\{ \varepsilon_{N}\left(Y_{t},U_{t}^{\prime};\theta+\xi/\sqrt{T}\right)-\varepsilon_{N}\left(Y_{t},U_{t};\theta\right)\right\} ^{4}\right)\\
 & \leq\frac{c}{T\epsilon^{2}}\left\{ \widetilde{\mathbb{E}}\left(\varepsilon_{N}\left(Y_{t},U_{t};\theta\right)^{4}\right)+\mathbb{E}\left(\varepsilon_{N}\left(Y_{t},U_{t}^{\prime};\theta+\xi/\sqrt{T}\right)^{4}\right)\right\} \text{ (}C_{p}\text{ inequality)}\\
 & \rightarrow0,
\end{align*}
so 
\[
T^{-1}\sum_{t=1}^{T}\overline{\mathbb{E}}\left(\left.\left\{ \varepsilon_{N}\left(Y_{t},U_{t}^{\prime};\theta+\xi/\sqrt{T}\right)-\varepsilon_{N}\left(Y_{t},U_{t};\theta\right)\right\} ^{2}\mathbb{I}_{\left\{ \left\vert \varepsilon_{N}\left(Y_{t};\theta\right)-\varepsilon_{N}\left(Y_{t},U_{t};\theta\right)\right\vert \geq\sqrt{T}\epsilon\right\} }\right\vert \mathcal{Y}^{T}\right)\overset{\overline{\mathbb{P}}}{\rightarrow}0
\]
and the limiting variance is given by 
\begin{align*}
 & \lim_{T\rightarrow\infty}\frac{1}{\beta T}\sum_{t=1}^{T}\Bigg[\overline{\mathbb{E}}\left(\left.\left\{ \varepsilon_{N}\left(Y_{t},U_{t}^{\prime};\theta+\xi/\sqrt{T}\right)-\varepsilon_{N}\left(Y_{t},U_{t};\theta\right)\right\} ^{2}\right\vert \mathcal{Y}^{T}\right)\\
 &\qquad\qquad\qquad\qquad -\overline{\mathbb{E}}\left(\left.\left\{ \varepsilon_{N}\left(Y_{t},U_{t}^{\prime};\theta+\xi/\sqrt{T}\right)-\varepsilon_{N}\left(Y_{t},U_{t};\theta\right)\right\} \right\vert \mathcal{Y}^{T}\right)^{2}\Bigg]\\
 & =\lim_{T\rightarrow\infty}\frac{1}{\beta T}\sum_{t=1}^{T}\widetilde{\mathbb{E}}\left(\left.\varepsilon_{N}\left(Y_{t},U_{t};\theta\right)^{2}\right\vert \mathcal{Y}^{T}\right)+\mathbb{E}\left(\left.\varepsilon_{N}\left(Y_{t},U_{t}^{\prime};\theta+\xi/\sqrt{T}\right)^{2}\right\vert \mathcal{Y}^{T}\right)-\widetilde{\mathbb{E}}\left(\left.\varepsilon_{N}\left(Y_{t},U_{t};\theta\right)\right\vert \mathcal{Y}^{T}\right)^{2}\\
 & =\lim_{T\rightarrow\infty}\frac{1}{\beta T}\sum_{t=1}^{T}\gamma\left(Y_{t};\theta\right)^{2}+\frac{\rho_{3}\left(Y_{t};\theta\right)}{N}-\frac{\gamma\left(Y_{t};\theta\right)^{4}}{N}+\gamma\left(Y_{t};\theta+\xi/\sqrt{T}\right)^{2}
\end{align*}
as $\mathbb{E}\left[\left.\varepsilon_{N}\left(Y_{t},U_{t}^{\prime};\theta+\xi/\sqrt{T}\right)\right\vert \mathcal{Y}^{T}\right]=0$.
Now we have 
\[
\lim_{T\rightarrow\infty}\frac{1}{\beta T}\sum_{t=1}^{T}\gamma\left(Y_{t};\theta\right)^{2}+\frac{\rho_{3}\left(Y_{t};\theta\right)}{N}-\frac{\gamma\left(Y_{t};\theta\right)^{4}}{N}=\beta^{-1}\gamma\left(\theta\right)^{2}
\]
by the SLLN as $\mathbb{E}\left[\gamma\left(Y_{t};\theta\right)^{4}\right]<\infty$.
We also have by the WLLN\ for triangular arrays that 
\begin{equation}
\lim_{T\rightarrow\infty}\frac{1}{T}\sum_{t=1}^{T}\gamma\left(Y_{t};\theta\right)^{2}-\gamma\left(Y_{t};\theta+\xi/\sqrt{T}\right)^{2}\overset{\overline{\mathbb{P}}}{\rightarrow}0\label{eq:variationingamma}
\end{equation}
and 
\begin{align}
T^{(1-\alpha)/2}\left\vert \gamma\left(\theta+\xi/\sqrt{T}\right)^{2}-\gamma\left(\theta\right)^{2}\right\vert  & =T^{(1-\alpha)/2}\left\vert \int_{\theta}^{\theta+\frac{\xi}{\sqrt{T}}}\frac{\partial\gamma\left(\vartheta\right)^{2}}{\partial\vartheta}\mathrm{d}\vartheta\right\vert \label{eq:remainderingammavanishes}\\
 & \leq T^{(1-\alpha)/2}\frac{\xi}{\sqrt{T}}\sup_{\vartheta\in\left[\theta\wedge\left(\theta+\frac{\xi}{\sqrt{T}}\right),\theta\vee\left(\theta+\frac{\xi}{\sqrt{T}}\right)\right]}\left\vert \frac{\partial\gamma\left(\vartheta\right)^{2}}{\partial\vartheta}\right\vert \rightarrow0.\nonumber 
\end{align}
We have already seen in the proof of Theorem\ \ref{Theorem:CLTmarginal},
equation (\ref{eq:CVsecondterminproba}), that 
\begin{equation}
\frac{1}{T^{\left(1+\alpha\right)/2}}\sum_{t=1}^{T}\varepsilon_{N}\left(Y_{t},U_{t};\theta\right)^{2}-T^{\left(1-\alpha\right)/2}\gamma\left(\theta\right)^{2}\overset{\overline{\mathbb{P}}}{\rightarrow}0,\label{eq:CVsecondterminproba3}
\end{equation}
and using a similar argument as the one used in the proof of (\ref{eq:CLTexpressionmarginalequilibrium})
in Theorem \ref{Theorem:CLTmarginal}, equation (\ref{eq:CVsecondterminproba2}),
we have 
\begin{equation}
\frac{1}{T^{\left(1+\alpha\right)/2}}\sum_{t=1}^{T}\varepsilon_{N}\left(Y_{t},U_{t}^{\prime};\theta+\xi/\sqrt{T}\right)^{2}-T^{\left(1-\alpha\right)/2}\gamma\left(\theta\right)^{2}\overset{\overline{\mathbb{P}}}{\rightarrow}0\label{eq:CVsecondterminproba4}
\end{equation}
as $\alpha>1/3$ and (\ref{eq:remainderingammavanishes}) holds.

Hence (\ref{eq:CLTtermTheorem3}) minus its mean satisfies a conditional
CLT\ of limiting variance $2\beta^{-1}\gamma\left(\theta\right)^{2}$
because of (\ref{eq:variationingamma})-(\ref{eq:remainderingammavanishes}).
Using (\ref{eq:firsttermdecomposition2converges}), its mean plus
$\beta^{-1}T^{\left(1-\alpha\right)/2}\gamma\left(\theta\right)^{2}$
converges to zero in probability and (\ref{eq:constanttermTheorem3})
vanishes in probability so the final result follows from Lemma \ref{propn1}.

\subsection{Proof of Theorem \ref{Theorem:conditionalCLTthetathetacand}\label{Section:ProofCLTnightmare}}

\subsubsection{Notation and continuous-time embedding}

For $\delta_{T}=\psi\frac{N}{T}$, we have $\rho_{T}=\exp\left(-\delta_{T}\right)$
and we can write for $t\in1:T$ and $i\in1:N$ 
\begin{equation}
U_{t,i}^{\prime}=e^{-\delta_{T}}U_{t,i}+\sqrt{1-e^{-2\delta_{T}}}\varepsilon_{t,i},\text{ \ }\varepsilon_{t,i}\sim\mathcal{N}\left(0_{p},I_{p}\right).\label{eq:originalAR}
\end{equation}
It will prove convenient for our proof to embed this discrete-time
process within the following Ornstein-Uhlenbeck process 
\begin{equation}
\mathrm{d}U_{t,i}\left(s\right)=-U_{t,i}\left(s\right)\mathrm{d}s+\sqrt{2}\mathrm{d}B_{t,i}\left(s\right),\label{eq:Orstein-Uhlenbeck}
\end{equation}
where $B_{t,i}$ are independent $p-$dimensional standard Brownian
motions for $t\in1:T$ and $i\in1:N$. It is easy to check that we
can set equivalently $U_{t,i}^{\prime}=U_{t,i}\left(\delta_{T}\right)$
as the value of the Ornstein-Uhlenbeck process at time $s=\delta_{T}$
which has been initialized at time $s=0$ using $U_{t,i}\left(0\right)=U_{t,i}.$

Whenever it is clear, we will drop the $T$ index to keep the notation
reasonable. We define 
\[
\widehat{W}_{t}^{T}\left(\theta\right)=\widehat{W}_{t}^{T}(Y_{t}\mid\theta;U_{t})=\frac{\widehat{p}\left(\left.Y_{t}\right\vert \theta,U_{t}\right)}{p\left(Y_{t}\mid\theta\right)}=\frac{1}{N}\sum_{i=1}^{N}\varpi\left(Y_{t},U_{t,i};\theta\right),
\]
where the full notation shall be retained when evaluating at the proposal
$\theta^{\prime},U_{t}^{\prime}$ and 
\begin{equation}
\eta_{t}^{T}=\frac{\widehat{W}_{t}^{T}(Y_{t}\mid\theta^{\prime};U_{t}^{\prime})-\widehat{W}_{t}^{T}\left(\theta\right)}{\widehat{W}_{t}^{T}\left(\theta\right)}.\label{eq:zdef}
\end{equation}
Let $\mathcal{F}^{T}\subset\mathcal{G}$ be the sigma-algebra spanned
by $\left\{ \mathcal{U}^{T},\mathcal{Y}^{T}\right\} $ where $\mathcal{U}^{T}=\sigma\{U_{t,i};t\in1:T,i\in1:N\}$
and $\mathcal{Y}^{T}=\mathcal{\sigma}\left\{ Y_{t};t\in1:T\right\} $.
Let $\widetilde{\mathbb{E}}\left[\left.\cdot\right\vert \mathcal{Y}^{T}\right]$
denotes the expectation w.r.t $\{U_{t,i}\left(0\right);t\in1:T,i\in1:N\}\sim\overline{\pi}\left(\left.\cdot\right\vert \theta\right)$
and the Brownian motions $\{(B_{t,i}\left(s\right);s\geq0);t\in1:T,i\in1:N\}$
where $\overline{\pi}(\{u_{t,i}^{T};t\in1:T,i\in1:N\}|\theta)$ is
given by (\ref{eq:posteriorpanelfactorizes}) whereas $\mathbb{E}\left[\left.\cdot\right\vert \mathcal{Y}^{T}\right]$
denotes the expectation w.r.t $U_{t,i}\overset{\text{i.i.d.}}{\sim}\mathcal{N}\left(0_{p},I_{p}\right)$
and the Brownian motions $\{(B_{t,i}\left(s\right);s\geq0);t\in1:T,i\in1:N\}$.
Finally, we define the Stein operator $\mathcal{S}$ for a real-valued
function $g(y,u;\theta)$ 
\begin{equation}
\mathcal{S}\left\{ g(y,u;\theta)\right\} :=\left\langle \nabla_{u},\nabla_{u}g(y,u;\theta)\right\rangle -\left\langle u,\nabla_{u}g(y,u;\theta)\right\rangle .\label{eq:Steinoperator}
\end{equation}

\subsubsection{Assumptions\label{subsec:AssumptionsTheoremconditionalCLT}}
\begin{assumption}
\label{ass:momentsofW} There exists $\epsilon>0$ such that 
\[
\limsup_{T}\text{ }\mathbb{E}\left[\left(\widehat{W}_{1}^{T}\left(\theta\right)\right)^{-3-\epsilon}\right]<\infty.
\]
\end{assumption}
\begin{assumption}
\label{ass:dthetaFubini}There exists $\chi:\mathsf{Y}\times\mathbb{R}^{p}\rightarrow\mathbb{R}^{+}$
such that $\vartheta\mapsto\nabla_{\vartheta}\varpi\left(y,u;\vartheta\right)$
is continuous at $\vartheta=\theta,$ $\left\Vert \nabla_{\theta}\varpi\left(y,u;\theta\right)\right\Vert \leq\chi\left(y,u\right)$
for all $y,u\in\mathsf{Y}\times\mathbb{R}^{p}$, and 
\[
\mathbb{E}\left[\chi\left(Y_{1},U_{1,1}\right)^{4}\right]<\infty.
\]
\end{assumption}
\begin{assumption}
\label{ass:dUStein} We have 
\[
\mathbb{E}\left[\left\vert \left\langle \nabla_{u},\nabla_{u}\varpi\left(Y_{1},U_{1,1};\theta\right)\right\rangle \right\vert \right]<\infty.
\]
\end{assumption}
\begin{assumption}
\label{ass:SteindU} We have 
\[
\mathbb{E}\left[\left(\mathcal{S}\left\{ \left\Vert \nabla_{u}\varpi(Y_{1},U_{1,1};\theta)\right\Vert ^{2}\right\} \right)^{2}\right]<\infty.
\]
\end{assumption}
\begin{assumption}
\label{ass:dU4moments} There exists $\varkappa>0$ such that 
\[
\mathbb{E}\left[\left\Vert \nabla_{u}\varpi(Y_{1},U_{1,1};\theta)\right\Vert ^{4+\varkappa}\right]<\infty.
\]
\end{assumption}
\begin{assumption}
\label{ass:duSteinfourthmoment} We have 
\[
\mathbb{E}\left[\left(\mathcal{S}\left\{ \varpi\left(Y_{1},U_{1,1};\theta\right)\right\} \right)^{4}\right]<\infty.
\]
\end{assumption}

\subsubsection{Details of the proof}

To simplify the presentation of the proof, we only consider the case
where $\theta$ is a scalar parameter, the dimension of $U_{t,i}$
is $p=1$ and $\psi=1$, the multivariate extension is straightforward
although much more tedious. Let $\theta^{\prime}=\theta+\xi/\sqrt{T}$.
Notice that by definition of $\widehat{W}_{t}^{T}(Y_{t}\mid\theta^{\prime};U_{t}^{\prime}),$
$\widehat{W}_{t}^{T}:=\widehat{W}_{t}^{T}\left(\theta\right)$ and
a Taylor expansion we have 
\begin{align}
\sum_{t=1}^{T}\log\left(\frac{\widehat{W}_{t}^{T}(Y_{t}\mid\theta^{\prime};U_{t}^{\prime})}{\widehat{W}_{t}^{T}}\right) & =\sum_{t=1}^{T}\log(1+\eta_{t}^{T})\nonumber \\
 & =\sum_{t=1}^{T}\eta_{t}^{T}-\frac{1}{2}\sum_{t=1}^{T}[\eta_{t}^{T}]^{2}+\sum_{t=1}^{T}h(\eta_{t}^{T})[\eta_{t}^{T}]^{2},\label{eq:etaTaylor}
\end{align}
as $\log(1+x)=x-x^{2}/2+h(x)x^{2}$ with $h(x)=o(x)$ as $x\rightarrow0$.

The proof proceeds through several auxiliary Lemmas in three main
steps. First, we prove that $\sum_{t=1}^{T}\eta_{t}^{T}$ converges
to a zero-mean normal conditional upon $\mathcal{F}^{T}$. Second,
we show that $\sum_{t=1}^{T}\left(\eta_{t}^{T}\right)^{2}$ converges
in probability towards a constant. Third, we show that high-order
terms vanish in probability. The result then follows from Proposition
\ref{propn1}.

Using Itô's formula, we decompose $\eta_{t}^{T}$ as follows 
\begin{equation}
\eta_{t}^{T}=J_{t}^{T}+L_{t}^{T}+M_{t}^{T},\label{eq:definitionetat}
\end{equation}
where 
\begin{align}
J_{t}^{T} & =\frac{1}{N\widehat{W}_{t}^{T}}\sum_{i=1}^{N}\{\varpi(Y_{t},U_{t,i}\left(\delta_{T}\right);\theta+\xi/\sqrt{T})-\varpi(Y_{t},U_{t,i}\left(\delta_{T}\right);\theta)\},\label{eq:Jtdef}\\
L_{t}^{T} & =\int_{0}^{\delta_{T}}\frac{1}{N\widehat{W}_{t}^{T}}\sum_{i=1}^{N}\left\{ -\partial_{u}\varpi\left(Y_{t},U_{t,i}\left(s\right);\theta\right)U_{t,i}\left(s\right)+\partial_{u,u}^{2}\varpi\left(Y_{t},U_{t,i}\left(s\right);\theta\right)\right\} \mathrm{d}s,\label{eq:Ltdef}\\
M_{t}^{T} & =\int_{0}^{\delta_{T}}\frac{\sqrt{2}}{N\widehat{W}_{t}^{T}}\sum_{i=1}^{N}\partial_{u}\varpi\left(Y_{t},U_{t,i}\left(s\right);\theta\right)\mathrm{d}B_{t,i}\left(s\right).\label{eq:Mtdef}
\end{align}
The following preliminary Lemmas establish various properties of the
terms $J_{t}^{T},$ $L_{t}^{T},$ $M_{t}^{T}$ and $\eta_{t}^{T}$. 
\begin{lem}
\label{lem:Jt} The sequence $\left\{ J_{t}^{T};t\geq1\right\} $
defined in (\ref{eq:Jtdef}) satisfies 
\[
\widetilde{\mathbb{E}}\left(J_{t}^{T}\right)=0,\text{ \  \  \ }\widetilde{\mathbb{V}}\left(\sum_{t=1}^{T}J_{t}^{T}\right)=T\widetilde{\mathbb{V}}\left(J_{1}^{T}\right)\rightarrow0
\]
and $\sum_{t=1}^{T}J_{t}^{T}\overset{\widetilde{\mathbb{P}}}{\rightarrow}0$,
$\sum_{t=1}^{T}(J_{t}^{T})^{2}\overset{\widetilde{\mathbb{P}}}{\rightarrow}0$. 
\end{lem}
\begin{lem}
\label{lem:Lt} The sequence $\left\{ L_{t}^{T};t\geq1\right\} $
defined in (\ref{eq:Ltdef}) satisfies 
\[
\widetilde{\mathbb{E}}\left(L_{t}^{T}\right)=0,\text{ \  \  \ }\widetilde{\mathbb{V}}\left(\sum_{t=1}^{T}L_{t}^{T}\right)=T\widetilde{\mathbb{V}}\left(L_{1}^{T}\right)\rightarrow0,
\]
and $\sum_{t=1}^{T}L_{t}^{T}\overset{\widetilde{\mathbb{P}}}{\rightarrow}0$. 
\end{lem}
\begin{lem}
\label{lem:Mt} The sequence $\left\{ M_{t}^{T};t\geq1\right\} $
defined in (\ref{eq:Mtdef}) satisfies 
\[
\widetilde{\mathbb{E}}[(M_{t}^{T})^{2}]=O(1/T),\text{ \  \ }\left.\sum_{t=1}^{T}M_{t}^{T}\right\vert \mathcal{F}^{T}\Rightarrow\mathcal{N}\left(0,\frac{\kappa\left(\theta\right)^{2}}{2}\right).
\]
\end{lem}
\begin{lem}
\label{lem:zsquared} The sequence $\left\{ \eta_{t}^{T};t\geq1\right\} $
defined in (\ref{eq:definitionetat}) satisfies 
\[
\sum_{t=1}^{T}(\eta_{t}^{T})^{2}\overset{\widetilde{\mathbb{P}}}{\rightarrow}\kappa^{2}\left(\theta\right).
\]
\end{lem}
Armed with the above results, we can now prove Theorem~\ref{Theorem:conditionalCLTthetathetacand}.
Combining Lemmas~\ref{lem:Jt}, \ref{lem:Lt}, \ref{lem:Mt} and
\ref{lem:zsquared} with Lemma~\ref{propn1} from Section \ref{Section:Conditionalweakconvergence},
we immediately obtain that 
\[
\sum_{t=1}^{T}\eta_{t}-\frac{1}{2}\left(\eta_{t}^{T}\right)^{2}\left\vert \mathcal{F}^{T}\right.\Rightarrow\mathcal{N}\left(-\frac{\kappa\left(\theta\right)^{2}}{2},\kappa\left(\theta\right)^{2}\right).
\]
It remains to control the remainder from the Taylor expansion (\ref{eq:etaTaylor}).
We bound it using Lemma~\ref{lem:zsquared} as 
\[
\left\vert \sum_{t=1}^{T}h\left(\eta_{t}^{T}\right)\left(\eta_{t}^{T}\right)^{2}\right\vert \leq\max_{t}\text{ }\left\vert h\left(\eta_{t}^{T}\right)\right\vert \text{ }\sum_{t=1}^{T}\eta_{t}^{2}=\max_{t}\text{ }|h(\eta_{t}^{T})|\text{ }O_{_{\widetilde{\mathbb{P}}}}(1).
\]

Without loss of generality we can assume that $\left\vert h\left(x\right)\right\vert \leq g\left(\left\vert x\right\vert \right)$
where $g$ is increasing on $\left[0,\infty\right)$ and $\lim_{x\rightarrow0^{+}}g\left(x\right)=0$
so that 
\[
\max_{t}\text{ }\left\vert h\left(\eta_{t}^{T}\right)\right\vert \leq g\left(\max_{t}\text{ }\left\vert \eta_{t}^{T}\right\vert \right)
\]
and 
\begin{align}
\widetilde{\mathbb{P}}\left(\max_{t}\left\vert \eta_{t}^{T}\right\vert \leq\varepsilon\right) & ={\prod\nolimits _{t=1}^{T}}\left\{ 1-\widetilde{\mathbb{P}}\left(\left\vert \eta_{t}^{T}\right\vert >\varepsilon\right)\right\} \geq\left(1-\varepsilon^{-2}\widetilde{\mathbb{E}}\left(\left(\eta_{1}^{T}\right)^{2}\mathbb{I}\left(\left\vert \eta_{1}^{T}\right\vert >\varepsilon\right)\right)\right)^{T}.\label{eq:probabound}
\end{align}
By using the decomposition of $\eta_{1}^{T}$, we have using the C$_{p}$
inequality 
\begin{align*}
T\widetilde{\mathbb{E}}\left(\left(\eta_{1}^{T}\right)^{2}\mathbb{I}\left(\left\vert \eta_{1}^{T}\right\vert >\varepsilon\right)\right) & \leq c\text{ }T\left(\widetilde{\mathbb{V}}(J_{1}^{T})+\widetilde{\mathbb{V}}(L_{1}^{T})+\widetilde{\mathbb{E}}\left[\left(M_{1}^{T}\right)^{2}\mathbb{I}\left(\left\vert \eta_{1}^{T}\right\vert >\varepsilon\right)\right]\right)\\
 & =o(1),
\end{align*}
where we have used Lemmas \ref{lem:Jt} and \ref{lem:Lt} for the
terms involving $J_{1}^{T}$ and $L_{1}^{T}$. The term involving
$M_{1}^{T}$ vanishes by uniform integrability of the family $\left\{ T(M_{1}^{T})^{2};T\geq1\right\} $,
the proof of which can be found in the proof of Lemma~\ref{lem:Mt}
where the Lindeberg condition is verified. Therefore overall we have
\[
\widetilde{\mathbb{E}}\left(\left\vert \eta_{1}^{T}\right\vert ^{2}\mathbb{I}\left(\left\vert \eta_{1}^{T}\right\vert >\varepsilon\right)\right)=o(1/T),
\]
and thus (\ref{eq:probabound})$\ $converges towards $1$ as $T\rightarrow\infty$.
Hence we have $g\left(\underset{t}{\max}\left\vert \eta_{t}^{T}\right\vert \right)=o_{\widetilde{\mathbb{P}}}(1)$
and the result follows.

\subsubsection{Proofs of Auxiliary Results}
\begin{proof}[Proof of Lemma ~\textit{\ref{lem:Jt}}]
From Assumption \ref{ass:dthetaFubini}, we obtain directly $\widetilde{\mathbb{E}}\left(J_{t}^{T}\right)=0$
and we can rewrite $J_{t}^{T}$ as follows 
\[
J_{t}^{T}=\frac{1}{N\widehat{W}_{t}^{T}}\sum_{i=1}^{N}\int_{\theta}^{\theta^{\prime}}\partial_{\vartheta}\varpi(Y_{t},U_{t,i}(\delta_{T});\vartheta)\mathrm{d}\vartheta,
\]
where $\theta^{\prime}=\theta+\xi/\sqrt{T}$. Thus we obtain 
\begin{align*}
\widetilde{\mathbb{V}}\left(\sum_{t=1}^{T}J_{t}^{T}\right) & =\sum_{t=1}^{T}\widetilde{\mathbb{V}}\left(J_{t}^{T}\right)=\sum_{t=1}^{T}\widetilde{\mathbb{E}}\left(\left(J_{t}^{T}\right)^{2}\right)\\
 & =\sum_{t=1}^{T}\widetilde{\mathbb{E}}\left(\left(\frac{1}{N\widehat{W}_{t}^{T}}\sum_{i=1}^{N}\int_{\theta}^{\theta^{\prime}}\partial_{\vartheta}\varpi(Y_{t},U_{t,i}(\delta_{T});\vartheta)\mathrm{d}\vartheta\right)^{2}\right)\\
 & =\sum_{t=1}^{T}\frac{1}{N^{2}}\mathbb{E}\left((\widehat{W}_{t}^{T})^{-1}\left(\sum_{i=1}^{N}\int_{\theta}^{\theta^{\prime}}\partial_{\vartheta}\varpi(Y_{t},U_{t,i}(\delta_{T});\vartheta)\mathrm{d}\vartheta\right)^{2}\right)\\
 & \leq\sum_{t=1}^{T}\frac{1}{N^{2}}\mathbb{E}\left((\widehat{W}_{t}^{T})^{-2}\right)^{1/2}\mathbb{E}\left(\left[\sum_{i=1}^{N}\int_{\theta}^{\theta^{\prime}}\partial_{\vartheta}\varpi(Y_{t},U_{t,i}(\delta_{T});\vartheta)\mathrm{d}\vartheta\right]^{4}\right)^{1/2}\\
 & =\frac{T}{N^{2}}\mathbb{E}\left((\widehat{W}_{1}^{T})^{-2}\right)^{1/2}\mathbb{E}\left(\left[\sum_{i=1}^{N}\int_{\theta}^{\theta^{\prime}}\partial_{\vartheta}\varpi(Y_{t},U_{t,i}(\delta_{T});\vartheta)\mathrm{d}\vartheta\right]^{4}\right)^{1/2}\\
 & \leq c\frac{TN}{N^{2}}\mathbb{E}\left((\widehat{W}_{1}^{T})^{-2}\right)^{1/2}\mathbb{E}\left(\left[\int_{\theta}^{\theta^{\prime}}\partial_{\vartheta}\varpi(Y_{1},U_{1,1}(\delta_{T});\vartheta)\mathrm{d}\vartheta\right]^{4}\right)^{1/2}
\end{align*}
where we have interchanged derivative and integration by Assumption~\ref{ass:dthetaFubini},
have used the fact that the expectation of the integrals over $\vartheta$
have zero mean and that under $\mathbb{P}$ the terms are i.i.d. over
index $i$. We also use the fact that for i.i.d. zero-mean random
variables $Z_{i}$
\[
\mathbb{E}\left[\left(\sum_{i=1}^{P}Z_{i}\right)^{4}\right]\leq cP^{2}\mathbb{E}\left(Z_{1}^{4}\right).
\]
Hence, we have using Assumptions~\ref{ass:momentsofW} and \ref{ass:dthetaFubini}
that $\left\vert \partial_{\vartheta}\varpi(Y_{1},U_{1,1}(\delta_{T});\vartheta)\right\vert \leq2\chi(Y_{1},U_{1,1}(\delta_{T}))$
for $\vartheta\in\left[\theta\wedge\theta^{\prime},\theta\vee\theta^{\prime}\right]$
and $T$ large enough. When $\theta$ is multidimensional, this can
be established using the fundamental theorem of calculus for line
integrals. It follows that 
\begin{align*}
\widetilde{\mathbb{V}}\Big(\sum_{t=1}^{T}J_{t}^{T}\Big) & \leq c\frac{TN}{N^{2}}\mathbb{E}\Big(\Big[\int_{\theta}^{\theta^{\prime}}\partial_{\vartheta}\varpi(Y_{1},U_{1,1}(\delta_{T});\vartheta)\mathrm{d}\vartheta\Big]^{4}\Big)^{1/2}\\
 & \leq c\frac{TN}{N^{2}}\mathbb{E}\Big(\Big[\int_{\theta}^{\theta^{\prime}}\chi(Y_{1},U_{1,1}(\delta_{T}))\mathrm{d}\vartheta\Big]^{4}\Big)^{1/2}\\
 & =c\frac{TN}{N^{2}}\frac{\xi^{2}}{T}\mathbb{E}\Big(\chi(Y_{1},U_{1,1})^{4}\Big)^{1/2}=O\Big(\frac{1}{N}\Big).
\end{align*}

This concludes the proof of the lemma.
\end{proof}
\begin{proof}[Proof of Lemma ~\textit{\ref{lem:Lt}}]
\textit{ }We can rewrite $L_{t}^{T}$ given by (\ref{eq:Ltdef})
as 
\[
L_{t}^{T}=\int_{0}^{\delta_{T}}\frac{1}{N\widehat{W}_{t}^{T}}\left(\sum_{i=1}^{N}\mathcal{S}\left\{ \varpi\left(Y_{t},U_{t,i}\left(s\right);\theta\right)\right\} \right)\mathrm{d}s,
\]
where $\mathcal{S}$ is the Stein operator defined in (\ref{eq:Steinoperator}).
By Assumption~\ref{ass:dUStein}, we can apply Fubini's theorem to
interchange the order of integration, and Stein's lemma \citep[Lemma 1]{Stein1981}
shows that 
\[
\widetilde{\mathbb{E}}\left(\left.L_{t}^{T}\right\vert \mathcal{Y}^{T}\right)=\frac{1}{N}\sum_{i=1}^{N}\mathbb{E}\left(\left.\int_{0}^{\delta_{T}}\mathcal{S}\left\{ \varpi\left(Y_{t},U_{t,i}\left(s\right);\theta\right)\right\} \mathrm{d}s\right\vert \mathcal{Y}^{T}\right)=0,
\]
so in particular $\widetilde{\mathbb{E}}\left(L_{t}^{T}\right)=0$.
Hence, we have 
\begin{align}
\widetilde{\mathbb{V}}\left(L_{t}^{T}\right) & =\widetilde{\mathbb{E}}\left(\left[\int_{0}^{\delta_{T}}\frac{1}{N\widehat{W}_{t}^{T}}\sum_{i=1}^{N}\mathcal{S}\left\{ \varpi\left(Y_{t},U_{t,i}\left(s\right);\theta\right)\right\} \mathrm{d}s\right]^{2}\right)\nonumber \\
 & =\widetilde{\mathbb{E}}\left(\frac{1}{\left(N\widehat{W}_{t}^{T}\right)^{2}}\int_{0}^{\delta_{T}}\int_{0}^{\delta_{T}}\left[\sum_{i=1}^{N}\mathcal{S}\left\{ \varpi\left(Y_{t},U_{t,i}\left(s\right);\theta\right)\right\} \right]\left[\sum_{j=1}^{N}\mathcal{S}\left\{ \varpi\left(Y_{t},U_{t,j}\left(r\right);\theta\right)\right\} \right]\mathrm{d}r\text{ }\mathrm{d}s\right)\nonumber \\
 & =\mathbb{E}\left(\frac{1}{N^{2}\widehat{W}_{t}^{T}}\int_{0}^{\delta_{T}}\int_{0}^{\delta_{T}}\left[\sum_{i=1}^{N}\mathcal{S}\left\{ \varpi\left(Y_{t},U_{t,i}\left(s\right);\theta\right)\right\} \right]\left[\sum_{j=1}^{N}\mathcal{S}\left\{ \varpi\left(Y_{t},U_{t,j}\left(r\right);\theta\right)\right\} \right]\mathrm{d}r\text{ }\mathrm{d}s\right).\label{eq:termvarianceVLt}
\end{align}
The term (\ref{eq:termvarianceVLt}) can be rewritten as 
\begin{align*}
 & \int_{0}^{\delta_{T}}\int_{0}^{\delta_{T}}\mathbb{E}\left[\frac{1}{\widehat{W}_{t}^{T}}\left(\frac{1}{N}\sum_{i=1}^{N}\mathcal{S}\left\{ \varpi\left(Y_{t},U_{t,i}\left(s\right);\theta\right)\right\} \right)\left(\frac{1}{N}\sum_{j=1}^{N}\mathcal{S}\left\{ \varpi\left(Y_{t},U_{t,j}\left(r\right);\theta\right)\right\} \right)\right]\mathrm{d}r\mathrm{d}s\\
 & \leq\int_{0}^{\delta_{T}}\int_{0}^{\delta_{T}}\mathbb{E}\left[\left(\widehat{W}_{t}^{T}\right)^{-2}\right]^{1/2}\mathbb{E}\left[\left(\frac{1}{N}\sum_{i=1}^{N}\mathcal{S}\left\{ \varpi\left(Y_{t},U_{t,i}\left(s\right);\theta\right)\right\} \right)^{2}\left(\frac{1}{N}\sum_{j=1}^{N}\mathcal{S}\left\{ \varpi\left(Y_{t},U_{t,j}\left(r\right);\theta\right)\right\} \right)^{2}\right]^{1/2}\mathrm{d}r\mathrm{d}s\\
 & \leq\int_{0}^{\delta_{T}}\int_{0}^{\delta_{T}}\mathbb{E}\left[\left(\widehat{W}_{t}^{T}\right)^{-2}\right]^{1/2}\mathbb{E}\left[\left(\frac{1}{N}\sum_{i=1}^{N}\mathcal{S}\left\{ \varpi\left(Y_{t},U_{t,i}\left(s\right);\theta\right)\right\} \right)^{4}\right]^{1/4}\mathbb{E}\left[\left(\frac{1}{N}\sum_{j=1}^{N}\mathcal{S}\left\{ \varpi\left(Y_{t},U_{t,j}\left(r\right);\theta\right)\right\} \right)^{4}\right]^{1/4}\mathrm{d}r\mathrm{d}s\\
 & =\int_{0}^{\delta_{T}}\int_{0}^{\delta_{T}}\mathbb{E}\left[\left(\widehat{W}_{t}^{T}\right)^{-2}\right]^{1/2}\mathbb{E}\left[\left(\frac{1}{N}\sum_{i=1}^{N}\mathcal{S}\left\{ \varpi\left(Y_{t},U_{t,i}\left(s\right);\theta\right)\right\} \right)^{4}\right]^{1/2}\mathrm{d}r\mathrm{d}s\\
 & \leq c\frac{\delta_{T}^{2}}{N}=O(\frac{N}{T^{2}}),
\end{align*}
byAssumption~\ref{ass:duSteinfourthmoment}, $\mathbb{E}\left(\mathcal{S}\left\{ \varpi\left(Y_{t},U_{t,i}\left(s\right);\theta\right)\right\} \right)=0$,
and the fact that $U_{t,i}\left(s\right)$ are stationary and independent
over $i$ under $\mathbb{P}$. Therefore we have 
\[
\widetilde{\mathbb{V}}\left(\sum_{t=1}^{T}L_{t}^{T}\right)=T\widetilde{\mathbb{V}}\left(L_{1}^{T}\right)=O\left(\frac{N}{T^{2}}T\right)=O\left(\frac{N}{T}\right).
\]
This concludes the proof of the lemma. 
\end{proof}
\begin{proof}[Proof of Lemma ~\textit{\ref{lem:Mt}}]
\ We check here that the conditions of the conditional CLT\ given
in Lemma \ref{lem:conditionalLindebergCLT} of Section \ref{Section:Conditionalweakconvergence}
are satisfied. Consider the term $M^{T}$ given by (\ref{eq:Mtdef})
which can be decomposed as 
\begin{equation}
M^{T}=\sum_{t=1}^{T}M_{t}^{T}=\sum_{t=1}^{T}\sum_{i=1}^{N}M_{t,i}^{T},\label{eq:doublesumMt}
\end{equation}
where 
\begin{equation}
M_{t,i}^{T}=\frac{\sqrt{2}}{N\widehat{W}_{t}^{T}}\int_{0}^{\delta_{T}}\partial_{u}\varpi\left(Y_{t},U_{t,i}\left(s\right);\theta\right)\mathrm{d}B_{t,i}\left(s\right).\label{eq:Mtidef}
\end{equation}
It is straightforward to see that 
\[
\widetilde{\mathbb{E}}\left(\left.M_{t}^{T}\right\vert \mathcal{F}^{T}\right)=\mathbb{E}\left(\left.\frac{\sqrt{2}}{N}\sum_{i=1}^{N}\int_{0}^{\delta_{T}}\partial_{u}\varpi\left(Y_{t},U_{t,i}\left(s\right);\theta\right)\mathrm{d}B_{t,i}\left(s\right)\right\vert \mathcal{F}^{T}\right)=0
\]
and 
\begin{equation}
s_{T}^{2}=\widetilde{\mathbb{V}}\left(\left.M^{T}\right\vert \mathcal{F}^{T}\right)=\sum_{t=1}^{T}\widetilde{\mathbb{V}}\left(\left.M_{t}^{T}\right\vert \mathcal{F}^{T}\right).\label{eq:definitionsTsquared}
\end{equation}
The term $\widetilde{\mathbb{V}}\left(\left.M_{t}^{T}\right\vert \mathcal{F}^{T}\right)$
satisfies 
\begin{align*}
\widetilde{\mathbb{V}}\left(\left.M_{t}^{T}\right\vert \mathcal{F}^{T}\right) & =\sum_{i=1}^{N}\widetilde{\mathbb{V}}\left(\left.M_{t,i}^{T}\right\vert \mathcal{F}^{T}\right)\\
 & =\sum_{i=1}^{N}\frac{2}{N^{2}\left(\widehat{W}_{t}^{T}\right)^{2}}\int_{0}^{\delta_{T}}\widetilde{\mathbb{E}}\left(\left.\left\{ \partial_{u}\varpi\left(Y_{t},U_{t,i}\left(s\right);\theta\right)\right\} ^{2}\right\vert \mathcal{F}^{T}\right)\mathrm{d}s.
\end{align*}
Letting 
\[
g(y,u;\theta)=\left\{ \partial_{u}\varpi(y,u;\theta)\right\} ^{2},
\]
and using Itô's formula, we obtain 
\begin{align*}
 & \int_{0}^{\delta_{T}}\widetilde{\mathbb{E}}\left(\left.\left\{ \partial_{u}\varpi(Y_{t},U_{t,i}\left(s\right);\theta)\right\} ^{2}\right\vert \mathcal{F}^{T}\right)\mathrm{d}s\\
 & \hspace{3cm}=\int_{0}^{\delta_{T}}\left\{ \partial_{u}\varpi(Y_{t},U_{t,i}\left(0\right);\theta)\right\} ^{2}\mathrm{d}s+\int_{0}^{\delta_{T}}\int_{0}^{s}\widetilde{\mathbb{E}}\left(\left.\mathcal{S}\left\{ g(Y_{t},U_{t,i}\left(s\right);\theta)\right\} \right\vert \mathcal{F}^{T}\right)\mathrm{d}r\mathrm{d}s\\
 & \hspace{3cm}=\delta_{T}\left\{ \partial_{u}\varpi(Y_{t},U_{t,i}\left(0\right);\theta)\right\} ^{2}+\int_{0}^{\delta_{T}}\int_{0}^{s}\widetilde{\mathbb{E}}\left(\left.\mathcal{S}\left\{ g(Y_{t},U_{t,i}\left(s\right);\theta)\right\} \right\vert \mathcal{F}^{T}\right)\mathrm{d}r\mathrm{d}s,
\end{align*}
where 
\[
\int_{0}^{s}\widetilde{\mathbb{E}}\left(\left.\mathcal{S}\left\{ g(Y_{t},U_{t,i}\left(s\right);\theta)\right\} \right\vert \mathcal{F}^{T}\right)\mathrm{d}r=\int_{0}^{s}\int\mathcal{S}\left\{ g(Y_{t},e^{-r}U_{t,i}+\sqrt{1-e^{-2r}}\varepsilon;\theta)\right\} \varphi\left(\varepsilon;0,1\right)\mathrm{d}\varepsilon\mathrm{d}r.
\]
Therefore 
\begin{align}
s_{T}^{2} & =\sum_{t=1}^{T}\sum_{i=1}^{N}\frac{2}{N^{2}\left(\widehat{W}_{t}^{T}\right)^{2}}\int_{0}^{\delta_{T}}\widetilde{\mathbb{E}}\left(\left\{ \partial_{u}\varpi(Y_{t},U_{t,i}\left(s\right);\theta)\right\} ^{2}\left\vert \mathcal{F}^{T}\right.\right)\mathrm{d}s\nonumber \\
 & =\sum_{t=1}^{T}\sum_{i=1}^{N}\frac{2\delta_{T}}{N^{2}\left(\widehat{W}_{t}^{T}\right)^{2}}\left\{ \partial_{u}\varpi(Y_{t},U_{t,i};\theta)\right\} ^{2}\label{eq:identitysTsquared1}\\
 & +\sum_{t=1}^{T}\sum_{i=1}^{N}\frac{2}{N^{2}\left(\widehat{W}_{t}^{T}\right)^{2}}\int_{0}^{\delta_{T}}\int_{0}^{s}\widetilde{\mathbb{E}}\left(\left.\mathcal{S}\left\{ g(Y_{t},U_{t,i}\left(s\right);\theta)\right\} \right\vert \mathcal{F}^{T}\right)\mathrm{d}r\mathrm{d}s.\label{eq:identitysTsquared2}
\end{align}
To show that the term (\ref{eq:identitysTsquared2}) vanishes in probability,
we show that it vanishes in absolute mean 
\begin{align*}
 & \frac{1}{N^{2}}\mathbb{E}\left(\left\vert \sum_{t=1}^{T}\sum_{i=1}^{N}\frac{2}{\left(\widehat{W}_{t}^{T}\right)^{2}}\int_{0}^{\delta_{T}}\int_{0}^{s}\widetilde{\mathbb{E}}\left(\left.\mathcal{S}\left\{ g(Y_{t},U_{t,i}\left(r\right);\theta)\right\} \right\vert \mathcal{F}^{T}\right)\mathrm{d}r\mathrm{d}s\right\vert \right)\\
 & \leq\frac{1}{N^{2}}\sum_{t=1}^{T}\sum_{i=1}^{N}\widetilde{\mathbb{E}}\left(\frac{2}{\left(\widehat{W}_{t}^{T}\right)^{2}}\int_{0}^{\delta_{T}}\int_{0}^{s}\left\vert \widetilde{\mathbb{E}}\left(\left.\mathcal{S}\left\{ g(Y_{t},U_{t,i}\left(r\right);\theta)\right\} \right\vert \mathcal{F}^{T}\right)\right\vert \mathrm{d}r\mathrm{d}s\right)\\
 & =\frac{1}{N^{2}}\sum_{t=1}^{T}\sum_{i=1}^{N}\widetilde{\mathbb{E}}\left(\frac{2}{\left(\widehat{W}_{t}^{T}\right)^{2}}\int_{0}^{\delta_{T}}\int_{0}^{s}\left\vert \widetilde{\mathbb{E}}\left(\left.\mathcal{S}\left\{ g(Y_{t},U_{t,i}\left(r\right);\theta)\right\} \right\vert \mathcal{F}^{T}\right)\right\vert \mathrm{d}r\mathrm{d}s\right)\\
 & =\frac{1}{N^{2}}\sum_{t=1}^{T}\sum_{i=1}^{N}\int_{0}^{\delta_{T}}\int_{0}^{s}\widetilde{\mathbb{E}}\left(\widetilde{\mathbb{E}}\left[\left.\frac{2}{\left(\widehat{W}_{t}^{T}\right)^{2}}\left\vert \mathcal{S}\left\{ g(Y_{t},U_{t,i}\left(r\right);\theta)\right\} \right\vert \right\vert \mathcal{F}^{T}\right]\right)\mathrm{d}r\mathrm{d}s\\
 & =\frac{2}{N^{2}}\sum_{t=1}^{T}\sum_{i=1}^{N}\int_{0}^{\delta_{T}}\int_{0}^{s}\mathbb{E}\left(\frac{1}{\widehat{W}_{t}^{T}}\left\vert \mathcal{S}\left\{ g(Y_{t},U_{t,i}\left(r\right);\theta)\right\} \right\vert \right)\mathrm{d}r\mathrm{d}s\\
 & =\frac{2}{N^{2}}\sum_{t=1}^{T}\sum_{i=1}^{N}\int_{0}^{\delta_{T}}\int_{0}^{s}\mathbb{E}\left[\left(\widehat{W}_{t}^{T}\right)^{-2}\right]^{1/2}\mathbb{E}\left[\left(\mathcal{S}\left\{ g(Y_{t},U_{t,i}\left(r\right);\theta)\right\} \right)^{2}\right]^{1/2}\mathrm{d}r\mathrm{d}s\\
 & =\delta_{T}^{2}\frac{NT}{N^{2}}\mathbb{E}\left[\left(\widehat{W}_{t}^{T}\right)^{-2}\right]^{1/2}\mathbb{E}\left[\left(\mathcal{S}\left\{ g(Y_{1},U_{1,1};\theta)\right\} \right)^{2}\right]^{1/2}\\
 & =\delta_{T}\mathbb{E}\left[\left(\widehat{W}_{t}^{T}\right)^{-2}\right]^{1/2}\mathbb{E}\left[\left(\mathcal{S}\left\{ g(Y_{1},U_{1,1};\theta)\right\} \right)^{2}\right]^{1/2}=O(\delta_{T}),
\end{align*}
by Assumption~\ref{ass:SteindU} and the fact that $U_{t,i}\left(r\right)$
are stationary and i.i.d. over $t,i$ under $\mathbb{P}$. Going back
to our calculation of $s_{T}^{2}$, we now treat the term (\ref{eq:identitysTsquared1})
\[
\sum_{t=1}^{T}\sum_{i=1}^{N}\frac{2\delta_{T}}{N^{2}\left(\widehat{W}_{t}^{T}\right)^{2}}\left\{ \partial_{u}\varpi(Y_{t},U_{t,i};\theta)\right\} ^{2}=\frac{2}{T}\sum_{t=1}^{T}g_{T}\left(Y_{t},U_{t}\right),
\]
where 
\[
g_{T}\left(Y_{t},U_{t}\right):=\frac{1}{N}\sum_{i=1}^{N}\left(\widehat{W}_{t}^{T}\right)^{-2}\left\{ \partial_{u}\varpi(Y_{t},U_{t,i};\theta)\right\} ^{2}.
\]
In order to apply the WLLN we have to check that 
\[
\sum_{t=1}^{T}\widetilde{\mathbb{E}}\left(\frac{\left\vert g_{T}\left(Y_{t},U_{t}\right)\right\vert }{T}\mathbb{I}\left\{ \left\vert g_{T}\left(Y_{t},U_{t}\right)\right\vert \geq\epsilon T\right\} \right)=\widetilde{\mathbb{E}}\left(\left\vert g_{T}\left(Y_{1},U_{1}^{T}\right)\right\vert \mathbb{I}\left\{ \left\vert g_{T}\left(Y_{1},U_{1}^{T}\right)\right\vert \geq\epsilon T\right\} \right)\rightarrow0,
\]
or equivalently that 
\[
\left\{ g_{T}\left(Y_{1},U_{1}^{T}\right)\right\} _{T\geq1}=\left\{ \frac{1}{N}\sum_{i=1}^{N}\left(\widehat{W}_{1}^{T}\right)^{-2}\left\{ \partial_{u}\varpi(Y_{1},U_{1,i}^{T};\theta)\right\} ^{2};T\geq1\right\} ,
\]
is uniformly integrable. We use de la Vallée-Poussin theorem; i.e.,
$\left\{ X_{n};n\geq1\right\} $ is uniformly integrable if and only
if there exists a non-negative increasing convex function $g$ such
that $g(x)/x\rightarrow\infty$ as $x\rightarrow\infty$ and $\sup_{n\geq1}\widetilde{\mathbb{E}}\left[g\left(\left\vert X_{n}\right\vert \right)\right]<\infty$.

If $g$ is convex then by Jensen's inequality 
\begin{align*}
\widetilde{\mathbb{E}}\left[g\left(\frac{1}{N}\sum_{i=1}^{N}\left(\widehat{W}_{1}^{T}\right)^{-2}\left\{ \partial_{u}\varpi(Y_{1},U_{1,i}^{T};\theta)\right\} ^{2}\right)\right] & \leq\widetilde{\mathbb{E}}\left[\frac{1}{N}\sum_{i=1}^{N}g\left(\left(\widehat{W}_{1}^{T}\right)^{-2}\left\{ \partial_{u}\varpi(Y_{1},U_{1,i}^{T};\theta)\right\} ^{2}\right)\right]\\
 & =\widetilde{\mathbb{E}}\left[g\left(\left(\widehat{W}_{1}^{T}\right)^{-2}\left\{ \partial_{u}\varpi(Y_{1},U_{1,1};\theta)\right\} ^{2}\right)\right],
\end{align*}
since the variables $\{U_{1,i}^{T};i\in1:N\}$ are exchangeable under
$\overline{\pi}\left(\left.\cdot\right\vert \theta\right)$. Therefore
it will suffice to assume that for some non-negative, increasing convex
function $g$ such that $g(x)/x\rightarrow\infty$ 
\[
\lim\sup_{T}\text{ }\widetilde{\mathbb{E}}\left[g\left(\left(\widehat{W}_{1}^{T}\right)^{-2}\left\{ \partial_{u}\varpi(Y_{1},U_{1,1};\theta)\right\} ^{2}\right)\right]<\infty
\]
or equivalently that the family 
\[
\left\{ \left(\widehat{W}_{1}^{T}\right)^{-2}\left\{ \partial_{u}\varpi(Y_{1},U_{1,1};\theta)\right\} ^{2};T\geq1\right\} ,
\]
is uniformly integrable under $\widetilde{\mathbb{P}}$. However,
this is satisfied as there exists $\varepsilon>0$ such that 
\[
\lim\sup_{T}\text{ }\widetilde{\mathbb{E}}\left[\left(\left(\widehat{W}_{1}^{T}\right)^{-2}\left\{ \partial_{u}\varpi(Y_{1},U_{1,1};\theta)\right\} ^{2}\right)^{1+\varepsilon}\right]<\infty,
\]
which can be verified by using Cauchy-Schwarz inequality and Assumptions~\ref{ass:momentsofW}
and \ref{ass:dU4moments}. By applying now the WLLN, we have 
\[
\frac{1}{T}\sum_{t=1}^{T}\left(g_{T}\left(Y_{t},U_{t}\right)-\widetilde{\mathbb{E}}\left[g_{T}\left(Y_{t},U_{t}\right)\right]\right)\overset{\mathbb{P}}{\rightarrow}0,
\]
where 
\begin{align*}
\widetilde{\mathbb{E}}\left[g_{T}\left(Y_{1},U_{1}^{T}\right)\right] & =\widetilde{\mathbb{E}}\left[\frac{1}{N}\sum_{i=1}^{N}\left(\widehat{W}_{1}^{T}\right)^{-2}\left\{ \partial_{u}\varpi(Y_{1},U_{1,i}^{T};\theta)\right\} ^{2}\right]\\
 & =\widetilde{\mathbb{E}}\left[\left(\widehat{W}_{1}^{T}\right)^{-2}\left\{ \partial_{u}\varpi(Y_{1},U_{1,1};\theta)\right\} ^{2}\right]=\mathbb{E}\left[\left(\widehat{W}_{1}^{T}\right)^{-1}\left\{ \partial_{u}\varpi(Y_{1},U_{1,1};\theta)\right\} ^{2}\right].
\end{align*}

By Cauchy-Schwarz, Assumptions~\ref{ass:momentsofW} and \ref{ass:dU4moments}
again, we have 
\[
\lim\sup_{T}\text{ }\mathbb{E}\left\{ \left(\widehat{W}_{1}^{T}\right)^{-1-\varepsilon}\left\{ \partial_{u}\varpi(Y_{1},U_{1,1};\theta)\right\} ^{2+2\varepsilon}\right\} <\infty.
\]
Therefore the family $\{\left(\widehat{W}_{1}^{T}\right)^{-1}\left\{ \partial_{u}\varpi(Y_{1},U_{1,1};\theta)\right\} ^{2};T\geq1\}$
is also uniformly integrable under $\mathbb{P}$ and, since $\widehat{W}_{t}^{T}\overset{\mathbb{P}}{\rightarrow}1$,
we have 
\[
\mathbb{E}\left(\left(\widehat{W}_{1}^{T}\right)^{-1}\left\{ \partial_{u}\varpi(Y_{1},U_{1,1};\theta)\right\} ^{2}\right)\rightarrow\mathbb{E}\left(\left\{ \partial_{u}\varpi(Y_{1},U_{1,1};\theta)\right\} ^{2}\right)=\frac{\kappa\left(\theta\right)^{2}}{2}.
\]
Hence, it follows that $s_{T}^{2}\overset{\widetilde{\mathbb{P}}}{\rightarrow}\kappa\left(\theta\right)^{2}$
and condition (\ref{eq:sumofvars}) of Lemma \ref{lem:conditionalLindebergCLT}
is satisfied.

We now need to check the Lindeberg condition (\ref{eq:lind}), i.e.,
that for any $\varepsilon>0$ 
\begin{equation}
\sum_{t=1}^{T}\widetilde{\mathbb{E}}\left(\left.\left\vert M_{t}^{T}\right\vert ^{2}\mathbb{I}\left\{ \left\vert M_{t}^{T}\right\vert \geq\varepsilon\right\} \right\vert \mathcal{F}^{T}\right)\overset{\widetilde{\mathbb{P}}}{\rightarrow}0.\label{eq:LindebergMt}
\end{equation}

Since the l.h.s. of (\ref{eq:LindebergMt}) is non-negative, it is
enough to show that its unconditional expectation vanishes or equivalently
that $T\left\vert M_{1}^{T}\right\vert ^{2}$ is uniformly integrable.
We have 
\begin{align*}
\sum_{t=1}^{T}\widetilde{\mathbb{E}}\left(\left\vert M_{t}^{T}\right\vert ^{2}\mathbb{I}\left\{ \left\vert M_{t}^{T}\right\vert \geq\varepsilon\right\} \right) & =T\widetilde{\mathbb{E}}\left(\left\vert M_{1}^{T}\right\vert ^{2}\mathbb{I}\left\{ \left\vert M_{1}^{T}\right\vert \geq\varepsilon\right\} \right)\\
 & =T\widetilde{\mathbb{E}}\left[\left\{ \frac{\sqrt{2}}{N\widehat{W}_{1}^{T}}\sum_{i=1}^{N}\int_{0}^{\delta_{T}}\partial_{u}\varpi\left(Y_{1},U_{1,i}^{T}\left(s\right);\theta\right)\mathrm{d}B_{1,i}\left(s\right)\right\} ^{2}\mathbb{I}\left\{ \left\vert M_{1}^{T}\right\vert \geq\varepsilon\right\} \right]\\
 & =\frac{2T}{N^{2}}\widetilde{\mathbb{E}}\left[\left\{ \frac{\mathbb{I}\left\{ \left\vert M_{1}^{T}\right\vert \geq\varepsilon\right\} }{\left(\widehat{W}_{1}^{T}\right)^{3/2}}\right\} \frac{1}{\left(\widehat{W}_{1}^{T}\right)^{1/2}}\left\{ \sum_{i=1}^{N}\int_{0}^{\delta_{T}}\partial_{u}\varpi\left(Y_{1},U_{1,i}^{T}\left(s\right);\theta\right)\mathrm{d}B_{1,i}\left(s\right)\right\} ^{2}\right]\\
 & \leq\frac{2T}{N^{2}}\widetilde{\mathbb{E}}\left(\frac{\mathbb{I}\left\{ \left\vert M_{1}^{T}\right\vert \geq\varepsilon\right\} }{\left(\widehat{W}_{1}^{T}\right)^{3}}\right)^{1/2}\widetilde{\mathbb{E}}\left(\frac{1}{\left(\widehat{W}_{1}^{T}\right)}\left\{ \sum_{i=1}^{N}\int_{0}^{\delta_{T}}\partial_{u}\varpi\left(Y_{1},U_{1,i}^{T}\left(s\right);\theta\right)\mathrm{d}B_{1,i}\left(s\right)\right\} ^{4}\right)^{1/2}
\end{align*}
by Cauchy-Schwartz and 
\begin{align*}
&\widetilde{\mathbb{E}}\left(\frac{1}{\left(\widehat{W}_{1}^{T}\right)}\left\{ \sum_{i=1}^{N}\int_{0}^{\delta_{T}}\partial_{u}\varpi\left(Y_{1},U_{1,i}^{T}\left(s\right);\theta\right)\mathrm{d}B_{1,i}\left(s\right)\right\} ^{4}\right)\\
& =\mathbb{E}\left(\left\{ \sum_{i=1}^{N}\int_{0}^{\delta_{T}}\partial_{u}\varpi\left(Y_{1},U_{1,i}^{T}\left(s\right);\theta\right)\mathrm{d}B_{1,i}\left(s\right)\right\} ^{4}\right)\\
 & =\mathbb{E}\left(\mathbb{E}\left(\left.\left\{ \sum_{i=1}^{N}\int_{0}^{\delta_{T}}\partial_{u}\varpi\left(Y_{1},U_{1,i}^{T}\left(s\right);\theta\right)\mathrm{d}B_{1,i}\left(s\right)\right\} ^{4}\right\vert \mathcal{Y}^{T}\right)\right)\\
 & \leq cN^{2}\mathbb{E}\left[\left(\int_{0}^{\delta_{T}}\partial_{u}\varpi\left(Y_{1},U_{1,1}\left(s\right);\theta\right)\mathrm{d}B_{1,1}\left(s\right)\right)^{4}\right]\\
 & \leq cN^{2}\left\{ 3\int_{0}^{\delta_{T}}\mathbb{E}\left[\left(\partial_{u}\varpi\left(Y_{1},U_{1,1}\left(s\right);\theta\right)\right)^{4}\right]^{1/2}\mathrm{d}s\right\} ^{2}\\
 & =c^{\prime}N^{2}\delta_{T}^{2}\mathbb{E}\left[\left(\partial_{u}\varpi\left(Y_{1},U_{1,1};\theta\right)\right)^{4}\right]<\infty,
\end{align*}
where the penultimate inequality follows from \citep[Theorem 1]{Zakai1967}
and the last one by Assumption~\ref{ass:dU4moments}. Therefore,
we have 
\begin{align}
\sum_{t=1}^{T}\widetilde{\mathbb{E}}\left(\left\vert M_{t}^{T}\right\vert ^{2}\mathbb{I}\left\{ \left\vert M_{t}^{T}\right\vert \geq\varepsilon\right\} \right) & \leq\sqrt{c^{\prime}}\frac{2T}{N^{2}}N\delta_{T}\mathbb{E}\left[\left(\partial_{u}\varpi\left(Y_{1},U_{1,1};\theta\right)\right)^{4}\right]^{1/2}\text{ }\widetilde{\mathbb{E}}\left(\frac{\mathbb{I}\left\{ \left\vert M_{1}^{T}\right\vert \geq\varepsilon\right\} }{\left(\widehat{W}_{1}^{T}\right)^{3}}\right)^{1/2}\nonumber \\
 & =2\sqrt{c^{\prime}}\mathbb{E}\left[\left(\partial_{u}\varpi\left(Y_{1},U_{1,1};\theta\right)\right)^{4}\right]^{1/2}\text{ }\widetilde{\mathbb{E}}\left(\frac{\mathbb{I}\left\{ \left\vert M_{1}^{T}\right\vert \geq\varepsilon\right\} }{\left(\widehat{W}_{1}^{T}\right)^{3}}\right)^{1/2}.\label{eq:expressionLin1}
\end{align}

Using Holder's inequality, Assumption~\ref{ass:momentsofW} then
Chebyshev's inequality, we have 
\begin{align}
\widetilde{\mathbb{E}}\left(\frac{\mathbb{I}\left\{ \left\vert M_{1}^{T}\right\vert \geq\varepsilon\right\} }{\left(\widehat{W}_{t}^{T}\right)^{3}}\right) & \leq\widetilde{\mathbb{E}}\left[\left(\widehat{W}_{t}^{T}\right)^{-3-3\epsilon}\right]^{1/(1+\epsilon)}\text{ }\widetilde{{\mathbb{P}}}\left(|M_{1}^{T}|\geq\epsilon\right)^{\epsilon/(1+\epsilon)}\nonumber \\
 & \leq c^{\prime\prime}\text{ }\widetilde{{\mathbb{P}}}\left(|M_{1}^{T}|\geq\epsilon\right)^{\epsilon/(1+\epsilon)}\nonumber \\
 & \leq c^{\prime\prime}\text{ }\left(\frac{\widetilde{\mathbb{E}}\left[\left(M_{1}^{T}\right)^{2}\right]}{\epsilon^{2}}\right)^{\epsilon/(1+\epsilon)}.\label{eq:expressionlin2}
\end{align}
To proceed, we need to control the second moment of $M_{1}^{T}$ 
\begin{align}
\widetilde{\mathbb{E}}\left[\left(M_{1}^{T}\right)^{2}\right] & =\widetilde{\mathbb{E}}\left[\left(\frac{\sqrt{2}}{N\widehat{W}_{1}^{T}}\sum_{i=1}^{N}\int_{0}^{\delta_{T}}\partial_{u}\varpi\left(Y_{1},U_{1,i}^{T}\left(s\right);\theta\right)\mathrm{d}B_{1,i}\left(s\right)\right)^{2}\right]\nonumber \\
 & =\frac{2}{N^{2}}\mathbb{E}\left[\frac{1}{\widehat{W}_{1}^{T}}\left(\sum_{i=1}^{N}\int_{0}^{\delta_{T}}\partial_{u}\varpi\left(Y_{1},U_{1,i}^{T}\left(s\right);\theta\right)\mathrm{d}B_{1,i}\left(s\right)\right)^{2}\right]\nonumber \\
 & \leq\frac{2}{N^{2}}\mathbb{E}\left[(\widehat{W}_{1}^{T})^{-2}\right]^{1/2}\mathbb{E}\left[\left(\sum_{i=1}^{N}\int_{0}^{\delta_{T}}\partial_{u}\varpi\left(Y_{1},U_{1,i}^{T}\left(s\right);\theta\right)\mathrm{d}B_{1,i}\left(s\right)\right)^{4}\right]^{1/2}\nonumber \\
 & =\frac{2}{N^{2}}\mathbb{E}\left[(\widehat{W}_{1}^{T})^{-2}\right]^{1/2}\mathbb{E}\left[\mathbb{E}\left\{ \left.\left(\sum_{i=1}^{N}\int_{0}^{\delta_{T}}\partial_{u}\varpi\left(Y_{1},U_{1,i}^{T}\left(s\right);\theta\right)\mathrm{d}B_{1,i}\left(s\right)\right)^{4}\right\vert \mathcal{Y}^{T}\right\} \right]^{1/2}\nonumber \\
 & \leq c\frac{N}{N^{2}}\mathbb{E}\left[(\widehat{W}_{1}^{T})^{-2}\right]^{1/2}\mathbb{E}\left[\mathbb{E}\left\{ \left.\left[\int_{0}^{\delta_{T}}\partial_{u}\varpi\left(Y_{1},U_{1,i}^{T}\left(s\right);\theta\right)\mathrm{d}B_{1,i}\left(s\right)\right]^{4}\right\vert \mathcal{Y}^{T}\right\} \right]^{1/2}\nonumber \\
 & \leq c\frac{1}{N}\mathbb{E}\left[(\widehat{W}_{1}^{T})^{-2}\right]^{1/2}\mathbb{E}\left[\left(\int_{0}^{\delta_{T}}\partial_{u}\varpi\left(Y_{1},U_{1,1}\left(s\right);\theta\right)\mathrm{d}B_{1,1}\left(s\right)\right)^{4}\right]^{1/2}\nonumber \\
 & \leq\frac{c}{N}\mathbb{E}\left[(\widehat{W}_{1}^{T})^{-2}\right]^{1/2}\delta_{T}\mathbb{E}\left\{ \left\vert \partial_{u}\varpi\left(Y_{1},U_{1,1};\theta\right)\right\vert ^{4}\right\} ^{1/2}=O(1/T),\label{eq:varianceMt}
\end{align}
by Cauchy-Schwartz, \citep[Theorem 1]{Zakai1967} and Assumptions~\ref{ass:momentsofW}
and \ref{ass:dU4moments}.

By combining (\ref{eq:expressionLin1}), (\ref{eq:expressionlin2})
and (\ref{eq:varianceMt}), it follows that (\ref{eq:LindebergMt})
holds. Therefore by the Lindeberg central limit theorem of Lemma \ref{lem:conditionalLindebergCLT}
applied conditionally on $\mathcal{F}^{T}$ and using $s_{T}^{2}\overset{\widetilde{\mathbb{P}}}{\rightarrow}\kappa\left(\theta\right)^{2}$,
we obtain 
\[
\left.\sum_{t=1}^{T}M_{t}^{T}\right\vert \mathcal{F}^{T}\Rightarrow\mathcal{N}\left(0,\frac{\kappa\left(\theta\right)^{2}}{2}\right).
\]
\end{proof}
\begin{proof}[Proof of Lemma ~\textit{\ref{lem:zsquared}}]
\ Notice that 
\[
\frac{1}{2}\sum_{t=1}^{T}\left(\eta_{t}\right)^{2}=\frac{1}{2}\sum_{t=1}^{T}\left\{ \frac{\widehat{W}_{t}^{T}(Y_{t}\mid\theta^{\prime};V_{t})-\widehat{W}_{t}^{T}}{\widehat{W}_{t}^{T}}\right\} ^{2}=\frac{1}{2}\sum_{t=1}^{T}\left\{ J_{t}^{T}+H_{t}^{T}\right\} ^{2}=\frac{1}{2}\sum_{t=1}^{T}\left\{ \left[J_{t}^{T}\right]^{2}+\left[H_{t}^{T}\right]^{2}+2J_{t}^{T}H_{t}^{T}\right\} .
\]
We know from Lemma~\ref{lem:Jt} that $\sum_{t=1}^{T}(J_{t}^{T})^{2}$
$\overset{\widetilde{\mathbb{P}}}{\rightarrow}0$. The $\left(H_{t}^{T}\right)^{2}$
terms are given by 
\begin{align*}
\sum_{t=1}^{T}\left(H_{t}^{T}\right)^{2} & =\sum_{t=1}^{T}\left(L_{t}^{T}+M_{t}^{T}\right)^{2}\\
 & =\sum_{t=1}^{T}\left(L_{t}^{T}\right)^{2}+2L_{t}^{T}M_{t}^{T}+\left(M_{t}^{T}\right)^{2}.
\end{align*}
The first term vanishes in probability since by Lemma~\ref{lem:Lt}
\[
\widetilde{\mathbb{E}}\left({\textstyle \sum\nolimits _{t=1}^{T}}\left(L_{t}^{T}\right)^{2}\right)={\textstyle \sum\nolimits _{t=1}^{T}}\widetilde{\mathbb{V}}\left(L_{t}^{T}\right)\rightarrow0.
\]
For the product term notice that by two applications of the Cauchy-Schwartz
inequality 
\begin{align*}
\widetilde{\mathbb{E}}\left(\left\vert {\textstyle \sum\nolimits _{t=1}^{T}}L_{t}^{T}M_{t}^{T}\right\vert \right) & \leq\widetilde{\mathbb{E}}\left(\left\{ {\textstyle \sum\nolimits _{t=1}^{T}}\left(L_{t}^{T}\right)^{2}\right\} ^{1/2}\left\{ {\textstyle \sum\nolimits _{t=1}^{T}}\left(M_{t}^{T}\right)^{2}\right\} ^{1/2}\right)\\
 & \leq\widetilde{\mathbb{E}}\left({\textstyle \sum\nolimits _{t=1}^{T}}\left(L_{t}^{T}\right)^{2}\right)^{1/2}\widetilde{\mathbb{E}}\left({\textstyle \sum\nolimits _{t=1}^{T}}\left(M_{t}^{T}\right)^{2}\right)^{1/2}\\
 & =\left({\textstyle \sum\nolimits _{t=1}^{T}}\widetilde{\mathbb{V}}(L_{t}^{T})\right)^{1/2}\left({\textstyle \sum\nolimits _{t=1}^{T}}\widetilde{\mathbb{V}}(M_{t}^{T})\right)^{1/2}\rightarrow0,
\end{align*}
by Lemmas~\ref{lem:Lt} and \ref{lem:Mt}. Finally, we also have
\[
{\textstyle \sum\nolimits _{t=1}^{T}}\widetilde{\mathbb{E}}\left(\left(M_{t}^{T}\right)^{2}\right)=O\left(1\right)
\]
by Lemma \ref{lem:Mt}.

For the term involving the product $J_{t}^{T}H_{t}^{T}$, we have
similarly by two applications of the Cauchy-Schwartz inequality 
\begin{align*}
\widetilde{\mathbb{E}}\left(\left\vert {\textstyle \sum\nolimits _{t=1}^{T}}J_{t}^{T}H_{t}^{T}\right\vert \right) & \leq\widetilde{\mathbb{E}}\left[\left({\textstyle \sum\nolimits _{t=1}^{T}}\left(J_{t}^{T}\right)^{2}\right)^{1/2}\left({\textstyle \sum\nolimits _{t=1}^{T}}\left(H_{t}^{T}\right)^{2}\right)^{1/2}\right]\\
 & \leq\widetilde{\mathbb{E}}\left({\textstyle \sum\nolimits _{t=1}^{T}}\left(J_{t}^{T}\right)^{2}\right)^{1/2}\widetilde{\mathbb{E}}\left({\textstyle \sum\nolimits _{t=1}^{T}}\left(H_{t}^{T}\right)^{2}\right)^{1/2}.
\end{align*}
By Lemmas~\ref{lem:Jt}, \ref{lem:Lt} and \ref{lem:Mt}, the first
factor vanishes, while we have just shown that the second factor is
$O\left(1\right)$.

Finally, conditionally on $\mathcal{F}^{T},$ the terms $\left(M_{t}^{T}\right)^{2}$
are independent. We want to apply the conditional WLLN to show that
\[
{\textstyle \sum\nolimits _{t=1}^{T}}\left(M_{t}^{T}\right)^{2}-\widetilde{\mathbb{E}}\left(\left.\left(M_{t}^{T}\right)^{2}\right\vert \mathcal{F}^{T}\right)\overset{\widetilde{\mathbb{P}}}{\rightarrow}0.
\]
As we have already shown that Lemma \ref{lem:Mt} holds, we only need
to check that for any $\epsilon>0$ 
\[
{\textstyle \sum\nolimits _{t=1}^{T}}\widetilde{\mathbb{E}}\left(\left.\left(M_{t}^{T}\right)^{2}\mathbb{I}\left\{ \left\vert M_{t}^{T}\right\vert \geq\varepsilon\right\} \right\vert \mathcal{F}^{T}\right)\overset{\widetilde{\mathbb{P}}}{\rightarrow}0.
\]
This has already been established in the proof of Lemma \ref{lem:Mt}. 
\end{proof}

\subsection{Sufficient conditions to ensure Assumption \ref{Assumption:uniformCLT}
\label{Appendix:uniformCLT}}

We will provide here sufficient conditions to ensure convergence happens
almost surely, hence in probability. In the notation of Section\,\ref{Section:ProofCLTnightmare},
let $\mu^{T}$ denote the conditional law of 
\[
R^{T}:=\sum_{t=1}^{T}\log\left(\frac{\widehat{W}_{t}^{T}(Y_{t}\mid\theta^{\prime};U_{t}^{\prime})}{\widehat{W}_{t}^{T}}\right)=\sum_{t=1}^{T}\log(1+\eta_{t}^{T})=\sum_{t=1}^{T}\eta_{t}^{T}-\frac{1}{2}\sum_{t=1}^{T}[\eta_{t}^{T}]^{2}+\sum_{t=1}^{T}h(\eta_{t}^{T})[\eta_{t}^{T}]^{2},
\]
given $\mathcal{F}^{T}$ where $\theta,\xi$ are fixed, $\theta^{\prime}=\theta+\xi/\sqrt{T}$
, $\xi\sim\upsilon(\cdot)$, $U\sim\overline{\pi}_{T}\left(\left.\cdot\right\vert \theta\right)$,
$U^{\prime}\sim K_{\rho_{T}}\left(U,\cdot\right)$ with $\rho_{T}$
given by (\ref{eq:correlationscaling}) and $N_{T}\rightarrow\infty$
as $T\rightarrow\infty$ with $N_{T}/T\rightarrow0$. We want to control
the term 
\begin{align*}
 & \sup_{\theta\in N(\bar{\theta})}\mathbb{\widetilde{\mathbb{E}}}\left[\left.d_{BL}(\mu^{T},\varphi\left(\cdot;-\frac{\kappa^{2}(\theta)}{2},\kappa^{2}(\theta)\right)\right\vert \mathcal{Y}^{T}\right]\\
 & =\sup_{\theta\in N(\bar{\theta})}\iint\Bigg\{ \overline{\pi}_{T}\left(\left.\mathrm{d}u_{0}\right\vert \theta\right)\upsilon\left(\mathrm{d}\xi\right)\\
 &\qquad \times\underset{f:\left\Vert f\right\Vert _{BL}\leq1}{\sup}
 \bigg\vert \int K_{\rho_{T}}\left(u_{0},\mathrm{d}u_{1}\right)f\left\{ \log\left(\frac{\hat{p}(Y_{1:T}\mid\theta_{0}+\xi/\sqrt{T},u_{1})/p(Y_{1:T}\mid\theta+\xi/\sqrt{T})}{\hat{p}(Y_{1:T}\mid\theta,u_{0})//p(Y_{1:T}\mid\theta)}\right)\right\}\\
 &\qquad \qquad   -\int\varphi(\mathrm{d}w;-\frac{\kappa^{2}(\theta)}{2},\kappa^{2}(\theta))f\left(w\right)\bigg\vert \Bigg\}.
\end{align*}

We state sufficient conditions under which this result holds in the
setting where $d=1,$ $p=1$ and $\psi=1$. The extension to the multivariate
scenario is straightforward albeit tedious. As in Theorem \ref{Theorem:conditionalCLTthetathetacand},
we define 
\[
\kappa\left(\theta\right)^{2}=2\mathbb{E}\left(\left\{ \partial_{u}\varpi(Y_{1},U_{1,1};\theta)\right\} ^{2}\right).
\]
We will also write
\[
\kappa(y,\theta)^{2}=2\mathbb{E}\left(\left.\left\{ \partial_{u}\varpi(Y_{1},U_{1,1};\theta)\right\} ^{2}\right|Y_{1}=y\right).
\]

\begin{assumption}
\label{assu:unifcondclt}Let $B:\mathcal{Y}\to\mathbb{R}^{+}$ be
a measurable function such that $\mathbb{E}B(Y_{1})^{10}<\infty$,
and let $\epsilon_{T}\to0$ as $T\to\infty$. Assume that $\int\xi^{10}\upsilon(\mathrm{d}\xi)<\infty$,
that $\kappa^{2}(\cdot,\theta)$ is measurable for all $\theta$and
$\kappa^{2}(y,\cdot)$ is continuous in $\theta$ for all $y$, that
$\kappa(\theta)$ is locally Lispschitz around $\bar{\theta}$ and
that the following inequalities hold:
\begin{align}
 & \kappa^{2}(y,\theta)\leq B(y),\label{ass:condcltboundonkappa}\\
 & \sup_{\theta\in N(\bar{\theta})}\mathbb{\mathbb{E}}\left[\left.\left(\widehat{W}_{1}^{T}\right)^{-6}\right|Y_{1}=y\right]\leq B(y),\label{ass:condcltinversesixthmoment}\\
 & \sup_{\theta\in N(\bar{\theta})}\mathbb{\mathbb{E}}^{1/2}\left[\left.\left|\left(\widehat{W}_{t}^{T}\right)^{2}\right|\right|Y_{t}=y\right]\leq B(y),\label{ass:condcltsecondmoment}\\
 & \sup_{\theta\in N(\bar{\theta})}\mathbb{E}^{1/2}\left[\left.\left\{ 2\partial_{\theta}\varpi(y,U_{t,1};\theta)^{2}\right\} ^{2}\right|\right]\leq B(y),\label{ass:condcltvarofg}\\
 & \sup_{\theta\in N(\bar{\theta})}\mathbb{E}\left[\left.\left|\frac{1}{\widehat{W}_{t}^{T}}-1\right|^{2}\right|Y_{t}=y\right],\sup_{\theta\in N(\bar{\theta})}\mathbb{E}\left[\left.\left|\frac{1}{\left(\widehat{W}_{t}^{T}\right)^{2}}-1\right|^{2}\right|Y_{t}=y\right]\leq\epsilon_{T}B(y),\label{ass:condcltinverseweightto1}\\
 & \sup_{\theta\in N(\bar{\theta})}\mathbb{E}\left[\left.\left|\widehat{W}_{t}^{T}-1\right|^{2}\right|Y_{t}=y\right]\leq\epsilon_{T}B(y),\label{ass:condcltweightto1}\\
 & \sup_{\theta\in N(\bar{\theta})}\mathbb{E}\left[\left.\frac{\eta_{t}^{T}}{1+\eta_{t}^{T}}\right|Y_{t}=y\right]\leq B(y),\label{ass:condcltetat}\\
 & \sup_{\theta\in N(\bar{\theta})}\mathbb{E}\left[\left.\partial_{\theta}\varpi(Y_{t},U_{t,1};\theta)^{10}\right|Y_{t}=y\right]\leq B(y),\label{ass:condcltdthetaten}\\
 & \sup_{\theta\in N(\bar{\theta})}\mathbb{E}\left[\left.\left(\partial_{u}\varpi\left(Y_{t},U_{t,1};\theta\right)\right)^{10}\right|Y_{t}=y\right]\leq B(y),\label{ass:condcltduten}\\
 & \sup_{\theta\in N(\bar{\theta})}\mathbb{E}\left[\left.\left[\partial_{uuu}^{3}\varpi\left(Y_{t},U_{t,1}\left(0\right);\theta\right)\right]^{4}\right|Y_{t}=y\right]\leq B(y),\label{ass:condcltdddufour}\\
 & \sup_{\theta\in N(\bar{\theta})}\mathbb{E}\left[\left.\mathcal{S}\left\{ -\partial_{u}\varpi\left(Y_{t},U_{t,1};\theta\right)U_{t,1}+\partial_{u,u}^{2}\varpi\left(Y_{t},U_{t,1};\theta\right)\right\} ^{10}\right|Y_{t}=y\right]\leq B(y).\label{ass:condcltsteinten}
\end{align}
\\
\end{assumption}
Under Assumption\,\ref{assu:unifcondclt} then Assumption \ref{Assumption:uniformCLT}
is satisfied as established in the following theorem.
\begin{thm}
\label{thm:UniformCLT}Under Assumption\,\ref{assu:unifcondclt},
we have as $T\to\infty$ 
\begin{align*}
 & \sup_{\theta\in N(\bar{\theta})}\mathbb{\widetilde{\mathbb{E}}}\left[\left.d_{BL}(\mu^{T},\varphi\left(\cdot;-\frac{\kappa^{2}(\theta)}{2},\kappa^{2}(\theta)\right)\right\vert \mathcal{Y}^{T}\right]\to0\ \mathbb{P}^{Y}-\mathrm{a.s.}
\end{align*}
\end{thm}
The proof of this result will require establish a few preliminary
lemmas. Let us first recall the decomposition 
\begin{equation}
\eta_{t}^{T}=J_{t}^{T}+L_{t}^{T}+M_{t}^{T},\label{eq:decompositioneta}
\end{equation}
where $J_{t}^{T},L_{t}^{T}$ and $M_{t}^{T}$ are defined in (\ref{eq:Jtdef})-(\ref{eq:Mtdef}).
We rearrange the above expression as 
\begin{equation}
R^{T}=M^{T}-\frac{1}{2}\sum\left(\eta_{t}^{T}\right)^{2}+\mathcal{R}_{1}^{T},\label{eq:R1remainderdefinition}
\end{equation}
where $M^{T}:=\sum_{t=1}^{T}M_{t}^{T}=\sum_{t=1}^{T}\sum_{i=1}^{N}M_{t,i}^{T}$
where $M_{t,i}^{T}$ is defined in (\ref{eq:Mtidef}). 

We can further decompose $M_{t}^{T}$ as 
\begin{align*}
\sum_{t=1}^{T}M_{t}^{T} & =\sum_{t=1}^{T}\int_{0}^{\delta_{T}}\frac{\sqrt{2}}{N\widehat{W}_{t}^{T}}\sum_{i=1}^{N}\partial_{u}\varpi\left(Y_{t},U_{t,i}\left(s\right);\theta\right)\mathrm{d}B_{t,i}\left(s\right)\\
 & =\sum_{t=1}^{T}\int_{0}^{\delta_{T}}\frac{\sqrt{2}}{N\widehat{W}_{t}^{T}}\sum_{i=1}^{N}\partial_{u}\varpi\left(Y_{t},U_{t,i}\left(0\right);\theta\right)\mathrm{d}B_{t,i}\left(s\right)\\
 & \quad+\sum_{t=1}^{T}\int_{0}^{\delta_{T}}\frac{\sqrt{2}}{N\widehat{W}_{t}^{T}}\sum_{i=1}^{N}\left[-\int_{0}^{s}\partial_{uu}^{2}\varpi\left(Y_{t},U_{t,i}\left(r\right);\theta\right)U_{t,i}(r)\mathrm{d}r+\int_{0}^{s}\partial_{uuu}^{3}(Y_{t},U_{t,i}(r);\theta)\mathrm{d}r\right]\mathrm{d}B_{t,i}\left(s\right)\\
 & \quad+\sum_{t=1}^{T}\int_{0}^{\delta_{T}}\frac{\sqrt{2}}{N\widehat{W}_{t}^{T}}\sum_{i=1}^{N}\sqrt{2}\int_{0}^{s}\partial_{uu}^{2}\varpi\left(Y_{t},U_{t,i}\left(r\right);\theta\right)\mathrm{d}B_{t,i}(r)\mathrm{d}B_{t,i}\left(s\right)\\
 & =\sum_{t=1}^{T}\frac{\sqrt{2}}{N\widehat{W}_{t}^{T}}\sum_{i=1}^{N}\partial_{u}\varpi\left(Y_{t},U_{t,i}\left(0\right);\theta\right)B_{t,i}\left(\delta_{T}\right)+\mathcal{R}_{2}^{T},
\end{align*}
where 
\begin{align*}
\mathcal{R}_{2}^{T} & :=-\sum_{t=1}^{T}\frac{\sqrt{2}}{N\widehat{W}_{t}^{T}}\sum_{i=1}^{N}\int_{0}^{\delta_{T}}\int_{0}^{s}\mathcal{S}\left\{ \partial_{u}\varpi\left(Y_{t},U_{t,i}\left(r\right);\theta\right)\right\} \mathrm{d}r\mathrm{d}B_{t,i}\left(s\right)\\
 & \quad+\sum_{t=1}^{T}\frac{\sqrt{2}}{N\widehat{W}_{t}^{T}}\sum_{i=1}^{N}\int_{0}^{\delta_{T}}\int_{0}^{s}\sqrt{2}\int_{0}^{s}\mathbb{\widetilde{\mathbb{E}}}\left[\left.d_{BL}(\mu^{T},\varphi\left(\cdot;-\frac{\kappa^{2}(\theta)}{2},\kappa^{2}(\theta)\right)\right\vert \mathcal{Y}^{T}\right]\\
 &\qquad \qquad\partial_{uuu}^{3}\varpi\left(Y_{t},U_{t,i}\left(r\right);\theta\right)\mathrm{d}B_{t,i}(r)\mathrm{d}B_{t,i}\left(s\right).
\end{align*}
Thus we can write 
\begin{align*}
M^{T} & =\left\{ \sum_{i,t}\frac{2}{NT\left(\widehat{W}_{t}^{T}\right)^{2}}\left[\partial_{u}\varpi\left(Y_{t},U_{t,i}\left(0\right);\theta\right)\right]^{2}\right\} ^{1/2}Z+\mathcal{R}_{2}^{T}\\
 & =\hat{s}_{T}(\theta)Z+\mathcal{R}_{2}^{T},
\end{align*}
where $Z\sim\mathcal{N}(0,1)$. Finally let 
\[
\mathcal{R}_{3}^{T}:=\frac{1}{2}\mathbb{\widetilde{\mathbb{E}}}\left[\left.\left|\sum_{t=1}^{T}\left(\eta_{t}^{T}\right)^{2}-\kappa^{2}(\theta)\right|\right|\mathcal{F}^{T}\right].
\]

\begin{lem}
\label{lem:dblbound}We have 
\begin{align}
\mathbb{\widetilde{\mathbb{E}}}\left[\left.d_{BL}\left(\mu^{T},\varphi\left(\cdot;-\frac{\kappa^{2}(\theta)}{2},\kappa^{2}(\theta)\right)\right)\right\vert \mathcal{Y}^{T}\right] & \leq\sqrt{\frac{2}{\pi}}\mathbb{\widetilde{\mathbb{E}}}\left[\left.\left|\hat{s}_{T}(\theta)-\kappa(\theta)\right|\right\vert \mathcal{Y}^{T}\right]\nonumber \\
 & +\widetilde{\mathbb{E}}\left[\left.\left|\mathcal{R}_{1}^{T}\right|\right|\mathcal{Y}^{T}\right]+\mathbb{\widetilde{\mathbb{E}}}\left[\left.\left|\mathcal{R}_{2}^{T}\right|\right|\mathcal{Y}^{T}\right]+\widetilde{\mathbb{E}}\left[\left.\left|\mathcal{R}_{3}^{T}\right|\right|\mathcal{Y}^{T}\right].\label{eq:firstboundondbl}
\end{align}
\end{lem}
\begin{proof}
We first notice that if $f\in BL(1)$ we have
\begin{align*}
 & \left|\mathbb{\widetilde{\mathbb{E}}}\left[\left.f\left(R^{T}\right)\right|\mathcal{F}^{T}\right]-\int f(z)\varphi\left(z;-\frac{\kappa^{2}(\theta)}{2},\kappa^{2}(\theta)\right)\mathrm{d}z\right|\\
 & =\left|\mathbb{\widetilde{\mathbb{E}}}\left[\left.f\left(M^{T}-\frac{1}{2}\sum\left(\eta_{t}^{T}\right)^{2}+\mathcal{R}_{1}^{T}\right)\right|\mathcal{F}^{T}\right]-\int f(z)\varphi\left(z;-\frac{\kappa^{2}(\theta)}{2},\kappa^{2}(\theta)\right)\mathrm{d}z\right|\\
 & \leq\left|\mathbb{\widetilde{\mathbb{E}}}\left[\left.f\left(M^{T}-\frac{1}{2}\sum\left(\eta_{t}^{T}\right)^{2}+\mathcal{R}_{1}^{T}\right)\right|\mathcal{F}^{T}\right]-\mathbb{\mathbb{\widetilde{\mathbb{E}}}}\left[\left.f\left(M^{T}-\frac{1}{2}\sum\left(\eta_{t}^{T}\right)^{2}\right)\right|\mathcal{F}^{T}\right]\right|\\
 & \quad+\left|\mathbb{\mathbb{\widetilde{\mathbb{E}}}}\left[\left.f\left(M^{T}-\frac{1}{2}\sum\left(\eta_{t}^{T}\right)^{2}\right)\right|\mathcal{F}^{T}\right]-\int f(z)\varphi\left(z;-\frac{\kappa^{2}(\theta)}{2},\kappa^{2}(\theta)\right)\mathrm{d}z\right|\\
 & \leq\left|\mathbb{\mathbb{\widetilde{\mathbb{E}}}}\left[\left.f\left(M^{T}-\frac{1}{2}\sum\left(\eta_{t}^{T}\right)^{2}\right)\right|\mathcal{F}^{T}\right]-\int f(z)\varphi\left(z;-\frac{\kappa^{2}(\theta)}{2},\kappa^{2}(\theta)\right)\mathrm{d}z\right|+\left|\mathbb{E}\left[\left.\mathcal{R}_{1}^{T}\right|\mathcal{F}^{T}\right]\right|,
\end{align*}
since $|f(x)-f(y)|\leq|x-y|.$ Notice that for all $\theta$ the function
$f_{\theta}(z):=f(z-\kappa^{2}(\theta)/2)$ also belongs to $BL(1)$.
Continuing with our estimate we therefore have 
\begin{align*}
 & \left|\mathbb{\widetilde{\mathbb{E}}}\left[\left.f\left(M^{T}-\frac{1}{2}\sum\left(\eta_{t}^{T}\right)^{2}\right)\right|\mathcal{F}^{T}\right]-\int f(z)\varphi\left(z;-\frac{\kappa^{2}(\theta)}{2},\kappa^{2}(\theta)\right)\mathrm{d}z\right|\\
 & =\left|\mathbb{\widetilde{\mathbb{E}}}\left[\left.f_{\theta}\left(M^{T}-\frac{1}{2}\sum\left(\eta_{t}^{T}\right)^{2}+\frac{\kappa^{2}(\theta)}{2}\right)\right|\mathcal{F}^{T}\right]-\int f_{\theta}(z)\varphi\left(z;0,\kappa^{2}(\theta)\right)\mathrm{d}z\right|\\
 & \leq\left|\mathbb{\widetilde{\mathbb{E}}}\left[\left.f_{\theta}\left(M^{T}-\frac{1}{2}\sum\left(\eta_{t}^{T}\right)^{2}+\frac{\kappa^{2}(\theta)}{2}\right)\right|\mathcal{F}^{T}\right]-\mathbb{\widetilde{\mathbb{E}}}\left[\left.f_{\theta}\left(M^{T}\right)\right|\mathcal{F}^{T}\right]\right|\\
 & \quad+\left|\mathbb{\widetilde{\mathbb{E}}}\left[\left.f_{\theta}\left(M^{T}\right)\right|\mathcal{F}^{T}\right]-\int f_{\theta}(z)\varphi\left(z;0,\kappa^{2}(\theta)\right)\mathrm{d}z\right|\\
 & \leq\frac{1}{2}\mathbb{\widetilde{\mathbb{E}}}\left[\left.\left|\sum_{t=1}^{T}\left(\eta_{t}^{T}\right)^{2}-\kappa^{2}(\theta)\right|\right|\mathcal{F}^{T}\right]+\left|\mathbb{\widetilde{\mathbb{E}}}\left[\left.f_{\theta}\left(M^{T}\right)\right|\mathcal{Y}^{T}\right]-\int f_{\theta}(z)\varphi\left(z;0,\kappa^{2}(\theta)\right)\mathrm{d}z\right|\\
 & \leq\left|\mathbb{\widetilde{\mathbb{E}}}\left[\left.\mathcal{R}_{3}^{T}\right|\mathcal{F}^{T}\right]\right|+\left|\mathbb{\widetilde{\mathbb{E}}}\left[\left.\mathcal{R}_{2}^{T}\right|\mathcal{F}^{T}\right]\right|+\left|\int f_{\theta}(z)\varphi\left(z;0,\hat{s}_{T}^{2}(\theta)\right)\mathrm{d}z-\int f_{\theta}(z)\varphi\left(z;0,\kappa^{2}(\theta)\right)\mathrm{d}z\right|
\end{align*}
Collecting terms and taking supremum over $BL(1)$, we have 
\begin{align*}
d_{BL}\left(\mu^{T},\mathcal{\varphi}\left(\cdot;-\frac{\kappa^{2}(\theta)}{2},\kappa^{2}(\theta)\right)\right) & :=\sup_{f\in BL(1)}\left|\mathbb{\widetilde{\mathbb{E}}}\left[\left.f(R^{T})\right|\mathcal{F}^{T}\right]-\int f(z)\varphi\left(z;-\frac{\kappa^{2}(\theta)}{2},\kappa^{2}(\theta)\right)\mathrm{d}z\right|\\
 & \leq\sup_{f\in BL(1)}\left|\int f_{\theta}(z)\varphi\left(z;0,\hat{s}_{T}^{2}(\theta)\right)\mathrm{d}z-\int f_{\theta}(z)\varphi\left(z;0,\kappa^{2}(\theta)\right)\mathrm{d}z\right|+\\
 & +\left|\mathbb{E}\left[\left.\mathcal{R}_{1}^{T}\right|\mathcal{F}^{T}\right]\right|+\left|\mathbb{E}\left[\left.\mathcal{R}_{2}^{T}\right|\mathcal{F}^{T}\right]\right|+\left|\mathbb{E}\left[\left.\mathcal{R}_{3}^{T}\right|\mathcal{F}^{T}\right]\right|\\
 & \leq\sqrt{\frac{2}{\pi}}\left|\hat{s}_{T}(\theta)-\kappa(\theta)\right|+\left|\mathbb{E}\left[\left.\mathcal{R}_{1}^{T}\right|\mathcal{F}^{T}\right]\right|+\left|\mathbb{E}\left[\left.\mathcal{R}_{2}^{T}\right|\mathcal{F}^{T}\right]\right|+\left|\mathbb{E}\left[\left.\mathcal{R}_{3}^{T}\right|\mathcal{F}^{T}\right]\right|
\end{align*}
since $\left\{ f_{\theta}:f\in BL(1)\right\} =BL(1)$. The result
follows by taking the conditional expectation w.r.t. $\mathcal{Y}^{T}$
and elementary manipulations. 
\end{proof}
We now need to control the four terms appearing on the r.h.s. of (\ref{eq:firstboundondbl}).
This is done in the following four subsections.

\subsubsection{Control of $\left|\hat{s}_{T}(\theta)-\kappa(\theta)\right|$}
\begin{lem}
\label{lem:controlst-kappa}As $T\to\infty$ we have that 
\[
\sup_{\theta\in N(\bar{\theta})}\widetilde{\mathbb{E}}\left[\left.\left|\hat{s}_{T}(\theta)-\kappa(\theta)\right|\right|\mathcal{Y}^{T}\right]\to0\ \mathbb{P}^{Y}-\mathrm{a.s.}
\]
\end{lem}
\begin{proof}
This result is established as follows. For any real numbers $\alpha,\beta$,
we have
\[
\sqrt{\alpha^{2}}-\sqrt{\beta^{2}}\leq\sqrt{\left|\alpha^{2}-\beta^{2}\right|},
\]
thus 
\[
\widetilde{\mathbb{E}}\left[\left.\left|\hat{s}_{T}(\theta)-\kappa(\theta)\right|\right|\mathcal{Y}^{T}\right]\leq\widetilde{\mathbb{E}}\left[\left.\sqrt{\left|\hat{s}_{T}^{2}(\theta)-\kappa^{2}(\theta)\right|}\right|\mathcal{Y}^{T}\right]\leq\widetilde{\mathbb{E}}^{1/2}\left[\left.\left|\hat{s}_{T}^{2}(\theta)-\kappa^{2}(\theta)\right|\right|\mathcal{Y}^{T}\right].
\]
We will control the last term in the above expression. Let us write
$g(y,u,\theta):=\left[\partial_{u}\varpi(y;u,\theta)\right]^{2}$
and define 
\[
\kappa^{2}(y,\theta):=2\mathbb{E}g(y,U_{1,1},\theta).
\]
Therefore $\kappa^{2}(\theta)=\mathbb{E}\kappa^{2}(Y_{1},\theta)$.
We next compute 
\begin{align*}
\hat{s}_{T}^{2}(\theta)-\kappa^{2}(\theta) & =\sum_{t=1}^{T}\frac{1}{T}\sum_{i=1}^{N}\frac{1}{N\left(\widehat{W}_{t}^{T}\right)^{2}}\left[2g\left(Y_{t},U_{t,i},\theta\right)-\kappa^{2}(Y_{t},\theta)\right]+\sum_{t=1}^{T}\frac{1}{T}\left[\frac{\kappa^{2}(Y_{t},\theta)}{\left(\widehat{W}_{t}^{T}\right)^{2}}-\kappa^{2}(\theta)\right].
\end{align*}
First we notice that 
\begin{align*}
 & \widetilde{\mathbb{E}}\left[\left.\left|\frac{1}{N\left(\widehat{W}_{t}^{T}\right)^{2}}\sum_{i=1}^{N}\left[2g\left(Y_{t},U_{t,i},\theta\right)-\kappa^{2}(Y_{t},\theta)\right]\right|\right|\mathcal{Y}^{T}\right]\\
 & =\mathbb{\mathbb{E}}\left[\left.\left|\frac{1}{N\left(\widehat{W}_{t}^{T}\right)}\sum_{i=1}^{N}\left[2g\left(Y_{t},U_{t,i},\theta\right)-\kappa^{2}(Y_{t},\theta)\right]\right|\right|\mathcal{Y}^{T}\right]\\
 & \leq\mathbb{\mathbb{E}}^{1/2}\left[\left.\left(\widehat{W}_{t}^{T}\right)^{-2}\right|\mathcal{Y}^{T}\right]\mathbb{\mathbb{E}}^{1/2}\left[\left.\left(\frac{1}{N}\sum_{i=1}^{N}\left[2g\left(Y_{t},U_{t,i},\theta\right)-\kappa^{2}(Y_{t},\theta)\right]\right)^{2}\right|\mathcal{Y}^{T}\right]\\
 & =\frac{1}{\sqrt{N}}\mathbb{\mathbb{E}}^{1/2}\left[\left.\left(\widehat{W}_{t}^{T}\right)^{-2}\right|Y_{t}\right]\mathbb{V}^{1/2}\left[\left.2g(Y_{t},U_{t,1},\theta)\right|Y_{t}\right],
\end{align*}
since the terms are mean zero and independent over $i$. Therefore
by Assumptions\ref{ass:condcltboundonkappa} and \ref{ass:condcltvarofg}
\begin{align*}
 & \widetilde{\mathbb{E}}\left[\left.\left|\frac{1}{T}\sum_{t=1}^{T}\frac{1}{N\left(\widehat{W}_{t}^{T}\right)^{2}}\sum_{i=1}^{N}\left[2g\left(Y_{t},U_{t,i},\theta\right)-\kappa^{2}(Y_{t},\theta)\right]\right|\right|\mathcal{Y}^{T}\right]\\
 & \leq\frac{1}{T}\sum_{t=1}^{T}\mathbb{\mathbb{E}}\left[\left.\left|\frac{1}{N\left(\widehat{W}_{t}^{T}\right)}\sum_{i=1}^{N}\left[2g\left(Y_{t},U_{t,i},\theta\right)-\kappa^{2}(Y_{t},\theta)\right]\right|\right|\mathcal{Y}^{T}\right]\\
 & \leq\frac{1}{\sqrt{N}}\frac{1}{T}\sum_{t=1}^{T}\mathbb{\mathbb{E}}^{1/2}\left[\left.\left(\widehat{W}_{t}^{T}\right)^{-2}\right|Y_{t}\right]\mathbb{V}^{1/2}\left[\left.2g(Y_{t},U_{t,1},\theta)\right|Y_{t}\right].\\
 & \leq\frac{1}{\sqrt{N}}\frac{1}{T}\sum_{t=1}^{T}B(Y_{t})^{1/3+1}.
\end{align*}
Continuing we have to control the remainder term 
\begin{align*}
 & \widetilde{\mathbb{E}}\left[\left.\left|\frac{1}{T}\sum_{t=1}^{T}\left[\frac{\kappa^{2}(Y_{t},\theta)}{\left(\widehat{W}_{t}^{T}\right)^{2}}-\kappa^{2}(\theta)\right]\right|\right|\mathcal{Y}^{T}\right]\\
 & \leq\widetilde{\mathbb{E}}\left[\left.\left|\frac{1}{T}\sum_{t=1}^{T}\left[\frac{\kappa^{2}(Y_{t},\theta)}{\left(\widehat{W}_{t}^{T}\right)^{2}}-\kappa^{2}(Y_{t},\theta)\right]\right|\right|\mathcal{Y}^{T}\right]+\left|\sum_{t=1}^{T}\frac{1}{T}\left[\kappa^{2}(Y_{t},\theta)-\kappa^{2}(\theta)\right]\right|\\
 & \leq\sum_{t=1}^{T}\frac{1}{T}\kappa^{2}(Y_{t},\theta)\mathbb{E}\left[\left.\left|\frac{1}{\widehat{W}_{t}^{T}}-\widehat{W}_{t}^{T}\right|\right|Y_{t}\right]+\left|\sum_{t=1}^{T}\frac{1}{T}\left[\kappa^{2}(Y_{t},\theta)-\kappa^{2}(\theta)\right]\right|\\
 & \leq\sum_{t=1}^{T}\frac{1}{T}\kappa^{2}(Y_{t},\theta)\left\{ \mathbb{E}\left[\left.\left|\frac{1}{\widehat{W}_{t}^{T}}-1\right|\right|Y_{t}\right]+\mathbb{E}\left[\left.\left|\widehat{W}_{t}^{T}-1\right|\right|Y_{t}\right]\right\} +\left|\sum_{t=1}^{T}\frac{1}{T}\left[\kappa^{2}(Y_{t},\theta)-\kappa^{2}(\theta)\right]\right|\\
 & \leq\sum_{t=1}^{T}\frac{2}{T}\kappa^{2}(Y_{t},\theta)\epsilon_{T}B(Y_{t})+\left|\sum_{t=1}^{T}\frac{1}{T}\left[\kappa^{2}(Y_{t},\theta)-\kappa^{2}(\theta)\right]\right|,
\end{align*}
by (\ref{ass:condcltinverseweightto1}) and (\ref{ass:condcltweightto1}).
Finally, by assumption $\kappa^{2}(y,\theta)$, defined for $\theta\in N(\bar{\theta})$,
is continuous in $\theta$ for all $y$, and a measurable function
of $y$ for each $\theta$ and Assumption \ref{ass:condcltboundonkappa}
ensures that $\kappa^{2}(y,\theta)\leq B(y)$ for all $y\in\mathcal{Y}$
and $\theta\in N(\bar{\theta})$ and we also have $\mathbb{E}B(Y_{1})<\infty$
by assumption. Thus, by \citep[Theorem 2,][]{Jennrich1969} it follows
that as $T\to\infty$ 
\begin{equation}
\sup_{\theta\in N(\bar{\theta})}\left|\sum_{t=1}^{T}\frac{1}{T}\left[\kappa^{2}(Y_{t},\theta)-\kappa^{2}(\theta)\right]\right|\to0\ \mathbb{P}^{Y}-\mathrm{a.s.},\label{eq:uniformSLLN}
\end{equation}
In addition we have that as $T\to\infty$ 
\[
\sum_{t=1}^{T}\frac{2}{T}\kappa^{2}(Y_{t},\theta)\epsilon_{T}B(Y_{t})\leq\frac{2\epsilon_{T}}{T}\sum_{t=1}^{T}B^{2}(Y_{t})\to0\ \mathbb{P}^{Y}-\mathrm{a.s.}
\]
\end{proof}

\subsubsection{Control of $\mathcal{R}_{1}^{T}$}
\begin{lem}
\label{lem:controlR1}As $T\to\infty$ we have that 
\[
\sup_{\theta\in N(\bar{\theta})}\mathbb{\widetilde{\mathbb{E}}}\left[\left.\left|\mathcal{R}_{1}^{T}\right|\right|\mathcal{Y}^{T}\right]\to0\ \mathbb{P}^{Y}-\mathrm{a.s.}
\]
\end{lem}
\begin{proof}
It follows from (\ref{eq:R1remainderdefinition}) that
\[
\mathcal{R}_{1}^{T}:=\sum_{t=1}^{T}J_{t}^{T}+\sum_{t=1}^{T}L_{t}^{T}+\sum_{t=1}^{T}h(\eta_{t}^{T})[\eta_{t}^{T}]^{2}
\]
with $h(x)=o(x)$. %
Recall that $h$ was defined through the Taylor expansion
\begin{align*}
\log(1+x)= & x-\frac{x^{2}}{2}+h(x)x^{2}\\
= & x-\frac{x^{2}}{2}+\int_{0}^{x}\left[\frac{y^{2}}{1+y}\right]\mathrm{d}y.
\end{align*}
Therefore we can write 
\begin{align}
\sum_{t=1}^{T}\widetilde{\mathbb{E}}\left[\left.h(\eta_{t}^{T})\left[\eta_{t}^{T}\right]^{2}\right|Y_{t}\right] & =\sum_{t=1}^{T}\widetilde{\mathbb{E}}\left[\left.\int_{0}^{\eta_{t}^{T}}\left[\frac{y^{2}}{1+y}\right]\mathrm{d}y\right|Y_{t}\right]\nonumber \\
 & \leq\sum_{t=1}^{T}\widetilde{\mathbb{E}}\left[\left.\left[\int_{0}^{\eta_{t}^{T}}y^{4}\mathrm{d}y\right]^{1/2}\left[\int_{0}^{\eta_{t}^{T}}\frac{\mathrm{d}y}{\left(1+y\right)^{2}}\right]^{1/2}\right|Y_{t}\right]\nonumber \\
 & \leq C\sum_{t=1}^{T}\mathbb{E}\left[\left.\widehat{W}_{t}^{T}\left[\eta_{t}^{T}\right]^{5/2}\frac{\left(\eta_{t}^{T}\right)^{1/2}}{\left(1+\eta_{t}^{T}\right)^{1/2}}\right|Y_{t}\right]\nonumber \\
 & \leq C\sum_{t=1}^{T}\mathbb{E}^{1/2}\left[\left.\left[\widehat{W}_{t}^{T}\right]^{2}\left[\eta_{t}^{T}\right]^{5}\right|Y_{t}\right]\mathbb{E}^{1/2}\left[\left.\frac{\eta_{t}^{T}}{\left(1+\eta_{t}^{T}\right)}\right|Y_{t}\right]\nonumber \\
 & \leq C\sum_{t=1}^{T}\mathbb{E}^{1/2}\left[\left.\left[\widehat{W}_{t}^{T}\right]^{2}\left[\eta_{t}^{T}\right]^{5}\right|Y_{t}\right]B(Y_{t})^{1/2}\label{eq:condcltfinalass1}
\end{align}
by (\ref{ass:condcltetat}). Letting $\widetilde{\eta}_{t}^{T}:=\widehat{W}_{t}^{T}\eta_{t}^{T}$
we have
\begin{align}
 & \mathbb{E}^{1/2}\left[\left.\left[\widehat{W}_{t}^{T}\right]^{2}\left[\eta_{t}^{T}\right]^{5}\right|Y_{t}\right]\nonumber \\
 & =\mathbb{E}^{1/2}\left[\left.\left[\widehat{W}_{t}^{T}\right]^{2-5}\left[\widetilde{\eta}_{t}^{T}\right]^{5}\right|Y_{t}\right]\nonumber \\
 & \leq\mathbb{E}^{1/4}\left[\left.\left[\widehat{W}_{t}^{T}\right]^{-6}\right|Y_{t}\right]\mathbb{E}^{1/4}\left[\left.\left[\widetilde{\eta}_{t}^{T}\right]^{10}\right|Y_{t}\right],\label{eq:condcltfinalass2}
\end{align}
and 
\begin{align}
\mathbb{E}\left[\left.\left[\widetilde{\eta}_{t}^{T}\right]^{10}\right|Y_{t}\right] & \leq C\left\{ \mathbb{E}\left[\left.\left[\widetilde{J}_{t}^{T}\right]^{10}\right|Y_{t}\right]+\mathbb{E}\left[\left.\left[\widetilde{L}_{t}^{T}\right]^{10}\right|Y_{t}\right]+\mathbb{E}\left[\left.\left[\widetilde{M}_{t}^{T}\right]^{10}\right|Y_{t}\right]\right\} ,\label{eq:condcltfinalass3}
\end{align}
where $\widetilde{J}_{t}^{T}:=\widehat{W}_{t}^{T}J_{t}^{T}$, $\widetilde{L}_{t}^{T}:=\widehat{W}_{t}^{T}L_{t}^{T}$
and $\widetilde{M}_{t}^{T}:=\widehat{W}_{t}^{T}M_{t}^{T}$. We now
control the terms $\widetilde{\mathbb{E}}\left[\left.\left|\sum_{t=1}^{T}J_{t}^{T}\right|\right|Y_{t}\right]^{2}$
and $\widetilde{\mathbb{E}}\left[\left.\left|\sum_{t=1}^{T}L_{t}^{T}\right|\right|Y_{t}\right]^{2}$
and the terms on the r.h.s. of (\ref{eq:condcltfinalass3}).

\emph{Terms} $\widetilde{J}_{t}^{T}$. Since $\widetilde{\mathbb{E}}\left[\left.J_{t}^{T}\right|Y_{t}\right]=0$
, and the terms are independent over $t$ we have
\begin{align*}
\widetilde{\mathbb{E}}\left[\left.\left|\sum_{t=1}^{T}J_{t}^{T}\right|\right|Y_{t}\right]^{2} & \leq\widetilde{\mathbb{E}}\left[\left.\left(\sum_{t=1}^{T}J_{t}^{T}\right)^{2}\right|Y_{t}\right]=\sum_{t=1}^{T}\widetilde{\mathbb{E}}\left[\left.\left(J_{t}^{T}\right)^{2}\right|Y_{t}\right]\\
 & =\sum_{t=1}^{T}\mathbb{E}\left[\left.\left[\widehat{W}_{t}^{T}\right]^{-1}\left(\widetilde{J}_{t}^{T}\right)^{2}\right|Y_{t}\right]\\
 & \leq\sum_{t=1}^{T}\mathbb{E}^{1/2}\left[\left.\left[\widehat{W}_{t}^{T}\right]^{-2}\right|Y_{t}\right]\mathbb{E}^{1/2}\left[\left.\left(\widetilde{J}_{t}^{T}\right)^{4}\right|Y_{t}\right]\\
 & \leq\sum_{t=1}^{T}B(Y_{t})^{1/2}\mathbb{E}^{1/5}\left[\left.\left(\widetilde{J}_{t}^{T}\right)^{10}\right|Y_{t}\right],
\end{align*}
by Holder's inequality. We thus have to control 
\begin{align*}
\mathbb{E}\left[\left.\left[\widetilde{J}_{t}^{T}\right]^{10}\right|Y_{t}\right] & \leq\mathbb{E}\left[\left.\left[\frac{1}{N}\sum_{i=1}^{N}\{\varpi(Y_{t},U_{t,i}\left(\delta_{T}\right);\theta+\xi/\sqrt{T})-\varpi(Y_{t},U_{t,i}\left(\delta_{T}\right);\theta)\}\right]^{10}\right|Y_{t}\right]\\
 & =\mathbb{E}\left[\left.\left[\frac{1}{N}\sum_{i=1}^{N}\{\varpi(Y_{t},U_{t,i};\theta+\xi/\sqrt{T})-\varpi(Y_{t},U_{t,i};\theta)\}\right]^{10}\right|Y_{t}\right].
\end{align*}
Since the terms are i.i.d. over $i$ and have zero mean we will use
the following fact: let $X_{1},\dots,X_{N}$ be i.i.d. and zero mean,
then 
\begin{align*}
\mathbb{E}\left[\left(\sum X_{i}\right)^{10}\right] & =\sum_{k=1}^{5}\sum_{i_{1}\neq\cdots\neq i_{k}}^{N}\sum_{\alpha\in A(k)}\prod_{j=1}^{k}\mathbb{E}\left[X_{i_{j}}^{2+\alpha_{j}}\right],
\end{align*}
where 
\[
A(k)=\left\{ \left(\alpha_{1},\dots,\alpha_{k}\right):\ \alpha_{1}+\cdots+\alpha_{k}=10-2k\right\} .
\]
Using Holder's inequality notice that since the factors are i.i.d.
we have
\begin{align*}
\prod_{j=1}^{k}\mathbb{E}\left[X_{i_{j}}^{2+\alpha_{j}}\right] & \leq\prod_{j=1}^{k}\mathbb{E}\left[X_{i_{j}}^{10}\right]^{\left(2+\alpha_{j}\right)/10}=\mathbb{E}\left[X_{1}^{10}\right].
\end{align*}
Therefore overall we have
\begin{align*}
\mathbb{E}\left[\left(\sum X_{i}\right)^{10}\right] & \leq\mathbb{E}\left[X_{1}^{10}\right]\sum_{k=1}^{5}\binom{N}{k}C(k)\leq C\mathbb{E}\left[X_{1}^{10}\right]\sum_{k=1}^{5}\binom{N}{k}\leq CN^{5}\mathbb{E}\left[X_{1}^{10}\right],
\end{align*}
since $C(k):=\sharp A(k)$ are combinatorial factors not depending
on $N$. Thus 
\begin{align*}
\mathbb{E}\left[\left.\left[\widetilde{J}_{t}^{T}\right]^{10}\right|Y_{t},\xi\right] & \leq\mathbb{E}\left[\left.\left[\frac{1}{N}\sum_{i=1}^{N}\{\varpi(Y_{t},U_{t,i};\theta+\xi/\sqrt{T})-\varpi(Y_{t},U_{t,i};\theta)\}\right]^{10}\right|Y_{t},\xi\right]\\
 & \leq\frac{C}{N^{10}}N^{5}\mathbb{E}\left[\left.\left(\varpi(Y_{t},U_{t,1};\theta+\xi/\sqrt{T})-\varpi(Y_{t},U_{t,1};\theta)\right)^{10}\right|Y_{t},\xi\right]\\
 & =\frac{C}{N^{10}}N^{5}\mathbb{E}\left[\left.\left(\int_{0}^{\xi/\sqrt{T}}\partial_{\theta}\varpi(Y_{t},U_{t,1};\theta+s)\mathrm{d}s\right)^{10}\right|Y_{t},\xi\right]\\
 & =\frac{C}{N^{5}}\left(\frac{\xi}{\sqrt{T}}\right)^{10}\mathbb{E}\left[\left.\left(\int_{0}^{\xi/\sqrt{T}}\partial_{\theta}\varpi(Y_{t},U_{t,1};\theta+s)\frac{\mathrm{d}s}{\xi/\sqrt{T}}\right)^{10}\right|Y_{t},\xi\right]\\
 & \leq\frac{C}{N^{5}}\left(\frac{\xi}{\sqrt{T}}\right)^{9}\mathbb{E}\left[\left.\int_{0}^{\xi/\sqrt{T}}\partial_{\theta}\varpi(Y_{t},U_{t,1};\theta+s)^{10}\mathrm{d}s\right|Y_{t},\xi\right]\\
 & =\frac{C}{N^{5}}\left(\frac{\xi}{\sqrt{T}}\right)^{9}\int_{0}^{\xi/\sqrt{T}}\mathbb{E}\left[\left.\partial_{\theta}\varpi(Y_{t},U_{t,1};\theta+s)^{10}\right|Y_{t}\right]\mathrm{d}s\\
 & \leq\frac{C}{N^{5}}\left(\frac{\xi}{\sqrt{T}}\right)^{10}B(Y_{t}),
\end{align*}
by (\ref{ass:condcltdthetaten}).

Since $\mathbb{E}[\xi^{10}]<\infty,$ we conclude that 
\begin{align}
\mathbb{E}\left[\left.\left[\widetilde{J}_{t}^{T}\right]^{10}\right|Y_{t}\right] & =\mathbb{E}\left[\left.\mathbb{E}\left[\left.\left(\widetilde{J}_{t}^{T}\right)^{10}\right|Y_{t},\xi\right]\right|Y_{t}\right]\nonumber \\
 & \leq\frac{C}{N^{5}T^{5}}B(Y_{t}),\label{eq:controlJ2}
\end{align}

Therefore we have 
\begin{align}
\widetilde{\mathbb{E}}\left[\left.\left|\sum_{t=1}^{T}J_{t}^{T}\right|\right|Y_{t}\right]^{2} & \leq\sum_{t=1}^{T}B(Y_{t})^{1/2}\mathbb{E}^{1/5}\left[\left.\left(\widetilde{J}_{t}^{T}\right)^{10}\right|Y_{t}\right]\nonumber \\
 & \leq\frac{1}{NT}\sum_{t=1}^{T}B(Y_{t})^{1/2}B(Y_{t})^{1/5}.\label{eq:controlJttilde}
\end{align}

\emph{Terms} $\widetilde{L}_{t}^{T}$. Since $\widetilde{\mathbb{E}}\left[\left.L_{t}^{T}\right|Y_{t}\right]=0$
, and the terms are independent over $t$ we have
\begin{align*}
\widetilde{\mathbb{E}}\left[\left.\left|\sum_{t=1}^{T}L_{t}^{T}\right|\right|Y_{t}\right]^{2} & \leq\widetilde{\mathbb{E}}\left[\left.\left(\sum_{t=1}^{T}L_{t}^{T}\right)^{2}\right|Y_{t}\right]=\sum_{t=1}^{T}\widetilde{\mathbb{E}}\left[\left.\left(L_{t}^{T}\right)^{2}\right|Y_{t}\right]\\
 & =\sum_{t=1}^{T}\mathbb{E}\left[\left.\left[\widehat{W}_{t}^{T}\right]^{-1}\left(\widetilde{L}_{t}^{T}\right)^{2}\right|Y_{t}\right]\\
 & \leq\sum_{t=1}^{T}\mathbb{E}^{1/2}\left[\left.\left[\widehat{W}_{t}^{T}\right]^{-2}\right|Y_{t}\right]\mathbb{E}^{1/2}\left[\left.\left(\widetilde{L}_{t}^{T}\right)^{4}\right|Y_{t}\right]\\
 & \leq\sum_{t=1}^{T}B(Y_{t})^{1/2}\mathbb{E}^{1/5}\left[\left.\left(\widetilde{L}_{t}^{T}\right)^{10}\right|Y_{t}\right].
\end{align*}

To proceed we estimate
\begin{align}
\mathbb{E}\left[\left.\left[\widetilde{L}_{t}^{T}\right]^{10}\right|Y_{t}\right] & =\mathbb{E}\left[\left.\left(\frac{1}{N}\sum_{i=1}^{N}\int_{0}^{\delta_{T}}\left\{ -\partial_{u}\varpi\left(Y_{t},U_{t,i}\left(s\right);\theta\right)U_{t,i}\left(s\right)+\partial_{u,u}^{2}\varpi\left(Y_{t},U_{t,i}\left(s\right);\theta\right)\right\} \mathrm{d}s\right)^{10}\right|Y_{t}\right]\nonumber \\
 & \leq\frac{C}{N^{5}}\mathbb{E}\left[\left.\left(\int_{0}^{\delta_{T}}\left\{ -\partial_{u}\varpi\left(Y_{t},U_{t,1}\left(s\right);\theta\right)U_{t,1}\left(s\right)+\partial_{u,u}^{2}\varpi\left(Y_{t},U_{t,1}\left(s\right);\theta\right)\right\} \mathrm{d}s\right)^{10}\right|Y_{t}\right]\nonumber \\
 & \leq\frac{C}{N^{5}}\delta_{T}^{10}\mathbb{E}\left[\left.\left(\int_{0}^{\delta_{T}}\left\{ -\partial_{u}\varpi\left(Y_{t},U_{t,1}\left(s\right);\theta\right)U_{t,1}\left(s\right)+\partial_{u,u}^{2}\varpi\left(Y_{t},U_{t,1}\left(s\right);\theta\right)\right\} \frac{\mathrm{d}s}{\delta_{T}}\right)^{10}\right|Y_{t}\right]\nonumber \\
 & \leq\frac{C}{N^{5}}\delta_{T}^{10}\mathbb{E}\left[\left.\int_{0}^{\delta_{T}}\left\{ -\partial_{u}\varpi\left(Y_{t},U_{t,1}\left(s\right);\theta\right)U_{t,1}\left(s\right)+\partial_{u,u}^{2}\varpi\left(Y_{t},U_{t,1}\left(s\right);\theta\right)\right\} ^{10}\frac{\mathrm{d}s}{\delta_{T}}\right|Y_{t}\right]\nonumber \\
 & =C\frac{N^{4}}{T^{9}}\mathbb{E}\left[\left.\int_{0}^{\delta_{T}}\left\{ -\partial_{u}\varpi\left(Y_{t},U_{t,1}\left(s\right);\theta\right)U_{t,1}\left(s\right)+\partial_{u,u}^{2}\varpi\left(Y_{t},U_{t,1}\left(s\right);\theta\right)\right\} ^{10}\mathrm{d}s\right|Y_{t}\right]\nonumber \\
 & =C\frac{N^{4}}{T^{9}}\int_{0}^{\delta_{T}}\mathbb{E}\left[\left.\left\{ -\partial_{u}\varpi\left(Y_{t},U_{t,1}\left(s\right);\theta\right)U_{t,1}\left(s\right)+\partial_{u,u}^{2}\varpi\left(Y_{t},U_{t,1}\left(s\right);\theta\right)\right\} ^{10}\right|Y_{t}\right]\mathrm{d}s\nonumber \\
 & \leq C\frac{N^{5}}{T^{10}}B(Y_{t}),\label{eq:controlL2}
\end{align}
by (\ref{ass:condcltsteinten}).Therefore we conclude that 
\begin{align}
\widetilde{\mathbb{E}}\left[\left.\left|\sum_{t=1}^{T}L_{t}^{T}\right|\right|Y_{t}\right]^{2} & \leq\sum_{t=1}^{T}B(Y_{t})^{1/2}\mathbb{E}^{1/5}\left[\left.\left(\widetilde{L}_{t}^{T}\right)^{10}\right|Y_{t}\right]\nonumber \\
 & \leq\left(\frac{N^{5}}{T^{10}}\right)^{1/5}\sum_{t=1}^{T}B(Y_{t})^{1/2}B(Y_{t})^{1/5}\nonumber \\
 & \leq\frac{N}{T}\frac{1}{T}\sum_{t=1}^{T}B(Y_{t})^{7/10}.\label{eq:controlLttilde}
\end{align}

\emph{Terms} $\widetilde{M}_{t}^{T}$. Finally we have, using \citep[Corollary 1, ][]{Zakai1967},
that 
\begin{align}
\mathbb{E}\left[\left.\left[\widetilde{M}_{t}^{T}\right]^{10}\right|Y_{t}\right] & =\mathbb{E}\left[\left.\left(\int_{0}^{\delta_{T}}\frac{\sqrt{2}}{N}\sum_{i=1}^{N}\partial_{u}\varpi\left(Y_{t},U_{t,i}\left(s\right);\theta\right)\mathrm{d}B_{t,i}\left(s\right)\right)^{10}\right|Y_{t}\right]\nonumber \\
 & \leq\frac{C2^{5}}{N^{5}}\mathbb{E}\left[\left.\left(\int_{0}^{\delta_{T}}\partial_{u}\varpi\left(Y_{t},U_{t,1}\left(s\right);\theta\right)\mathrm{d}B_{t,1}\left(s\right)\right)^{10}\right|Y_{t}\right]\nonumber \\
 & \leq\frac{C}{N^{5}}\delta_{T}^{4}\int_{0}^{\delta_{T}}\mathbb{E}\left[\left.\left(\partial_{u}\varpi\left(Y_{t},U_{t,1}\left(s\right);\theta\right)\right)^{10}\right|Y_{t}\right]\mathrm{d}s\nonumber \\
 & =\frac{C}{N^{5}}\delta_{T}^{5}\mathbb{E}\left[\left.\left(\partial_{u}\varpi\left(Y_{t},U_{t,1};\theta\right)\right)^{10}\right|Y_{t}\right]\leq\frac{C}{T^{5}}B(Y_{t}),\label{eq:controlMttilde}
\end{align}
by (\ref{ass:condcltduten}).

\emph{Overall control. }Bringing all terms together we thus have using
(\ref{eq:condcltfinalass1}), (\ref{eq:condcltfinalass2}), (\ref{eq:condcltfinalass3})
and (\ref{eq:controlMttilde}), (\ref{eq:controlJ2}) and (\ref{eq:controlL2})
that 
\begin{align}
\sum_{t=1}^{T}\widetilde{\mathbb{E}}\left[\left.h(\eta_{t}^{T})\left[\eta_{t}^{T}\right]^{2}\right|Y_{t}\right] & \leq C\sum_{t=1}^{T}B(Y_{t})^{1/2}\mathbb{E}^{1/2}\left[\left.\left[\widehat{W}_{t}^{T}\right]^{2}\left[\eta_{t}^{T}\right]^{5}\right|Y_{t}\right]\nonumber \\
 & \leq C\sum_{t=1}^{T}B(Y_{t})^{1/2}\mathbb{E}^{1/4}\left[\left.\left[\widehat{W}_{t}^{T}\right]^{-6}\right|Y_{t}\right]\mathbb{E}^{1/4}\left[\left.\left[\widetilde{\eta}_{t}^{T}\right]^{10}\right|Y_{t}\right]\nonumber \\
 & \leq C\sum_{t=1}^{T}B(Y_{t})^{3/2}\mathbb{E}^{1/4}\left[\left.\left[\widetilde{\eta}_{t}^{T}\right]^{10}\right|Y_{t}\right]\nonumber \\
 & \leq C\sum_{t=1}^{T}B(Y_{t})^{3/2}\left(\mathbb{E}\left[\left.\left[\widetilde{J}_{t}^{T}\right]^{10}\right|Y_{t}\right]+\mathbb{E}\left[\left.\left[\widetilde{L}_{t}^{T}\right]^{10}\right|Y_{t}\right]+\mathbb{E}\left[\left.\left[\widetilde{M}_{t}^{T}\right]^{10}\right|Y_{t}\right]\right)^{1/4}\nonumber \\
 & \leq C\sum_{t=1}^{T}B(Y_{t})\left(\frac{1}{\left(NT\right)^{5}}+\frac{N^{5}}{T^{10}}+\frac{1}{T^{5}}\right)^{1/4}\label{eq:controlheta}\\
 & \leq\frac{C}{T^{1/4}}\frac{1}{T}\sum_{t=1}^{T}B(Y_{t})\nonumber 
\end{align}
for $T$ large enough as $N_{T}/T\rightarrow0$. Hence by combining
the bounds (\ref{eq:controlheta}), (\ref{eq:controlJttilde}) and
(\ref{eq:controlLttilde}), Lemma \ref{lem:controlR1} follows.
\end{proof}

\subsubsection{Control of $\mathcal{R}_{2}^{T}$\label{subsec:controlofR2}}
\begin{lem}
\label{lem:controlR2}As $T\to\infty$ we have that 
\[
\sup_{\theta\in N(\bar{\theta})}\mathbb{\widetilde{\mathbb{E}}}\left[\left.\left|\mathcal{R}_{2}^{T}\right|\right|\mathcal{Y}^{T}\right]\to0\ \mathbb{P}^{Y}-\mathrm{a.s.}
\]
\end{lem}
\begin{proof}
We have
\begin{align*}
\mathcal{R}_{2}^{T} & :=-\sum_{t=1}^{T}\frac{\sqrt{2}}{N\widehat{W}_{t}^{T}}\sum_{i=1}^{N}\int_{0}^{\delta_{T}}\int_{0}^{s}\mathcal{S}\left\{ \partial_{u}\varpi\left(Y_{t},U_{t,i}\left(r\right);\theta\right)\right\} \mathrm{d}r\mathrm{d}B_{t,i}\left(s\right)\\
 & \quad+\sum_{t=1}^{T}\frac{\sqrt{2}}{N\widehat{W}_{t}^{T}}\sum_{i=1}^{N}\int_{0}^{\delta_{T}}\int_{0}^{s}\sqrt{2}\int_{0}^{s}\partial_{uuu}^{3}\varpi\left(Y_{t},U_{t,i}\left(r\right);\theta\right)\mathrm{d}B_{t,i}(r)\mathrm{d}B_{t,i}\left(s\right)\\
 & =:\mathcal{R}_{21}^{T}+\mathcal{R}_{22}^{T}.
\end{align*}
 We control these two terms separately. We have 
\begin{align*}
\widetilde{\mathbb{E}}\left[\left.\left|\mathcal{R}_{21}^{T}\right|\right|\mathcal{Y}^{T}\right] & \leq\widetilde{\mathbb{E}}^{1/2}\left[\left.\left|\sum_{t=1}^{T}\frac{\sqrt{2}}{N\widehat{W}_{t}^{T}}\sum_{i=1}^{N}\int_{0}^{\delta_{T}}\int_{0}^{s}\mathcal{S}\left\{ \partial_{u}\varpi\left(Y_{t},U_{t,i}\left(r\right);\theta\right)\right\} \mathrm{d}r\mathrm{d}s\right|^{2}\right|\mathcal{Y}^{T}\right].
\end{align*}
Since conditionally on $\mathcal{Y}^{T}$ the vectors $\left\{ U_{t,i}:i\right\} $
are independent across $t$ and 
\begin{align*}
 & \widetilde{\mathbb{E}}\left[\left.\frac{\sqrt{2}}{N\widehat{W}_{t}^{T}}\sum_{i=1}^{N}\int_{0}^{\delta_{T}}\int_{0}^{s}\mathcal{S}\left\{ \partial_{u}\varpi\left(Y_{t},U_{t,i}\left(r\right);\theta\right)\right\} \mathrm{d}r\mathrm{d}s\right|\mathcal{Y}^{T}\right]\\
= & \mathbb{E}\left[\left.\frac{\sqrt{2}}{N}\sum_{i=1}^{N}\int_{0}^{\delta_{T}}\int_{0}^{s}\mathcal{S}\left\{ \partial_{u}\varpi\left(Y_{t},U_{t,i}\left(r\right);\theta\right)\right\} \mathrm{d}r\mathrm{d}s\right|\mathcal{Y}^{T}\right]=0,
\end{align*}
it follows that 
\begin{align*}
 & \widetilde{\mathbb{E}}\left[\left.\left|\sum_{t=1}^{T}\frac{\sqrt{2}}{N\widehat{W}_{t}^{T}}\sum_{i=1}^{N}\int_{0}^{\delta_{T}}\int_{0}^{s}\mathcal{S}\left\{ \partial_{u}\varpi\left(Y_{t},U_{t,i}\left(r\right);\theta\right)\right\} \mathrm{d}r\mathrm{d}s\right|^{2}\right|\mathcal{Y}^{T}\right]\\
 & =\sum_{t=1}^{T}\widetilde{\mathbb{E}}\left[\left.\frac{2}{N^{2}\left[\widehat{W}_{t}^{T}\right]^{2}}\left(\sum_{i=1}^{N}\int_{0}^{\delta_{T}}\int_{0}^{s}\mathcal{S}\left\{ \partial_{u}\varpi\left(Y_{t},U_{t,i}\left(r\right);\theta\right)\right\} \mathrm{d}r\mathrm{d}s\right)^{2}\right|\mathcal{Y}^{T}\right]\\
 & =\sum_{t=1}^{T}\mathbb{E}\left[\left.\frac{2}{\widehat{W}_{t}^{T}}\left(\frac{1}{N}\sum_{i=1}^{N}\int_{0}^{\delta_{T}}\int_{0}^{s}\mathcal{S}\left\{ \partial_{u}\varpi\left(Y_{t},U_{t,i}\left(r\right);\theta\right)\right\} \mathrm{d}r\mathrm{d}s\right)^{2}\right|\mathcal{Y}^{T}\right]\\
 & \leq2\sum_{t=1}^{T}\mathbb{E}^{1/2}\left[\left.\left(\widehat{W}_{t}^{T}\right)^{-2}\right|\mathcal{Y}^{T}\right]\mathbb{E}^{1/2}\left[\left.\left(\frac{1}{N}\sum_{i=1}^{N}\int_{0}^{\delta_{T}}\int_{0}^{s}\mathcal{S}\left\{ \partial_{u}\varpi\left(Y_{t},U_{t,i}\left(r\right);\theta\right)\right\} \mathrm{d}r\mathrm{d}s\right)^{4}\right|\mathcal{Y}^{T}\right]\\
 & \leq2\sum_{t=1}^{T}B(Y_{t})\times\mathbb{E}^{1/2}\left[\left.\left(\frac{1}{N}\sum_{i=1}^{N}\int_{0}^{\delta_{T}}\int_{0}^{s}\mathcal{S}\left\{ \partial_{u}\varpi\left(Y_{t},U_{t,i}\left(r\right);\theta\right)\right\} \mathrm{d}r\mathrm{d}s\right)^{4}\right|\mathcal{Y}^{T}\right]\\
 & \leq2\sum_{t=1}^{T}B(Y_{t})\times\left(\frac{C}{N^{2}}\mathbb{E}\left[\left.\left(\int_{0}^{\delta_{T}}\int_{0}^{s}\mathcal{S}\left\{ \partial_{u}\varpi\left(Y_{t},U_{t,1}\left(r\right);\theta\right)\right\} \mathrm{d}r\mathrm{d}s\right)^{4}\right|\mathcal{Y}^{T}\right]\right)^{1/2}\\
 & \leq C\sum_{t=1}^{T}B(Y_{t})\frac{1}{N}\frac{\delta_{T}^{4}}{4}\mathbb{E}^{1/2}\left[\left.\left(\int_{0}^{\delta_{T}}\int_{0}^{s}\mathcal{S}\left\{ \partial_{u}\varpi\left(Y_{t},U_{t,1}\left(r\right);\theta\right)\right\} \frac{\mathrm{d}r\mathrm{d}s}{\delta_{T}^{2}/2}\right)^{4}\right|\mathcal{Y}^{T}\right]\\
 & \leq C\sum_{t=1}^{T}B(Y_{t})\frac{1}{N}\frac{\delta_{T}^{4}}{4}\left[\int_{0}^{\delta_{T}}\int_{0}^{s}\mathbb{E}\left[\left.\mathcal{S}\left\{ \partial_{u}\varpi\left(Y_{t},U_{t,1}\left(r\right);\theta\right)\right\} ^{4}\right|\mathcal{Y}^{T}\right]\frac{\mathrm{d}r\mathrm{d}s}{\delta_{T}^{2}/2}\right]^{1/2}\\
 & \leq C\sum_{t=1}^{T}B(Y_{t})\frac{1}{N}\frac{\delta_{T}^{3}}{2\sqrt{2}}\left[\int_{0}^{\delta_{T}}\int_{0}^{s}\mathbb{E}\left[\left.\mathcal{S}\left\{ \partial_{u}\varpi\left(Y_{t},U_{t,1}\left(0\right);\theta\right)\right\} ^{4}\right|\mathcal{Y}^{T}\right]\mathrm{d}r\mathrm{d}s\right]^{1/2}\\
 & =C\sum_{t=1}^{T}B(Y_{t})\frac{1}{N}\frac{\delta_{T}^{4}}{4}\mathbb{E}^{1/2}\left[\left.\mathcal{S}\left\{ \partial_{u}\varpi\left(Y_{t},U_{t,1}\left(0\right);\theta\right)\right\} ^{4}\right|Y_{t}\right]\\
 & =C\sum_{t=1}^{T}B(Y_{t})^{2}\frac{1}{N}\frac{\delta_{T}^{4}}{4}=C\sum_{t=1}^{T}B(Y_{t})^{2}\frac{1}{N}\frac{N^{4}}{T^{4}}=\frac{N^{3}}{T^{3}}\frac{C}{T}\sum_{t=1}^{T}B(Y_{t})^{2},
\end{align*}
by (\ref{ass:condcltsteinten}). Thus 
\[
\widetilde{\mathbb{E}}\left[\left.\left|\mathcal{R}_{21}^{T}\right|\right|\mathcal{Y}^{T}\right]^{2}\leq C\delta_{T}^{3}\frac{1}{T}\sum_{t=1}^{T}B(Y_{t})^{2}.
\]
On the other hand using \citep[Corollary 1,][]{Zakai1967} twice,
we have 
\begin{align*}
&\widetilde{\mathbb{E}}\left[\left.\left|\mathcal{R}_{22}^{T}\right|\right|\mathcal{Y}^{T}\right]^{2}\\
& \leq\widetilde{\mathbb{E}}\left[\left.\left|\mathcal{R}_{22}^{T}\right|^{2}\right|\mathcal{Y}^{T}\right]\\
 & =\widetilde{\mathbb{E}}\left[\left.\left|\sum_{t=1}^{T}\frac{\sqrt{2}}{N\widehat{W}_{t}^{T}}\sum_{i=1}^{N}\int_{0}^{\delta_{T}}\int_{0}^{s}\int_{0}^{s}\partial_{uuu}^{3}\varpi\left(Y_{t},U_{t,i}\left(r\right);\theta\right)\mathrm{d}B_{t,i}(r)\mathrm{d}B_{t,i}\left(s\right)\right|^{2}\right|\mathcal{Y}^{T}\right]\\
 & =\sum_{t=1}^{T}\widetilde{\mathbb{E}}\left[\left.\left(\frac{\sqrt{2}}{N\widehat{W}_{t}^{T}}\sum_{i=1}^{N}\int_{0}^{\delta_{T}}\int_{0}^{s}\int_{0}^{s}\partial_{uuu}^{3}\varpi\left(Y_{t},U_{t,i}\left(r\right);\theta\right)\mathrm{d}B_{t,i}(r)\mathrm{d}B_{t,i}\left(s\right)\right)^{2}\right|\mathcal{Y}^{T}\right]\\
 & \leq\sum_{t=1}^{T}\mathbb{E}^{1/2}\left[\left.\left(\widehat{W}_{t}^{T}\right)^{-2}\right|\mathcal{Y}^{T}\right]\mathbb{E}^{1/2}\left[\left.\left(\frac{\sqrt{2}}{N}\sum_{i=1}^{N}\int_{0}^{\delta_{T}}\int_{0}^{s}\partial_{uuu}^{3}\varpi\left(Y_{t},U_{t,i}\left(r\right);\theta\right)\mathrm{d}B_{t,i}(r)\mathrm{d}B_{t,i}\left(s\right)\right)^{4}\right|\mathcal{Y}^{T}\right]\\
 & \leq\sum_{t=1}^{T}B(Y_{t})^{1/2}\mathbb{E}^{1/2}\left[\left.\left(\frac{\sqrt{2}}{N}\sum_{i=1}^{N}\int_{0}^{\delta_{T}}\int_{0}^{s}\partial_{uuu}^{3}\varpi\left(Y_{t},U_{t,1}\left(r\right);\theta\right)\mathrm{d}B_{t,1}(r)\mathrm{d}B_{t,1}\left(s\right)\right)^{4}\right|\mathcal{Y}^{T}\right]\\
 & \leq\sum_{t=1}^{T}B(Y_{t})^{1/2}\left\{ \frac{C}{N^{2}}\mathbb{E}\left[\left.\left(\int_{0}^{\delta_{T}}\int_{0}^{s}\sqrt{2}\int_{0}^{s}\partial_{uuu}^{3}\varpi\left(Y_{t},U_{t,1}\left(r\right);\theta\right)\mathrm{d}B_{t,1}(r)\mathrm{d}B_{t,1}\left(s\right)\right)^{4}\right|Y_{t}\right]\right\} ^{1/2}\\
 & \leq\sum_{t=1}^{T}B(Y_{t})^{1/2}\left\{ C\frac{4}{N^{2}}\delta_{T}\int_{0}^{\delta_{T}}\mathbb{E}\left[\left.\left[\int_{0}^{s}\partial_{uuu}^{3}\varpi\left(Y_{t},U_{t,1}\left(r\right);\theta\right)\mathrm{d}B_{t,1}(r)\right]^{4}\right|Y_{t}\right]\mathrm{d}s\right\} ^{1/2}\\
 & \leq\sum_{t=1}^{T}B(Y_{t}))^{1/2}\left\{ C\frac{4}{N^{2}}\delta_{T}\int_{0}^{\delta_{T}}s\int_{0}^{s}\mathbb{E}\left[\left.\left[\partial_{uuu}^{3}\varpi\left(Y_{t},U_{t,1}\left(r\right);\theta\right)\right]^{4}\right|Y_{t}\right]\mathrm{d}r\mathrm{d}s\right\} ^{1/2}\\
 & \leq\sum_{t=1}^{T}B(Y_{t}))^{1/2}\left\{ C\frac{4}{N^{2}}\mathbb{E}\left[\left.\left[\partial_{uuu}^{3}\varpi\left(Y_{t},U_{t,1}\left(0\right);\theta\right)\right]^{4}\right|Y_{t}\right]\delta_{T}\int_{0}^{\delta_{T}}s\int_{0}^{s}\mathrm{d}r\mathrm{d}s\right\} ^{1/2}\\
 & \leq\sum_{t=1}^{T}B(Y_{t}))^{1/2}\left\{ C\frac{4}{N^{2}}\mathbb{E}\left[\left.\left[\partial_{uuu}^{3}\varpi\left(Y_{t},U_{t,1}\left(0\right);\theta\right)\right]^{4}\right|Y_{t}\right]\delta_{T}\int_{0}^{\delta_{T}}s^{2}\mathrm{d}s\right\} ^{1/2}\\
 & \leq\sum_{t=1}^{T}B(Y_{t}))^{1/2}\left\{ C\frac{4}{N^{2}}\mathbb{E}\left[\left.\left[\partial_{uuu}^{3}\varpi\left(Y_{t},U_{t,1}\left(0\right);\theta\right)\right]^{4}\right|Y_{t}\right]\delta_{T}^{4}\right\} ^{1/2}\\
 & \leq\sum_{t=1}^{T}B(Y_{t}))^{1/2}\left\{ C\frac{4}{N^{2}}\frac{N^{4}}{T^{4}}\mathbb{E}\left[\left.\left[\partial_{uuu}^{3}\varpi\left(Y_{t},U_{t,1}\left(0\right);\theta\right)\right]^{4}\right|Y_{t}\right]\right\} ^{1/2}\\
 & \leq\frac{N}{T}\frac{1}{T}\sum_{t=1}^{T}B(Y_{t}),
\end{align*}
by (\ref{ass:condcltdddufour}). Since all bounds obtained are independent
of $\theta$ we have the following result. 
\end{proof}

\subsubsection{Control of $\mathcal{R}_{3}^{T}$}
\begin{lem}
\label{lem:controlR3}As $T\to\infty$ we have that 
\[
\sup_{\theta\in N(\bar{\theta})}\mathbb{\widetilde{\mathbb{E}}}\left[\left.\left|\mathcal{R}_{3}^{T}\right|\right|\mathcal{Y}^{T}\right]\to0\ \mathbb{P}^{Y}-\mathrm{a.s.}
\]
\end{lem}
\begin{proof}
This result is established as follows. Let us write $K_{t}^{T}=J_{t}^{T}+L_{t}^{T}$,
then we have 
\begin{align*}
\mathbb{\widetilde{\mathbb{E}}}\left[\left.\left|\sum_{t=1}^{T}\left(\eta_{t}^{T}\right)^{2}-\kappa^{2}(\theta)\right|\right|\mathcal{Y}^{T}\right] & \leq\mathbb{\widetilde{\mathbb{E}}}\left[\left.\left|\sum_{t=1}^{T}\left(\eta_{t}^{T}\right)^{2}-\kappa^{2}(\theta)\right|\right|\mathcal{Y}^{T}\right]\\
 & \leq\mathbb{\widetilde{\mathbb{E}}}\left[\left.\left|\sum_{t=1}^{T}\left(K_{t}^{T}\right)^{2}\right|+\left|\sum_{t=1}^{T}K_{t}^{T}M_{t}^{T}\right|\right|\mathcal{Y}^{T}\right]+\mathbb{\widetilde{\mathbb{E}}}\left[\left.\left|\sum_{t=1}^{T}\left(M_{t}^{T}\right)^{2}-\kappa^{2}(\theta)\right|\right|\mathcal{Y}^{T}\right]\\
 & \leq\mathbb{\widetilde{\mathbb{E}}}\left[\left.\left|\sum_{t=1}^{T}\left(K_{t}^{T}\right)^{2}\right|\right|\mathcal{Y}^{T}\right]+\sum_{t=1}^{T}\widetilde{\mathbb{E}}^{1/2}\left[\left.\left(K_{t}^{T}\right)^{2}\right|\mathcal{F}^{T}\right]\widetilde{\mathbb{E}}^{1/2}\left[\left.\left(M_{t}^{T}\right)^{2}\right|\mathcal{Y}^{T}\right]\\
 & +\mathbb{\widetilde{\mathbb{E}}}\left[\left.\left|\sum_{t=1}^{T}\left(M_{t}^{T}\right)^{2}-\kappa^{2}(\theta)\right|\right|\mathcal{Y}^{T}\right].
\end{align*}
Now from (\ref{eq:controlJttilde}) and (\ref{eq:controlLttilde})
it easily follows that 
\[
\mathbb{\widetilde{\mathbb{E}}}\left[\left.\left|\sum_{t=1}^{T}\left(K_{t}^{T}\right)^{2}\right|\right|\mathcal{Y}^{T}\right]\leq C\left(\frac{1}{NT}+\frac{N}{T^{2}}\right)\sum_{t=1}^{T}B(Y_{t})^{7/10}.
\]
In addition, from (\ref{eq:controlMttilde}) and (\ref{ass:condcltboundonkappa})
\begin{align*}
\widetilde{\mathbb{E}}\left[\left.\left(M_{t}^{T}\right)^{2}\right|\mathcal{Y}^{T}\right] & =\mathbb{E}\left[\left.\left(\widehat{W}_{t}^{T}\right)^{-1}\left(\widetilde{M}_{t}^{T}\right)^{2}\right|\mathcal{Y}^{T}\right]\\
 & \leq\mathbb{E}^{1/2}\left[\left.\left(\widehat{W}_{t}^{T}\right)^{-2}\right|\mathcal{Y}^{T}\right]\mathbb{E}^{1/2}\left[\left.\left(\widetilde{M}_{t}^{T}\right)^{4}\right|\mathcal{Y}^{T}\right]\\
 & \leq B(Y_{t})^{1/2}\mathbb{E}^{1/5}\left[\left.\left(\widetilde{M}_{t}^{T}\right)^{10}\right|\mathcal{Y}^{T}\right]\\
 & \leq\frac{1}{T}B(Y_{t})^{1/2}B(Y_{t})^{1/5},
\end{align*}
and thus we find that 
\begin{align*}
\sum_{t=1}^{T}\widetilde{\mathbb{E}}^{1/2}\left[\left.\left(K_{t}^{T}\right)^{2}\right|\mathcal{Y}^{T}\right]\widetilde{\mathbb{E}}^{1/2}\left[\left.\left(M_{t}^{T}\right)^{2}\right|\mathcal{Y}^{T}\right] & \leq\sum_{t=1}^{T}\left[C\left(\frac{1}{NT}+\frac{N}{T^{2}}\right)\frac{1}{T}\right]^{1/2}B(Y_{t})^{7/10}.
\end{align*}
Finally we have that 
\begin{align*}
M_{t}^{T} & =\frac{\sqrt{2}}{N\widehat{W}_{t}^{T}}\sum_{i=1}^{N}\partial_{u}\varpi\left(Y_{t},U_{t,i}\left(0\right);\theta\right)\sqrt{\delta_{T}}\xi_{t,i}\\
 & \quad-\frac{\sqrt{2}}{N\widehat{W}_{t}^{T}}\sum_{i=1}^{N}\int_{0}^{\delta_{T}}\int_{0}^{s}\mathcal{S}\left\{ \partial_{u}\varpi\left(Y_{t},U_{t,i}\left(r\right);\theta\right)\right\} \mathrm{d}r\mathrm{d}B_{t,i}\left(s\right)\\
 & \quad\quad+\int_{0}^{\delta_{T}}\frac{\sqrt{2}}{N\widehat{W}_{t}^{T}}\sum_{i=1}^{N}\sqrt{2}\int_{0}^{s}\partial_{uu}^{2}\varpi\left(Y_{t},U_{t,i}\left(r\right);\theta\right)\mathrm{d}B_{t,i}(r)\mathrm{d}B_{t,i}\left(s\right)\\
 & =\frac{\sqrt{2}}{\widehat{W}_{t}^{T}}\frac{1}{\sqrt{T}}\left[\frac{1}{N}\sum_{i=1}^{N}\left(\partial_{u}\varpi\left(Y_{t},U_{t,i}\left(0\right);\theta\right)\right)^{2}\right]^{1/2}\xi_{t}+\mathcal{R}_{2,t}^{T}.
\end{align*}
Therefore 
\[
\sum_{t=1}^{T}(M_{t}^{T})^{2}=\sum_{t=1}^{T}\frac{2}{\left[\widehat{W}_{t}^{T}\right]^{2}}\frac{1}{TN}\sum_{i=1}^{N}\left(\partial_{u}\varpi\left(Y_{t},U_{t,i}\left(0\right);\theta\right)\right)^{2}\xi_{t}^{2}+\mathcal{R}_{2,t}^{T}.
\]
From Section\,\ref{subsec:controlofR2} it follows that 
\[
\sup_{\theta\in B(\bar{\theta,}\epsilon)}\widetilde{\mathbb{E}}\left[\left.\sum_{t=1}^{T}\mathcal{R}_{2,t}^{T}\right|\mathcal{Y}^{T}\right]\to0\ \mathbb{P}^{Y}-\mathrm{a.s.}
\]
Thus we can focus on the remaining term. We have
\begin{align*}
 & \mathbb{\widetilde{\mathbb{E}}}\left[\left.\left|\frac{1}{T}\sum_{t=1}^{T}\frac{2}{\left[\widehat{W}_{t}^{T}\right]^{2}}\frac{1}{N}\sum_{i=1}^{N}\left(\partial_{u}\varpi\left(Y_{t},U_{t,i}\left(0\right);\theta\right)\right)^{2}\xi_{t}^{2}-\kappa^{2}(\theta)\right|\right|\mathcal{Y}^{T}\right]\\
 & \leq\mathbb{\widetilde{\mathbb{E}}}\left[\left.\left|\frac{2}{T}\sum_{t=1}^{T}\xi_{t}^{2}\left(\frac{1}{\left[\widehat{W}_{t}^{T}\right]^{2}}-1\right)\frac{1}{N}\sum_{i=1}^{N}\left(\partial_{u}\varpi\left(Y_{t},U_{t,i}\left(0\right);\theta\right)\right)^{2}\right|\right|\mathcal{Y}^{T}\right]\\
 & \quad+\mathbb{\widetilde{\mathbb{E}}}\left[\left.\left|\frac{2}{T}\sum_{t=1}^{T}\xi_{t}^{2}\frac{1}{N}\sum_{i=1}^{N}\left(\partial_{u}\varpi\left(Y_{t},U_{t,i}\left(0\right);\theta\right)\right)^{2}-\kappa^{2}(\theta)\right|\right|\mathcal{Y}^{T}\right]\\
 & \leq\frac{2}{T}\sum_{t=1}^{T}\mathbb{\widetilde{\mathbb{E}}}\left[\left.\left|\xi_{t}^{2}\left(\frac{1}{\left[\widehat{W}_{t}^{T}\right]^{2}}-1\right)\frac{1}{N}\sum_{i=1}^{N}\left(\partial_{u}\varpi\left(Y_{t},U_{t,i}\left(0\right);\theta\right)\right)^{2}\right|\right|\mathcal{Y}^{T}\right]\\
 & \quad+\mathbb{\widetilde{\mathbb{E}}}\left[\left.\left|\frac{1}{T}\sum_{t=1}^{T}\xi_{t}^{2}\frac{1}{N}\sum_{i=1}^{N}2\left(\partial_{u}\varpi\left(Y_{t},U_{t,i}\left(0\right);\theta\right)\right)^{2}-\kappa^{2}(\theta)\right|\right|\mathcal{Y}^{T}\right]\\
 & \leq\frac{2}{T}\sum_{t=1}^{T}\mathbb{\widetilde{\mathbb{E}}}^{1/2}\left[\left.\left(\frac{1}{\left[\widehat{W}_{t}^{T}\right]^{2}}-1\right)^{2}\right|Y_{t}\right]\mathbb{\widetilde{\mathbb{E}}}^{1/2}\left[\left.\left|\frac{1}{N}\sum_{i=1}^{N}\left(\partial_{u}\varpi\left(Y_{t},U_{t,i}\left(0\right);\theta\right)\right)^{2}\right|^{2}\right|Y_{t}\right]\\
 & \quad+\mathbb{\widetilde{\mathbb{E}}}\left[\left.\left|\frac{1}{T}\sum_{t=1}^{T}\xi_{t}^{2}\frac{1}{N}\sum_{i=1}^{N}2\left(\partial_{u}\varpi\left(Y_{t},U_{t,i}\left(0\right);\theta\right)\right)^{2}-\kappa^{2}(\theta)\right|\right|\mathcal{Y}^{T}\right]\\
 & =:J_{1}+J_{2}.
\end{align*}
since $\xi_{t}$ is independent of the remaining terms and $\mathbb{E}\xi_{t}^{2}=1$.
We control the first term using the Cauchy-Schwarz inequality, (\ref{ass:condcltsecondmoment}),
(\ref{ass:condcltinverseweightto1}) and the triangle inequality
\begin{align*}
J_{1} & \leq\frac{2}{T}\sum_{t=1}^{T}\mathbb{\mathbb{E}}^{1/2}\left[\left.\widehat{W}_{t}^{T}\left|\left(\frac{1}{\left[\widehat{W}_{t}^{T}\right]^{2}}-1\right)^{2}\right|\right|Y_{t}\right]\mathbb{\widetilde{\mathbb{E}}}^{1/2}\left[\left.\left|\frac{1}{N}\sum_{i=1}^{N}\left(\partial_{u}\varpi\left(Y_{t},U_{t,i}\left(0\right);\theta\right)\right)^{2}\right|^{2}\right|Y_{t}\right]\\
 & \leq\frac{2}{T}\sum_{t=1}^{T}\sqrt{\epsilon_{T}}B(Y_{t})\mathbb{\widetilde{\mathbb{E}}}^{1/2}\left[\left.\left|\left(\partial_{u}\varpi\left(Y_{t},U_{t,i}\left(0\right);\theta\right)\right)\right|^{4}\right|Y_{t}\right]\\
 & \leq\frac{2}{T}\sum_{t=1}^{T}\sqrt{\epsilon_{T}}B(Y_{t})\mathbb{\mathbb{E}}^{1/4}\left[\left.\left|\left(\partial_{u}\varpi\left(Y_{t},U_{t,i}\left(0\right);\theta\right)\right)\right|^{8}\right|Y_{t}\right]\mathbb{\mathbb{E}}^{1/2}\left[\left.\left|\left(\widehat{W}_{t}^{T}\right)^{2}\right|\right|Y_{t}\right]\\
 & \leq\frac{2}{T}\sum_{t=1}^{T}\sqrt{\epsilon_{T}}B(Y_{t})^{k}.
\end{align*}
Next we control 
\begin{align*}
J_{2} & =\mathbb{\widetilde{\mathbb{E}}}\left[\left.\left|\frac{1}{T}\sum_{t=1}^{T}\xi_{t}^{2}\frac{1}{N}\sum_{i=1}^{N}2\left(\partial_{u}\varpi\left(Y_{t},U_{t,i}\left(0\right);\theta\right)\right)^{2}-\kappa^{2}(\theta)\right|\right|\mathcal{Y}^{T}\right]\\
 & =\mathbb{\widetilde{\mathbb{E}}}\left[\left.\left|\frac{1}{T}\sum_{t=1}^{T}\xi_{t}^{2}\left[\frac{1}{N}\sum_{i=1}^{N}2g(Y_{t},U_{t,i};\theta)-\kappa^{2}(Y_{t},\theta)\right]\right|\right|\mathcal{Y}^{T}\right]+\mathbb{\widetilde{\mathbb{E}}}\left[\left.\left|\frac{1}{T}\sum_{t=1}^{T}\xi_{t}^{2}\kappa^{2}(Y_{t},\theta)-\kappa^{2}(\theta)\right|\right|\mathcal{Y}^{T}\right]\\
 & \leq\frac{1}{T}\sum_{t=1}^{T}\mathbb{\widetilde{\mathbb{E}}}\left[\left.\left|\xi_{t}^{2}\left[\frac{1}{N}\sum_{i=1}^{N}2g(Y_{t},U_{t,i};\theta)-\kappa^{2}(Y_{t},\theta)\right]\right|\right|\mathcal{Y}^{T}\right]+\mathbb{\widetilde{\mathbb{E}}}\left[\left.\left|\frac{1}{T}\sum_{t=1}^{T}\xi_{t}^{2}\kappa^{2}(Y_{t},\theta)-\kappa^{2}(\theta)\right|\right|\mathcal{Y}^{T}\right]\\
 & \leq\frac{1}{T}\sum_{t=1}^{T}\mathbb{E}^{1/2}\left[\left.\left|\left(\widehat{W}_{t}^{T}\right)^{2}\right|\right|\mathcal{Y}^{T}\right]\mathbb{\mathbb{E}}^{1/2}\left[\left.\left|\left[\frac{2}{N}\sum_{i=1}^{N}g(Y_{t},U_{t,i};\theta)-\kappa^{2}(Y_{t},\theta)\right]^{2}\right|\right|Y_{t}\right]\\
 & \qquad+\mathbb{\widetilde{\mathbb{E}}}\left[\left.\left|\frac{1}{T}\sum_{t=1}^{T}\xi_{t}^{2}\kappa^{2}(Y_{t},\theta)-\kappa^{2}(\theta)\right|\right|\mathcal{Y}^{T}\right]\\
 & \leq\frac{1}{T}\sum_{t=1}^{T}B(Y_{t})^{1/2}\frac{1}{\sqrt{N}}\mathbb{\mathbb{V}}^{1/2}\left[\left.\left|g(Y_{t},U_{t,i};\theta)\right|\right|Y_{t}\right]+\mathbb{\widetilde{\mathbb{E}}}\left[\left.\left|\frac{1}{T}\sum_{t=1}^{T}\left[\xi_{t}^{2}-1\right]\kappa^{2}(Y_{t},\theta)\right|\right|Y_{t}\right]\\
 & \qquad\quad+\left|\frac{1}{T}\sum_{t=1}^{T}\kappa(Y_{t},\theta)-\kappa^{2}(\theta)\right|\\
 & \leq\frac{1}{\sqrt{N}}\frac{1}{T}\sum_{t=1}^{T}B(Y_{t})+\mathbb{\widetilde{\mathbb{E}}}^{1/2}\left[\left.\left|\frac{1}{T}\sum_{t=1}^{T}\left[\xi_{t}^{2}-1\right]\kappa^{2}(Y_{t},\theta)\right|^{2}\right|Y_{t}\right]+\left|\frac{1}{T}\sum_{t=1}^{T}\kappa(Y_{t},\theta)-\kappa^{2}(\theta)\right|\\
 & \leq\frac{1}{\sqrt{N}}\frac{1}{T}\sum_{t=1}^{T}B(Y_{t})+\frac{1}{T^{2}}\sum_{t=1}^{T}\kappa^{4}(Y_{t},\theta)\mathbb{V}(\xi_{t}^{2})+\left|\frac{1}{T}\sum_{t=1}^{T}\kappa(Y_{t},\theta)-\kappa^{2}(\theta)\right|.
\end{align*}
Except from the very last term, everything else is independent of
$\theta.$ Therefore from (\ref{eq:uniformSLLN}) and the strong law
of large numbers we have the following result. 
\end{proof}

\subsubsection{Proof of Theorem \ref{thm:UniformCLT}}

The result follows now directly from Lemmas \ref{lem:dblbound}, \ref{lem:controlst-kappa},
\ref{lem:controlR1}, \ref{lem:controlR2} and \ref{lem:controlR3}.

\subsection{Proof of Theorem \ref{Theorem:weakconvergence}}

Let $\{\widetilde{\vartheta}_{n}^{T};n\geq0\}$ be the projection
on the first component of the stationary Markov chain $\{(\widetilde{\vartheta}_{n}^{T},\mathsf{U}_{n}^{T});n\geq0\}$
of invariant distribution $\widetilde{\pi}_{T}$ defined in (\ref{eq:rescaledtargetCPM})
and transition kernel $Q_{T}$ defined in (\ref{eq:rescaledtransitionkernelr})
and $\{\widetilde{\vartheta}_{n};n\geq0\}$ the stationary Markov
chain of invariant distribution $\varphi(\mathrm{d}\widetilde{\theta};0,\overline{\Sigma})$
and transition kernel $P$ defined in (\ref{eq:penaltykernel}). The
proof of Theorem \ref{Theorem:weakconvergence} relies on a set of
preliminary propositions. All the expectations in this section have
to be understood as conditional expectations w.r.t. $\mathcal{Y}^{T}$. 
\begin{prop}
\label{Proposition:Kernelconverges}Under the assumptions of Theorem
\ref{Theorem:weakconvergence}, for any subsequence $\left\{ T_{k};k\geq0\right\} $
we can extract a further subsequence $\left\{ T_{k}^{\prime};k\geq0\right\} $
such that almost surely along this subsequence we have 
\begin{equation}
\mathbb{E}\left(\left\vert Q_{T}f(\widetilde{\vartheta}_{0}^{T},\mathsf{U}_{0}^{T})-Pf(\widetilde{\vartheta}_{0}^{T})\right\vert \right)\rightarrow0\label{eq:kernelconvprop}
\end{equation}
for all $f\in B\left(\mathbb{R}^{d}\right)$. 
\end{prop}
\begin{rem*}
We emphasize that the subset on which this almost sure convergence
occurs is independent of $f$. 
\end{rem*}
\begin{proof}[Proof of Proposition \textit{\ref{Proposition:Kernelconverges}}]

We define 
\[
\widetilde{r}_{T}\left(\widetilde{\theta}_{0},\widetilde{\theta}_{1}\right)=\frac{\widetilde{\pi}_{T}(\widetilde{\theta}_{1})\widetilde{q}(\widetilde{\theta}_{1},\widetilde{\theta}_{0})}{\widetilde{\pi}_{T}(\widetilde{\theta}_{0})\widetilde{q}(\widetilde{\theta}_{0},\widetilde{\theta}_{1})},\text{ \  \  \ }\widetilde{r}\left(\widetilde{\theta}_{0},\widetilde{\theta}_{1}\right)=\frac{\varphi(\widetilde{\theta}_{1};0,\overline{\Sigma})\widetilde{q}(\widetilde{\theta}_{1},\widetilde{\theta}_{0})}{\varphi(\widetilde{\theta}_{0};0,\overline{\Sigma})\widetilde{q}(\widetilde{\theta}_{0},\widetilde{\theta}_{1})},
\]
and write 
\[
\overline{p}(Y_{1:T}\mid\theta_{i},u_{i})=\frac{\widehat{p}(Y_{1:T}\mid\theta_{i},u_{i})}{p(Y_{1:T}\mid\theta_{i})},
\]
where $\theta_{i}=\widehat{\theta}_{T}+\widetilde{\theta}_{i}/\sqrt{T}$
for $i=0,1$. As Assumption\,\ref{Assumption:proposal} holds, we
have 
\begin{align}
 & Q_{T}f(\widetilde{\theta}_{0},u_{0})=\iint\widetilde{q}(\widetilde{\theta}_{0},\mathrm{d}\widetilde{\theta}_{1})f(\widetilde{\theta}_{1})K_{\rho_{T}}\left(u_{0},\mathrm{d}u_{1}\right)\min\left(1,\widetilde{r}_{T}\left(\widetilde{\theta}_{0},\widetilde{\theta}_{1}\right)\frac{\overline{p}(Y_{1:T}\mid\theta_{1},u_{1})}{\overline{p}(Y_{1:T}\mid\theta_{0},u_{0})}\right)\label{eq:kernelCPMH}\\
 & \qquad+f(\widetilde{\theta}_{0})\left[1-\iint\widetilde{q}(\widetilde{\theta}_{0},\mathrm{d}\widetilde{\theta}_{1})K_{\rho_{T}}\left(u_{0},\mathrm{d}u_{1}\right)\min\left(1,\widetilde{r}_{T}\left(\widetilde{\theta}_{0},\widetilde{\theta}_{1}\right)\frac{\overline{p}(Y_{1:T}\mid\theta_{1},u_{1})}{\overline{p}(Y_{1:T}\mid\theta_{0},u_{0})}\right)\right]\nonumber 
\end{align}
and 
\begin{align}
 & Pf(\widetilde{\theta}_{0})=\iint\widetilde{q}(\widetilde{\theta}_{0},\mathrm{d}\widetilde{\theta}_{1})f(\widetilde{\theta}_{1})\varphi\left(\mathrm{d}w;-\frac{\kappa^{2}}{2},\kappa^{2}\right)\min\left(1,\widetilde{r}\left(\widetilde{\theta}_{0},\widetilde{\theta}_{1}\right)\exp\left(w\right)\right)\label{eq:kernelpenalty}\\
 & \qquad+f(\widetilde{\theta}_{0})\left[1-\iint\widetilde{q}(\widetilde{\theta}_{0},\mathrm{d}\widetilde{\theta}_{1})\varphi\left(\mathrm{d}w;-\frac{\kappa^{2}}{2},\kappa^{2}\right)\min\left(1,\widetilde{r}\left(\widetilde{\theta}_{0},\widetilde{\theta}_{1}\right)\exp\left(w\right)\right)\right].\nonumber 
\end{align}
It follows that 
\begin{align*}
\lefteqn{\frac{1}{2}\mathbb{E}\left(\left\vert Q_{T}f(\widetilde{\vartheta}_{0}^{T},\mathsf{U}_{0}^{T})-Pf(\widetilde{\vartheta}_{0}^{T})\right\vert \right)}\\
 & \leq\frac{1}{2}\iiint\widetilde{\pi}_{T}(\mathrm{d}\widetilde{\theta}_{0},\mathrm{d}u_{0})\widetilde{q}(\widetilde{\theta}_{0},\mathrm{d}\widetilde{\theta}_{1})\left\vert f(\widetilde{\theta}_{1})\right\vert \left\vert \int K_{\rho_{T}}\left(u_{0},\mathrm{d}u_{1}\right)\min\left(1,\widetilde{r}_{T}\left(\widetilde{\theta}_{0},\widetilde{\theta}_{1}\right)\frac{\overline{p}(Y_{1:T}\mid\theta_{1},u_{1})}{\overline{p}(Y_{1:T}\mid\theta_{0},u_{0})}\right)\right.\\
 & \quad\qquad-\left.\int\varphi\left(\mathrm{d}w;-\frac{\kappa^{2}}{2},\kappa^{2}\right)\min\left(1,\widetilde{r}\left(\widetilde{\theta}_{0},\widetilde{\theta}_{1}\right)\exp\left(w\right)\right)\right\vert \\
 & \qquad+\frac{1}{2}\iiint\widetilde{\pi}_{T}(\mathrm{d}\widetilde{\theta}_{0},\mathrm{d}u_{0})\widetilde{q}(\widetilde{\theta}_{0},\mathrm{d}\widetilde{\theta}_{1})\left\vert f\left(\widetilde{\theta}_{0}\right)\right\vert \left\vert \int K_{\rho_{T}}\left(u_{0},\mathrm{d}u_{1}\right)\min\left(1,\widetilde{r}_{T}\left(\widetilde{\theta}_{0},\widetilde{\theta}_{1}\right)\frac{\overline{p}(Y_{1:T}\mid\theta_{1},u_{1})}{\overline{p}(Y_{1:T}\mid\theta_{0},u_{0})}\right)\right.\\
 & \qquad\quad\left.-\int\varphi\left(\mathrm{d}w;-\frac{\kappa^{2}}{2},\kappa^{2}\right)\min\left(1,\widetilde{r}\left(\widetilde{\theta}_{0},\widetilde{\theta}_{1}\right)\exp\left(w\right)\right)\right\vert \\
 & \leq\iiint\widetilde{\pi}_{T}(\mathrm{d}\widetilde{\theta}_{0},\mathrm{d}u_{0})\widetilde{q}(\widetilde{\theta}_{0},\mathrm{d}\widetilde{\theta}_{1})\left\vert \int K_{\rho_{T}}\left(u_{0},\mathrm{d}u_{1}\right)\min\left(1,\widetilde{r}_{T}\left(\widetilde{\theta}_{0},\widetilde{\theta}_{1}\right)\frac{\overline{p}(Y_{1:T}\mid\theta_{1},u_{1})}{\overline{p}(Y_{1:T}\mid\theta_{0},u_{0})}\right)\right.\\
 & \qquad-\left.\int\varphi\left(\mathrm{d}w;-\frac{\kappa^{2}}{2},\kappa^{2}\right)\min\left(1,\widetilde{r}\left(\widetilde{\theta}_{0},\widetilde{\theta}_{1}\right)\exp\left(w\right)\right)\right\vert .
\end{align*}
Hence, we have
\begin{align}
 & \frac{1}{2}\mathbb{E}\left(\left|Q_{T}f(\widetilde{\vartheta}_{0}^{T},\mathsf{U}_{0}^{T})-Pf(\widetilde{\vartheta}_{0}^{T})\right|\right)\nonumber \\
 & =\iiint\widetilde{\pi}_{T}(\left.\mathrm{d}u_{0}\right\vert \widetilde{\theta}_{0})\widetilde{q}(\widetilde{\theta}_{0},\mathrm{d}\widetilde{\theta}_{1})\left\vert \int K_{\rho_{T}}\left(u_{0},\mathrm{d}u_{1}\right)\min\left(\widetilde{\pi}_{T}(\widetilde{\theta}_{0}),\widetilde{\pi}_{T}(\widetilde{\theta}_{1})\frac{\widetilde{q}(\widetilde{\theta}_{1},\widetilde{\theta}_{0})}{\widetilde{q}(\widetilde{\theta}_{0},\widetilde{\theta}_{1})}\frac{\overline{p}(Y_{1:T}\mid\theta_{1},u_{1})}{\overline{p}(Y_{1:T}\mid\theta_{0},u_{0})}\right)\right.\nonumber \\
 & \qquad-\left.\widetilde{\pi}_{T}(\widetilde{\theta}_{0})\int\varphi\left(\mathrm{d}w;-\frac{\kappa^{2}}{2},\kappa^{2}\right)\min\left(1,\widetilde{r}\left(\widetilde{\theta}_{0},\widetilde{\theta}_{1}\right)\exp\left(w\right)\right)\right\vert \mathrm{d}\widetilde{\theta}_{0}\nonumber \\
 & \leq\iiint\widetilde{\pi}_{T}(\left.\mathrm{d}u_{0}\right\vert \widetilde{\theta}_{0})\int K_{\rho_{T}}\left(u_{0},\mathrm{d}u_{1}\right)\left\vert \min\left(\widetilde{\pi}_{T}(\widetilde{\theta}_{0})\widetilde{q}(\widetilde{\theta}_{0},\widetilde{\theta}_{1}),\widetilde{\pi}_{T}(\widetilde{\theta}_{1})\widetilde{q}(\widetilde{\theta}_{1},\widetilde{\theta}_{0})\frac{\overline{p}(Y_{1:T}\mid\theta_{1},u_{1})}{\overline{p}(Y_{1:T}\mid\theta_{0},u_{0})}\right)\right.\label{eq:firsttermofdecomposition}\\
 & \qquad\quad-\left.\min\left(\varphi(\widetilde{\theta}_{0};0,\overline{\Sigma})\widetilde{q}(\widetilde{\theta}_{0},\widetilde{\theta}_{1}),\varphi(\widetilde{\theta}_{1};0,\overline{\Sigma})\widetilde{q}(\widetilde{\theta}_{1},\widetilde{\theta}_{0})\frac{\overline{p}(Y_{1:T}\mid\theta_{1},u_{1})}{\overline{p}(Y_{1:T}\mid\theta_{0},u_{0})}\right)\right\vert \mathrm{d}\widetilde{\theta}_{0}\mathrm{d}\widetilde{\theta}_{1}\nonumber \\
 & \qquad+\iiint\widetilde{\pi}_{T}(\left.\mathrm{d}u_{0}\right\vert \widetilde{\theta}_{0})\widetilde{q}(\widetilde{\theta}_{0},\mathrm{d}\widetilde{\theta}_{1})\left\vert \varphi(\widetilde{\theta}_{0};0,\overline{\Sigma})\int K_{\rho_{T}}\left(u_{0},\mathrm{d}u_{1}\right)\min\left(1,\widetilde{r}\left(\widetilde{\theta}_{0},\widetilde{\theta}_{1}\right)\frac{\overline{p}(Y_{1:T}\mid\theta_{1},u_{1})}{\overline{p}(Y_{1:T}\mid\theta_{0},u_{0})}\right)\right.\label{eq:secondtermofdecomposition}\\
 & \qquad\quad-\left.\widetilde{\pi}_{T}\left(\widetilde{\theta}_{0}\right)\int\varphi\left(\mathrm{d}w;-\frac{\kappa^{2}}{2},\kappa^{2}\right)\min\left(1,\widetilde{r}\left(\widetilde{\theta}_{0},\widetilde{\theta}_{1}\right)\exp\left(w\right)\right)\right\vert \mathrm{d}\widetilde{\theta}_{0}\nonumber 
\end{align}
For the first term given in (\ref{eq:firsttermofdecomposition}),
using the inequality $\left\vert \min\left(x,y\right)-\min\left(w,z\right)\right\vert \leq\left\vert x-w\right\vert +\left\vert y-z\right\vert $,
we obtain the bound 
\begin{align}
\lefteqn{\iiint\widetilde{\pi}_{T}(\left.\mathrm{d}u_{0}\right\vert \widetilde{\theta}_{0})\int K_{\rho_{T}}\left(u_{0},\mathrm{d}u_{1}\right)\left\vert \min\left(\widetilde{\pi}_{T}(\widetilde{\theta}_{0})\widetilde{q}(\widetilde{\theta}_{0},\widetilde{\theta}_{1}),\widetilde{\pi}_{T}(\widetilde{\theta}_{1})\widetilde{q}(\widetilde{\theta}_{1},\widetilde{\theta}_{0})\frac{\overline{p}(Y_{1:T}\mid\theta_{1},u_{1})}{\overline{p}(Y_{1:T}\mid\theta_{0},u_{0})}\right)\right.}\nonumber \\
 & \qquad\quad-\left.\min\left(\varphi(\widetilde{\theta}_{0};0,\overline{\Sigma})\widetilde{q}(\widetilde{\theta}_{0},\widetilde{\theta}_{1}),\varphi(\widetilde{\theta}_{1};0,\overline{\Sigma})\widetilde{q}(\widetilde{\theta}_{1},\widetilde{\theta}_{0})\frac{\overline{p}(Y_{1:T}\mid\theta_{1},u_{1})}{\overline{p}(Y_{1:T}\mid\theta_{0},u_{0})}\right)\right\vert \mathrm{d}\widetilde{\theta}_{0}\mathrm{d}\widetilde{\theta}_{1}\nonumber \\
 & \leq\iiiint\widetilde{\pi}_{T}(\left.\mathrm{d}u_{0}\right\vert \widetilde{\theta}_{0})\widetilde{q}(\widetilde{\theta}_{0},\mathrm{d}\widetilde{\theta}_{1})K_{\rho_{T}}\left(u_{0},\mathrm{d}u_{1}\right)\left\vert \widetilde{\pi}_{T}(\widetilde{\theta}_{0})-\varphi(\widetilde{\theta}_{0};0,\overline{\Sigma})\right\vert \mathrm{d}\widetilde{\theta}_{0}\label{eq:firsttermequality}\\
 & \qquad+\iiiint\widetilde{\pi}_{T}(\left.\mathrm{d}u_{0}\right\vert \widetilde{\theta}_{0})K_{\rho_{T}}\left(u_{0},\mathrm{d}u_{1}\right)\left\vert \widetilde{\pi}_{T}(\widetilde{\theta}_{1})-\varphi(\widetilde{\theta}_{1};0,\overline{\Sigma})\right\vert \widetilde{q}(\widetilde{\theta}_{1},\mathrm{d}\widetilde{\theta}_{0})\frac{\overline{p}(Y_{1:T}\mid\theta_{1},u_{1})}{\overline{p}(Y_{1:T}\mid\theta_{0},u_{0})}\mathrm{d}\widetilde{\theta}_{1}\label{eq:secondtermequality}
\end{align}
The term\ (\ref{eq:firsttermequality}) satisfies 
\begin{align*}
 & \iiiint\widetilde{\pi}_{T}(\left.\mathrm{d}u_{0}\right\vert \widetilde{\theta}_{0})\widetilde{q}(\widetilde{\theta}_{0},\mathrm{d}\widetilde{\theta}_{1})K_{\rho_{T}}\left(u_{0},\mathrm{d}u_{1}\right)\left\vert \widetilde{\pi}_{T}(\widetilde{\theta}_{0})-\varphi(\widetilde{\theta}_{0};0,\overline{\Sigma})\right\vert \mathrm{d}\widetilde{\theta}_{0}\\
 & =\int\left\vert \widetilde{\pi}_{T}(\widetilde{\theta}_{0})-\varphi(\widetilde{\theta}_{0};0,\overline{\Sigma})\right\vert \mathrm{d}\widetilde{\theta}_{0}\overset{\mathbb{P}^{Y}}{\rightarrow}0
\end{align*}
by Assumption \ref{ass:BVM}. Therefore, for any subsequence\ $\left\{ T_{k};k\geq0\right\} $
we can extract a further subsequence $\left\{ T_{k}^{1};k\geq0\right\} $
such that along this subsequence 
\[
\int\left\vert \widetilde{\pi}_{T}(\widetilde{\theta}_{0})-\varphi(\widetilde{\theta}_{0};0,\overline{\Sigma})\right\vert \mathrm{d}\widetilde{\theta}_{0}\rightarrow0
\]
 almost surely. Since 
\[
\widetilde{\pi}_{T}(\left.\mathrm{d}u_{0}\right\vert \widetilde{\theta}_{0})=m_{T}\left(\mathrm{d}u_{0}\right)\overline{p}(Y_{1:T}\mid\theta_{0},u_{0})
\]
then the term (\ref{eq:secondtermequality}) satisfies along $\left\{ T_{k}^{1};k\geq0\right\} $
\begin{align*}
\lefteqn{\iiiint\widetilde{\pi}_{T}(\left.\mathrm{d}u_{0}\right\vert \widetilde{\theta}_{0})\widetilde{q}(\widetilde{\theta}_{1},\mathrm{d}\widetilde{\theta}_{0})K_{\rho_{T}}(u_{0},\mathrm{d}u_{1})\left\vert \widetilde{\pi}_{T}(\widetilde{\theta}_{1})-\varphi(\widetilde{\theta}_{1};0,\overline{\Sigma})\right\vert \frac{\overline{p}(Y_{1:T}\mid\theta_{1},u_{1})}{\overline{p}(Y_{1:T}\mid\theta_{0},u_{0})}\mathrm{d}\widetilde{\theta}_{1}}\\
 & =\iiiint m_{T}\left(\mathrm{d}u_{0}\right)\widetilde{q}(\widetilde{\theta}_{1},\mathrm{d}\widetilde{\theta}_{0})K_{\rho_{T}}\left(u_{0},\mathrm{d}u_{1}\right)\left\vert \widetilde{\pi}_{T}(\widetilde{\theta}_{1})-\varphi(\widetilde{\theta}_{1};0,\overline{\Sigma})\right\vert \overline{p}(Y_{1:T}\mid\theta_{1},u_{1})\mathrm{d}\widetilde{\theta}_{1}\\
 & =\iiint\widetilde{q}(\widetilde{\theta}_{1},\mathrm{d}\widetilde{\theta}_{0})m_{T}\left(\mathrm{d}u_{1}\right)\left\vert \widetilde{\pi}_{T}(\widetilde{\theta}_{1})-\varphi(\widetilde{\theta}_{1};0,\overline{\Sigma})\right\vert \overline{p}(Y_{1:T}\mid\theta_{1},u_{1})\mathrm{d}\widetilde{\theta}_{1}\text{ \ (}K_{\rho_{T}}\text{ }m\text{-invariant)}\\
 & =\iiint\widetilde{q}(\widetilde{\theta}_{1},\mathrm{d}\widetilde{\theta}_{0})\widetilde{\pi}_{T}(\left.\mathrm{d}u_{1}\right\vert \widetilde{\theta}_{1})\left\vert \widetilde{\pi}_{T}(\widetilde{\theta}_{1})-\varphi(\widetilde{\theta}_{1};0,\overline{\Sigma})\right\vert \mathrm{d}\widetilde{\theta}_{1}\\
 & =\int\left\vert \widetilde{\pi}_{T}(\widetilde{\theta}_{1})-\varphi(\widetilde{\theta}_{1};0,\overline{\Sigma})\right\vert \mathrm{d}\widetilde{\theta}_{1}\rightarrow0
\end{align*}
almost surely.

Going back to the term given by (\ref{eq:secondtermofdecomposition}),
we note that 
\begin{align}
\lefteqn{\iiint\widetilde{\pi}_{T}(\left.\mathrm{d}u_{0}\right\vert \widetilde{\theta}_{0})\widetilde{q}(\widetilde{\theta}_{0},\mathrm{d}\widetilde{\theta}_{1})\left\vert \varphi(\widetilde{\theta}_{0};0,\overline{\Sigma})\int K_{\rho_{T}}(u_{0},\mathrm{d}u_{1})\min\left(1,\widetilde{r}\left(\widetilde{\theta}_{0},\widetilde{\theta}_{1}\right)\frac{\overline{p}(Y_{1:T}\mid\theta_{1},u_{1})}{\overline{p}(Y_{1:T}\mid\theta_{0},u_{0})}\right)\right.}\nonumber \\
 & \qquad-\left.\widetilde{\pi}_{T}(\widetilde{\theta}_{0})\int\varphi\left(\mathrm{d}w;-\frac{\kappa^{2}}{2},\kappa^{2}\right)\min\left(1,\widetilde{r}\left(\widetilde{\theta}_{0},\widetilde{\theta}_{1}\right)\exp\left(w\right)\right)\right\vert \mathrm{d}\widetilde{\theta}_{0}\nonumber \\
 & \leq\iiint\widetilde{\pi}_{T}(\left.\mathrm{d}u_{0}\right\vert \widetilde{\theta}_{0})\widetilde{q}(\widetilde{\theta}_{0},\mathrm{d}\widetilde{\theta}_{1})\left\vert \varphi(\widetilde{\theta}_{0};0,\overline{\Sigma)}\int K_{\rho_{T}}\left(u_{0},\mathrm{d}u_{1}\right)\min\left(1,\widetilde{r}\left(\widetilde{\theta}_{0},\widetilde{\theta}_{1}\right)\frac{\overline{p}(Y_{1:T}\mid\theta_{1},u_{1})}{\overline{p}(Y_{1:T}\mid\theta_{0},u_{0})}\right)\right.\label{eq:firsttermofsecondterm}\\
 & \qquad\qquad-\left.\varphi\left(\widetilde{\theta}_{0};0,\overline{\Sigma}\right)\int\varphi\left(\mathrm{d}w;-\frac{\kappa^{2}}{2},\kappa^{2}\right)\min\left(1,\widetilde{r}\left(\widetilde{\theta}_{0},\widetilde{\theta}_{1}\right)\exp\left(w\right)\right)\right\vert \mathrm{d}\widetilde{\theta}_{0}\nonumber \\
 & \qquad+\iiint\left\vert \widetilde{\pi}_{T}(\widetilde{\theta}_{0})-\varphi(\widetilde{\theta}_{0};0,\overline{\Sigma})\right\vert \widetilde{q}(\widetilde{\theta}_{0},\mathrm{d}\widetilde{\theta}_{1})\varphi\left(\mathrm{d}w;-\frac{\kappa^{2}}{2},\kappa^{2}\right)\min\left(1,\widetilde{r}\left(\widetilde{\theta}_{0},\widetilde{\theta}_{1}\right)\exp\left(w\right)\right)\mathrm{d}\widetilde{\theta}_{0}\label{eq:secondtermofsecondterm}
\end{align}
where (\ref{eq:secondtermofsecondterm}) satisfies
\begin{align*}
 & \iiint\left\vert \widetilde{\pi}_{T}(\widetilde{\theta}_{0})-\varphi(\widetilde{\theta}_{0};0,\overline{\Sigma})\right\vert \widetilde{q}(\widetilde{\theta}_{0},\mathrm{d}\widetilde{\theta}_{1})\varphi\left(\mathrm{d}w;-\frac{\kappa^{2}}{2},\kappa^{2}\right)\min\left(1,\widetilde{r}\left(\widetilde{\theta}_{0},\widetilde{\theta}_{1}\right)\exp\left(w\right)\right)\mathrm{d}\widetilde{\theta}_{0}\\
 & \leq\iiint\left\vert \widetilde{\pi}_{T}(\widetilde{\theta}_{0})-\varphi(\widetilde{\theta}_{0};0,\overline{\Sigma})\right\vert \widetilde{q}(\widetilde{\theta}_{0},\mathrm{d}\widetilde{\theta}_{1})\varphi\left(\mathrm{d}w;-\frac{\kappa^{2}}{2},\kappa^{2}\right)\mathrm{d}\widetilde{\theta}_{0}\\
 & =\int\left\vert \widetilde{\pi}_{T}(\widetilde{\theta}_{0})-\varphi(\widetilde{\theta}_{0};0,\overline{\Sigma})\right\vert \mathrm{d}\widetilde{\theta}_{0}\rightarrow0
\end{align*}
almost surely along $\left\{ T_{k}^{1};k\geq0\right\} $. %
We can rewrite (\ref{eq:firsttermofsecondterm}) as 
\begin{align*}
 & \iiint\varphi\left(\mathrm{d}\theta_{0};\widehat{\theta}_{T},\frac{\overline{\Sigma}}{T}\right)\overline{\pi}_{T}\left(\left.\mathrm{d}u_{0}\right\vert \theta_{0}\right)q\left(\theta_{0},\mathrm{d}\theta_{1}\right)\left\vert \int K_{\rho_{T}}\left(u_{0},\mathrm{d}u_{1}\right)\min\left(1,\frac{\varphi(\sqrt{T}(\theta_{1}-\widehat{\theta}_{T});0,\overline{\Sigma})}{\varphi(\sqrt{T}(\theta_{0}-\widehat{\theta}_{T});0,\overline{\Sigma})}\frac{\overline{p}(Y_{1:T}\mid\theta_{1},u_{1})}{\overline{p}(Y_{1:T}\mid\theta_{0},u_{0})}\right)\right.\\
 & -\left.\int\varphi\left(\mathrm{d}w;-\frac{\kappa^{2}}{2},\kappa^{2}\right)\min\left(1,\frac{\varphi(\sqrt{T}(\theta_{1}-\widehat{\theta}_{T});0,\overline{\Sigma})}{\varphi(\sqrt{T}(\theta_{0}-\widehat{\theta}_{T});0,\overline{\Sigma})}\exp\left(w\right)\right)\right\vert \\
 & =\iiint\varphi\left(\mathrm{d}\theta_{0};\widehat{\theta}_{T},\frac{\overline{\Sigma}}{T}\right)\overline{\pi}_{T}\left(\left.\mathrm{d}u_{0}\right\vert \theta_{0}\right)\upsilon\left(\mathrm{d}\xi\right)\left\vert \int K_{\rho_{T}}\left(u_{0},\mathrm{d}u_{1}\right)\min\left(1,\frac{\varphi(\sqrt{T}(\theta_{0}+\xi/\sqrt{T}-\widehat{\theta}_{T});0,\overline{\Sigma})}{\varphi(\sqrt{T}(\theta_{0}-\widehat{\theta}_{T});0,\overline{\Sigma})}\right.\right.\\
 & \times\left.\frac{\overline{p}(Y_{1:T}\mid\theta_{0}+\xi/\sqrt{T},u_{1})}{\overline{p}(Y_{1:T}\mid\theta_{0},u_{0})}\right)-\left.\int\varphi\left(\mathrm{d}w;-\frac{\kappa^{2}}{2},\kappa^{2}\right)\min\left(1,\frac{\varphi(\sqrt{T}(\theta_{0}+\xi/\sqrt{T}-\widehat{\theta}_{T});0,\overline{\Sigma})}{\varphi(\sqrt{T}(\theta_{0}-\widehat{\theta}_{T});0,\overline{\Sigma})}\exp\left(w\right)\right)\right\vert .
\end{align*}
As $\widehat{\theta}_{T}\overset{\mathbb{P}^{Y}}{\rightarrow}\overline{\theta}$,
we extract a further subsequence $\left\{ T_{k}^{2};k\geq0\right\} $
of $\left\{ T_{k}^{1};k\geq0\right\} $ such that along this subsequence
$\widehat{\theta}_{T}\overset{}{\rightarrow}\overline{\theta}$ almost
surely. Hence if we let $A^{T}\left(\varepsilon\right)=\left\{ Y_{1:T}:\left\Vert \widehat{\theta}_{T}-\overline{\theta}\right\Vert <\varepsilon/2\right\} $
which satisfies $\mathbb{P}^{Y}\left(\left(A^{T}\left(\varepsilon\right)\right)^{\mathtt{C}}\right)=o\left(1\right)$
then along this subsequence $\mathbb{I}\left(A^{T}\left(\varepsilon\right)\right)\rightarrow1$
almost surely and therefore (\ref{eq:firsttermofsecondterm}) is equal
to 
\begin{align*}
 & \mathbb{I}\left(A^{T}\left(\varepsilon\right)\right)\iiint\varphi\left(\mathrm{d}\theta_{0};\widehat{\theta}_{T},\frac{\overline{\Sigma}}{T}\right)\overline{\pi}_{T}\left(\left.\mathrm{d}u_{0}\right\vert \theta_{0}\right)\upsilon\left(\mathrm{d}\xi\right)\left\vert \int K_{\rho_{T}}\left(u_{0},\mathrm{d}u_{1}\right)\min\left(1,\frac{\varphi(\sqrt{T}(\theta_{0}+\xi/\sqrt{T}-\widehat{\theta}_{T});0,\overline{\Sigma})}{\varphi(\sqrt{T}(\theta_{0}-\widehat{\theta}_{T});0,\overline{\Sigma})}\right.\right.\\
 & \times\left.\frac{\overline{p}(Y_{1:T}\mid\theta_{0}+\xi/\sqrt{T},u_{1})}{\overline{p}(Y_{1:T}\mid\theta_{0},u_{0})}\right)-\left.\int\varphi(\mathrm{d}w;-\kappa{}^{2}/2,\kappa{}^{2})\min\left(1,\frac{\varphi(\sqrt{T}(\theta_{0}+\xi/\sqrt{T}-\widehat{\theta}_{T});0,\overline{\Sigma})}{\varphi(\sqrt{T}(\theta_{0}-\widehat{\theta}_{T});0,\overline{\Sigma})}\exp\left(w\right)\right)\right\vert \text{ }\\
 & +o\left(1\right)
\end{align*}
almost surely. Along $\left\{ T_{k}^{2};k\geq0\right\} $, we can
rewrite the integral in the above expression as 
\begin{align*}
 & \mathbb{I}\left(A^{T}\left(\varepsilon\right)\right)\iiint\varphi\left(\mathrm{d}\theta_{0};\widehat{\theta}_{T},\frac{\overline{\Sigma}}{T}\right)\mathbb{I}\left(\left\Vert \widehat{\theta}_{T}-\theta_{0}\right\Vert <\varepsilon/2\right)\overline{\pi}_{T}\left(\left.\mathrm{d}u_{0}\right\vert \theta_{0}\right)\upsilon\left(\mathrm{d}\xi\right)\\
 & \left\vert \int K_{\rho_{T}}\left(u_{0},\mathrm{d}u_{1}\right)\min\left(1,\frac{\varphi(\sqrt{T}(\theta_{0}+\xi/\sqrt{T}-\widehat{\theta}_{T});0,\overline{\Sigma})}{\varphi(\sqrt{T}(\theta_{0}-\widehat{\theta}_{T});0,\overline{\Sigma})}\frac{\overline{p}(Y_{1:T}\mid\theta_{0}+\xi/\sqrt{T},u_{1})}{\overline{p}(Y_{1:T}\mid\theta_{0},u_{0})}\right)\right.\\
 & -\left.\int\varphi(\mathrm{d}w;-\kappa^{2}/2,\kappa^{2})\min\left(1,\frac{\varphi(\sqrt{T}(\theta_{0}+\xi/\sqrt{T}-\widehat{\theta}_{T});0,\overline{\Sigma})}{\varphi(\sqrt{T}(\theta_{0}-\widehat{\theta}_{T});0,\overline{\Sigma})}\exp\left(w\right)\right)\right\vert +o\left(1\right).
\end{align*}
Notice that the functions 
\[
x\mapsto\min\left(1,\frac{\varphi(\sqrt{T}(\theta_{0}+\xi/\sqrt{T}-\widehat{\theta}_{T});0,\overline{\Sigma})}{\varphi(\sqrt{T}(\theta_{0}-\widehat{\theta}_{T});0,\overline{\Sigma})}\exp\left(x\right)\right)
\]
are bounded above by $1$ and Lipschitz, with Lipschitz constants
bounded by $1$ uniformly in all parameters. Therefore (\ref{eq:firsttermofsecondterm})
is bounded almost surely along $\left\{ T_{k}^{2};k\geq0\right\} $
by 
\begin{align}
 & \mathbb{I}\left(A^{T}\left(\varepsilon\right)\right)\iiint\varphi\left(\mathrm{d}\theta_{0};\widehat{\theta}_{T},\frac{\overline{\Sigma}}{T}\right)\mathbb{I}\left(\left\Vert \widehat{\theta}_{T}-\theta_{0}\right\Vert <\varepsilon/2\right)\overline{\pi}_{T}\left(\left.\mathrm{d}u_{0}\right\vert \theta_{0}\right)\upsilon\left(\mathrm{d}\xi\right)\nonumber \\
 & \times\underset{f:\left\Vert f\right\Vert _{BL}\leq2}{\sup}\text{ }\left\vert \int K_{\rho_{T}}\left(u_{0},\mathrm{d}u_{1}\right)f\left\{ \log\left(\frac{\overline{p}(Y_{1:T}\mid\theta_{0}+\xi/\sqrt{T},u_{1})}{\overline{p}(Y_{1:T}\mid\theta_{0},u_{0})}\right)\right\} -\int\varphi(\mathrm{d}w;-\kappa{}^{2}/2,\kappa{}^{2})f\left(w\right)\right\vert +o\left(1\right),\label{eq:applicationconditionalCLT}
\end{align}
where $\left\Vert f\right\Vert _{BL}$ is defined in (\ref{def:BLfunction}). 

We further decompose (\ref{eq:applicationconditionalCLT}) as 
\begin{align}
 & \mathbb{I}\left(A^{T}\left(\varepsilon\right)\right)\iiint\varphi\left(\mathrm{d}\theta_{0};\widehat{\theta}_{T},\frac{\overline{\Sigma}}{T}\right)\mathbb{I}\left(\left\Vert \widehat{\theta}_{T}-\theta_{0}\right\Vert <\varepsilon/2\right)\overline{\pi}_{T}\left(\left.\mathrm{d}u_{0}\right\vert \theta_{0}\right)\upsilon\left(\mathrm{d}\xi\right)\nonumber \\
 & \times\underset{f:\left\Vert f\right\Vert _{BL}\leq2}{\sup}\text{ }\left\vert \int K_{\rho_{T}}\left(u_{0},\mathrm{d}u_{1}\right)f\left\{ \log\left(\frac{\overline{p}(Y_{1:T}\mid\theta_{0}+\xi/\sqrt{T},u_{1})}{\overline{p}(Y_{1:T}\mid\theta_{0},u_{0})}\right)\right\} -\int\varphi\left(\mathrm{d}w;-\kappa{}^{2}(\theta_{0})/2,\kappa{}^{2}(\theta_{0})\right)f\left(w\right)\right\vert \nonumber \\
 & +\mathbb{I}\left(A^{T}\left(\varepsilon\right)\right)\iiint\left.\left\{ \varphi\left(\mathrm{d}\theta_{0};\widehat{\theta}_{T},\frac{\overline{\Sigma}}{T}\right)\mathbb{I}\left(\left\Vert \widehat{\theta}_{T}-\theta_{0}\right\Vert <\varepsilon/2\right)\overline{\pi}_{T}\left(\left.\mathrm{d}u_{0}\right\vert \theta_{0}\right)\upsilon\left(\mathrm{d}\xi\right)\right.\right.\label{eq:justbeforeunifcondclt}\\
 & \times\left.d_{BL}\left(\mathcal{N}\left(-\kappa{}^{2}(\theta_{0})/2,\kappa{}^{2}(\theta_{0})\right),\mathcal{N}\left(-\kappa{}^{2}/2,\kappa{}^{2}\right)\right)\right\} +o(1).\nonumber 
\end{align}
The second term can be easily bounded above by
\begin{align*}
 & \mathbb{I}\left(A^{T}\left(\varepsilon\right)\right)\iiint\varphi\left(\mathrm{d}\theta_{0};\widehat{\theta}_{T},\frac{\overline{\Sigma}}{T}\right)\mathbb{I}\left(\left\Vert \widehat{\theta}_{T}-\theta_{0}\right\Vert <\varepsilon/2\right)\left[\frac{1}{2}\left|\kappa^{2}(\theta_{0})-\kappa^{2}(\bar{\theta})\right|+\left|\kappa(\theta_{0})-\kappa(\bar{\theta})\right|\sqrt{\frac{2}{\pi}}\right]\\
 & \leq\mathbb{I}\left(A^{T}\left(\varepsilon\right)\right)\iiint\varphi\left(\mathrm{d}\theta_{0};\widehat{\theta}_{T},\frac{\overline{\Sigma}}{T}\right)\mathbb{I}\left(\left\Vert \widehat{\theta}_{T}-\theta_{0}\right\Vert <\varepsilon/2\right)\left|\kappa(\theta_{0})-\kappa(\bar{\theta})\right|\left[\frac{1}{2}\left(\kappa(\theta_{0})+\kappa(\bar{\theta})\right)+\sqrt{\frac{2}{\pi}}\right]\\
 & \leq C\mathbb{I}\left(A^{T}\left(\varepsilon\right)\right)\iiint\varphi\left(\mathrm{d}\theta_{0};\widehat{\theta}_{T},\frac{\overline{\Sigma}}{T}\right)\mathbb{I}\left(\left\Vert \widehat{\theta}_{T}-\theta_{0}\right\Vert <\varepsilon/2\right)\left(\left|\kappa(\theta_{0})-\kappa(\widehat{\theta}_{T})\right|+\left|\kappa(\bar{\theta})-\kappa(\widehat{\theta}_{T})\right|\right)\\
\leq & C\mathbb{I}\left(A^{T}\left(\varepsilon\right)\right)\int\varphi\left(\mathrm{d}\theta_{0};\widehat{\theta}_{T},\frac{\overline{\Sigma}}{T}\right)\mathbb{I}\left(\left\Vert \widehat{\theta}_{T}-\theta_{0}\right\Vert <\varepsilon/2\right)\left[\left\Vert \theta_{0}-\widehat{\theta}_{T}\right\Vert +\left\Vert \widehat{\theta}_{T}-\bar{\theta}\right\Vert \right]\\
\leq & C\mathbb{I}\left(A^{T}\left(\varepsilon\right)\right)\left[T^{-1/2}\int\varphi\left(\mathrm{d}\zeta;0,I_{d}\right)\left\Vert \left(\overline{\Sigma}\right)^{1/2}\zeta\right\Vert +\left\Vert \widehat{\theta}_{T}-\bar{\theta}\right\Vert \right],
\end{align*}
where we have used the fact that $\kappa$ is locally Lipschitz around
$\bar{\theta}$. As we are in a subsequence along which $\widehat{\theta}_{T}\overset{}{\rightarrow}\overline{\theta}$
almost surely, then this quantity converges to zero almost surely
along this subsequence. Finally the first term of (\ref{eq:justbeforeunifcondclt})
can be controlled for $\varepsilon$ small enough by
\begin{align*}
 & \mathbb{I}\left(A^{T}\left(\varepsilon\right)\right)\int\varphi\left(\mathrm{d}\theta_{0};\widehat{\theta}_{T},\frac{\overline{\Sigma}}{T}\right)\mathbb{I}\left(\left\Vert \widehat{\theta}_{T}-\theta_{0}\right\Vert <\varepsilon/2\right)\\
\times & \sup_{\theta_{0}\in N(\bar{\theta})}\iint\left\{ \overline{\pi}_{T}\left(\left.\mathrm{d}u_{0}\right\vert \theta_{0}\right)\upsilon\left(\mathrm{d}\xi\right)\right.\\
\times & \left.\underset{f:\left\Vert f\right\Vert _{BL}\leq2}{\sup}\text{ }\left\vert \int K_{\rho_{T}}\left(u_{0},\mathrm{d}u_{1}\right)f\left\{ \log\left(\frac{\overline{p}(Y_{1:T}\mid\theta_{0}+\xi/\sqrt{T},u_{1})}{\overline{p}(Y_{1:T}\mid\theta_{0},u_{0})}\right)\right\} -\int\varphi(\mathrm{d}w;-\kappa{}^{2}(\theta_{0})/2,\kappa{}^{2}(\theta_{0}))f\left(w\right)\right\vert \right\} \\
= & \mathbb{I}\left(A^{T}\left(\varepsilon\right)\right)\int\varphi\left(\mathrm{d}\theta_{0};\widehat{\theta}_{T},\frac{\overline{\Sigma}}{T}\right)\mathbb{I}\left(\left\Vert \widehat{\theta}_{T}-\theta_{0}\right\Vert <\varepsilon/2\right)\\
\times & \sup_{\theta_{0}\in N(\bar{\theta})}\iint\left\{ \overline{\pi}_{T}\left(\left.\mathrm{d}u_{0}\right\vert \theta_{0}\right)\upsilon\left(\mathrm{d}\xi\right),\right.\\
\times & \left.\underset{f:\left\Vert f\right\Vert _{BL}\leq2}{\sup}\text{ }\left\vert \int K_{\rho_{T}}\left(u_{0},\mathrm{d}u_{1}\right)f\left\{ \log\left(\frac{\overline{p}(Y_{1:T}\mid\theta_{0}+\xi/\sqrt{T},u_{1})}{\overline{p}(Y_{1:T}\mid\theta_{0},u_{0})}\right)\right\} -\int\varphi(\mathrm{d}w;-\kappa{}^{2}(\theta_{0})/2,\kappa{}^{2}(\theta_{0}))f\left(w\right)\right\vert \right\} ,
\end{align*}
which vanishes in probability by Assumption \ref{Assumption:uniformCLT}.
Hence we can extract a further subsequence along which this convergence
happens almost surely. The result follows. 
\end{proof}

\begin{lem}
\label{Lemma:CVofPimpliesCVofPk}If along a subsequence $\left\{ T_{k};k\geq0\right\} $
we have almost surely 
\[
\mathbb{E}\left(\left\vert Q_{T}f(\widetilde{\vartheta}_{0}^{T},\mathsf{U}_{0}^{T})-Pf(\widetilde{\vartheta}_{0}^{T})\right\vert \right)\rightarrow0
\]
for all $f\in B\left(\mathbb{R}^{d}\right)$, then along $\left\{ T_{k};k\geq0\right\} $
we have almost surely 
\[
\mathbb{E}\left(\left\vert Q_{T}^{k}f(\widetilde{\vartheta}_{0}^{T},\mathsf{U}_{0}^{T})-P^{k}f(\widetilde{\vartheta}_{0}^{T})\right\vert \right)\rightarrow0
\]
for all $f\in B\left(\mathbb{R}^{d}\right)$ and all $k\geq1$. 
\end{lem}
\begin{rem*}
We emphasize again that the subset on which this almost sure convergence
occurs is independent of $f$ and $k$. 
\end{rem*}
\begin{proof}[Proof of Lemma \textit{\ref{Lemma:CVofPimpliesCVofPk}}]
We prove the result by induction. For $k=1$, this follows from the
assumption. Now we have 
\begin{align*}
\lefteqn{Q_{T}^{k+1}f(\widetilde{\theta}_{0},u_{0})-P^{k+1}f(\widetilde{\theta}_{0})}\\
 & =Q_{T}^{k+1}f(\widetilde{\theta}_{0},u_{0})-Q_{T}(P^{k}f)(\widetilde{\theta}_{0},u_{0})+Q_{T}(P^{k}f)(\widetilde{\theta}_{0},u_{0})-P^{k+1}f(\widetilde{\theta}_{0}).
\end{align*}
and therefore 
\begin{align*}
 & \mathbb{E}\left(\left\vert Q_{T}^{k+1}f(\widetilde{\vartheta}_{0}^{T},\mathsf{U}_{0}^{T})-P^{k+1}f(\widetilde{\vartheta}_{0}^{T})\right\vert \right)\\
 & \leq\mathbb{E}\left(\left\vert Q_{T}^{k+1}f(\widetilde{\vartheta}_{0}^{T},\mathsf{U}_{0}^{T})-Q_{T}(P^{k}f)(\widetilde{\vartheta}_{0}^{T},\mathsf{U}_{0}^{T})\right\vert \right)+\mathbb{E}\left(\left\vert Q_{T}(P^{k}f)(\widetilde{\vartheta}_{0}^{T},\mathsf{U}_{0}^{T})-P^{k+1}f(\widetilde{\vartheta}_{0}^{T})\right\vert \right)\\
 & \leq\mathbb{E}\left(\left\vert Q_{T}^{k}f(\widetilde{\vartheta}_{0}^{T},\mathsf{U}_{0}^{T})-P^{k}f(\widetilde{\vartheta}_{0}^{T},\mathsf{U}_{0}^{T})\right\vert \right)+\mathbb{E}\left(\left\vert Q_{T}(P^{k}f)(\widetilde{\vartheta}_{0}^{T},\mathsf{U}_{0}^{T})-P(P^{k}f)(\widetilde{\vartheta}_{0}^{T})\right\vert \right),
\end{align*}
since $Q_{T}$ is $\widetilde{\pi}_{T}$-invariant. We can now apply
the induction hypothesis to the functions $f$ and $P^{k}f$ as $P^{k}f\in B\left(\mathbb{R}^{d}\right)$. 
\end{proof}
\begin{prop}
\label{Proposition:KernelCVplusBVMequalweakCV}Under the assumptions
of Theorem \ref{Theorem:weakconvergence}, for any subsequence $\left\{ T_{k};k\geq0\right\} $
we can extract a further subsequence $\left\{ T_{k}^{\prime};k\geq0\right\} $
such that almost surely along this subsequence we have 
\[
\mathbb{E}\left[{\textstyle \prod\limits _{i=0}^{n}}f_{i}\left(\widetilde{\vartheta}_{k_{i}}^{T}\right)\right]\rightarrow\mathbb{E}\left[{\textstyle \prod\limits _{i=0}^{n}}f_{i}\left(\widetilde{\vartheta}_{k_{i}}\right)\right]
\]
for any $n\geq0$, any $0\leq k_{0}<k_{1}<k_{2}<\dots<k_{n}\in\mathbb{N}$
and $f_{0},\dots,f_{n}\in B\left(\mathbb{R}^{d}\right)$. 
\end{prop}
\begin{proof}[Proof of Proposition \textit{\ref{Proposition:KernelCVplusBVMequalweakCV}}]
In Proposition \ref{Proposition:Kernelconverges}, we have extracted
a subsequence $\left\{ T_{k}^{\prime};k\geq0\right\} $ of $\left\{ T_{k};k\geq0\right\} $
such that along this subsequence 
\[
\int\left\vert \widetilde{\pi}_{T}(\widetilde{\theta}_{0})-\varphi(\widetilde{\theta}_{0};0,\overline{\Sigma})\right\vert \mathrm{d}\widetilde{\theta}_{0}\rightarrow0
\]
almost surely. Hence, along this subsequence, the result\ holds for
$n=0$. For $n=1,$ we have 
\begin{align*}
\lefteqn{\left\vert \mathbb{E}\left[f_{0}(\widetilde{\vartheta}_{k_{0}}^{T})f_{1}(\widetilde{\vartheta}_{k_{1}}^{T})\right]-\mathbb{E}\left[f_{0}(\widetilde{\vartheta}_{k_{0}})f_{1}(\widetilde{\vartheta}_{k_{1}})\right]\right\vert }\\
 & =\left\vert \int f_{0}(\widetilde{\theta}_{0})\widetilde{\pi}_{T}(\widetilde{\theta}_{0},u_{0})Q_{T}^{k_{1}-k_{0}}f_{1}(\widetilde{\theta}_{0},u_{0})\mathrm{d}\widetilde{\theta}_{0}\mathrm{d}u_{0}-\int f_{0}(\widetilde{\theta}_{0})\varphi(\widetilde{\theta}_{0};0,\overline{\Sigma})P^{k_{1}-k_{0}}f_{1}(\widetilde{\theta}_{0})\mathrm{d}\widetilde{\theta}_{0}\right\vert \\
 & \leq\left\vert \int f_{0}(\widetilde{\theta}_{0})\widetilde{\pi}_{T}(\widetilde{\theta}_{0},u_{0})\{Q_{T}^{k_{1}-k_{0}}f_{1}(\widetilde{\theta}_{0},u_{0})-P^{k_{1}-k_{0}}f_{1}(\widetilde{\theta}_{0})\}\mathrm{d}\widetilde{\theta}_{0}\mathrm{d}u_{0}\right\vert \\
 & +\left\vert \int f_{0}(\widetilde{\theta}_{0})\{\widetilde{\pi}_{T}(\widetilde{\theta}_{0})-\varphi(\widetilde{\theta}_{0};0,\overline{\Sigma})\}P^{k_{1}-k_{0}}f_{1}(\widetilde{\theta}_{0})\mathrm{d}\widetilde{\theta}_{0}\right\vert \\
 & \leq\mathbb{E}\left[\left\vert Q_{T}^{k_{1}-k_{0}}f_{1}(\widetilde{\vartheta}_{0}^{T},\mathsf{U}_{0}^{T})-P^{k_{1}-k_{0}}f_{1}(\widetilde{\vartheta}_{0}^{T})\right\vert \right]+\int\left\vert \widetilde{\pi}_{T}(\widetilde{\theta}_{0})-\varphi(\widetilde{\theta}_{0};0,\overline{\Sigma})\right\vert \mathrm{d}\widetilde{\theta}_{0}\text{.}
\end{align*}
Hence from Lemma \ref{Lemma:CVofPimpliesCVofPk}, the result also
follows for $n=1$. Now for any $n\geq1$, we have 
\begin{align}
\mathbb{E}\left[{\textstyle \prod\nolimits _{j=0}^{n+1}}f_{j}(\widetilde{\vartheta}_{k_{j}}^{T})\right] & =\mathbb{E}\left[{\textstyle \prod\nolimits _{j=0}^{n}}f_{j}(\widetilde{\vartheta}_{k_{j}}^{T})Q_{T}^{k_{n+1}-k_{n}}f_{n+1}(\widetilde{\vartheta}_{k_{n}}^{T},U_{k_{n}}^{T})\right]\nonumber \\
 & =\mathbb{E}\left[{\textstyle \prod\nolimits _{j=0}^{n}}f_{j}(\widetilde{\vartheta}_{k_{j}}^{T})P^{k_{n+1}-k_{n}}f_{n+1}(\widetilde{\vartheta}_{k_{n}}^{T})\right]\label{eq:firsttermWEAKCV}\\
 & \qquad+\mathbb{E}\left[{\textstyle \prod\nolimits _{j=0}^{n}}f_{j}(\widetilde{\vartheta}_{k_{j}}^{T})\{Q_{T}^{k_{n+1}-k_{n}}f_{n+1}(\widetilde{\vartheta}_{k_{n}}^{T},U_{k_{n}}^{T})-P^{k_{n+1}-k_{n}}f_{n+1}(\widetilde{\vartheta}_{k_{n}}^{T})\}\right].\label{eq:secondtermWEAKCV}
\end{align}
By the induction hypothesis, the first term (\ref{eq:firsttermWEAKCV})
converges to 
\[
\mathbb{E}\left[{\textstyle \prod\nolimits _{j=0}^{n}}f_{j}(\widetilde{\vartheta}_{k_{j}})P^{k_{n+1}-k_{n}}f_{n+1}(\widetilde{\vartheta}_{k_{n}})\right]=\mathbb{E}\left[{\textstyle \prod\nolimits _{j=0}^{n+1}}f_{j}(\widetilde{\vartheta}_{k_{j}})\right].
\]
So it remains to show that the remainder (\ref{eq:secondtermWEAKCV})
vanishes. We have 
\begin{align*}
 & \left\vert \mathbb{E}\left[{\textstyle \prod\nolimits _{j=0}^{n}}f_{j}(\widetilde{\vartheta}_{k_{j}}^{T})\{Q_{T}^{k_{n+1}-k_{n}}f_{n+1}(\widetilde{\vartheta}_{k_{n}}^{T},U_{k_{n}}^{T})-P^{k_{n+1}-k_{n}}f_{n+1}(\widetilde{\vartheta}_{k_{n}}^{T})\}\right]\right\vert \\
 & \leq\mathbb{E}\left[\left\vert Q_{T}^{k_{n+1}-k_{n}}f_{n+1}(\widetilde{\vartheta}_{k_{n}}^{T},U_{k_{n}}^{T})-P^{k_{n+1}-k_{n}}f_{n+1}(\widetilde{\vartheta}_{k_{n}}^{T})\right\vert \right].
\end{align*}
So using Lemma \ref{Lemma:CVofPimpliesCVofPk}, this term vanishes
and the result follows. 
\end{proof}
\begin{proof}[Proof of Theorem \ref{Theorem:weakconvergence}]
We have shown that for any subsequence $\left\{ T_{k};k\geq0\right\} $
there exists a further subsequence $\left\{ T_{k}^{\prime};k\geq0\right\} $
such that almost surely we have 
\begin{equation}
\mathbb{E}\left({\textstyle \prod\nolimits _{j=0}^{n}}f_{j}(\widetilde{\vartheta}_{k_{j}}^{T})\right)\rightarrow\mathbb{E}\left({\textstyle \prod\nolimits _{j=0}^{n}}f_{j}(\widetilde{\vartheta}_{k_{j}})\right),\label{eq:productbounded}
\end{equation}
for any $n\geq0$, any $0\leq k_{0}<k_{1}<k_{2}<\dots<k_{n}\in\mathbb{N}$
and any bounded functions $f_{0},\dots,f_{n}$. Therefore, we have
by \citep[Proposition 3.4.6]{ethier2005} that on this subsequence
the probability measures on $\left(\mathbb{R}^{d}\right)^{\infty}$
given by the laws of $\left\{ \Theta_{T};T\geq1\right\} $ converge
weakly towards the probability measure induced by the law of $\left\{ \widetilde{\vartheta}_{n};n\geq0\right\} $
almost surely. From this, the result follows from a standard argument;
see, e.g., \citep[Theorem 2.3.2]{durrett2010}. 
\end{proof}

\subsection{Proofs for the bounding chain \ref{sec: optimisation}}
\begin{proof}[Proof of Proposition \textit{\ref{Proposition:Qstarproperties}}]
It is straightforward to check that $Q^{\ast}$ is $\pi-$reversible
as it follows from (\ref{eq:Boundingchainkernel}) that 
\begin{align*}
\pi\left(\mathrm{d}\theta\right)Q^{\ast}\left(\theta,\mathrm{d}\theta^{\prime}\right) & =\varrho_{\text{\textsc{U}}}\left(\kappa\right)\pi\left(\mathrm{d}\theta\right)Q_{\textsc{ex}}\left(\theta,\mathrm{d}\theta^{\prime}\right)+\left\{ 1-\varrho_{\text{\textsc{U}}}\left(\kappa\right)\right\} \pi\left(\mathrm{d}\theta\right)\delta_{\theta}\left(\mathrm{d}\theta^{\prime}\right)\\
 & =\varrho_{\text{\textsc{U}}}\left(\kappa\right)\pi\left(\mathrm{d}\theta^{\prime}\right)Q_{\textsc{ex}}\left(\theta^{\prime},\mathrm{d}\theta\right)+\left\{ 1-\varrho_{\text{\textsc{U}}}\left(\kappa\right)\right\} \pi\left(\mathrm{d}\theta^{\prime}\right)\delta_{\theta^{\prime}}\left(\mathrm{d}\theta\right)\\
 & =\pi\left(\mathrm{d}\theta^{\prime}\right)Q^{\ast}\left(\theta^{\prime},\mathrm{d}\theta\right)
\end{align*}
given $Q_{\textsc{ex}}$ is $\pi-$reversible. Now, we can also rewrite
$Q^{\ast}$ as 
\[
Q^{\ast}\left(\theta,\mathrm{d}\theta^{\prime}\right)=\varrho_{\text{\textsc{U}}}\left(\kappa\right)\alpha_{\textsc{ex}}(\theta,\theta^{\prime})q\left(\theta,\mathrm{d}\theta^{\prime}\right)+\left\{ 1-\varrho_{\text{\textsc{U}}}\left(\kappa\right)\varrho_{\textsc{ex}}\left(\theta\right)\right\} \delta_{\theta}\left(\mathrm{d}\theta^{\prime}\right)
\]
so the acceptance probability of a proposal is given by $\varrho_{\text{\textsc{U}}}\left(\kappa\right)\alpha_{\textsc{ex}}(\theta,\theta^{\prime})\leq\alpha_{\widehat{Q}}\left(\theta,\theta^{\prime}\right)$
for any $\left(\theta,\theta^{\prime}\right)$ as \linebreak{}
$\min\left(1,a\right)\min\left(1,b\right)\leq\min\left(1,ab\right)$
for $a,b\geq0$. The inequality (\ref{eq:inequalityaverageacceptanceproba})
follows directly from this result. Moreover, it also follows that
$\mathrm{IF}\left(h,Q^{\ast}\right)\leq\mathrm{IF}(h,\widehat{Q})$
from \citep[Theorem 4]{Tierney1998}, which is a general state-space
version of the main result in \citep{peskun1973optimum}. To establish
the expression of $\mathrm{IF}\left(h,Q^{\ast}\right)$, we first
note that there exists a probability measure $e(h,Q_{\textsc{ex}})$
on $\left[-1,1\right]$ such that 
\[
\phi_{n}(h,Q_{\textsc{ex}})=\int_{-1}^{1}\lambda^{n}e(h,Q_{\textsc{ex}})(\mathrm{d}\lambda),\quad\mathrm{IF}(h,Q_{\textsc{ex}})=\int_{-1}^{1}\frac{1+\lambda}{1-\lambda}e(h,Q_{\textsc{ex}})(\mathrm{d}\lambda).
\]
This follows from the spectral representation of reversible Markov
chains; see e.g. \citep{KipnisVaradhan86}. From the expression (\ref{eq:Boundingchainkernel})
of $Q^{\ast}$, we have 
\[
\left(Q^{\ast}\right)^{n}=\sum_{k=0}^{n}\left(\begin{array}{c}
n\\
k
\end{array}\right)\varrho_{\text{\textsc{U}}}^{k}\left(\kappa\right)\left\{ 1-\varrho_{\text{\textsc{U}}}\left(\kappa\right)\right\} ^{n-k}Q_{\textsc{ex}}^{k}.
\]
Therefore, if we denote by $X\sim$Bin$(n;\varrho_{\text{\textsc{U}}}\left(\kappa\right))$
the number of acceptances from $0$ to $n$, we have 
\begin{equation}
\phi_{n}(h,Q^{\ast})=\sum\nolimits _{k=0}^{n}\int\lambda^{k}\Pr(X=k)e(h,Q_{\textsc{ex}})(\mathrm{d}\lambda)=\int\{(1-\varrho_{\text{\textsc{U}}}\left(\kappa\right))+\varrho_{\text{\textsc{U}}}\left(\kappa\right)\lambda\}^{n}e(h,Q_{\textsc{ex}})(\mathrm{d}\lambda),\label{eq:spectralrepresentationautocorrPstar}
\end{equation}
Hence, it follows that 
\begin{align*}
\mathrm{IF}\left(h,Q^{\ast}\right) & =1+2{\textstyle \sum\nolimits _{n=1}^{\infty}}\phi_{n}(h,Q^{\ast})\\
 & =\int\frac{1+(1-\varrho_{\text{\textsc{U}}}\left(\kappa\right))+\varrho_{\text{\textsc{U}}}\left(\kappa\right)\lambda}{\varrho_{\text{\textsc{U}}}\left(\kappa\right)-\varrho_{\text{\textsc{U}}}\left(\kappa\right)\lambda}e(h,Q_{\textsc{ex}})(\mathrm{d}\lambda)=\frac{1}{\varrho_{\text{\textsc{U}}}\left(\kappa\right)}(\mathrm{IF}(h,Q_{\textsc{ex}})+1)-1.
\end{align*}

Assuming $Q_{\textsc{ex}}$ is geometrically ergodic, then $\phi_{n}(h,Q_{\textsc{ex}})=\int_{-1+\epsilon}^{1-\epsilon}\lambda^{n}e(h,Q_{\textsc{ex}})(\mathrm{d}\lambda),$
where $0<\epsilon<1$ is the spectral gap. From (\ref{eq:spectralrepresentationautocorrPstar}),
a simple change of variables yields 
\[
\phi_{n}(h,Q^{\ast})=\int_{-1+\epsilon}^{1-\epsilon}\left[(1-\varrho_{\text{\textsc{U}}}\left(\kappa\right))+\varrho_{\text{\textsc{U}}}\left(\kappa\right)\lambda\right]^{n}e(h,Q_{\textsc{ex}})(\mathrm{d}\lambda)=\int_{1-2\varrho_{\text{\textsc{U}}}\left(\kappa\right)+\epsilon\varrho_{\text{\textsc{U}}}\left(\kappa\right)}^{1-\epsilon\varrho_{\text{\textsc{U}}}\left(\kappa\right)}\widetilde{\lambda}^{n}\widetilde{e}(h,Q^{\ast})(\mathrm{d}\widetilde{\lambda}).
\]
Thus $Q^{\ast}$ is also geometrically ergodic. 
\end{proof}
\begin{proof}[Proof of Proposition \textit{\ref{cor:RIF}}]
Parts (i) and (ii) are immediate. To simplify notation, we write
here $\mathrm{ARCT}=\mathrm{ARCT}\left(h,Q^{\ast}\right)$, $\mathrm{\mathrm{IF}=IF}(h,Q_{\textsc{ex}}),$
$\varrho_{\text{\textsc{U}}}(\kappa)=\varrho(\kappa)$, $\varphi\left(x\right)=\varphi\left(x;0,1\right)$.
We note that 
\begin{align*}
\log\left(\mathrm{ARCT}\right) & =\frac{1}{2}\log\left\{ \mathrm{IF}+1-\varrho\left(\kappa\right)\right\} -\log\left(\kappa\right)-\log\left(\varrho\left(\kappa\right)\right)-\frac{1}{2}\log\left\{ \mathrm{IF}\right\} ,
\end{align*}
so we obtain
\[
\frac{\partial\log\left(\mathrm{ARCT}\right)}{\partial\mathrm{IF}}=\frac{1}{2}\left\{ \frac{1}{\mathrm{IF}+1-\varrho\left(\kappa\right)}-\frac{1}{\mathrm{IF}}\right\} <0,
\]
which shows that $\mathrm{ARCT}$ decreases with $\mathrm{IF}$. We
also have 
\begin{align*}
\frac{\partial\log\left(\mathrm{ARCT}\right)}{\partial\kappa} & =\frac{1}{2}\varphi\left(\frac{\kappa}{2}\right)G\left(\kappa\right)-\frac{1}{\kappa},\quad G\left(\kappa\right)=\frac{1}{2-\varrho\left(\kappa\right)}+\frac{2}{\varrho\left(\kappa\right)}.
\end{align*}
The minimizing argument $\hat{\kappa}$ satisfies 
\[
J\left(\hat{\kappa},\mathrm{IF}\right)=\left.\frac{\partial\log\left(\mathrm{ARCT}\right)}{\partial\kappa}\right|_{\hat{\kappa}}=0.
\]
By implicit differentiation 
\begin{equation}
\frac{\mathrm{d}\mathrm{\hat{\kappa}}}{\mathrm{d}\mathrm{IF}}=-\frac{\partial J\left(\hat{\kappa},\mathrm{IF}\right)}{\partial\mathrm{IF}}/\frac{\partial J\left(\hat{\kappa},\mathrm{IF}\right)}{\partial\hat{\kappa}},\label{eq:implicitdiff}
\end{equation}
where
\begin{equation}
\frac{\partial J\left(\widehat{\kappa},\mathrm{IF}\right)}{\partial\mathrm{IF}}=-\frac{1}{2}\varphi\left(\frac{\widehat{\kappa}}{2}\right)\frac{1}{\left(IF+1-\varrho\left(\widehat{\kappa}\right)\right)^{2}}<0\label{eq:numnegative}
\end{equation}
and 
\begin{align}
\frac{\partial J\left(\widehat{\kappa},\mathrm{IF}\right)}{\partial\widehat{\kappa}} & =\frac{1}{2}\varphi\left(\frac{\widehat{\kappa}}{2}\right)\left\{ \frac{\partial G\left(\widehat{\kappa},IF\right)}{\partial\widehat{\kappa}}-\frac{\widehat{\kappa}}{4}G\left(\widehat{\kappa},IF\right)\right\} +\frac{1}{\widehat{\kappa}^{2}}\nonumber \\
 & =\frac{1}{2}\left\{ \varphi\left(\frac{\widehat{\kappa}}{2}\right)\frac{\partial G\left(\widehat{\kappa},IF\right)}{\partial\widehat{\kappa}}-\frac{1}{2}\right\} +\frac{1}{\widehat{\kappa}^{2}}\nonumber \\
 & \geq\frac{1}{2}\left\{ \varphi\left(\frac{\widehat{\kappa}}{2}\right)^{2}\left[\frac{2}{\varrho\left(\widehat{\kappa}\right)^{2}}-\frac{1}{\left(2-\varrho\left(\widehat{\kappa}\right)\right)^{2}}\right]-\frac{1}{2}\right\} +\frac{1}{\widehat{\kappa}^{2}}>0\label{eq:denpositive}
\end{align}
where we have used $J\left(\hat{\kappa},\mathrm{IF}\right)=0$ to
simplify the expression of this derivative. It follows from (\ref{eq:implicitdiff}),
(\ref{eq:numnegative}) and (\ref{eq:denpositive}) that the minimizing
argument of $\mathrm{ARCT}$ increases with $\mathrm{IF}$.
\end{proof}

\subsection{Conditional weak convergence\label{Section:Conditionalweakconvergence}}

Let $C_{b}$ the space of bounded continuous functions and $BL\left(1\right)$
the space of bounded Lipschitz functions $f$ with $\left\Vert f\right\Vert _{BL}\leq1$
where 
\begin{equation}
\left\Vert f\right\Vert _{BL}=\left\Vert f\right\Vert _{\infty}+\sup_{x,y:x\neq y}\text{ }\left\vert \frac{f\left(y\right)-f\left(x\right)}{y-x}\right\vert .\label{def:BLfunction}
\end{equation}
For sake of completeness, we present a version of the conditional
CLT for triangular arrays which allows us to conclude that the expectations
of any function $f\in C_{b}$ converge in probability. We have not
been able to find this precise statement in the literature so we present
a proof mimicking the steps of the proof of \citep[Theorem 7.2]{billingsley1968}
without any claim of originality. 
\begin{lem}
\label{lem:conditionalLindebergCLT} Let $\{X_{n,i}\}_{1\leq i\leq k_{n}}$
be a triangular array of real-valued random variables on a common
probability space $\left(\Omega,\mathcal{F},P\right)$ and $\{\mathcal{F}_{n}\}_{n\geq0}$
a sequence of sub-$\sigma$-algebras of $\mathcal{F}$ such that $\{X_{n,i}\}_{1\leq i\leq k_{n}}$
are conditionally independent given $\mathcal{F}_{n}$, $\mathbb{E}\left(X_{n,i}\mid\mathcal{F}_{n}\right)=0$
and $\sigma_{n,i}^{2}:=\mathbb{E}\left(X_{n,i}^{2}\mid\mathcal{F}_{n}\right)<\infty$.
Suppose also that for some $\sigma^{2}>0$, as $n\rightarrow\infty$
\begin{equation}
s_{n}^{2}:=\sum\nolimits _{i=1}^{k_{n}}\sigma_{n,i}^{2}\overset{P}{\rightarrow}\sigma^{2},\label{eq:sumofvars}
\end{equation}
and that for all $\epsilon>0$, 
\begin{equation}
\sum\nolimits _{i=1}^{k_{n}}\mathbb{E}\left(X_{n,i}^{2}\mathbf{1}\left\{ |X_{n,i}|\geq\epsilon\right\} \Big|\mathcal{F}_{n}\right)\overset{P}{\rightarrow}0.\label{eq:lind}
\end{equation}
Then we have 
\[
\sum\nolimits _{i=1}^{k_{n}}X_{n,i}|\mathcal{F}_{n}\Rightarrow\mu,
\]
where $\mu\left(\mathrm{d}x\right)=\varphi(\mathrm{d}x;0,\sigma^{2})$
in the sense that for all $f\in C_{b}$ 
\[
\mathbb{E}\left[f\left(\sum\nolimits _{i=1}^{k_{n}}X_{n,i}\right)\Big|\mathcal{F}_{n}\right]\overset{P}{\rightarrow}\mu\left(f\right).
\]
In particular, the random measures $\mu_{n}$ defined by a regular
version of the conditional probability distributions 
\begin{equation}
\mu_{n}(A)=\mathbb{P}\left(\left.\sum\nolimits _{i=1}^{k_{n}}X_{n,i}\in A\right\vert \mathcal{F}_{n}\right)\text{ for }A\in\mathcal{B}(\mathbb{R})\label{eq:convergencerandommeasures}
\end{equation}
converge weakly to $\mu$ in probability in the sense that 
\begin{equation}
d_{\mathrm{BL}}\left(\mu_{n},\mu\right):=\underset{f\in BL\left(1\right)}{\sup}\left\vert \mu_{n}\left(f\right)-\mu\left(f\right)\right\vert \overset{P}{\rightarrow}0.\label{eq:convergencerandommeasuresboundedLipschitz}
\end{equation}
\end{lem}
\begin{rem}
A random probability measure $\mu$ on a metric space $\mathcal{X}$
equipped with its Borel $\sigma$-algebra $\mathcal{B}\left(\mathcal{X}\right)$
is usually defined as a map $\mu$ from some probability space $\left(\Omega,\mathcal{F},P\right)$
to the space $\mathcal{P}\left(\mathcal{X}\right)$ of probability
measures on $\left(\mathcal{X},\mathcal{B}\left(\mathcal{X}\right)\right)$
such that for all $\omega\in\Omega$, $\mu\left(\omega,\cdot\right)\in\mathcal{P}\left(\mathcal{X}\right)$
and for all $A\in\mathcal{B}\left(\mathcal{X}\right)$, the map $\omega\longmapsto\mu\left(\omega,A\right)$
is measurable. As explained in \citep[Remark 3.20]{crauel2003}, such
a random measure is a measurable map from $\Omega$ to $\mathcal{P}\left(\mathcal{X}\right)$
w.r.t. the Borel $\sigma$-algebra on $\mathcal{P}\left(\mathcal{X}\right)$
induced by the topology of weak convergence. Indeed, from the above
definition of random measures, it follows that for any function $g\in C_{b}\left(\mathcal{X}\right)$,
the map $\omega\longmapsto\mu\left(\omega\right)\left(g\right)$ is
measurable. Since the map $\omega\longmapsto\mu\left(\omega\right)\left(g\right)$
can be written as a composition of 
\[
\omega\longmapsto\mu\left(\omega,\cdot\right)\overset{\mathcal{I}_{g}}{\longmapsto}\mu\left(\omega\right)\left(g\right),
\]
measurability of this map implies that for any $B\in$$\mathcal{B}\left(\mathbb{R}\right)$
we have $(\mathcal{I}_{g}\circ\mu)^{-1}\left(B\right)\in\mathcal{F}$
or equivalently that $\mu^{-1}(\mathcal{I}_{g}^{-1}\left(B\right))\in\mathcal{F}$.
Since the collection of sets $\left\{ \mathcal{I}_{g}^{-1}\left(B\right);B\in\mathcal{B}\left(\mathbb{R}\right),g\in C_{b}\left(\mathcal{X}\right)\right\} $
generates $\mathcal{B}\left(\mathcal{P}\left(\mathcal{X}\right)\right)$,
the mapping $\omega\longmapsto\mu\left(\omega,\cdot\right)$ is measurable
w.r.t. $\mathcal{P}\left(\mathcal{X}\right)$. In particular if $\mathcal{X}$
is Polish, the topology of weak convergence is metrized by the bounded
Lipschitz metric which is then continuous and therefore measurable.
One can easily check that the random probability measures specified
in Lemma \ref{lem:conditionalLindebergCLT} falls within this context.
Therefore the quantity on the l.h.s. of (\ref{eq:convergencerandommeasuresboundedLipschitz})
is measurable. 
\end{rem}
\begin{proof}[Proof of Lemma \textit{\ref{lem:conditionalLindebergCLT}}]
We first prove the result for $f$ bounded and infinitely differentiable,
with bounded derivatives of all orders. Without loss of generality,
we can assume that the probability space also supports a triangular
array of independent standard normal random variables $\{\xi_{n,i}\}_{1\leq i\leq k_{n}}$,
independent of $\{X_{n,i}\}_{n,i}$ and of $\mathcal{F}_{n}$ for
all $n$. For all $n$ and $1\leq i\leq k_{n}$ define $\eta_{n,i}:=\sigma_{n,i}\xi_{n,i}$.

Then using the standard Lindeberg approach, as employed in the proof
of \citep[Theorem 7.2]{billingsley1968}, we use the following telescoping
identity 
\begin{align*}
f\left(\sum\nolimits _{i=1}^{k_{n}}X_{n,i}\right) & =f\left(\sum\nolimits _{i=1}^{k_{n}}X_{n,i}\right)-f\left(\sum\nolimits _{i=1}^{k_{n}-1}X_{n,i}+\eta_{n,k_{n}}\right)\\
 & \qquad+f\left(\sum\nolimits _{i=1}^{k_{n}-1}X_{n,i}+\eta_{n,k_{n}}\right)-f\left(\sum\nolimits _{i=1}^{k_{n}-2}X_{n,i}+\eta_{n,k_{n}-1}+\eta_{n,k_{n}}\right)\\
 & \qquad+\cdots\\
 & \qquad+f\left(X_{n,1}+\sum\nolimits _{j=2}^{k_{n}}\eta_{n,j}\right)-f\left(\sum\nolimits _{i=1}^{k_{n}}\eta_{n,i}\right)\\
 & \qquad+f\left(\sum\nolimits _{i=1}^{k_{n}}\eta_{n,i}\right).
\end{align*}

Writing $Z$ for a standard normal, independent of all other variables
and $\mathcal{F}_{n}$ and $\{X_{n,i}\}_{i}$, notice first that 
\[
\mathbb{E}\left[f\left(\sum\nolimits _{i=1}^{k_{n}}\eta_{n,i}\right)\Big|\mathcal{F}_{n}\right]=\mathbb{E}\left[f(s_{n}Z)\Big|\mathcal{F}_{n}\right]\overset{P}{\rightarrow}\mathbb{E}\left[f(\sigma Z)\right].
\]
Therefore 
\[
\mathbb{E}\left[f\left(\sum\nolimits _{i=1}^{k_{n}}X_{n,i}\right)\Big|\mathcal{F}_{n}\right]=o_{P}(1)+\mathbb{E}\left[f(\sigma Z)\right]+\sum\nolimits _{i=1}^{k_{n}}\mathbb{E}[\mathcal{E}_{n,i}|\mathcal{F}_{n}],
\]
where 
\begin{align*}
\mathbb{E}[\mathcal{E}_{n,i}|\mathcal{F}_{n}] & :=\mathbb{E}\left[f\left(\sum\nolimits _{j=1}^{i}X_{n,j}+\sum\nolimits _{j=i+1}^{k_{n}}\eta_{n,j}\right)\Big|\mathcal{F}_{n}\right]-\mathbb{E}\left[f\left(\sum\nolimits _{j=1}^{i-1}X_{n,j}+\sum\nolimits _{j=i}^{k_{n}}\eta_{n,j}\right)\Big|\mathcal{F}_{n}\right]\\
 & =\mathbb{E}\left[f\left(\sum\nolimits _{j=1}^{i-1}X_{n,j}+\sum\nolimits _{j=i+1}^{k_{n}}\eta_{n,j}+X_{n,i}\right)\Big|\mathcal{F}_{n}\right]\\
 &\qquad -\mathbb{E}\left[f\left(\sum\nolimits _{j=1}^{i-1}X_{n,j}+\sum\nolimits _{j=i+1}^{k_{n}}\eta_{n,j}+\eta_{n,i}\right)\Big|\mathcal{F}_{n}\right].
\end{align*}
Letting 
\[
g(h):=\sup_{x}|f(x+h)-f(x)-f^{\prime}(x)h-\frac{1}{2}f^{\prime\prime}\left(x\right)h^{2}|,
\]
we have by the mean value theorem, and the fact that $f$ has bounded
derivative of order two that 
\begin{align*}
f(x+h)-f(x) & =\int_{x}^{x+h}f^{\prime}(s)\mathrm{d}s=f^{\prime}(x)h+\int_{x}^{x+h}\int_{x}^{s}f^{\prime\prime}(t)\mathrm{d}t\mathrm{d}s\\
 & =f^{\prime}(x)h+\frac{1}{2}f^{\prime\prime}\left(x\right)h^{2}+\int_{x}^{x+h}\int_{x}^{s}f^{\prime\prime}(t)-f^{\prime\prime}(x)\mathrm{d}t\mathrm{d}s,
\end{align*}
and the last term can be bounded above by 
\[
\left\vert \int_{s=x}^{x+h}\int_{t=x}^{s}f^{\prime\prime}(t)-f^{\prime\prime}(x)\mathrm{d}t\mathrm{d}s\right\vert \leq\int_{s=x}^{x+h}\int_{t=x}^{s}\left\vert f^{\prime\prime}(t)-f^{\prime\prime}(x)\right\vert \mathrm{d}t\mathrm{d}s\leq{h^{2}}\Vert f^{\prime\prime}\Vert_{\infty},
\]
and by 
\[
\int_{s=x}^{x+h}\int_{t=x}^{s}\left\vert f^{\prime\prime}(t)-f^{\prime\prime}(x)\right\vert \mathrm{d}t\mathrm{d}s\leq c{h^{3}}\Vert f^{\prime\prime\prime}\Vert_{\infty}
\]
Therefore there exists $K$ such that 
\[
g(h)\leq K\min\{h^{2},|h|^{3}\}.
\]

Let us look at one of these remainder terms. Write 
\begin{align*}
\mathcal{E}_{n,i} & =f\left(\sum\nolimits _{j=1}^{i-1}X_{n,j}+\sum\nolimits _{j=i+1}^{k_{n}}\eta_{n,j}+X_{n,i}\right)-f\left(\sum\nolimits _{j=1}^{i-1}X_{n,j}+\sum\nolimits _{j=i+1}^{k_{n}}\eta_{n,j}+\eta_{n,i}\right)\\
 & =f^{\prime}\left(\sum\nolimits _{j=1}^{i-1}X_{n,j}+\sum\nolimits _{j=i+1}^{k_{n}}\eta_{n,j}\right)(X_{n,i}-\eta_{n,i})\\
 & \qquad+\frac{1}{2}f^{\prime\prime}\left(\sum\nolimits _{j=1}^{i-1}X_{n,j}+\sum\nolimits _{j=i+1}^{k_{n}}\eta_{n,j}\right)(X_{n,i}^{2}-\eta_{n,i}^{2})+R_{n,i},
\end{align*}
where 
\[
|R_{n,i}|\leq g(X_{n,i})+g(\eta_{n,i}).
\]
Taking conditional expectations we observe that 
\begin{align*}
\lefteqn{\mathbb{E}\left[f^{\prime}\left(\sum\nolimits _{j=1}^{i-1}X_{n,j}+\sum\nolimits _{j=i+1}^{k_{n}}\eta_{n,j}\right)(X_{n,i}-\eta_{n,i})\Big|\mathcal{F}_{n}\right]}\\
 & =\mathbb{E}\left[f^{\prime}\left(\sum\nolimits _{j=1}^{i-1}X_{n,j}+\sum\nolimits _{j=i+1}^{k_{n}}\eta_{n,j}\right)\Big|\mathcal{F}_{n}\right]\times\mathbb{E}\left[(X_{n,i}-\eta_{n,i})\Big|\mathcal{F}_{n}\right]=0,
\end{align*}
by independence, conditional independence and the fact that $X_{n,i}$
are conditionally centred. Similarly 
\begin{align*}
\lefteqn{\mathbb{E}\left[f^{\prime\prime}\left(\sum\nolimits _{j=1}^{i-1}X_{n,j}+\sum\nolimits _{j=i+1}^{k_{n}}\eta_{n,j}\right)(X_{n,i}^{2}-\eta_{n,i}^{2})\Big|\mathcal{F}_{n}\right]}\\
 & =\mathbb{E}\left[f^{\prime}\left(\sum\nolimits _{j=1}^{i-1}X_{n,j}+\sum\nolimits _{j=i+1}^{k_{n}}\eta_{n,j}\right)\Big|\mathcal{F}_{n}\right]\times\mathbb{E}\left[(X_{n,i}^{2}-\eta_{n,i}^{2})\Big|\mathcal{F}_{n}\right]=0,
\end{align*}
since 
\[
\mathbb{E}[\eta_{n,i}^{2}|\mathcal{F}_{n}]=\sigma_{n,i}^{2}\mathbb{E}[\xi_{n,i}^{2}|\mathcal{F}_{n}]=\sigma_{n,i}^{2}.
\]

It remains to control the expression 
\[
\sum_{i=1}^{k_{n}}\mathbb{E}[g(X_{n,i})|\mathcal{F}_{n}]+\mathbb{E}[g(\eta_{n,i})|\mathcal{F}_{n}].
\]
For the first term, letting $\epsilon>0$ we have 
\begin{align*}
\sum\nolimits _{i=1}^{k_{n}}\mathbb{E}[g(X_{n,i})|\mathcal{F}_{n}] & =\sum\nolimits _{i=1}^{k_{n}}\mathbb{E}\left[g(X_{n,i})\mathbf{1}\{|X_{n,i}|<\epsilon\}\Big|\mathcal{F}_{n}\right]+\sum\nolimits _{i=1}^{k_{n}}\mathbb{E}\left[g(X_{n,i})\mathbf{1}\{|X_{n,i}|\geq\epsilon\}\Big|\mathcal{F}_{n}\right]\\
 & \leq K\sum\nolimits _{i=1}^{k_{n}}\mathbb{E}\left[|X_{n,i}|^{3}\mathbf{1}\{|X_{n,i}|<\epsilon\}\Big|\mathcal{F}_{n}\right]+K\sum\nolimits _{i=1}^{k_{n}}\mathbb{E}\left[|X_{n,i}|^{2}\mathbf{1}\{|X_{n,i}|\geq\epsilon\}\Big|\mathcal{F}_{n}\right]\\
 & \leq K\epsilon\sum\nolimits _{i=1}^{k_{n}}\mathbb{E}\left[|X_{n,i}|^{2}\mathbf{1}\{|X_{n,i}|<\epsilon\}\Big|\mathcal{F}_{n}\right]+K\sum\nolimits _{i=1}^{k_{n}}\mathbb{E}\left[|X_{n,i}|^{2}\mathbf{1}\{|X_{n,i}|\geq\epsilon\}\Big|\mathcal{F}_{n}\right]\\
 & \leq K\epsilon\sum\nolimits _{i=1}^{k_{n}}\sigma_{n,i}^{2}+K\sum\nolimits _{i=1}^{k_{n}}\mathbb{E}\left[|X_{n,i}|^{2}\mathbf{1}\{|X_{n,i}|\geq\epsilon\}\Big|\mathcal{F}_{n}\right]\overset{P}{\rightarrow}0,
\end{align*}
because $\epsilon>0$ is arbitrary, and the second term vanishes in
probability by hypothesis.

For the second term, we obtain similarly 
\begin{align*}
\sum\nolimits _{i=1}^{k_{n}}\mathbb{E}[g(\eta_{n,i})|\mathcal{F}_{n}] & \leq K\epsilon\sum\nolimits _{i=1}^{k_{n}}\mathbb{E}\left[|\eta_{n,i}|^{2}\mathbf{1}\{|\eta_{n,i}|<\epsilon\}\Big|\mathcal{F}_{n}\right]+K\sum\nolimits _{i=1}^{k_{n}}\mathbb{E}\left[|\eta_{n,i}|^{2}\mathbf{1}\{|\eta_{n,i}|\geq\epsilon\}\Big|\mathcal{F}_{n}\right]\\
 & \leq KC\epsilon\sum\nolimits _{i=1}^{k_{n}}\sigma_{n,i}^{2}+K\sum\nolimits _{i=1}^{k_{n}}\mathbb{E}\left[\sigma_{n,i}^{2}|Z|^{2}\mathbf{1}\{\sigma_{n,i}|Z|\geq\epsilon\}\Big|\mathcal{F}_{n}\right],
\end{align*}
where the second term on the r.h.s. of this inequality satisfies 
\[
K\sum\nolimits _{i=1}^{k_{n}}\epsilon^{2}\mathbb{E}\left[\epsilon^{-2}\sigma_{n,i}^{2}|Z|^{2}\mathbf{1}\{\sigma_{n,i}|Z|\geq\epsilon\}\Big|\mathcal{F}_{n}\right]\leq\frac{K}{\epsilon}\sum\nolimits _{i=1}^{k_{n}}\mathbb{E}\left[\sigma_{n,i}^{3}|Z|^{3}\Big|\mathcal{F}_{n}\right]=\frac{K}{\epsilon}\sum\nolimits _{i=1}^{k_{n}}\sigma_{n,i}^{3}\mathbb{E}[|Z|^{3}].
\]
Since 
\begin{align*}
\sigma_{n,i}^{2} & =\mathbb{E}[X_{n,i}^{2}|\mathcal{F}_{n}]\\
 & =\mathbb{E}\left[X_{n,i}^{2}\mathbf{1}\left\{ |X_{n,i}|\leq\epsilon\right\} \Big|\mathcal{F}_{n}\right]+\mathbb{E}\left[X_{n,i}^{2}\mathbf{1}\left\{ |X_{n,i}|>\epsilon\right\} \Big|\mathcal{F}_{n}\right]\\
 & =\epsilon^{2}+\mathbb{E}\left[X_{n,i}^{2}\mathbf{1}\left\{ |X_{n,i}|>\epsilon\right\} \Big|\mathcal{F}_{n}\right],
\end{align*}

we have that 
\[
\max_{i\leq k_{n}}\sigma_{n,i}^{2}\leq\epsilon^{2}+\sum\nolimits _{i=1}^{k_{n}}\mathbb{E}\left[X_{n,i}^{2}\mathbf{1}\left\{ |X_{n,i}|>\epsilon\right\} \Big|\mathcal{F}_{n}\right].
\]
Since $\epsilon>0$ is arbitrary, $\max_{i\leq k_{n}}\sigma_{n,i}^{2}\overset{P}{\rightarrow}0$,
and therefore 
\[
\sum\nolimits _{i=1}^{k_{n}}\sigma_{n,i}^{3}\leq\max_{i\leq k_{n}}\sigma_{n,i}^{2}\sum\nolimits _{i=1}^{k_{n}}\sigma_{n,i}^{2}\overset{P}{\rightarrow}0.
\]

To complete the proof, let $f\in C_{b}$ and $Z_{n}:=\sum X_{n,i}$.
Let $K>0$ be arbitrary and notice that 
\begin{align*}
\mathbb{E}\left[\left.f(Z_{n})\right\vert \mathcal{F}_{n}\right] & =\mathbb{E}\left[\left.f_{K}(Z_{n})\right\vert \mathcal{F}_{n}\right]+E_{1},\\
|E_{1}| & =\left\vert \mathbb{E}\left[\left.f_{K}(Z_{n})\right\vert \mathcal{F}_{n}\right]-\mathbb{E}\left[\left.f(Z_{n})\right\vert \mathcal{F}_{n}\right]\right\vert \\
 & \leq\mathbb{E}\left[\left.\left\vert f_{K}(Z_{n})-f(Z_{n})\right\vert \right\vert \mathcal{F}_{n}\right]\\
 & \leq(2\Vert f\Vert_{\infty}+1)\mathbb{P}\left(\left.|Z_{n}|\geq K\right\vert \mathcal{F}_{n}\right).
\end{align*}
Since $f_{K}$ is continuous and compactly supported, for any $\epsilon>0$
we can find $g_{K,\epsilon}\in C_{b}^{\infty}$, the space of continuous
functions with continuous bounded derivatives of all orders, such
that $\sup_{x}|g_{K,\epsilon}(x)-f_{K}(x)|<\epsilon$. Therefore we
also have 
\[
\mathbb{E}\left[\left.f_{K}(Z_{n})\right\vert \mathcal{F}_{n}\right]=\mathbb{E}\left[\left.g_{K,\epsilon}(Z_{n})\right\vert \mathcal{F}_{n}\right]+E_{2},
\]
where 
\[
|E_{2}|=\left\vert \mathbb{E}\left[\left.f_{K}(Z_{n})\right\vert \mathcal{F}_{n}\right]-\mathbb{E}\left[\left.g_{K,\epsilon}(Z_{n})\right\vert \mathcal{F}_{n}\right]\right\vert <\epsilon.
\]
Since $g_{K,\epsilon}\in C_{b}^{\infty}$ we know by the first result
that 
\[
\mathbb{E}\left[\left.g_{K,\epsilon}(Z_{n})\right\vert \mathcal{F}_{n}\right]=\mathbb{E}\left[g_{K,\epsilon}(\sigma Z)\right]+E_{3}(n),
\]
where $E_{3}(n)\overset{P}{\rightarrow}0$.

Moreover, we also have that 
\begin{align*}
\mathbb{E}\left[f(\sigma Z)\right] & =\mathbb{E}\left[g_{K,\epsilon}(\sigma Z)\right]+D_{1}+D_{2},\\
\left\vert D_{1}\right\vert  & \leq\mathbb{P}\left(\left\vert \sigma Z\right\vert \geq K\right),\\
\left\vert D_{2}\right\vert  & \leq\epsilon.
\end{align*}
Thus, overall we get that, for any $K>0$ and $\epsilon>0$ 
\begin{align*}
\lefteqn{\left\vert \mathbb{E}\left[\left.f(Z_{n})\right\vert \mathcal{F}_{n}\right]-\mathbb{E}\left[f(\sigma Z)\right]\right\vert }\\
 & \leq2\epsilon+E_{3}(n)+(2\Vert f\Vert_{\infty}+1)\mathbb{P}\left(\left.|Z_{n}|\geq K\right\vert \mathcal{F}_{n}\right)+\mathbb{P}\left(\left\vert \sigma Z\right\vert \geq K\right).
\end{align*}
We know that for any $K,\epsilon>0$, $E_{3}(n)\overset{P}{\rightarrow}0$.
It is clear that as $K\rightarrow\infty$ the last term vanishes,
while we also have that 
\begin{align*}
\mathbb{P}\left(\left.|Z_{n}|\geq K\right\vert \mathcal{F}_{n}\right) & \leq\frac{\mathbb{E}\left(\left.Z_{n}^{2}\right\vert \mathcal{F}_{n}\right)}{K^{2}}\\
 & =\frac{\sum\nolimits _{i=1}^{k_{n}}\sigma_{n,i}^{2}}{K^{2}}\overset{P}{\rightarrow}\frac{\sigma^{2}}{K^{2}},
\end{align*}
as $n\rightarrow\infty$ by assumption. Letting $K\rightarrow\infty$
we obtain the result.

Result (\ref{eq:convergencerandommeasures}) follows from Corollary
2.4 in \citep{bertipratellirigo2006} while (\ref{eq:convergencerandommeasuresboundedLipschitz})
follows from the discussion after Eq. (3) in this paper since $\mu_{n}$
and $\mu$ are measures on $\mathbb{R}$. 
\end{proof}
\begin{lem}
\label{propn1}Suppose that $Z_{n}:=\sum\nolimits _{i=1}^{k_{n}}X_{n,i}$
and $\mathcal{F}_{n}$ are as in Lemma \ref{lem:conditionalLindebergCLT}.
If $T_{n}\overset{P}{\rightarrow}c$, then 
\[
Z_{n}+T_{n}|\mathcal{F}_{n}\Rightarrow\mathcal{N}(c,\sigma^{2}).
\]
\end{lem}
\begin{proof}[Proof of Lemma ~\ref{propn1}]
\textit{ }Let $f\in C_{b}$. Let $K>0$ be arbitrary, and let $f_{K}$
be continuous so that $f_{K}(x)=f(x)$ for $|x|\leq K$, $f_{K}(x)=0$
for $|x|>K+1$ and $\Vert f_{K}\Vert_{\infty}\leq\Vert f\Vert_{\infty}$.
Then $f_{K}$ is continuous, and compactly supported, so also bounded
and uniformly continuous. Then 
\begin{align*}
\mathbb{E}\left[\left.f(Z_{n}+T_{n})\right\vert \mathcal{F}_{n}\right] & =\mathbb{E}\left[\left.f_{K}(Z_{n}+T_{n})\right\vert \mathcal{F}_{n}\right]+E_{1}(n),\\
|E_{1}(n)| & \leq(2\Vert f\Vert_{\infty}+1)\mathbb{P}\left(\left.|Z_{n}+T_{n}|\geq K\right\vert \mathcal{F}_{n}\right).
\end{align*}
Then 
\begin{align*}
\lefteqn{\left\vert \mathbb{E}\left[\left.f_{K}(Z_{n}+T_{n})\right\vert \mathcal{F}_{n}\right]-\mathbb{E}\left[f_{K}(\sigma Z+c)\right]\right\vert }\\
 & \leq\left\vert \mathbb{E}\left[\left.f_{K}(Z_{n}+T_{n})\right\vert \mathcal{F}_{n}\right]-\mathbb{E}\left[\left.f_{K}(Z_{n}+c)\right\vert \mathcal{F}_{n}\right]\right\vert +\left\vert \mathbb{E}\left[\left.f_{K}(Z_{n}+c)\right\vert \mathcal{F}_{n}\right]-\mathbb{E}\left[f_{K}(\sigma Z+c)\right]\right\vert .
\end{align*}
For the first term notice that since $f_{K}$ is uniformly continuous,
for any $\epsilon>0$, we can find $\epsilon^{\prime}>0$, so that
$|x-y|<\epsilon^{\prime}$ implies that $|f_{K}(x)-f_{K}(y)|<\epsilon$.
Therefore 
\begin{align*}
\lefteqn{\left\vert \mathbb{E}\left[\left.f_{K}(Z_{n}+T_{n})\right\vert \mathcal{F}_{n}\right]-\mathbb{E}\left[\left.f_{K}(Z_{n}+c)\right\vert \mathcal{F}_{n}\right]\right\vert }\\
 & \leq2\Vert f\Vert_{\infty}\mathbb{P}\left(\left.|T_{n}-c|\geq\epsilon^{\prime}\right\vert \mathcal{F}_{n}\right)+\mathbb{E}\left[\left.|f_{K}(Z_{n}+T_{n})-f_{K}(Z_{n}+c)|\mathbf{1}\left\{ |T_{n}-c|\leq\epsilon^{\prime}\right\} \right\vert \mathcal{F}_{n}\right]\\
 & \leq2\Vert f\Vert_{\infty}\mathbb{P}\left(\left.|T_{n}-c|\geq\epsilon^{\prime}\right\vert \mathcal{F}_{n}\right)+\epsilon.
\end{align*}
We know that 
\[
\mathbb{E}\left[\mathbb{P}\left(\left.|T_{n}-c|\geq\epsilon^{\prime}\right\vert \mathcal{F}_{n}\right)\right]\rightarrow0,
\]
and thus 
\[
\mathbb{P}\left(\left.|T_{n}-c|\geq\epsilon^{\prime}\right\vert \mathcal{F}_{n}\right)\overset{P}{\rightarrow}0.
\]
This proves that the first term vanishes in probability. For the second
term notice that $z\mapsto f_{K}(\cdot+c)$ is continuous and bounded,
and therefore the second term also vanishes in probability. 
\end{proof}

\subsection{Proof of Proposition \ref{Proposition:slowdown}\label{Section:Breakdownproof}}

We want to study $\mathrm{IF}\left(\Psi,Q_{T}\right)$ where $\Psi\left(u\right)=\nabla_{\vartheta}\log W(\widehat{\theta},u)$
is only a function of the auxiliary variables. To be precise, we should
write $\Psi^{T}\left(u^{T}\right)=\nabla_{\vartheta}\log W^{T}(\widehat{\theta}_{T},u^{T})$.
However, for presentation brevity, we drop the index $T$ also in
the following proof whenever there is no possible confusion. The kernel
$Q_{T}$ has been designed as a pseudo-marginal-like algorithm targetting
$\pi\left(\mathrm{d}\theta\right)$ while $U$ are auxiliary variables.
However, we can also think of $Q_{T}$ as a pseudo-marginal algorithm
targeting 
\[
\overline{\pi}\left(\mathrm{d}u\right)=\int\overline{\pi}\left(\mathrm{d}\theta,\mathrm{d}u\right)=m\left(\mathrm{d}u\right)\int\frac{\widehat{p}(y\mid\theta,u)}{p(y\mid\theta)}\pi\left(\mathrm{d}\theta\right),
\]
while $\theta$ is an auxiliary variable. In particular, the acceptance
probability of the CPM kernel (\ref{eq:transitionCPM}) can be rewritten
as 
\[
\alpha_{Q}\left\{ \left(\theta,u\right),\left(\theta^{\prime},u^{\prime}\right)\right\} =\min\left\{ 1,r\left(u,u^{\prime}\right)\frac{\overline{\pi}(\left.\theta^{\prime}\right\vert u^{\prime})q\left(\theta^{\prime},\theta\right)}{\overline{\pi}(\left.\theta\right\vert u)q\left(\theta,\theta^{\prime}\right)}\right\} ,
\]
with 
\[
r\left(u,u^{\prime}\right)=\frac{\overline{\pi}(u^{\prime})m\left(u\right)}{\overline{\pi}(u)m\left(u^{\prime}\right)}.
\]
Let us consider the following MH\ algorithm 
\[
\overline{Q}\left\{ \left(\theta,u\right),\left(\mathrm{d}\theta^{\prime},\mathrm{d}u^{\prime}\right)\right\} =K\left(u,\mathrm{d}u^{\prime}\right)\overline{\pi}(\left.\mathrm{d}\theta^{\prime}\right\vert u^{\prime})\alpha_{\overline{Q}}\left(u,u^{\prime}\right)+\left\{ 1-\varrho_{\overline{Q}}\left(u\right)\right\} \delta_{\left(\theta,u\right)}\left(\mathrm{d}\theta^{\prime},\mathrm{d}u^{\prime}\right),
\]
where 
\[
\alpha_{\overline{Q}}\left(u,u^{\prime}\right)=\min\left\{ 1,r\left(u,u^{\prime}\right)\right\} 
\]
and $1-\varrho_{\overline{Q}}\left(u\right)$ is the corresponding
rejection probability. This kernel admits the same invariant distribution
as $Q$ and we have 
\[
\int_{\Theta}\overline{Q}\left\{ \left(\theta,u\right),\left(\mathrm{d}\theta^{\prime},\mathrm{d}u^{\prime}\right)\right\} =\overline{Q}(u,\mathrm{d}u^{\prime})
\]
where 
\[
\overline{Q}\left(u,\mathrm{d}u^{\prime}\right)=K\left(u,\mathrm{d}u^{\prime}\right)\alpha_{\overline{Q}}\left(u,u^{\prime}\right)+\left\{ 1-\varrho_{\overline{Q}}\left(u\right)\right\} \delta_{u}\left(\mathrm{d}u^{\prime}\right)
\]
is the `ideal' marginal MH\ algorithm. The following lemma is an
adaptation from \citep[Proposition 2]{andrieuvihola2015}. 
\begin{lem}
\label{Proposition:Deltadifference}Let $g:\mathcal{U}^{2}\rightarrow\mathbb{R}^{+}$
be a measurable function. Define 
\begin{align*}
\Delta_{\overline{Q}}\left(g\right) & =\iint\overline{\pi}\left(\mathrm{d}\theta,\mathrm{d}u\right)\iint K\left(u,\mathrm{d}u^{\prime}\right)\overline{\pi}(\left.\mathrm{d}\theta^{\prime}\right\vert u^{\prime})\alpha_{\overline{Q}}\left(u,u^{\prime}\right)g\left(u,u^{\prime}\right),\\
\Delta_{Q}\left(g\right) & =\iint\overline{\pi}\left(\mathrm{d}\theta,\mathrm{d}u\right)\iint K\left(u,\mathrm{d}u^{\prime}\right)q\left(\theta,\mathrm{d}\theta^{\prime}\right)\alpha_{Q}\left\{ \left(\theta,u\right),\left(\theta^{\prime},u^{\prime}\right)\right\} g\left(u,u^{\prime}\right).
\end{align*}
Then we have $\Delta_{\overline{Q}}\left(g\right)\geq\Delta_{Q}\left(g\right).$ 
\end{lem}
\begin{proof}[Proof of Lemma \textit{\ref{Proposition:Deltadifference}}]
We can write for a bounded function $g$ 
\[
\Delta_{\overline{Q}}\left(g\right)-\Delta_{Q}\left(g\right)=\iint\overline{\pi}\left(\mathrm{d}u\right)K\left(u,\mathrm{d}u^{\prime}\right)g\left(u,u^{\prime}\right)\iint\overline{\pi}\left(\left.\mathrm{d}\theta\right\vert u\right)q\left(\theta,\mathrm{d}\theta^{\prime}\right)\left[\alpha_{\overline{Q}}\left(u,u^{\prime}\right)-\alpha_{Q}\left\{ \left(\theta,u\right),\left(\theta^{\prime},u^{\prime}\right)\right\} \right].
\]
Now we have by Jensen's inequality 
\begin{align*}
\iint\overline{\pi}\left(\left.\mathrm{d}\theta\right\vert u\right)q\left(\theta,\mathrm{d}\theta^{\prime}\right)\alpha_{Q}\left\{ \left(\theta,u\right),\left(\theta^{\prime},u^{\prime}\right)\right\}  & =\iint\overline{\pi}\left(\left.\mathrm{d}\theta\right\vert u\right)q\left(\theta,\mathrm{d}\theta^{\prime}\right)\min\left\{ 1,r\left(u,u^{\prime}\right)\frac{\overline{\pi}(\left.\theta^{\prime}\right\vert u^{\prime})q\left(\theta^{\prime},\theta\right)}{\overline{\pi}(\left.\theta\right\vert u)q\left(\theta,\theta^{\prime}\right)}\right\} \\
 & \leq\min\left\{ 1,r\left(u,u^{\prime}\right)\iint\overline{\pi}\left(\left.\mathrm{d}\theta\right\vert u\right)q\left(\theta,\mathrm{d}\theta^{\prime}\right)\frac{\overline{\pi}(\left.\theta^{\prime}\right\vert u^{\prime})q\left(\theta^{\prime},\theta\right)}{\overline{\pi}(\left.\theta\right\vert u)q\left(\theta,\theta^{\prime}\right)}\right\} \\
 & =\alpha_{\overline{Q}}\left(u,u^{\prime}\right).
\end{align*}
Hence $\Delta_{\overline{Q}}\left(g\right)\geq\Delta_{Q}\left(g\right)$
for bounded $g$. Monotone convergence and a truncation argument shows
this is true for general $g$. 
\end{proof}
The following Proposition follows now directly from Lemma \ref{Proposition:Deltadifference}
and by checking that the arguments of the proof of Theorem 7 in \citep[Proposition 2]{andrieuvihola2015}
are still valid in our scenario. 
\begin{prop}
\label{Proposition:InequalityMattiDrieu}Let $h:\mathcal{U}\rightarrow\mathbb{R}$
satisfying $\overline{\pi}\left(h^{2}\right)<\infty$ then $\mathrm{IF}(h,Q)\geq$
$\mathrm{IF}(h,\overline{Q}).$ 
\end{prop}
Armed with Proposition \ref{Proposition:InequalityMattiDrieu}, we
will show that $\mathrm{IF}(\Psi,\overline{Q})\geq C\mathbb{V}_{\overline{\pi}}\left(\Psi\right)$
which implies that $\mathrm{IF}(\Psi,Q)\geq C\mathbb{V}_{\overline{\pi}}\left(\Psi\right)$
almost surely. Let $e\left(\Psi,\overline{Q}\right)\left(\mathrm{d}\lambda\right)$
denote the spectral measure of $\Psi$ w.r.t $Q$; see e.g. \citep{geyer1992},
\citep{KipnisVaradhan86}. This measure $e\left(\Psi,\overline{Q}\right)$
is supported on $\left[-1,1\right]$ as $\overline{Q}$ is reversible
and $\int_{-1}^{1}e\left(\Psi,\overline{Q}\right)\left(\mathrm{d}\lambda\right)=\mathbb{V}_{\overline{\pi}}\left(\Psi\right).$
We will show that 
\begin{equation}
\int\left(1-\lambda\right)e\left(\Psi,\overline{Q}\right)\left(\mathrm{d}\lambda\right)\leq C,\label{eq:ESJDClaim}
\end{equation}
almost surely where the l.h.s. of (\ref{eq:ESJDClaim}) is the Expected
Square Jump Distance (ESJD) of $\Psi$ . By applying Jensen's inequality
w.r.t. the probability measure $e\left(\Psi,\overline{Q}\right)\left(\mathrm{d}\lambda\right)/\mathbb{V}_{\overline{\pi}}\left(\Psi\right)$,
the above inequality will imply that 
\begin{align*}
\mathrm{IF}(\Psi,\overline{Q}) & =2\int\frac{1}{1-\lambda}\frac{e\left(\Psi,\overline{Q}\right)\left(\mathrm{d}\lambda\right)}{\mathbb{V}_{\overline{\pi}}\left(\Psi\right)}-1\\
 & \geq\frac{2}{\int\left(1-\lambda\right)\frac{e\left(\Psi,\overline{Q}\right)\left(\mathrm{d}\lambda\right)}{\mathbb{V}_{\overline{\pi}}\left(\Psi\right)}}-1=2C\mathbb{V}_{\overline{\pi}}\left(\Psi\right)-1
\end{align*}
almost surely. We now show that (\ref{eq:ESJDClaim}) holds, at least
under severe regularity conditions listed in the proof which however
make the calculations tractables. We postulate that this result holds
under much weaker assumptions. 
\begin{prop}
\label{Proposition:ESJDupperbounded}Under Assumptions \ref{ass:BVM}
and \ref{assumption:Taylor}-\ref{ass:THEEND}, the inequality (\ref{eq:ESJDClaim})
holds almost surely. 
\end{prop}
The lengthy proof of this proposition is deferred to the next section.

\subsection{Proof of Proposition \ref{Proposition:ESJDupperbounded}\label{Appendix:breakdown}}

For presentation brevity, we will only prove the result for $d=1$
and $p=1$. Using the notation of Section \ref{Section:ProofCLTnightmare}
and a similar continuous-time embedding approach, we have for $U^{T}\left(\delta_{T}\right)\sim K_{\rho_{T}}\left(U^{T}\left(0\right),\cdot\right)$
\begin{align*}
 & \nabla\log W^{T}(\widehat{\theta}_{T},U\left(\delta_{T}\right))-\nabla\log W^{T}(\widehat{\theta}_{T},U\left(0\right))\\
 & =\sum_{t=1}^{T}\nabla\log\widehat{W}_{t}^{T}(Y_{t}\mid\widehat{\theta}_{T};U\left(\delta_{T}\right))-\nabla\log\widehat{W}_{t}^{T}(Y_{t}\mid\widehat{\theta}_{T};U_{t}\left(0\right))\\
 & =\sum_{t=1}^{T}\nabla\log\left(1+\eta_{t}^{T}\right)=\sum_{t=1}^{T}\frac{\nabla\eta_{t}^{T}}{1+\eta_{t}^{T}}
\end{align*}
where
\[
\eta_{t}^{T}=\frac{\widehat{W}_{t}^{T}(Y_{t}\mid\widehat{\theta}_{T};U\left(\delta_{T}\right)}{\widehat{W}_{t}^{T}(Y_{t}\mid\widehat{\theta}_{T};U_{t}\left(0\right)}-1.
\]
To simplify notation, we have written $\nabla$ to denote the derivative
w.r.t. $\vartheta$ evaluated at $\widehat{\theta}_{T}$. 

We will make here the following assumptions. Here $B\left(\overline{\theta}\right)$
denotes a neighbourhood of $\overline{\theta}$. 
\begin{assumption}
\label{assumption:Taylor}There exists $\epsilon>0$ such that for
$T$ large enough we have $\eta_{t}^{T}>-1+\epsilon$ for all $t$
in probability. 
\end{assumption}
\begin{assumption}
\label{assumption:boundedeverywhere}The function $u\longmapsto\pi_{T}\left(u\right)/\pi_{T}(\left.u\right\vert \widehat{\theta}_{T})$
is bounded w.r.t $u$ for $T$ large enough in probability. 
\end{assumption}
\begin{assumption}
\label{ass:momentsofWhigher}We have 
\[
\limsup_{T}\sup_{\theta\in B\left(\overline{\theta}\right)}\mathbb{E}\left[\left.\left(\widehat{W}_{1}^{T}\left(\theta\right)\right)^{-14}\right\vert Y_{1}\right]<B\left(Y_{1}\right).
\]
\end{assumption}
\begin{assumption}
\label{ass:momentsofgradWhigher}We have 
\[
\limsup_{T}\sup_{\theta\in B\left(\overline{\theta}\right)}\mathbb{E}\left[\left.\left(\nabla\widehat{W}_{1}^{T}\left(\theta\right)\right)^{16}\right\vert Y_{1}\right]<B\left(Y_{1}\right).
\]
\end{assumption}
\begin{assumption}
\label{ass:dut}We have 
\begin{align*}
\limsup_{T}\sup_{\theta\in B\left(\overline{\theta}\right)} \bigg\{
\mathbb{E}\Big[\left\vert \partial_{u}\varpi\left(Y_{1},U_{1,1}\left(0\right);\theta\right)\right\vert ^{8}
& +\left\vert \partial_{u}\varpi\left(Y_{1},U_{1,1}\left(0\right);\theta\right)U_{1,1}\left(0\right)\right\vert ^{8}\Big\vert Y_{1}\Big]\\
&+\mathbb{E}\Big[\left\vert \partial_{u,u}\varpi\left(Y_{1},U_{1,1}\left(0\right);\theta\right)\right\vert ^{8}\Big\vert Y_{1}\Big]\bigg\}<B\left(Y_{1}\right).
\end{align*}
\end{assumption}
\begin{assumption}
\label{ass:dutheta}We have 
\begin{align*}
\limsup_{T}\sup_{\theta\in B\left(\overline{\theta}\right)}\mathbb{E}\bigg[\left\vert \partial_{u,\vartheta}\varpi\left(Y_{1},U_{1,1}\left(0\right);\theta\right)\right\vert ^{16}
&+\left\vert \partial_{u,\vartheta}\varpi\left(Y_{1},U_{1,1}\left(0\right);\theta\right)U_{1,1}\left(0\right)\right\vert ^{8}\\
&+\left\vert \partial_{u,u,\vartheta}\varpi\left(Y_{1},U_{1,1}\left(0\right);\theta\right)\right\vert ^{8}\Big| Y_{1}\bigg]<B\left(Y_{1}\right).
\end{align*}
\end{assumption}
\begin{assumption}
\label{ass:THEEND}We have 
\[
\mathbb{E}_{Y_{1}\sim\mu}\left[B\left(Y_{1}\right)^{4}\right]<\infty.
\]
\end{assumption}
To establish the result of the proposition, it is enough to show that
the ESJD\ is $O\left(1\right)$ almost surely. \emph{All the expectations
in this section have to be understood conditional expectations w.r.t.}
$\mathcal{Y}^{T}$. Under Assumption \ref{assumption:Taylor}, we
have 
\begin{align*}
 & \nabla\log W^{T}(\widehat{\theta},U\left(\delta_{T}\right))-\nabla\log W^{T}(\widehat{\theta},U\left(0\right))=\sum_{t=1}^{T}\nabla\eta_{t}^{T}+\nabla\eta_{t}^{T}.f(\eta_{t}^{T}),
\end{align*}
with$\left\vert f(x)\right\vert \lesssim x$. In the sequel, the generic
notation $c$ is used to denote a constant that is independent of
$T$. To alleviate notations, we do not use distinct indices each
time such a constant appears, and keep using the notation $c$ even
though the corresponding constant may vary from one statement to the
other. However, to avoid confusion, we sometimes make a distinction
between such constants by using $c,$ $c^{\prime}$, $c^{\prime\prime}$
inside an argument. We also further drop the dependence of $W^{T}$,
$\widehat{\theta}_{T}$ and $\delta_{T}$ on $T$ when no confusion
is possible.

Using Assumption \ref{assumption:boundedeverywhere}, the ESJD\ satisfies
\begin{align*}
 & \mathbb{E}_{U\left(0\right)\sim\overline{\pi}}\left[\left(\frac{\nabla W\left(\widehat{\theta},U\left(0\right)\right)}{W\left(\widehat{\theta},U\left(0\right)\right)}-\frac{\nabla W\left(\widehat{\theta},U\left(\delta\right)\right)}{W\left(\widehat{\theta},U\left(\delta\right)\right)}\right)^{2}\cdot1\wedge\left(\frac{\overline{\pi}\left(U\left(\delta\right)\right)}{m\left(U\left(\delta\right)\right)}\right)/\left(\frac{\overline{\pi}\left(U\left(0\right)\right)}{m\left(U\left(0\right)\right)}\right)\right]\\
 & =\widetilde{\mathbb{E}}\left[\frac{\overline{\pi}\left(U\left(0\right)\right)}{\overline{\pi}\left(\left.U\left(0\right)\right\vert \widehat{\theta}\right)}\left(\frac{\nabla W\left(\widehat{\theta},U\left(0\right)\right)}{W\left(\widehat{\theta},U\left(0\right)\right)}-\frac{\nabla W\left(\widehat{\theta},U\left(\delta\right)\right)}{W\left(\widehat{\theta},U\left(\delta\right)\right)}\right)^{2}\cdot1\wedge\left(\frac{\overline{\pi}\left(U\left(\delta\right)\right)}{m\left(U\left(\delta\right)\right)}\right)/\left(\frac{\overline{\pi}\left(U\left(0\right)\right)}{m\left(U\left(0\right)\right)}\right)\right]\\
 & \leq\widetilde{\mathbb{E}}\left[\frac{\overline{\pi}\left(U\left(0\right)\right)}{\overline{\pi}\left(\left.U\left(0\right)\right\vert \widehat{\theta}\right)}\left(\frac{\nabla W\left(\widehat{\theta},U\left(0\right)\right)}{W\left(\widehat{\theta},U\left(0\right)\right)}-\frac{\nabla W\left(\widehat{\theta},U\left(\delta\right)\right)}{W\left(\widehat{\theta},U\left(\delta\right)\right)}\right)^{2}\right].\\
 & \leq c\text{ }\widetilde{\mathbb{E}}\left[\left(\sum_{t=1}^{T}\nabla\eta_{t}^{T}+\nabla\eta_{t}^{T}.f(\eta_{t}^{T})\right)^{2}\right]\\
 & \leq c^{\prime}\left(\widetilde{\mathbb{E}}\left[\left(\sum_{t=1}^{T}\nabla\eta_{t}^{T}\right)^{2}\right]+\widetilde{\mathbb{E}}\left[\left(\sum_{t=1}^{T}\nabla\eta_{t}^{T}.f(\eta_{t}^{T})\right)^{2}\right]\right)
\end{align*}
where $\widetilde{\mathbb{E}}$ is to be understood in the rest of
this section as having $U\left(0\right)\sim\overline{\pi}\left(\left.\cdot\right\vert \widehat{\theta}\right)$.

\subsubsection{Decomposition of $\eta_{t}^{T}$}

We have 
\[
\eta_{t}^{T}=L_{t}^{T}+M_{t}^{T},
\]
where 
\begin{align*}
L_{t}^{T} & =\int_{0}^{\delta_{T}}\frac{1}{N\widehat{W}_{t}^{T}\left(\widehat{\theta}\right)}\sum_{i=1}^{N}\left\{ -\partial_{u}\varpi\left(Y_{t},U_{t,i}\left(s\right);\widehat{\theta}\right)U_{t,i}^{T}\left(s\right)+\partial_{u,u}^{2}\varpi\left(Y_{t},U_{t,i}\left(s\right);\widehat{\theta}\right)\right\} \mathrm{d}s,\\
M_{t}^{T} & =\int_{0}^{\delta_{T}}\frac{\sqrt{2}}{N\widehat{W}_{t}^{T}\left(\widehat{\theta}\right)}\sum_{i=1}^{N}\partial_{u}\varpi\left(Y_{t},U_{t,i}\left(s\right);\widehat{\theta}\right)\mathrm{d}B_{t,s}^{i}.
\end{align*}
Here we write 
\[
\nabla L_{t}^{T}=\nabla L_{t,1}^{T}+\nabla L_{t,2}^{T},\text{ }\nabla M_{t}^{T}=\nabla M_{t,1}^{T}+\nabla M_{t,2}^{T},
\]
where 
\begin{align*}
\nabla L_{t,1}^{T} & =\int_{0}^{\delta_{T}}\frac{1}{N\widehat{W}_{t}^{T}\left(\widehat{\theta}\right)}\sum_{i=1}^{N}\left\{ -\partial_{u,\vartheta}\varpi\left(Y_{t},U_{t,i}\left(s\right);\widehat{\theta}\right)U_{t,i}^{T}\left(s\right)+\partial_{u,u,\vartheta}\varpi\left(Y_{t},U_{t,i}\left(s\right);\widehat{\theta}\right)\right\} \mathrm{d}s,\\
\nabla L_{t,2}^{T} & =-\int_{0}^{\delta_{T}}\frac{\sum_{i=1}^{N}\left\{ -\partial_{u}\varpi\left(Y_{t},U_{t,i}\left(s\right);\widehat{\theta}\right)U_{t,i}^{T}\left(s\right)+\partial_{u,u}^{2}\varpi\left(Y_{t},U_{t,i}\left(s\right);\widehat{\theta}\right)\right\} \nabla\log\widehat{W}_{t}^{T}\left(\widehat{\theta}\right)}{N\widehat{W}_{t}^{T}\left(\widehat{\theta}\right)}\mathrm{d}s,\\
\nabla M_{t,1}^{T} & =\int_{0}^{\delta_{T}}\frac{\sqrt{2}}{N\widehat{W}_{t}^{T}\left(\widehat{\theta}\right)}\sum_{i=1}^{N}\partial_{u,\vartheta}\varpi\left(Y_{t},U_{t,i}\left(s\right);\widehat{\theta}\right)\mathrm{d}B_{t,s}^{i},\\
\nabla M_{t,2}^{T} & =-\int_{0}^{\delta_{T}}\frac{\sqrt{2}\nabla\widehat{W}_{t}^{T}\left(\widehat{\theta}\right)}{N\left(\widehat{W}_{t}^{T}\left(\widehat{\theta}\right)\right)^{2}}\left(\sum_{i=1}^{N}\partial_{u}\varpi\left(Y_{t},U_{t,i}\left(s\right);\widehat{\theta}\right)\right)\mathrm{d}B_{t,s}^{i}.
\end{align*}

\subsubsection{Control of the term $\left(\sum_{t=1}^{T}\nabla\eta_{t}^{T}\right)^{2}$}

By the C$_{p}$ inequality, we have 
\begin{equation}
\widetilde{\mathbb{E}}\left[\left(\sum_{t=1}^{T}\nabla L_{t,1}^{T}+\nabla L_{t,2}^{T}+\nabla M_{t}^{T}\right)^{2}\right]\leq c\text{ }\left(\widetilde{\mathbb{E}}\left[\left(\sum_{t=1}^{T}\nabla L_{t,1}^{T}\right)^{2}\right]+\widetilde{\mathbb{E}}\left[\left(\sum_{t=1}^{T}\nabla L_{t,2}^{T}\right)^{2}\right]+\widetilde{\mathbb{E}}\left[\left(\sum_{t=1}^{T}\nabla M_{t}^{T}\right)^{2}\right]\right).\label{eq:sumgrad}
\end{equation}
We now need to control the three terms appearing on the r.h.s. of
(\ref{eq:sumgrad}).

\emph{Term} $\nabla L_{t,1}^{T}.$ We have 
\begin{equation}
\widetilde{\mathbb{E}}\left[\left(\sum_{t=1}^{T}\nabla L_{t,1}^{T}\right)^{2}\right]=\sum_{t=1}^{T}\widetilde{\mathbb{E}}\left[\left(\nabla L_{t,1}^{T}\right)^{2}\right]+\sum_{t,s:t\neq s}^{T}\widetilde{\mathbb{E}}\left[\nabla L_{t,1}^{T}.\nabla L_{s,1}^{T}\right].\label{eq:termLttilda1}
\end{equation}
We have for $s\neq t$ 
\[
\widetilde{\mathbb{E}}\left[\nabla L_{t,1}^{T}\text{ }\nabla L_{s,1}^{T}\right]=\widetilde{\mathbb{E}}\left[\nabla L_{t,1}^{T}\text{ }\nabla L_{s,1}^{T}\right]=\widetilde{\mathbb{E}}\left[\nabla L_{t,1}^{T}\right]\text{ }\widetilde{\mathbb{E}}\left[\nabla L_{s,1}^{T}\right]=0.
\]
Now we have 
\begin{align*}
 & \widetilde{\mathbb{E}}\left[\left(\nabla L_{t,1}^{T}\right)^{2}\right]\\
 & =\widetilde{\mathbb{E}}\left[\left(\int_{0}^{\delta_{T}}\frac{1}{N\widehat{W}_{t}^{T}\left(\widehat{\theta}\right)}\sum_{i=1}^{N}\left\{ -\partial_{u,\vartheta}\varpi\left(Y_{t},U_{t,i}\left(s\right);\widehat{\theta}\right)U_{t,i}^{T}\left(s\right)+\partial_{u,u,\vartheta}\varpi\left(Y_{t},U_{t,i}\left(s\right);\widehat{\theta}\right)\right\} \mathrm{d}s\right)^{2}\right]\\
 & =\mathbb{E}\Bigg[\left({\displaystyle \prod\nolimits _{r\neq t}}\widehat{W}_{r}^{T}\left(\widehat{\theta}\right)\right)\left(\widehat{W}_{t}^{T}\left(\widehat{\theta}\right)\right)^{-1}\\
 &\qquad \times
 \left(\int_{0}^{\delta_{T}}\frac{1}{N}\sum_{i=1}^{N}\left\{ -\partial_{u,\vartheta}\varpi\left(Y_{t},U_{t,i}\left(s\right);\widehat{\theta}\right)U_{t,i}\left(s\right)
 +\partial_{u,u,\vartheta}\varpi\left(Y_{t},U_{t,i}\left(s\right);\widehat{\theta}\right)\right\} \mathrm{d}s\right)^{2}\Bigg]\\
 & \leq\sup_{\theta\in B\left(\overline{\theta}\right)}\mathbb{E}\left[\left(\int_{0}^{\delta_{T}}\frac{1}{N}\sum_{i=1}^{N}\left\{ -\partial_{u,\vartheta}\varpi\left(Y_{t},U_{t,i}\left(s\right);\theta\right)U_{t,i}^{T}\left(s\right)+\partial_{u,u,\vartheta}\varpi\left(Y_{t},U_{t,i}\left(s\right);\theta\right)\right\} \mathrm{d}s\right)^{4}\right]^{1/2}\\
 & \times\text{ }\sup_{\theta\in B\left(\overline{\theta}\right)}\mathbb{E}\left[\left(\widehat{W}_{t}^{T}\left(\theta\right)\right)^{-2}\right]^{1/2}\text{ }\\
 & \leq c\frac{\left(\delta_{T}\right)^{2}}{N}B\left(Y_{t}\right)^{1/14+1/4}\leq c^{\prime}\frac{N}{T^{2}}B\left(Y_{t}\right)^{1/14+1/4},
\end{align*}
where we have used Assumptions \ref{ass:momentsofWhigher} and \ref{ass:dutheta}
and Cauchy-Schwarz inequality. To establish the last inequality, we
have also used the fact that 
\begin{align*}
\mathbb{E}\left[\left(\int_{0}^{\delta_{T}}f\left(U_{t}\left(s\right)\right)\mathrm{d}s\right)^{4}\right] & =\delta_{T}^{4}\mathbb{E}\left[\left(\int_{0}^{\delta_{T}}f\left(U_{t}\left(s\right)\right)\frac{\mathrm{d}s}{\delta_{T}}\right)^{4}\right]\\
 & \leq\delta_{T}^{4}\mathbb{E}\left[\int_{0}^{\delta_{T}}f^{4}\left(U_{t}\left(s\right)\right)\frac{\mathrm{d}s}{\delta_{T}}\right]\\
 & =\delta_{T}^{4}\mathbb{E}\left[f^{4}\left(U_{t}\left(0\right)\right)\right]\text{ (by stationarity).}
\end{align*}
Further on, we will not emphasize that the constant appearing in our
upper bounds are a power of $B\left(Y_{t}\right)$, which is assumed
to have a finite expectation under the distribution $\mu$ of the
observations under Assumption \ref{ass:THEEND}. Overall (\ref{eq:termLttilda1})
thus contributes $O\left(N/T\right)$ almost surely by the SLLN. To
shorten presentation, we keep the details to a minimum from now on.

\emph{Term} $\nabla L_{t,2}^{T}.$We have 
\begin{equation}
\widetilde{\mathbb{E}}\left[\left(\sum_{t=1}^{T}\nabla L_{t,2}^{T}\right)^{2}\right]=\sum_{t=1}^{T}\widetilde{\mathbb{E}}\left[\left(\nabla L_{t,2}^{T}\right)^{2}\right]+\sum_{t,s:t\neq s}^{T}\widetilde{\mathbb{E}}\left[\nabla L_{t,2}^{T}.\nabla L_{s,2}^{T}\right].\label{eq:termgradLt2}
\end{equation}
We have for $s\neq t$ 
\[
\widetilde{\mathbb{E}}\left[\nabla L_{t,2}^{T}.\nabla L_{s,2}^{T}\right]=\widetilde{\mathbb{E}}\left[\nabla L_{t,2}^{T}\right]\widetilde{\mathbb{E}}\left[\nabla L_{s,2}^{T}\right]
\]
and 
\begin{align*}
 & \left\vert \widetilde{\mathbb{E}}\left[\nabla L_{t,2}^{T}\right]\right\vert =\frac{1}{N}\left\vert \mathbb{E}\left[-\nabla\log\widehat{W}_{t}^{T}\left(\widehat{\theta}\right)\int_{0}^{\delta_{T}}\sum_{i=1}^{N}\left\{ -\partial_{u}\varpi\left(Y_{t},U_{t,i}^{T}\left(s\right);\widehat{\theta}\right)U_{t,i}^{T}\left(s\right)+\partial_{u,u}^{2}\varpi\left(Y_{t},U_{t,i}^{T}\left(s\right);\widehat{\theta}\right)\right\} \mathrm{d}s\right]\right\vert \\
 & \leq\frac{1}{N}\sup_{\theta\in B\left(\overline{\theta}\right)}\mathbb{E}\left[\left(\nabla\log\widehat{W}_{t}^{T}\left(\theta\right)\right)^{2}\right]^{1/2}\\
 & \times\sup_{\theta\in B\left(\overline{\theta}\right)}\mathbb{E}\left[\left(\int_{0}^{\delta_{T}}\sum_{i=1}^{N}\left\{ -\partial_{u}\varpi\left(Y_{t},U_{t,i}^{T}\left(s\right);\theta\right)U_{t,i}^{T}\left(s\right)+\partial_{u,u}^{2}\varpi\left(Y_{t},U_{t,i}^{T}\left(s\right);\theta\right)\right\} \mathrm{d}s\right)^{2}\right]^{^{1/2}}\\
 & =\frac{1}{N}\sup_{\theta\in B\left(\overline{\theta}\right)}\mathbb{E}\left[\frac{\left(\nabla\widehat{W}_{t}^{T}\left(\theta\right)\right)^{2}}{\left(\widehat{W}_{t}^{T}\left(\theta\right)\right)^{2}}\right]^{^{1/2}}\\
 & \times\sup_{\theta\in B\left(\overline{\theta}\right)}\mathbb{E}\left[\left(\int_{0}^{\delta_{T}}\sum_{i=1}^{N}\left\{ -\partial_{u}\varpi\left(Y_{t},U_{t,i}^{T}\left(s\right);\widehat{\theta}\right)U_{t,i}^{T}\left(s\right)+\partial_{u,u}^{2}\varpi\left(Y_{t},U_{t,i}^{T}\left(s\right);\widehat{\theta}\right)\right\} \mathrm{d}s\right)^{2}\right]^{1/2}\\
 & =\frac{1}{N}\sup_{\theta\in B\left(\overline{\theta}\right)}\mathbb{E}\left[\left(\nabla\widehat{W}_{t}^{T}\left(\theta\right)\right)^{4}\right]^{1/2}\sup_{\theta\in B\left(\overline{\theta}\right)}\mathbb{E}\left[\left(\widehat{W}_{t}^{T}\left(\theta\right)\right)^{-4}\right]^{1/2}\\
 & \times\sqrt{N}\sup_{\theta\in B\left(\overline{\theta}\right)}\mathbb{E}\left[\left(\int_{0}^{\delta_{T}}\left\{ -\partial_{u}\varpi\left(Y_{t},U_{t,1}^{T}\left(s\right);\widehat{\theta}\right)U_{t,i}^{T}\left(s\right)+\partial_{u,u}^{2}\varpi\left(Y_{t},U_{t,1}^{T}\left(s\right);\widehat{\theta}\right)\right\} \mathrm{d}s\right)^{2}\right]^{^{1/2}}\\
 & \leq c\frac{1}{N\sqrt{N}}\sqrt{N}\delta_{T}\leq c^{\prime}\frac{1}{T}.
\end{align*}
We have 
\begin{align*}
 & \widetilde{\mathbb{E}}\left[\left(\nabla L_{t,2}^{T}\right)^{2}\right]\\
 & =\mathbb{E}\left[W\left(\widehat{\theta},U\left(0\right)\right).\left(-\int_{0}^{\delta_{T}}\frac{\sum_{i=1}^{N}\left\{ -\partial_{u}\varpi\left(Y_{t},U_{t,i}^{T}\left(s\right);\widehat{\theta}\right)U_{t,i}^{T}\left(s\right)+\partial_{u,u}^{2}\varpi\left(Y_{t},U_{t,i}^{T}\left(s\right);\widehat{\theta}\right)\right\} \nabla\widehat{W}_{t}^{T}}{N\left(\widehat{W}_{t}^{T}\right)^{2}}\mathrm{d}s\right)^{2}\right]\\
 & =\mathbb{E}\left[\frac{\left(\nabla\log\widehat{W}_{t}^{T}\left(\widehat{\theta}\right)\right)^{2}}{\widehat{W}_{t}^{T}\left(\widehat{\theta}\right)}.\left(-\int_{0}^{\delta_{T}}\frac{1}{N}\sum_{i=1}^{N}\left\{ -\partial_{u}\varpi\left(Y_{t},U_{t,i}^{T}\left(s\right);\widehat{\theta}\right)U_{t,i}^{T}\left(s\right)+\partial_{u,u}^{2}\varpi\left(Y_{t},U_{t,i}^{T}\left(s\right);\widehat{\theta}\right)\right\} \mathrm{d}s\right)^{2}\right]\\
 & =\sup_{\theta\in B\left(\overline{\theta}\right)}\mathbb{E}\left[\frac{\left(\nabla\log\widehat{W}_{t}^{T}\left(\theta\right)\right)^{4}}{\left(\widehat{W}_{t}^{T}\left(\theta\right)\right)^{2}}\right]^{1/2}\\
 & \times\sup_{\theta\in B\left(\overline{\theta}\right)}\mathbb{E}\left[\left(-\int_{0}^{\delta_{T}}\frac{1}{N}\sum_{i=1}^{N}\left\{ -\partial_{u}\varpi\left(Y_{t},U_{t,i}^{T}\left(s\right);\theta\right)U_{t,i}^{T}\left(s\right)+\partial_{u,u}^{2}\varpi\left(Y_{t},U_{t,i}^{T}\left(s\right);\theta\right)\right\} \mathrm{d}s\right)^{4}\right]^{^{1/2}}\\
 & \leq c\frac{N}{T^{2}}\sup_{\theta\in B\left(\overline{\theta}\right)}\mathbb{E}\left[\left(\widehat{W}_{t}^{T}\left(\theta\right)\right)^{-12}\right]^{1/4}\sup_{\theta\in B\left(\overline{\theta}\right)}\mathbb{E}\left[\left(\nabla\widehat{W}_{t}^{T}\left(\theta\right)\right)^{8}\right]^{1/4}\\
 & \leq c^{\prime}\frac{N}{T^{2}}.
\end{align*}
Thus the term (\ref{eq:termgradLt2}) is $O\left(1\right)$ almost
surely.

\emph{Term} $\nabla M_{t,1}^{T}.$ We have 
\begin{equation}
\widetilde{\mathbb{E}}\left[\left(\sum_{t=1}^{T}\nabla M_{t,1}^{T}\right)^{2}\right]=\sum_{t=1}^{T}\widetilde{\mathbb{E}}\left[\left(\nabla M_{t,1}^{T}\right)^{2}\right]+\sum_{t,s:t\neq s}^{T}\widetilde{\mathbb{E}}\left[\nabla M_{t,1}^{T}.\nabla M_{s,1}^{T}\right].\label{eq:termgradMt1}
\end{equation}
We have for $s\neq t$ 
\[
\widetilde{\mathbb{E}}\left[\nabla M_{t,1}^{T}.\nabla M_{s,1}^{T}\right]=\widetilde{\mathbb{E}}\left[\nabla M_{t,1}^{T}\right]\widetilde{\mathbb{E}}\left[\nabla M_{s,1}^{T}\right]=0.
\]
Now we have 
\begin{align*}
\widetilde{\mathbb{E}}\left[\left(\nabla M_{t,1}^{T}\right)^{2}\right] & =\mathbb{E}\left[W\left(\widehat{\theta},U\left(0\right)\right).\left(\int_{0}^{\delta_{T}}\frac{\sqrt{2}}{N\widehat{W}_{t}^{T}}\sum_{i=1}^{N}\partial_{u,\vartheta}\varpi\left(Y_{t},U_{t,i}^{T}\left(s\right);\widehat{\theta}\right)\mathrm{d}B_{t,s}^{i}\right)^{2}\right]\\
 & =\mathbb{E}\left[\left(\int_{0}^{\delta_{T}}\frac{\sqrt{2}}{N\sqrt{\widehat{W}_{t}^{T}\left(\widehat{\theta}\right)}}\sum_{i=1}^{N}\partial_{u,\vartheta}\varpi\left(Y_{t},U_{t,i}^{T}\left(s\right);\widehat{\theta}\right)\mathrm{d}B_{t,s}^{i}\right)^{2}\right]\\
 & =\frac{2}{N^{2}}\sum_{i=1}^{N}\int_{0}^{\delta_{T}}\mathbb{E}\left[\frac{1}{\widehat{W}_{t}^{T}\left(\widehat{\theta}\right)}\left\{ \partial_{u,\vartheta}\varpi\left(Y_{t},U_{t,i}^{T}\left(s\right);\widehat{\theta}\right)\right\} ^{2}\right]\mathrm{d}s\\
 & \leq c\frac{\sup_{\theta\in B\left(\overline{\theta}\right)}\mathbb{E}\left[\left(\widehat{W}_{t}^{T}\left(\theta\right)\right)^{-2}\right]^{1/2}}{N^{2}}\sum_{i=1}^{N}\int_{0}^{\delta_{T}}\sup_{\theta\in B\left(\overline{\theta}\right)}\mathbb{E}\left[\left\{ \partial_{u,\vartheta}\varpi\left(Y_{t},U_{t,i}^{T}\left(s\right);\theta\right)\right\} ^{4}\right]^{1/2}\mathrm{d}s\\
 & \leq\frac{c^{\prime}}{T}.
\end{align*}
Hence the term (\ref{eq:termgradMt1}) is overall $O\left(1\right)$
almost surely.

\emph{Term }$\nabla M_{t,2}^{T}.$ We have 
\begin{equation}
\widetilde{\mathbb{E}}\left[\left(\sum_{t=1}^{T}\nabla M_{t,2}^{T}\right)^{2}\right]=\sum_{t=1}^{T}\widetilde{\mathbb{E}}\left[\left(\nabla M_{t,2}^{T}\right)^{2}\right]+\sum_{t,s:t\neq s}^{T}\widetilde{\mathbb{E}}\left[\nabla M_{t,2}^{T}\text{ }\nabla M_{s,2}^{T}\right].\label{eq:termgradMt2}
\end{equation}
We have for $s\neq t$ 
\[
\widetilde{\mathbb{E}}\left[\nabla M_{t,2}^{T}\text{ }\nabla M_{s,2}^{T}\right]=\widetilde{\mathbb{E}}\left[\nabla M_{t,2}^{T}\right]\text{ }\widetilde{\mathbb{E}}\left[\nabla M_{s,2}^{T}\right]
\]
where 
\begin{align*}
 & \widetilde{\mathbb{E}}\left[-\int_{0}^{\delta_{T}}\frac{\sqrt{2}\nabla\widehat{W}_{t}^{T}\left(\widehat{\theta}\right)}{N\left(\widehat{W}_{t}^{T}\left(\widehat{\theta}\right)\right)^{2}}\left(\sum_{i=1}^{N}\partial_{u}\varpi\left(Y_{t},U_{t,i}^{T}\left(s\right);\widehat{\theta}\right)\right)\mathrm{d}B_{t,s}^{i}\right]\\
 & =\mathbb{E}\left[-\int_{0}^{\delta_{T}}\frac{\sqrt{2}\nabla\widehat{W}_{t}^{T}\left(\widehat{\theta}\right)}{N\widehat{W}_{t}^{T}\left(\widehat{\theta}\right)}\left(\sum_{i=1}^{N}\partial_{u}\varpi\left(Y_{t},U_{t,i}^{T}\left(s\right);\widehat{\theta}\right)\right)\mathrm{d}B_{t,s}^{i}\right]=0
\end{align*}

We have 
\begin{align*}
\widetilde{\mathbb{E}}\left[\left(\nabla M_{t,2}^{T}\right)^{2}\right] & =\widetilde{\mathbb{E}}\left[\left(\sum_{i=1}^{N}\frac{\sqrt{2}\nabla\widehat{W}_{t}^{T}\left(\widehat{\theta}\right)}{N\left(\widehat{W}_{t}^{T}\left(\widehat{\theta}\right)\right)^{2}}\int_{0}^{\delta_{T}}\partial_{u}\varpi\left(Y_{t},U_{t,i}^{T}\left(s\right);\widehat{\theta}\right)\mathrm{d}B_{t,s}^{i}\right)^{2}\right]\\
 & =\mathbb{E}\left[\left(\sum_{i=1}^{N}\frac{\sqrt{2}\nabla\widehat{W}_{t}^{T}}{N\left(\widehat{W}_{t}^{T}\right)^{3/2}}\int_{0}^{\delta_{T}}\partial_{u}\varpi\left(Y_{t},U_{t,i}^{T}\left(s\right);\widehat{\theta}\right)\mathrm{d}B_{t,s}^{i}\right)^{2}\right]\\
 & =\frac{1}{N}\mathbb{E}\left[\frac{2\left(\nabla\widehat{W}_{t}^{T}\left(\widehat{\theta}\right)\right)^{2}}{\left(\widehat{W}_{t}^{T}\left(\widehat{\theta}\right)\right)^{3}}\int_{0}^{\delta_{T}}\left\{ \partial_{u}\varpi\left(Y_{t},U_{t,i}^{T}\left(s\right);\widehat{\theta}\right)\right\} ^{2}\mathrm{d}s\right]\\
 & =\frac{1}{N}\int_{0}^{\delta_{T}}\mathbb{E}\left[\frac{2\left(\nabla\widehat{W}_{t}^{T}\left(\widehat{\theta}\right)\right)^{2}}{\left(\widehat{W}_{t}^{T}\left(\widehat{\theta}\right)\right)^{3}}\left\{ \partial_{u}\varpi\left(Y_{t},U_{t,i}^{T}\left(s\right);\widehat{\theta}\right)\right\} ^{2}\right]\mathrm{d}s\\
 & \leq\frac{\delta_{T}}{N}\sup_{\theta\in B\left(\overline{\theta}\right)}\mathbb{E}\left[\left(\frac{\sqrt{2}\left(\nabla\widehat{W}_{t}^{T}\left(\theta\right)\right)^{2}}{\left(\widehat{W}_{t}^{T}\left(\theta\right)\right)^{3}}\right)^{2}\right]^{1/2}\sup_{\theta\in B\left(\overline{\theta}\right)}\mathbb{E}\left[\left\{ \partial_{u}\varpi\left(Y_{t},U_{t,1}^{T}\left(0\right);\theta\right)\right\} ^{4}\right]^{1/2}\\
 & \leq\frac{c}{T}.
\end{align*}
Hence the term (\ref{eq:termgradMt2}) is overall $O\left(1\right)$
almost surely.

\subsubsection{Control of $\sum_{t=1}^{T}\nabla\eta_{t}^{T}.f(\eta_{t}^{T})$}

We have 
\begin{align*}
\widetilde{\mathbb{E}}\left[\left(\sum_{t=1}^{T}\nabla\eta_{t}^{T}.f(\eta_{t}^{T})\right)^{2}\right] & \leq\widetilde{\mathbb{E}}\left[\sum_{t=1}^{T}\left(\nabla\eta_{t}^{T}\right)^{2}.\sum_{t=1}^{T}f(\eta_{t}^{T})^{2}\right]\\
 & \leq c\text{ }\widetilde{\mathbb{E}}\left[\sum_{t=1}^{T}\left(\nabla\eta_{t}^{T}\right)^{2}.\sum_{t=1}^{T}\eta_{t}^{T}{}^{2}\right]\\
 & \leq c^{\prime}\text{ }\widetilde{\mathbb{E}}\left[\left(\sum_{t=1}^{T}\left(\nabla\eta_{t}^{T}\right)^{2}\right)^{2}\right]^{1/2}\widetilde{\mathbb{E}}\left[\left(\sum_{t=1}^{T}\left(\eta_{t}^{T}\right)^{2}\right)^{2}\right]^{1/2}.
\end{align*}

\emph{Control of} $\left(\sum_{t=1}^{T}\left(\nabla\eta_{t}^{T}\right)^{2}\right)^{2}$.
We have 
\begin{align*}
\widetilde{\mathbb{E}}\left[\left(\sum_{t=1}^{T}\left(\nabla\eta_{t}^{T}\right)^{2}\right)^{2}\right] & =\sum_{t=1}^{T}\widetilde{\mathbb{E}}\left[\left(\nabla\eta_{t}^{T}\right)^{4}\right]+\sum_{t,s:t\neq s}^{T}\widetilde{\mathbb{E}}\left[\left(\nabla\eta_{t}^{T}\right)^{2}\left(\nabla\eta_{s}^{T}\right)^{2}\right]\\
 & =\sum_{t=1}^{T}\widetilde{\mathbb{E}}\left[\left(\nabla\eta_{t}^{T}\right)^{4}\right]+\sum_{t,s:t\neq s}^{T}\widetilde{\mathbb{E}}\left[\left(\nabla\eta_{t}^{T}\right)^{2}\right]\widetilde{\mathbb{E}}\left[\left(\nabla\eta_{s}^{T}\right)^{2}\right].
\end{align*}

We have 
\begin{align*}
\widetilde{\mathbb{E}}\left[\left(\nabla\eta_{t}^{T}\right)^{4}\right] & =\widetilde{\mathbb{E}}\left[\left(\nabla L_{t,1}^{T}+\nabla L_{t,2}^{T}+\nabla M_{t,1}^{T}+\nabla M_{t,2}^{T}\right)^{4}\right]\\
 & \leq c\text{ }\left(\widetilde{\mathbb{E}}\left[\left(\nabla L_{t,1}^{T}\right)^{4}\right]+\widetilde{\mathbb{E}}\left[\left(\nabla L_{t,2}^{T}\right)^{4}\right]+\widetilde{\mathbb{E}}\left[\left(\nabla M_{t,1}^{T}\right)^{4}\right]+\widetilde{\mathbb{E}}\left[\left(\nabla M_{t,2}^{T}\right)^{4}\right]\right).
\end{align*}
And now 
\begin{align*}
 & \widetilde{\mathbb{E}}\left[\left(\nabla L_{t,1}^{T}\right)^{4}\right]\\
 & =\mathbb{E}\left[W\left(\widehat{\theta},U\left(0\right)\right).\left(\nabla L_{t,1}^{T}\right)^{4}\right]\\
 & =\mathbb{E}\left[\left(\widehat{W}_{t}^{T}\left(\widehat{\theta}\right)\right)^{-3}.\left(\int_{0}^{\delta_{T}}\frac{1}{N}\sum_{i=1}^{N}\left\{ -\partial_{u,\vartheta}\varpi\left(Y_{t},U_{t,i}^{T}\left(s\right);\widehat{\theta}\right)U_{t,i}^{T}\left(s\right)+\partial_{u,u,\vartheta}\varpi\left(Y_{t},U_{t,i}^{T}\left(s\right);\widehat{\theta}\right)\right\} \mathrm{d}s\right)^{4}\right]\\
 & \leq\sup_{\theta\in B\left(\overline{\theta}\right)}\mathbb{E}\left[\left(\widehat{W}_{t}^{T}\left(\theta\right)\right)^{-6}\right]^{1/2}\\
 & \times\sup_{\theta\in B\left(\overline{\theta}\right)}\mathbb{E}\left[\left(\int_{0}^{\delta_{T}}\frac{1}{N}\sum_{i=1}^{N}\left\{ -\partial_{u,\vartheta}\varpi\left(Y_{t},U_{t,i}^{T}\left(s\right);\theta\right)U_{t,i}^{T}\left(s\right)+\partial_{u,u,\vartheta}\varpi\left(Y_{t},U_{t,i}^{T}\left(s\right);\theta\right)\right\} \mathrm{d}s\right)^{8}\right]^{1/2}\\
 & \leq c\frac{N^{2}}{T^{4}}.
\end{align*}
We also have 
\begin{align*}
 & \widetilde{\mathbb{E}}\left[\left(\nabla L_{t,2}^{T}\right)^{4}\right]\\
 & =\mathbb{E}\left[W\left(\widehat{\theta},U\left(0\right)\right).\left(\nabla L_{t,2}^{T}\right)^{4}\right]\\
 & =\mathbb{E}\left[W\left(\widehat{\theta},U\left(0\right)\right).\left(\int_{0}^{\delta_{T}}\frac{\sum_{i=1}^{N}\left\{ -\partial_{u}\varpi\left(Y_{t},U_{t,i}^{T}\left(s\right);\widehat{\theta}\right)U_{t,i}^{T}\left(s\right)+\partial_{u,u}^{2}\varpi\left(Y_{t},U_{t,i}^{T}\left(s\right);\widehat{\theta}\right)\right\} \nabla\widehat{W}_{t}^{T}\left(\widehat{\theta}\right)}{N\left(\widehat{W}_{t}^{T}\left(\widehat{\theta}\right)\right)^{2}}\mathrm{d}s\right)^{4}\right]\\
 & =\mathbb{E}\left[\nabla\widehat{W}_{t}^{T}\left(\widehat{\theta}\right)\left(\widehat{W}_{t}^{T}\left(\widehat{\theta}\right)\right)^{-3}.\left(\int_{0}^{\delta_{T}}\frac{\sum_{i=1}^{N}\left\{ -\partial_{u}\varpi\left(Y_{t},U_{t,i}^{T}\left(s\right);\widehat{\theta}\right)U_{t,i}^{T}\left(s\right)+\partial_{u,u}^{2}\varpi\left(Y_{t},U_{t,i}^{T}\left(s\right);\widehat{\theta}\right)\right\} }{N}\mathrm{d}s\right)^{4}\right]\\
 & \leq\sup_{\theta\in B\left(\overline{\theta}\right)}\mathbb{E}\left[\left(\nabla\widehat{W}_{t}^{T}\left(\theta\right)\right)^{2}\left(\widehat{W}_{t}^{T}\left(\theta\right)\right)^{-6}\right]^{1/2}\\
 & \times\sup_{\theta\in B\left(\overline{\theta}\right)}\mathbb{E}\left[\left(\int_{0}^{\delta_{T}}\frac{\sum_{i=1}^{N}\left\{ -\partial_{u}\varpi\left(Y_{t},U_{t,i}^{T}\left(s\right);\theta\right)U_{t,i}^{T}\left(s\right)+\partial_{u,u}^{2}\varpi\left(Y_{t},U_{t,i}^{T}\left(s\right);\theta\right)\right\} }{N}\mathrm{d}s\right)^{8}\right]^{1/2}\\
 & \leq\sup_{\theta\in B\left(\overline{\theta}\right)}\mathbb{E}\left[\left(\nabla\widehat{W}_{t}^{T}\left(\theta\right)\right)^{4}\right]^{1/4}\sup_{\theta\in B\left(\overline{\theta}\right)}\mathbb{E}\left[\left(\widehat{W}_{t}^{T}\left(\theta\right)\right)^{-12}\right]^{1/4}\\
 & \times\sup_{\theta\in B\left(\overline{\theta}\right)}\mathbb{E}\left[\left(\int_{0}^{\delta_{T}}\frac{\sum_{i=1}^{N}\left\{ -\partial_{u}\varpi\left(Y_{t},U_{t,i}^{T}\left(s\right);\theta\right)U_{t,i}^{T}\left(s\right)+\partial_{u,u}^{2}\varpi\left(Y_{t},U_{t,i}^{T}\left(s\right);\theta\right)\right\} }{N}\mathrm{d}s\right)^{8}\right]^{1/2}\\
 & \leq c\frac{N^{2}}{T^{4}}.
\end{align*}
We have 
\begin{align*}
\widetilde{\mathbb{E}}\left[\left(\nabla M_{t,1}^{T}\right)^{4}\right] & =\mathbb{E}\left[W\left(\widehat{\theta},U\left(0\right)\right).\left(\int_{0}^{\delta_{T}}\frac{\sqrt{2}}{N\widehat{W}_{t}^{T}}\sum_{i=1}^{N}\partial_{u,\vartheta}\varpi\left(Y_{t},U_{t,i}^{T}\left(s\right);\widehat{\theta}\right)\mathrm{d}B_{t,s}^{i}\right)^{4}\right]\\
 & \leq\sup_{\theta\in B\left(\overline{\theta}\right)}\mathbb{E}\left[\left(\widehat{W}_{t}^{T}\left(\theta\right)\right)^{-6}\right]^{1/2}\sup_{\theta\in B\left(\overline{\theta}\right)}\mathbb{E}\left[\left(\frac{\sqrt{2}}{N}\sum_{i=1}^{N}\int_{0}^{\delta_{T}}\partial_{u,\vartheta}\varpi\left(Y_{t},U_{t,i}^{T}\left(s\right);\theta\right)\mathrm{d}B_{t,s}^{i}\right)^{8}\right]^{1/2}\\
 & \leq c\text{ }\sup_{\theta\in B\left(\overline{\theta}\right)}\mathbb{E}\left[\left(\widehat{W}_{t}^{T}\left(\theta\right)\right)^{-6}\right]^{1/2}\frac{1}{N^{2}}\sup_{\theta\in B\left(\overline{\theta}\right)}\mathbb{E}\left[\left(\int_{0}^{\delta_{T}}\partial_{u,\vartheta}\varpi\left(Y_{t},U_{t,1}^{T}\left(s\right);\theta\right)\mathrm{d}B_{t,s}^{i}\right)^{8}\right]^{1/2}\\
 & \leq c^{\prime}\text{ }\sup_{\theta\in B\left(\overline{\theta}\right)}\mathbb{E}\left[\left(\widehat{W}_{t}^{T}\left(\theta\right)\right)^{-6}\right]^{1/2}\frac{1}{N^{2}}\left[\int_{0}^{\delta_{T}}\sup_{\theta\in B\left(\overline{\theta}\right)}\mathbb{E}\left(\left(\partial_{u,\vartheta}\varpi\left(Y_{t},U_{t,1}^{T}\left(0\right);\theta\right)\right)^{8}\right)^{1/4}\mathrm{d}s\right]^{2}\\
 & \leq\frac{c^{\prime\prime}}{T^{2}}
\end{align*}
and 
\begin{align*}
\widetilde{\mathbb{E}}\left[\left(\nabla M_{t,2}^{T}\right)^{4}\right] & =\mathbb{E}\left[W\left(\widehat{\theta},U\left(0\right)\right).\left(\int_{0}^{\delta_{T}}\frac{\sqrt{2}\nabla\widehat{W}_{t}^{T}}{N\left(\widehat{W}_{t}^{T}\right)^{2}}\left(\sum_{i=1}^{N}\partial_{u}\varpi\left(Y_{t},U_{t,i}^{T}\left(s\right);\widehat{\theta}\right)\right)\mathrm{d}B_{t,s}^{i}\right)^{4}\right]\\
 & \leq\sup_{\theta\in B\left(\overline{\theta}\right)}\mathbb{E}\left[\left(\widehat{W}_{t}^{T}\left(\theta\right)\right)^{-14}\right]^{1/2}\sup_{\theta\in B\left(\overline{\theta}\right)}\mathbb{E}\left[\left(\frac{\sqrt{2}\nabla\widehat{W}_{t}^{T}\left(\theta\right)}{N}\sum_{i=1}^{N}\int_{0}^{\delta_{T}}\left(\partial_{u}\varpi\left(Y_{t},U_{t,i}^{T}\left(s\right);\theta\right)\right)\mathrm{d}B_{t,s}^{i}\right)^{8}\right]^{1/2}\\
 & \leq c\text{ }\sup_{\theta\in B\left(\overline{\theta}\right)}\mathbb{E}\left[\left(\widehat{W}_{t}^{T}\left(\theta\right)\right)^{-14}\right]^{1/2}\sup_{\theta\in B\left(\overline{\theta}\right)}\mathbb{E}\left[\left(\nabla\widehat{W}_{t}^{T}\left(\theta\right)\right)^{16}\right]^{1/4}\\
 & \times\sup_{\theta\in B\left(\overline{\theta}\right)}\mathbb{E}\left[\left(\frac{1}{N}\sum_{i=1}^{N}\int_{0}^{\delta_{T}}\left(\partial_{u}\varpi\left(Y_{t},U_{t,i}^{T}\left(s\right);\theta\right)\right)\mathrm{d}B_{t,s}^{i}\right)^{16}\right]^{1/4}\\
 & \leq\frac{c}{N^{2}}\sup_{\theta\in B\left(\overline{\theta}\right)}\mathbb{E}\left[\left(\int_{0}^{\delta_{T}}\partial_{u,\vartheta}\varpi\left(Y_{t},U_{t,i}^{T}\left(s\right);\theta\right)\mathrm{d}B_{t,s}^{i}\right)^{16}\right]^{1/4}\\
 & \leq c^{\prime}\frac{N^{2}}{T^{4}}.
\end{align*}

\emph{Control of} $\left(\sum_{t=1}^{T}\left(\eta_{t}^{T}\right)^{2}\right)^{2}$.
We have 
\begin{align*}
\widetilde{\mathbb{E}}\left[\left(\sum_{t=1}^{T}\left(\eta_{t}^{T}\right)^{2}\right)^{2}\right] & =\sum_{t=1}^{T}\widetilde{\mathbb{E}}\left[\left(\eta_{t}^{T}\right)^{4}\right]+\sum_{t,s:t\neq s}^{T}\widetilde{\mathbb{E}}\left[\left(\eta_{t}^{T}\right)^{2}\left(\eta_{s}^{T}\right)^{2}\right]\\
 & =\sum_{t=1}^{T}\widetilde{\mathbb{E}}\left[\left(\eta_{t}^{T}\right)^{4}\right]+\sum_{t,s:t\neq s}^{T}\widetilde{\mathbb{E}}\left[\left(\eta_{t}^{T}\right)^{2}\right]\widetilde{\mathbb{E}}\left[\left(\eta_{s}^{T}\right)^{2}\right]
\end{align*}
where 
\[
\widetilde{\mathbb{E}}\left[\left(\eta_{t}^{T}\right)^{2}\right]\widetilde{\mathbb{E}}\left[\left(\eta_{s}^{T}\right)^{2}\right]\leq\widetilde{\mathbb{E}}\left[\left(\eta_{t}^{T}\right)^{4}\right]^{1/2}\widetilde{\mathbb{E}}\left[\left(\eta_{t}^{T}\right)^{4}\right]^{1/2},
\]
with 
\[
\eta_{t}^{T}=L_{t}^{T}+M_{t}^{T}.
\]

Now we have 
\begin{align*}
\widetilde{\mathbb{E}}\left[\left(L_{t}^{T}\right)^{4}\right] & =\mathbb{E}\left[W\left(\widehat{\theta},U\left(0\right)\right)\left(L_{t}^{T}\right)^{4}\right]\\
 & =\mathbb{E}\left[W\left(\widehat{\theta},U\left(0\right)\right)\left(\int_{0}^{\delta_{T}}\frac{1}{N\widehat{W}_{t}^{T}}\sum_{i=1}^{N}\left\{ -\partial_{u}\varpi\left(Y_{t},U_{t,i}^{T}\left(s\right);\widehat{\theta}\right)U_{t,i}^{T}\left(s\right)+\partial_{u,u}^{2}\varpi\left(Y_{t},U_{t,i}^{T}\left(s\right);\widehat{\theta}\right)\right\} \mathrm{d}s\right)^{4}\right]\\
 & =\mathbb{E}\left[\left(\widehat{W}_{t}^{T}\left(\widehat{\theta}\right)\right)^{-3}\left(\frac{1}{N}\sum_{i=1}^{N}\int_{0}^{\delta_{T}}\left\{ -\partial_{u}\varpi\left(Y_{t},U_{t,i}^{T}\left(s\right);\widehat{\theta}\right)U_{t,i}^{T}\left(s\right)+\partial_{u,u}^{2}\varpi\left(Y_{t},U_{t,i}^{T}\left(s\right);\widehat{\theta}\right)\right\} \mathrm{d}s\right)^{4}\right]\\
 & \leq c\text{ }\sup_{\theta\in B\left(\overline{\theta}\right)}\mathbb{E}\left[\left(\widehat{W}_{t}^{T}\left(\theta\right)\right)^{-6}\right]^{1/2}\\
 & \times\sup_{\theta\in B\left(\overline{\theta}\right)}\mathbb{E}\left[\left(\frac{1}{N}\sum_{i=1}^{N}\int_{0}^{\delta_{T}}\left\{ -\partial_{u}\varpi\left(Y_{t},U_{t,i}^{T}\left(s\right);\theta\right)U_{t,i}^{T}\left(s\right)+\partial_{u,u}^{2}\varpi\left(Y_{t},U_{t,i}^{T}\left(s\right);\theta\right)\right\} \mathrm{d}s\right)^{8}\right]^{1/2}\\
 & \leq c^{\prime}\frac{N^{2}}{T^{4}}
\end{align*}
and 
\begin{align*}
\widetilde{\mathbb{E}}\left[\left(M_{t}^{T}\right)^{4}\right] & =\mathbb{E}\left[W\left(\widehat{\theta},U\left(0\right)\right)\left(M_{t}^{T}\right)^{4}\right]\\
 & =\mathbb{E}\left[W\left(\widehat{\theta},U\left(0\right)\right)\left(\int_{0}^{\delta_{T}}\frac{\sqrt{2}}{N\widehat{W}_{t}^{T}}\sum_{i=1}^{N}\partial_{u}\varpi\left(Y_{t},U_{t,i}^{T}\left(s\right);\widehat{\theta}\right)\mathrm{d}B_{t,s}^{i}\right)^{4}\right]\\
 & =\mathbb{E}\left[\left(\widehat{W}_{t}^{T}\left(\widehat{\theta}\right)\right)^{-3}\left(\frac{\sqrt{2}}{N}\sum_{i=1}^{N}\int_{0}^{\delta_{T}}\partial_{u}\varpi\left(Y_{t},U_{t,i}^{T}\left(s\right);\widehat{\theta}\right)\mathrm{d}B_{t,s}^{i}\right)^{4}\right]\\
 & \leq\sup_{\theta\in B\left(\overline{\theta}\right)}\mathbb{E}\left[\left(\widehat{W}_{t}^{T}\left(\theta\right)\right)^{-6}\right]^{1/2}\sup_{\theta\in B\left(\overline{\theta}\right)}\mathbb{E}\left[\left(\frac{\sqrt{2}}{N}\sum_{i=1}^{N}\int_{0}^{\delta_{T}}\partial_{u}\varpi\left(Y_{t},U_{t,i}^{T}\left(s\right);\theta\right)\mathrm{d}B_{t,s}^{i}\right)^{8}\right]^{1/2}\\
 & \leq\sup_{\theta\in B\left(\overline{\theta}\right)}\mathbb{E}\left[\left(\widehat{W}_{t}^{T}\left(\theta\right)\right)^{-6}\right]^{1/2}\frac{1}{N^{2}}\sup_{\theta\in B\left(\overline{\theta}\right)}\mathbb{E}\left[\left(\int_{0}^{\delta_{T}}\partial_{u}\varpi\left(Y_{t},U_{t,1}^{T}\left(s\right);\theta\right)\mathrm{d}B_{t,s}^{i}\right)^{8}\right]^{1/2}\\
 & \leq\sup_{\theta\in B\left(\overline{\theta}\right)}\mathbb{E}\left[\left(\widehat{W}_{t}^{T}\left(\theta\right)\right)^{-6}\right]^{1/2}\frac{1}{N^{2}}\left(\int_{0}^{\delta_{T}}\sup_{\theta\in B\left(\overline{\theta}\right)}\mathbb{E}\left[\left(\partial_{u}\varpi\left(Y_{t},U_{t,1}^{T}\left(0\right);\theta\right)^{8}\right)\right]^{1/2}\mathrm{d}s\right)^{2}\\
 & \leq c\frac{\delta_{T}^{2}}{N^{2}}\leq\frac{c^{\prime}}{T^{2}}.
\end{align*}

Combining all the terms, we have shown that the ESJD\ is $O\left(1\right)$
almost surely.

\begin{thebibliography}{99}
\setlength{\bibsep}{-15pt}


\bibitem{andrieu:doucet:holenstein2010}Andrieu, C., Doucet, A. and
Holenstein, R. (2010). Particle Markov chain Monte Carlo methods (with
discussion). \emph{Journal of the Royal Statistical Society, Series
}B \textbf{72}, 269\textendash 342.

\bibitem{andrieudoucetlee2012}Andrieu, C., Doucet, A. and Lee, A.
(2012). Discussion of ``Constructing summary statistics for approximate
Bayesian computation: semi-automatic approximate Bayesian computation\textquotedblright \ by
P. Fearnhead and D. Prangle. \emph{Journal of the Royal Statistical
Society, Series }B \textbf{72}, 451\textendash 452.

\bibitem{andrieu2009pseudo}Andrieu, C. and Roberts G.O. (2009). The
pseudo-marginal approach for efficient Monte Carlo computations. \textit{The
Annals of Statistics} \textbf{37}, 697\textendash 725.

\bibitem{andrieuvihola2015}Andrieu, C. and Vihola, M. (2015). Convergence
properties of pseudo-marginal Markov chain Monte Carlo algorithms.
\emph{The Annals of Applied Probability} \textbf{25}, 1030\textendash 1077.

\bibitem{beaumont2003estimation}Beaumont, M. (2003). Estimation of
population growth or decline in genetically monitored populations.
\textit{Genetics} \textbf{164}, 1139\textendash 1160.

\bibitem{berarddelmoraldoucet2014}Bérard, J., Del Moral, P. and\ Doucet,
A. (2014). A lognormal central limit theorem for particle approximations
of normalizing constants. \emph{Electronic Journal of Probability}
\textbf{19}, 1\textendash 28.

\bibitem{bertipratellirigo2006}Berti, P., Pratelli, L. and Rigo,
P. (2006). Almost sure weak convergence of random probability measures.
\emph{Stochastics} \textbf{78}, 91\textendash 97.

\bibitem{billingsley1968}Billingsley, P. (1968). \emph{Convergence
of Probability Measures}. Wiley:\ New York.

\bibitem{carpenter1999}Carpenter, J., P. Clifford and Fearnhead,
P. (1999). Improved particle filter for nonlinear problems. \emph{IEE
Proceedings-F} \textbf{146}, 2\textendash 7.

\bibitem{ceperley1999}Ceperley, D.M. and Dewing, M. (1999). The penalty
method for random walks with uncertain energies. \emph{Journal of
Chemical Physics} \textbf{110}, 9812\textendash 9820.

\bibitem{chopin2017}Chopin, N. and Gerber, M. (2017). Sequential
quasi-Monte Carlo: Introduction for non-experts, dimension reduction,
application to partly observed diffusion processes. arXiv preprint
arXiv:1706.05305v1.

\bibitem{chopin2013}Chopin, N., Jacob, P.E. and Papaspiliopoulos,
O. (2013). SMC$^{2}$: an efficient algorithm for sequential analysis
of state space models. \emph{Journal of the Royal Statistical Society,
Series }B \textbf{75}, 397\textendash 426.

\bibitem{crauel2003}Crauel, H. (2003). \emph{Random Probability Measures
on Polish Spaces}. CRC\ Press.

\bibitem{Dahlin2015}Dalhin, J., Lindsten, F., Kronander, J. and Schön,
T.B. (2015). Accelerating pseudo-marginal Metropolis-Hastings by correlating
auxiliary variables. arXiv preprint arXiv:1511.05483.

\bibitem{DelMoral2004}Del Moral, P. (2004). \emph{Feynman-Kac Formulae:
Genealogical and Interacting Particle Systems with Applications}.
Springer-Verlag:\ New York.

\bibitem{doucmoulinesstoffer2014}Douc, R., Moulines, E. and Stoffer,
D.S. (2014). \emph{Nonlinear Time Series}. CRC\ Press.

\bibitem{doucet2013}Doucet, A., Jacob, P.E. and Rubenthaler, S. (2013).
Derivative-free estimation of the score vector and observed information
matrix with applications to state-space models. arXiv:1304:5768.

\bibitem{doucet2015efficient}Doucet, A., Pitt, M.K., Deligiannidis,
G. and\ Kohn, R. (2015)\textsc{.} Efficient implementation of Markov
chain Monte Carlo when using an unbiased likelihood estimator.\textit{\ Biometrika}\textsl{\ }\textbf{102},
295-313.

\bibitem{durrett2010}Durrett, R. (2010).\ \emph{Probability:\ Theory
and Examples}. Cambridge University Press.

\bibitem{ethier2005}Ethier, S.N. and Kurtz, T.G. (2005). \emph{Markov
Processes: Characterization and Convergence}. Wiley:\ New York.

\bibitem{fluryshephard2011}Flury, T. and\ Shephard, N. (2011). Bayesian
inference based only on simulated likelihood: particle filter analysis
of dynamic economic models. \emph{Econometric Theory} \textbf{27},
933\textendash 956.

\bibitem{gentil2008}Gentil, I. and Rémillard, B. (2008). Using systematic
sampling selection for Monte Carlo solutions of Feynman-Kac equations.
\emph{Advances in Applied Probability} \textbf{40}, 454-472.

\bibitem{gerber2015}Gerber, M. and Chopin, N. (2015). Sequential
quasi Monte Carlo (with discussion). \emph{Journal of the Royal Statistical
Society, Series} B \textbf{77}, 509-579.

\bibitem{geyer1992}Geyer, C.J. (1992). Practical Markov chain Monte
Carlo. \emph{Statistical Science} \textbf{7}, 473-483.

\bibitem{Gordon93}Gordon, N. J., Salmond, D. and Smith, A.F.M. (1993).
Novel approach to nonlinear/non-Gaussian Bayesian state estimation.
\emph{IEE Proceedings F}, \textbf{140}, 107\textendash 113.

\bibitem{AdamAnthony2016}Guarniero, P., Johansen, A.M., and Lee,
A. (2017). The iterated auxiliary particle filter. \emph{Journal of
the American Statistical Association}, to appear.

\bibitem{Heston1993}Heston, S.L. (1993). A closed-form solution for
options with stochastic volatility with applications to bound and
currency options. \emph{The Review of Financial Studies}, \textbf{6},
327\textendash 343.

\bibitem{ionides2006}Ionides, E.L., Breto,\ C. and King, A.A. (2006).\ Inference
for nonlinear dynamical systems. \emph{Proceedings of the National
Academy of Science USA} \textbf{103}, 18438\textendash 18443.

\bibitem{Jacob2016}Jacob, P.E., Lindsten, F. and Schön, T.B. (2016).
Coupling of particle filters. arXiv preprint arXiv:1606.01156.

\bibitem{Jennrich1969}Jennrich, R.I. (1969). Asymptotic properties
of non-linear least squares estimators. \textit{Ann. Math. Statist.}
\textbf{40}:2, 633-643.

\bibitem{Johndrow2016}Johndrow, J.E., Smith, A., Pillai, N.S. and
Dunson, D.B. (2016). Inefficiency of data augmentation for large sample
imbalanced data. arXiv preprint arXiv:1605.05798.

\bibitem{KipnisVaradhan86}Kipnis, C. and Varadhan S. R. (1986). Central
limit theorem for additive functionals of reversible Markov processes
and applications to simple exclusions. \emph{Communications in Mathematical
Physics} \textbf{104}, 1-19.

\bibitem{Lecuyer2016}L'Ecuyer, P., Munger, D., Lécot, C. and Tuffin,
B. (2016). Sorting methods and convergence rates for array-RQMC:\ Some
empirical comparisons. \emph{Mathematics and Computers in Simulation},
to appear.

\bibitem{Lee2008}Lee, A. (2008). Towards smooth particle filters
for likelihood estimation with multivariate latent variables. M.Sc.
Thesis, Department of Computer Science, University of British Columbia.

\bibitem{Lee2010}Lee, A. and\ Holmes, C.C. (2010).\ Discussion
of ``Particle Markov chain Monte Carlo methods\textquotedblright \ by
C. Andrieu, A. Doucet and R. Holenstein. \textit{Journal of the Royal
Statistical Society, Series }B \textbf{72}, 327.

\bibitem{linliuSloan2000}Lin, L., Liu, K.F., Sloan, J. (2000). A
noisy Monte Carlo algorithm. \textit{Physical Review D} \textbf{61},
074505.

\bibitem{Lindsten2016}Lindsten, F. and Doucet, A. (2016). Pseudo-marginal
Hamiltonian Monte Carlo. arXiv preprint arXiv:1607.02516.

\bibitem{Lindsten2014}Lindsten, F., Jordan, M.I. and Schön, T.B.
(2014). Particle Gibbs with ancestor sampling. \emph{Journal of Machine
Learning Research} \textbf{15}, 2145\textendash 2184.

\bibitem{liu2001monte}Liu, J.S. (2001). \emph{Monte Carlo Strategies
in Scientific Computing}. Springer-Verlag, New York.

\bibitem{MalikPitt2011}Malik, S. and Pitt, M.K. (2011). Particle
filters for continuous likelihood evaluation and maximisation. \emph{Journal
of Econometrics} \textbf{165}, 190-209.

\bibitem{nicholls2012}Nicholls, G.K., Fox, C. and\ Watt, A.M. (2012).
Coupled MCMC\ with a randomized acceptance probability. arXiv preprint
arXiv:1205.6857.

\bibitem{papaspiliopoulos2007}Papaspiliopoulos, O., Roberts, G.O.
and Sköld, M. (2007). A general framework for the parametrization
of hierarchical models. \emph{Statistical Science} \textbf{22}, 59\textendash 73.

\bibitem{peskun1973optimum}Peskun, P.H. (1973). Optimum Monte\textendash Carlo
sampling using Markov chains. \textit{Biometrika} \textbf{60}, 607\textendash 612.

\bibitem{PittShephard1999}Pitt, M.K. and\ Shephard, N. (1999). Filtering
via simulation: auxiliary particle filters. \emph{Journal of the American
Statistical Association} \textbf{94}, 590-599.

\bibitem{PittSilvaGiordaniKohn(12)}Pitt, M.K., Silva, R., Giordani,
P. and Kohn, R. (2012). On some properties of Markov chain Monte Carlo
simulation methods based on the particle filter. \textit{Journal of
Econometrics} \textbf{171}, 134\textendash 151.

\bibitem{RobertsGelmanGilks1997}Roberts, G.O., Gelman, A. and Gilks,
W.R. (1997). Weak convergence and optimal scaling of random walk Metropolis
algorithms. \textit{The Annals of Applied Probability} \textbf{7},
110\textendash 120.

\bibitem{SenThieryJasra2016}Sen, D., Thiery, A.H. and Jasra, A. (2017).
On coupling particle filter trajectories. arXiv preprint arXiv:1606.01016,
to appear in \emph{Statistics and Computing}.

\bibitem{Sherlock2015efficiency}Sherlock, C., Thiery, A., Roberts,
G.O. and Rosenthal, J.S. (2015). On the efficiency of pseudo-marginal
random walk Metropolis algorithms. \textit{The Annals of Statistics
}\textbf{43}, 238\textendash 275.

\bibitem{Stein1981}Stein, C.M. (1981). Estimation of the mean of
a multivariate normal distribution. \emph{The Annals of Statistics}
\textbf{9}, 1135\textendash 1151.

\bibitem{Tierney1998}Tierney, L. (1998). A note on Metropolis\textendash Hastings
kernels for general state spaces. \emph{The Annals of Applied Probability}
\textbf{8}, 1\textendash 9.

\bibitem{Tran2016b}Tran, M-N., Kohn, R., Quiroz, M. and Villani,
M. (2017). The block pseudo-marginal sampler. arXiv preprint arXiv:1603.024845v4.

\bibitem{Tran2016}Tran, M-N., Pitt, M.K. and Kohn, R. (2016). Adaptive
Metropolis-Hastings sampling using reversible dependent mixture proposals.
\emph{Statistics and Computing} \textbf{26}, 361\textendash 381.

\bibitem{vandervaart2000}van der Vaart, A.W. (2000). \emph{Asymptotic
Statistics}. Cambridge University Press.

\bibitem{Zakai1967}Zakai, M. (1967). Some moment inequalities for
stochastic integrals and for solutions of stochastic differential
equations. \emph{Israel Journal of Mathematics} \textbf{5}, 170\textendash 176. 
\end{thebibliography}
\end{document}